%% file: YourName-Dissertation.tex
\pdfoutput=1

\documentclass[phd,12pt]{psuthesis}
\usepackage[T1]{fontenc}
\usepackage{lmodern}
\usepackage{textcomp}
\usepackage{microtype}

\usepackage{amsmath}
\usepackage{amssymb}
\usepackage{gensymb}
\usepackage{amsthm}
\usepackage{amsfonts}
\usepackage{longtable}
\usepackage{acro}
\usepackage[final]{graphicx}
\usepackage[dvipsnames]{xcolor}
\DeclareGraphicsExtensions{.pdf, .jpg}
\usepackage{subfigure}
\usepackage{enumerate}

\theoremstyle{definition}
\newtheorem{thm}{Theorem}[chapter]
\newtheorem*{thm*}{Theorem}
\newtheorem{ithm}{Theorem}[chapter]
\newtheorem{cor}{Corollary}[chapter]
\newtheorem{conj}{Conjecture}[chapter]
\newtheorem{lem}{Lemma}[chapter]
\newtheorem{prop}{Proposition}[chapter]
\newtheorem{defn}{Definition}[chapter]
\newtheorem{ex}{Example}[chapter]
\newtheorem{obs}{Observation}[chapter]
\newtheorem{question}{Question}[chapter]
\newtheorem*{question*}{Question}
\newtheorem{const}{Construction}[chapter]

\newcommand{\scC}{\mathcal C}
\newcommand{\scD}{\mathcal D}

\newcommand{\scH}{\mathcal H}
\newcommand{\scM}{\mathcal M}

\newcommand{\cC}{\mathcal C}
\newcommand{\cU}{\mathcal U}
\newcommand{\cN}{\mathcal N}
\newcommand{\cT}{\mathcal T}
\newcommand{\cM}{\mathcal M}
\newcommand{\cA}{\mathcal A}
\newcommand{\cL}{\mathcal L}
\newcommand{\cE}{\mathcal E}
\newcommand{\cD}{\mathcal D}
\newcommand{\cF}{\mathcal F}
\newcommand{\cV}{\mathcal V}
\newcommand{\cH}{\mathcal H}
\newcommand{\cW}{\mathcal W}
\newcommand{\cP}{\mathcal P}
\newcommand{\cR}{\mathcal R}
\newcommand{\cQ}{\mathcal Q}

\def\sfL{\mathsf{L}}

\newcommand{\lra}{\leftrightarrow}

\newcommand{\R}{\mathbb R}
\newcommand{\E}{\mathbb E}

\newcommand{\bS}{\mathbb S}

\newcommand{\xdot}{\dot x}

\DeclareMathOperator{\odim}{odim}
\DeclareMathOperator{\cdim}{cdim}
\DeclareMathOperator{\ndim}{ndim}
\DeclareMathOperator{\conv}{conv}
\DeclareMathOperator{\cl}{cl}
\DeclareMathOperator{\Int}{int}
\DeclareMathOperator{\code}{code}
\DeclareMathOperator{\nerve}{nerve}
\DeclareMathOperator{\link}{link}
\DeclareMathOperator{\trunk}{trunk}
\DeclareMathOperator{\sign}{sign}
\DeclareMathOperator{\FP}{FP}
\DeclareMathOperator{\ur}{urank}
\DeclareMathOperator{\monr}{mrank}
\DeclareMathOperator{\order}{order}
\DeclareMathOperator{\rank}{rank}
\DeclareMathOperator{\radr}{radrank}
\DeclareMathOperator{\vertices}{Vert}
\DeclareMathOperator{\sep}{sep}
\DeclareMathOperator{\tk}{trunk}

\newcommand{\aff}{\mathrm{aff}}

\def \shatter{\Delta_{\mathrm{sh}}}

\newcommand{\pcode}{P_{\code}}

\newcommand{\wt}{\widetilde}
\newcommand{\inv}{^{-1}}

\renewcommand{\emptyset}{\varnothing}

\usepackage{url}

\usepackage{cite}

\usepackage{titlesec}
\input{SupplementaryMaterial/UserDefinedCommands}

\title{Combinatorial geometry of neural codes, neural data analysis, and neural networks}

\author{Caitlin Lienkaemper}
\dept{Mathematics}
\degreedate{August 2022}
\copyrightyear{2022}

\bachelorsdegreeinfo{for a baccalaureate degree \\ in Engineering Science \\ with honors in Engineering Science}

\documenttype{Dissertation}
\submittedto{The Graduate School}

\collegesubmittedto{}

\numberofreaders{4}

\honorsadvisor{Honors P. Advisor}
{Associate Professor of Engineering Science and Mechanics}
\honorsadvisortwo{Honors P. Advisor, Jr.}
{Professor of Engineering Science and Mechanics}

\secondthesissupervisor{Second T. Supervisor}

\escdepthead{Mark Levi}
\escdeptheadtitle{Professor of Mathematics}

\advisor[Dissertation Advisor][Chair of Committee]
		{Carina Curto }
		{Professor of Mathematics}

\readerone
			{Vladimir Itskov}
			{Associate Professor of Mathematics}

\readertwo
			{Jason Morton}
			{Associate Professor of Mathematics}

\readerthree
			{Réka Albert}
			{Distinguished Professor of Physics and Biology}

\readerfour
			{Alexei Novikov}
			{Professor of Mathematics, 
Associate Head for Graduate Studies}
\definecolor{gray75}{gray}{0.75}
\newcommand{\hsp}{\hspace{15pt}}
\titleformat{\chapter}[display]{\fontsize{30}{30}\selectfont\bfseries\sffamily}{Chapter \thechapter\hsp\textcolor{gray75}{\raisebox{3pt}{|}}}{0pt}{}{}

\titleformat{\section}[block]{\Large\bfseries\sffamily}{\thesection}{12pt}{}{}
\titleformat{\subsection}[block]{\large\bfseries\sffamily}{\thesubsection}{12pt}{}{}

\usepackage{listings}
\begin{document}
\pagestyle{fancy}
\fancyhead[L,C,R]{}
\fancyfoot[L,R]{}
\fancyfoot[C]{\thepage}
\renewcommand{\headrulewidth}{0pt}
\renewcommand{\footrulewidth}{0pt}
%%%%%%%%%%%%%%%%%%%%%%%%
% Preliminary Material %
%%%%%%%%%%%%%%%%%%%%%%%%
% This command is needed to properly set up the frontmatter.
\frontmatter

%%%%%%%%%%%%%%%%%%%%%%%%%%%%%%%%%%%%%%%%%%%%%%%%%%%%%%%%%%%%%%
% IMPORTANT
%
% The following commands allow you to include all the
% frontmatter in your thesis. If you don't need one or more of
% these items, you can comment it out. Most of these items are
% actually required by the Grad School -- see the Thesis Guide
% for details regarding what is and what is not required for
% your particular degree.
%%%%%%%%%%%%%%%%%%%%%%%%%%%%%%%%%%%%%%%%%%%%%%%%%%%%%%%%%%%%%%
% !!! DO NOT CHANGE THE SEQUENCE OF THESE ITEMS !!!
%%%%%%%%%%%%%%%%%%%%%%%%%%%%%%%%%%%%%%%%%%%%%%%%%%%%%%%%%%%%%%

% Generates the title page based on info you have provided
% above.
\psutitlepage

% Generates the committee page -- this is bound with your
% thesis. If this is an baccalaureate honors thesis, then
% comment out this line.
\psucommitteepage

% Generates the abstract. The argument should point to the
% file containing your abstract. 
%
%\pagestyle{plain}
\chapter*{Abstract}
    \begin{singlespace}
        \input{SupplementaryMaterial/Abstract}
    \end{singlespace}
\newpage

% Generates the Table of Contents
\thesistableofcontents

% Generates the List of Figures
\begin{singlespace}
\renewcommand{\listfigurename}{\sffamily\Huge List of Figures}
\setlength{\cftparskip}{\baselineskip}
\addcontentsline{toc}{chapter}{List of Figures}
%\fancypagestyle{plain}{%
%\fancyhf{} % clear all header and footer fields
%\fancyfoot[C]{\thepage}} % except the center
\listoffigures
\end{singlespace}
\clearpage

%% Generates the List of Tables
%\begin{singlespace}
%\renewcommand{\listtablename}{\sffamily\Huge List of Tables}
%\setlength{\cftparskip}{\baselineskip}
%\addcontentsline{toc}{chapter}{List of Tables}
%\listoftables
%\end{singlespace}
%\clearpage

% Generates the List of Symbols. The argument should point to
% the file containing your List of Symbols. 
%\thesislistofsymbols{SupplementaryMaterial/ListOfSymbols}

% Generates the Acknowledgments. The argument should point to
% the file containing your Acknowledgments. 
%
\chapter{Acknowledgments}
%\addcontentsline{toc}{chapter}{Acknowledgments}
\begin{singlespace}
    \input{SupplementaryMaterial/Acknowledgments}
\end{singlespace}

% Generates the Epigraph/Dedication. The first argument should
% point to the file containing your Epigraph/Dedication and
% the second argument should be the title of this page. 
%\thesisdedication{SupplementaryMaterial/Dedication}{Dedication}

%%%%%%%%%%%%%%%%%%%%%%%%%%%%%%%%%%%%%%%%%%%%%%%%%%%%%%
% This command is needed to get the main part of the %
% document going.                                    %
%%%%%%%%%%%%%%%%%%%%%%%%%%%%%%%%%%%%%%%%%%%%%%%%%%%%%%
\thesismainmatter

%%%%%%%%%%%%%%%%%%%%%%%%%%%%%%%%%%%%%%%%%%%%%%%%%%
% This is an AMS-LaTeX command to allow breaking %
% of displayed equations across pages. Note the  %
% closing the "}" just before the bibliography.  %
%%%%%%%%%%%%%%%%%%%%%%%%%%%%%%%%%%%%%%%%%%%%%%%%%%
\allowdisplaybreaks{
%\pagestyle{fancy}
%\fancyhead{}
%

%%%%%%%%%%%%%%%%%%%%%%
% THE ACTUAL CONTENT %
%%%%%%%%%%%%%%%%%%%%%%
% Chapters
\part{Preliminaries }
\include{Introduction/Introduction}

\include{CombinatoricsBackground/CombinatoricsBackground}

\part{Convex Neural Codes}
\include{CombinatoricsBackground/ConvexCodesBackground/ConvexCodesBackground}

\include{OrderForcing/OrderForcing}

\include{MatroidsCodes/MatroidsCodes}

\part{Underlying Rank}
\include{UnderlyingRank/UnderlyingRank}

\include{UnderlyingRank/UnderlyingRankMatroids}

\part{Threshold-Linear Networks} 
\include{ThresholdLinear/ThresholdLinearIntro }
\include{ThresholdLinear/ThresholdLinearResults }

%%%%%%%%%%%%%%%%%%%%%%%%%%%%%%%%%%%%%%%%%%%%%%%%%%%%%%%%%%%%%%%
% Appendices
%
% Because of a quirk in LaTeX (see p. 48 of The LaTeX
% Companion, 2e), you cannot use \include along with
% \addtocontents if you want things to appear the proper
% sequence.
%%%%%%%%%%%%%%%%%%%%%%%%%%%%%%%%%%%%%%%%%%%%%%%%%%%%%%%%%%%%%%%
\appendix
\titleformat{\chapter}[display]{\fontsize{30}{30}\selectfont\bfseries\sffamily}{Appendix \thechapter\textcolor{gray75}{\raisebox{3pt}{|}}}{0pt}{}{}
% If you have a single appendix, then to prevent LaTeX from
% calling it ``Appendix A'', you should uncomment the following two
% lines that redefine the \thechapter and \thesection:
%\renewcommand\thechapter{}
%\renewcommand\thesection{\arabic{section}}

%%%%%%%%%%%%%%%%%%%%%%%%%%%%%%%%%%%%%%%%%%%%%%%%%%%%%%%%%%%%%%%
% ESM students need to include a Nontechnical Abstract as the %
% last appendix.                                              %
%%%%%%%%%%%%%%%%%%%%%%%%%%%%%%%%%%%%%%%%%%%%%%%%%%%%%%%%%%%%%%%
% This \include command should point to the file containing
% that abstract.
%\include{nontechnical-abstract}
%%%%%%%%%%%%%%%%%%%%%%%%%%%%%%%%%%%%%%%%%%%
} % End of the \allowdisplaybreak command %
%%%%%%%%%%%%%%%%%%%%%%%%%%%%%%%%%%%%%%%%%%%

%%%%%%%%%%%%%%%%
% BIBLIOGRAPHY %
%%%%%%%%%%%%%%%%
% You can use BibTeX or other bibliography facility for your
% bibliography. LaTeX's standard stuff is shown below. If you
% bibtex, then this section should look something like:
%	\begin{singlespace}
%	\bibliographystyle{GLG-bibstyle}
%	\addcontentsline{toc}{chapter}{Bibliography}
%	\bibliography{Biblio-Database}
%	\end{singlespace}

\begin{singlespace}

\end{singlespace}

\backmatter

% Vita
\input{SupplementaryMaterial/Vita}

\end{document}

%% file: SupplementaryMaterial/Abstract.tex
% Place abstract below.

\vspace{-0.3in}

This dissertation explores several applications of discrete geometry --- the study of convex sets, hyperplanes arrangements, and point arrangements --- in mathematical neuroscience.
We begin with \emph{convex neural codes}. 
These were introduced to model the activity of hippocampal place cells and other neurons with convex receptive fields. 
In Chapter \ref{chapter:order_forcing}, we introduce \emph{order-forcing}, a new tool for constraining convex realizations of codes. 
We use order-forcing to construct several new examples of non-convex codes with no local obstructions. 
In Chapter \ref{chapter:matroids_codes}, we relate oriented matroids to convex neural codes.
In particular, we show that a neural code has a realization with convex polytopes if and only if it is the image of a representable oriented matroid under a neural code morphism. 
We also show that determining whether a code is convex is at least as difficult as determining whether an oriented matroid is representable--this implies that the problem of determining whether a neural code is convex is NP-hard. 
Next, we turn to the problem of the \emph{underlying rank} of a matrix. 
This problem is motivated by the problem of determining the dimensionality of (neural) data which has been corrupted by an unknown monotone transformation. 
In Chapter \ref{chapter:urank}, we introduce two tools for computing underlying rank, the \emph{minimal nodes} and the \emph{Radon rank}. 
We apply these to analyze calcium imaging data from a larval zebrafish. 
Next, in Chapter \ref{chapter:urank_math}, we explore the underlying rank in more mathematical detail, establish connections to oriented matroid theory, and show that computing underlying rank is also NP-hard. 
Finally, we turn to studying the dynamics of threshold-linear networks (TLNs), a simple model of the activity of neural circuits. 
In Chapter \ref{chapter:nullclines}, we describe the nullcline arrangement of a threshold linear network, and show that a subset of its chambers are an attracting set. 
In Chapter \ref{chapter:TLNs2}, we focus on combinatorial threshold linear networks (CTLNs), which are TLNs defined from a directed graph. 
We prove that if the graph of a CTLN is a directed acyclic graph, then all trajectories of the CTLN approach a fixed point. 

%% file: SupplementaryMaterial/Acknowledgments.tex
%auto-ignore 

I would like to thank everyone who supported me, both academically and personally, throughout my education. 

First, I would like to thank my advisor Carina Curto, for giving me advice when I needed it and, the freedom to pursue my own ideas, and support at all times. I would also like to thank Vladimir Itskov, for his feedback and guidance. Finally,  I would like to thank Jason Morton and Réka Albert for their volunteering their time by serving on my committee. I'm grateful to all of the math department staff, in particular Kelli Jones, Katie Miller,  Amy Stover, and  Allyson Borger for all of their help.

I'm grateful for support from Mathematical Neuroscience Lab alumni, especially Alex Kunin and Nora Youngs. I'm also grateful to Anne Shiu and Mohamed Omar, for introducing me to neural codes way back when, and to Francis Su for making me love math. 
I'm grateful to my coauthors, Alex Kunin, Zvi Rosen, Amzi Jeffs, Nora Youngs, Hannah Rocio Santa Cruz, Juliana Londono-Alvarez,  Katie Morrison, and Carina Curto, who contributed to chapters here. 
I'm grateful to Isabella Novik and Bernd Sturmfels, for inviting me on very productive research visits during and immediately before graduate school. 
I'm grateful to the broader combinatorics community for adopting me as one of their own. 

I'm grateful to  my mother, my father, and my extended family for supporting me unconditionally throughout my life and my education. I want to thank my partner Damie Pak, who has kept me sane while I was writing this and who made living in a studio apartment during the first year of the pandemic a good time.
I want to thank my friends, at Penn State and beyond, for making graduate school an enjoyable experience. At risk of leaving many people out, I'll single out Sarah Henderson, Kenny Gill,  Ayush Khaitan, Jesus Sanchez, Michael Francis, Dominic Veconi, Sergio Zamora Barrera, Ana Chavez Caliz, Ahmad Reza Haj Saeedi Sadegh, and Matthew Bernstein. I should also thank Zhitai Li for keeping me well caffeinated. Outside of Penn State, Micah Pedrick deserves special mention for constant support, and for even having to listen to me talk about math.

 Finally, I acknowledge funding from the NSF GRFP (DGE-1255832). The findings and conclusions do not necessarily reflect the view of the NSF.

%% file: Introduction/Introduction.tex
%auto-ignore 

% !TEX root = ../YourName-Dissertation.tex

\chapter{Introduction } \label{chapter1:introduction}

%You have more than two brain cells \cite{stringer2019high}. Therefore, we need new geometric tools to analyze the activity of our brain cells over time.  On the other hand, you only have finitely many brain cells. So combinatorics could be relevant. 

Neural activity both encodes information about the outside world and arises from the activity of interacting neurons. Here, we will explore how both of these things happen. Which features of the world can be recovered from neural activity alone? What is the relationship between the structure of neural connectivity and the dynamics of neural activity? Answering these questions presents mathematical challenges, which we consider here. Quantitative data is often hidden not just by noise, but by nonlinear distortion or missing information that makes traditional data analysis and modeling techniques less useful. When quantitative details are lost, combinatorial information remains. With tools from discrete geometry, we can use this combinatorial information to constrain the underlying geometry. Using nonlinear dynamics, we show how the graph structure of a neural circuit constrains its dynamics. 

The first problem we consider, in Chapters \ref{sec:codes},  \ref{chapter:order_forcing},  and \ref{chapter:matroids_codes}, concerns \emph{convex neural codes}. One common way neurons encode information about the world is via receptive field coding:  each neuron has a set of stimuli, known as its receptive field, and fires when the animal receives a stimulus within this receptive field. %For instance, head direction cells are neurons found in multiple areas of the brain which help animals keep track of which direction they are facing \cite{ranck1984head, taube1990head1,taube1990head2}. %no shit sherlock
%The receptive field of a head direction cell is the contiguous range of angles which it fires in response to.
 In 2014, John O'Keefe won the Nobel prize for the 1976 discovery of \textit{place cells}, neurons whose receptive fields, termed \emph{place fields}, are regions of the environment \cite{o1976place}. When an animal is located within the place field of a given neuron, that place cell fires at an increased rate. Though receptive field coding occurs in many other contexts, we will focus on place cells as our primary example of receptive field coding. 
 
 \begin{figure}[ht!]
\includegraphics[width = 6 in]{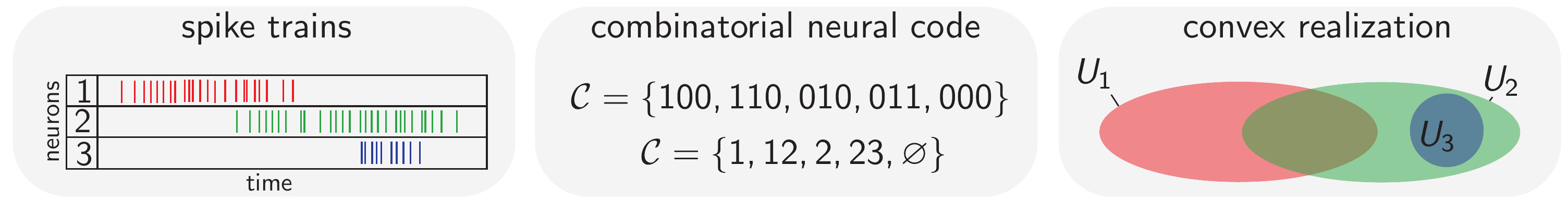}
\caption[Overview of convex neural codes.]{Convex neural codes arise from the activity of neurons with convex receptive fields. Spike trains (left) can be proccessed into a combinatorial neural code (center), which records which sets of neurons fire at the same time. This code is \emph{convex} if it is the intersection pattern of a family of convex open sets (right). Convex sets correspond to receptive fields, regions of the stimulus space to which the neurons respond.  \label{fig:cvx}}
\end{figure}

Notice that it is in principle possible to recover from neural activity alone which place fields overlap: a set of place fields has a nonempty intersection if the corresponding place cells are active at the same time. We formalize this using a simplified model of place cells in which each neuron's receptive field is a subset of Euclidean space $U_i \subseteq \R^d$. Each neuron is active if and only if the animal is located which this neuron's place field. We refer to the set of active neurons corresponding to each point in space as a codeword $\sigma \subseteq [n] := \{1, \ldots, n\}$. We can describe the set of codewords arising from a family of place fields $U_1, \ldots, U_n$  by 
\[\code(U_1, \ldots, U_n) := \{\sigma \mid \bigcap_{i\in \sigma}U_i \setminus \bigcup_{j\notin\sigma} U_j \neq \emptyset\},\]
as illustrated in Figure \ref{fig:cvx}.

Within small environments, place fields are roughly convex sets. 
A growing body of work \cite{curto2013neural, curto2017can, curto2017makes, franke2017every, cruz2019open, chen2019neural, lienkaemper2017obstructions, jeffs2019sunflowers, jeffs2020morphisms, kunin2020oriented, itskov2020hyperplane, gambacini2021non, 
chan2020nondegenerate, jeffs2021open, jeffs2022embedding, jeffs2021convex, goldrup2020classification, johnston2020neural
} explores what consequences this constraint on receptive field geometry has for combinatorial neural codes: which combinatorial neural codes arise from neurons with convex receptive fields? 
Given a combinatorial neural code $\cC\subseteq 2^{[n]}$, when do there exist convex open sets $U_1, \ldots, U_n$ such that $\cC = \code(U_1, \ldots, U_n)$? 
This question turns out to be mathematically rich: it has connections to classic problems in discrete geometry which ask when a simplicial complex is representable as the nerve of a family of convex sets in $\R^d$  \cite{kalai1984characterization, kalai1986characterization,matouvsek2009dimension, tancer2013intersection} or when an oriented matroid is representable as a hyperplane arrangement \cite{sturmfels1987decidability, mnev1988universality, shor1991stretchability}. 
In Chapter \ref{chapter:order_forcing}, we introduce \emph{order-forcing} as a technique for proving that codes are not convex. Our main result in this section is Theorem \ref{thm:order-forcing}, which gives conditions for when a list of codewords must correspond to a straight line in every convex realization of a code. We use this result to construct several new examples of non-convex codes. Material in this paper is taken from \cite{jeffs2020order}. 
In Chapter \ref{chapter:matroids_codes}, we make the connection between convex codes and oriented matroids explicit. Our main results in this section are Theorem \ref{thm:polytope_matroid}, which relates convex codes to oriented matroids via the code morphisms of \cite{jeffs2020morphisms}, and Theorem \ref{thm:hard}, which states that recognizing convex codes is computationally intractable. Material in this chapter is taken from \cite{kunin2020oriented}.

Next in Chapters \ref{chapter:urank} and \ref{chapter:urank_math}, we investigate a different aspect of the neural code, its dimensionality.  
The dimensionality of neural activity varies across experimental conditions, but is often observed to be low relative to the number of neurons recorded  \cite{gao2015simplicity}. 
Low-dimensional activity can reflect the low-dimensional input or low-dimensional intrinsic dynamics \cite{cowley2016stimulus, pang2016dimensionality, chaudhuri2019intrinsic, zhou2018hyperbolic}. 
The optimal dimensionality of neural activity is subject to computational trade-offs: higher-dimensional neural representations allow more complex readouts by downstream networks, while lower-dimensional representations have better generalization properties \cite{fusi2016neurons, farrell2019dynamic}.

However, common measurement techniques such as calcium imaging can distort firing rates in a nonlinear way, which can cause problems for typical methods of estimating dimensionality \cite{akerboom2012optimization, nauhaus2012nonlinearity}.  However, we do expect this distortion to roughly monotone, and thus to preserve the ordering between measurements. Can we use this information to estimate dimensionality?

 \begin{figure}[ht!]
\includegraphics[width = 6 in]{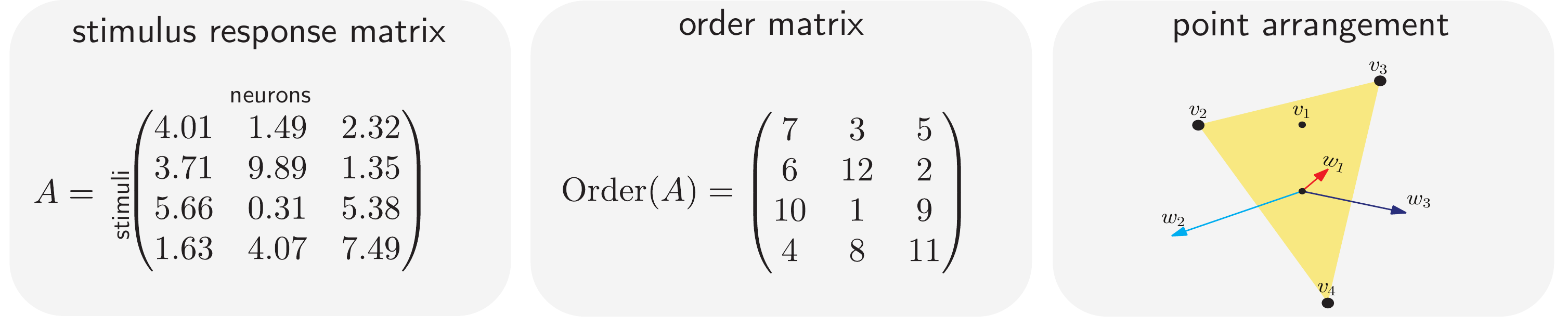}
\caption
[Overview of underlying rank.]
{The underlying rank of a matrix is the lowest rank consistent with the ordering of entries.  The only reliable information contained in neural data (left) after an unknown monotone transformation is the ordering of entries (center). We introduce techniques for estimating underlying rank by using the order matrix to infer information about an underlying point arrangement (right). \label{fig:urank}}
\end{figure}

Motivated by this problem, we introduce the  \emph{underlying rank} of a matrix  $A$ as the minimal value of $d$ such that there is a rank $d$ matrix whose entries are in the same order as those of $A$.  The general idea of using the order of entries in a matrix to determine information about geometric structure is introduced in \cite{giusti2015clique}, and is approached topologically in that paper and in \cite{curto2021betti}. See Figure \ref{fig:urank} for an overview of the underlying rank?
 We show that matrices with underlying rank $d$ correspond to point arrangements in $\R^d$, and that it is possible to recover information about this point arrangement using the ordering of entries in $A$. Much like the convex neural codes problem, underlying rank has natural connections to the theory of allowable sequences \cite{goodman1991complexity} and oriented matroids \cite{bjorner1999oriented}. In Chapter \ref{chapter:urank}, we introduce the underlying rank and some tools for estimating it. Our main contributions in this chapter are as follows: We define \emph{minimal nodes}, and prove Propositions \ref{prop:expected_cube}, \ref{prop:expected_gauss}, and \ref{prop:expected_ball}, which relate the expected number of minimal nodes to the rank of a random matrix. We also define the \emph{Radon rank} of a matrix and prove that it is a lower bound for underlying rank in Proposition \ref{prop:radon_bound}. In Chapter \ref{chapter:urank_math}, we explore underlying rank in greater mathematical detail. The main results of this chapter are Examples \ref{ex:strict} and \ref{ex:allowable}, matrices whose underlying rank exceeds their Radon rank. Example \ref{ex:strict} arises from the relationship between underlying rank and oriented matroid theory which we describe in Theorem \ref{thm:potential}. Example \ref{ex:allowable} arises from the relationship between underlying rank and allowable sequences, which we describe in Observation  \ref{obs:allowable}. We also exploit this relationship between allowable sequences and underlying rank to prove that computing underlying rank is computationally intractable in Corollary \ref{cor:urank_hard}.  
 
Finally, in Chapters \ref{chapter:TLNs1}, \ref{chapter:nullclines}, and \ref{chapter:TLNs2}, we turn to the relationship between the connectivity of a neural circuit and its dynamics. See Figure \ref{fig:ctln} for an overview of this relationship.  Different computational tasks require different patterns of network activity: for instance, central pattern generators are networks which generate the periodic activity needed for walking, breathing, and other rhythmic activities. Viewed as dynamical systems, these networks need to have limit cycles. On the other hand, networks whose activity always converges to a stable fixed point, such as the Hopfield model, are used to model pattern completion in memory. How do different patterns of network connectivity contribute to these different types of activity? 

In general, this question is difficult because neural circuits have nonlinear dynamics. We thus consider a simple, nonlinear model, threshold-linear networks (TLNs). In a TLN, the firing rate $x_i$ of neuron $i$ is determined by
\begin{align} \label{eqn:tln}
\frac{dx_i}{dt} &= -x_i + \left[\sum_{j = 1}^n W_{ij}x_j + b_i\right]_+,
\end{align}
where $\left[y\right]_+ = \max\{y, 0\}$. 
To isolate the role of connectivity, we consider a further restriction to combinatorial threshold-linear networks (CTLNs), whose dynamics are fully determined by a directed graph. A fair amount is known about how the stable and unstable fixed points of combinatorial threshold-linear networks are constrained by the graph, but less is known about their dynamic attractors more generally. In particular if a TLN is symmetric, then all trajectories approach stable fixed points by \cite{hahnloser2000permitted}. At the opposite extreme, if a graph has no bidirectional edges and no sinks, its CTLN  has no stable fixed points, and thus must have a dynamic attractor.

\begin{figure}[ht!]
\includegraphics[width = 6 in]{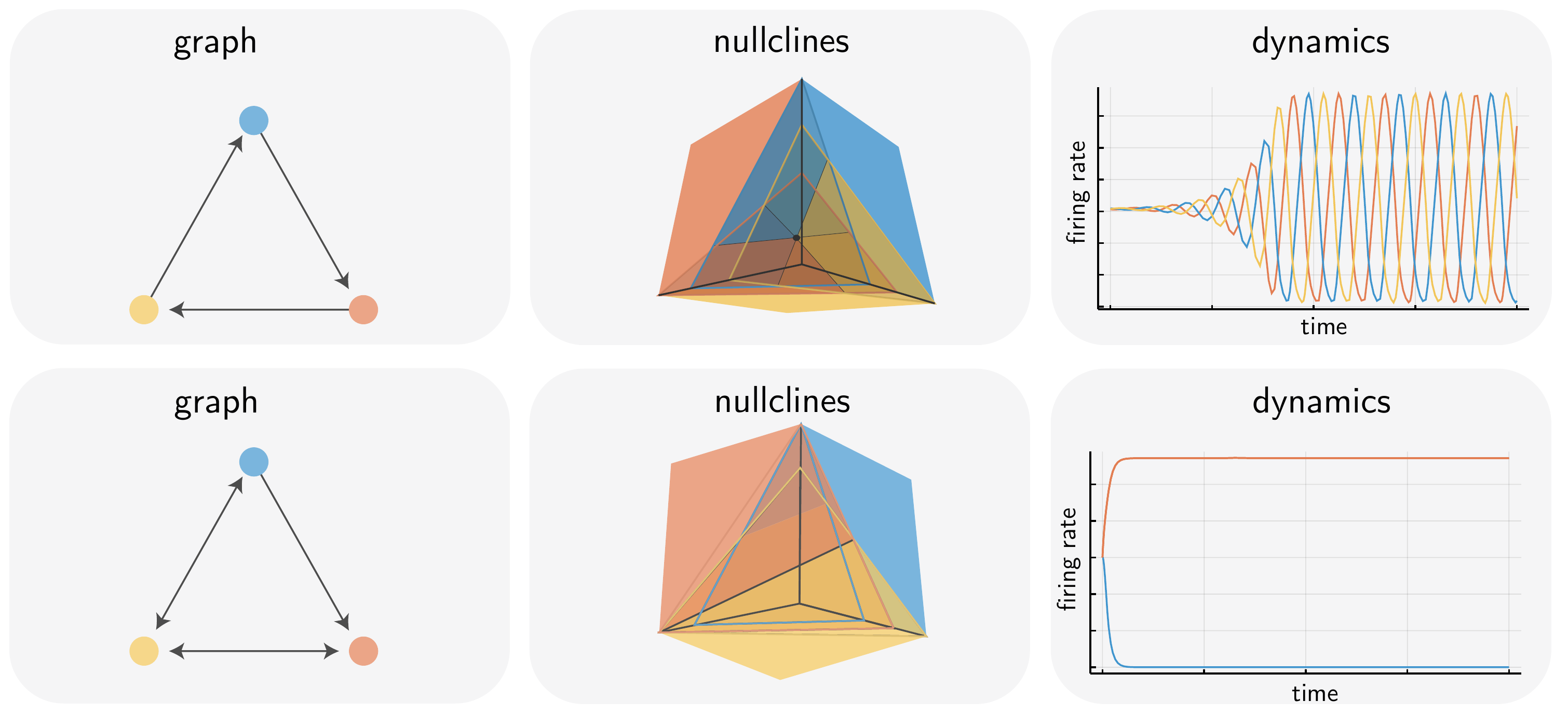}
\caption
[Overview of Combinatorial Threshold Linear Networks]
{CTLNs are dynamical systems determined by directed graph. The directed graph of a CTLN (left) constrains the nullcline arrangement, which in turn constrains the dynamics.  \label{fig:ctln}}
\end{figure}

 Here, we begin to classify which CTLNs have dynamic attractors, which do not. In particular, we prove that the CTLN of a directed acyclic graph must have all trajectories converge to a fixed point. This result can be combined with the result about symmetric threshold-linear networks to prove that if a graph can be decomposed into a directed acyclic graph with edges onto a symmetric graph, all trajectories of its CTLN converge to a fixed point. Our results are sufficient to classify which three neuron graphs have dynamic attractors. Chapter \ref{chapter:TLNs1} gives an introduction to CTLNs. Chapter \ref{chapter:nullclines}  explores the nullcline arrangements of TLNs. The main result in this chapter is Theorem \ref{thm:mixed_sign}, which states that all trajectories of a competitive TLN approach in a small region defined by the nullclines. This results in Corollary \ref{cor:total_pop}, which bounds the total population activity of a TLN in terms of the weights.  The bulk of our main results about CTLNs appear in Chapter \ref{chapter:TLNs2}. In particular, this chapter contains Theorems \ref{thm:dag}  and  \ref{thm:dag_onto_symmetric}, which give conditions which guarantee that all trajectories of a CTLN converge to a fixed point. Theorem \ref{thm:dag} covers the case of directed acyclic graphs, while Theorem \ref{thm:dag_onto_symmetric} strengthens this result to include graphs which contain a directed acyclic part and a symmetric part arranged in a particular way. 

The structure of this dissertation is as follows: in Chapter \ref{chapter:combo_background}, we give the common background on convex sets, hyperplane arrangements, point arrangements, and oriented matroids required for the subsequent chapters. In Part II, Chapters \ref{sec:codes}, \ref{chapter:order_forcing}, and \ref{chapter:matroids_codes}, we discuss convex neural codes. In Part III, Chapters \ref{chapter:urank} and \ref{chapter:urank_math} we discuss underlying rank.  In Part IV, Chapters \ref{chapter:TLNs1}, \ref{chapter:nullclines}, and \ref{chapter:TLNs2}, we discuss threshold-linear networks. 

My novel contributions are concentrated in Chapters 4, 5, 6, 7, 9, and 10. The results in Chapters 4 and 5 can be found in the papers \emph{Order Forcing in Neural Codes}, written with Amzi Jeffs and Nora Youngs \cite{jeffs2020order} and \emph{Oriented Matroids and Neural Codes}, written with Alexander Kunin and Zvi Rosen \cite{kunin2020oriented}. The results on underlying rank in Chapters 6 and 7 and the results on TLNs in Chapters 9 and 10 are currently being written up for publication.

%% file: CombinatoricsBackground/CombinatoricsBackground.tex
%auto-ignore 

 \chapter{Combinatorial Background} \label{chapter:combo_background}

In this chapter, we give background information on the combinatorial objects and theorems which are used here. 
In particular, we discuss convex neural codes, hyperplane arrangements, and oriented matroids. 
We use the notation $[n] := \{1, \ldots, n\}$, and use $2^{[n]} := \mathcal P ([n])$ to denote the powerset of $[n]$. 

\section{Intersection patterns of convex sets }

The \emph{nerve} of a cover records the intersection pattern of a family of sets. 

\begin{defn}
Let $\cU = \{U_1, \ldots, U_n\}$ be a family of subsets of a set $X$. The \emph{nerve} of $\cU$, denoted $\nerve(\cU)$, is the simplicial complex 
\begin{align*}
\nerve(\cU) := \{\sigma \subseteq [n] \mid \bigcap_{i\in \sigma} U_i \neq \varnothing\}.
\end{align*}
We use the notation $U_\sigma := \bigcap_{i\in \sigma} U_i$, with $U_\emptyset = X$. 
\end{defn}

Notice that $\nerve(\cU)$ records less detail about a family of sets than $\code(\cU)$, defined in the previous section. 
Various versions of the nerve theorem, which relate the topology of the nerve to that of the underlying space, were proved in  \cite{ borsuk1948imbedding, leray1950anneau, weil1952theoremes}.  
The version of the nerve theorem we use is \cite[Corollary 4G.3]{hatcher2001algebraic}, and holds when the members of $\cU$ form a good cover. 

\begin{defn}
A family of sets $\cU = \{U_1, \ldots, U_n\}$ is a \emph{good cover} if for all $\sigma \subseteq [n]$, $U_\sigma$ is either empty or contractible. 
\end{defn}

Notice that, since intersections of convex sets are convex, and convex sets are contractible, any collection of convex sets forms a good cover.

\begin{thm}[The Nerve Lemma \cite{hatcher2001algebraic}]
Let $\cU = \{U_1, \ldots, U_n\}$ be a good cover, with all sets open or all sets closed. Then $\nerve(U)$ is homotopy equivalent to $\bigcup_{i = 1}^n U_i$. 
\end{thm}

Every simplicial complex arises as the nerve of a family of convex sets \cite{tancer2013intersection}. 
However, characterizing the dimension required is more complicated. A simplicial complex $\Delta$ is $d$-representable if $\Delta = \nerve(\cU)$ where $\cU$ is a family of convex open sets in $\R^{d}$. 

A classic theorem in the vein is Helly's theorem, which constrains the intersection pattern of convex sets in $\R^d$. 

\begin{thm}[Helly's theorem \cite{helly1923mengen}]
Let $\cU = \{U_1, \ldots, U_n\}$ be a family of convex sets in $\R^{d}$. 
Define $U_\sigma = \bigcap_{i\in \sigma} U_i$.  
Then if for each $\sigma \subseteq [n]$ with $|\sigma | = d+1$ has $U_\sigma \neq \emptyset$, then $U_{[n]} \neq \emptyset$. 
\end{thm}

Interpreted as a result about $d$-representability, Helly's theorem states that if a $d$-representable simplicial complex contains every $d$-simplex, then it is a simplex. 
Helly's theorem holds when the family of convex sets is replaced with a good cover, and is a consequence of the nerve theorem. 
 More general results in this vein exist, characterizing $f$-vectors of $d$-representable complexes \cite{kalai1984characterization, kalai1986characterization, kalai1984intersection}.
 Other results characterize $d$-representable complexes topologically and combinatorially: $d$-representable complexes must be $d$-collapsible \cite{wegner1975d}. 

In particular, the problem of determining the minimal dimension $d$ such that $\Delta$ arises as the nerve of convex sets in $\R^d$ is is NP-hard \cite{tancer2010d}. 
Further, the minimal dimension   
$d$ for which  $\Delta$ arises as the nerve of a good cover may be lower than that in which $\Delta$ arises as the nerve of a family of convex sets \cite{tancer2010counterexample}. 

Note that if a combinatorial code $\cC$ is a simplicial complex, then convexity in a certain dimension corresponds to $d$-representability. Thus, the theory of convex codes which we will discuss in Chapters \ref{sec:codes}, \ref{chapter:order_forcing}, and \ref{chapter:matroids_codes} strictly generalizes the theory of $d$-representability.

\section{Hyperplane and point arrangements}

Hyperplane arrangements and point arrangements are both ways of giving a geometric structure to the relationships between the columns of a matrix. 
We will make extensive use of hyperplane arrangements, via oriented matroid theory, in Chapter \ref{chapter:matroids_codes}, and will use them to a lesser extent in Chapters \ref{chapter:TLNs1}, \ref{chapter:nullclines}, and \ref{chapter:TLNs2}. We will use point arrangements heavily in Chapters \ref{chapter:urank} and \ref{chapter:urank_math}.

Arrangements of hyperplanes, and the half-spaces they define, form an important special case of arrangements of convex sets. 
A vector $h\in \R^{d}$ defines a \emph{hyperplane} $H$ and two open half-spaces $H^+$ and $H^-$ by 
\begin{align*}
H := \{x\mid h\cdot x = 0\}&&
H^+ := \{ x \mid h \cdot  x > 0\}&&
H^- := \{ x \mid h\cdot x < 0\}.
\end{align*}
$H^+$ and $H^-$ are referred to as the positive and negative sides of $H$, respectively. 
A set of hyperplanes is called a \emph{hyperplane arrangement}.
A hyperplane arrangement is \emph{essential} if the matrix with columns $h_1, \ldots, h_n$ has rank $d$. Notice that under this definition, all hyperplanes meet at the origin. We can also define \emph{affine} hyperplane arrangements. 
An affine hyperplane is defined by 
\begin{align*}
H := \{x\mid h\cdot x + b = 0\}&&
H^+ := \{ x \mid h \cdot  x + b > 0\}&&
H^- := \{ x \mid h\cdot x +b < 0\}.
\end{align*}
An arrangement which is not affine as \emph{central}. 
We can translate between central and affine hyperplane arrangements, embedding any affine arrangement in $\R^d$ as a central arrangement in $\R^{d+1}$. 
More specifically, let $h = (h_1, \ldots, h_d), b$ define an affine hyperplane. 
Then $\hat h =  (h_1, \ldots, h_d, b)$ defines a central hyperplane in $\R^{d+1}$. 
Restricting to the plane $x_{d+1} = 1$ recovers our original affine hyperplane arrangement.

An arrangement of $n$ hyperplanes in $\R^{d}$ divides space into a union of at most  $2^{n}$ full dimensional chambers.
Each of these chambers corresponds to a facet of $\nerve(H_1^+, H_1^-, \ldots, H_n^+, H_n^-)$, a simplicial complex on the vertex set $\{1, \ldots, n\} \cup \{\bar 1, \ldots, \bar n\}$ known as the polar complex in \cite{itskov2020hyperplane}.  A set $\sigma\cup\tau$, $\sigma\subseteq \{1, \ldots, n\}, \tau\subseteq \{\bar 1, \ldots, \bar n\}$ is a face of $\nerve(H_1^+, H_1^-, \ldots, H_n^+, H_n^-)$ if and only if there is some point $p\in \R^{d}$ such that $h_i \cdot p > 0$ for $i\in \sigma$, $h_j\cdot p <0$ for $\bar j\in \tau$.  While every simplicial complex arises as the nerve of some arrangement of convex sets, this is not true when we replace ``convex sets" with half spaces. In general, it is difficult to determine whether a simplicial complex is the nerve of an arrangement of half-spaces. On the other hand, it is possible to recover the dimension of an essential hyperplane arrangement from its nerve.  

To see this, we notice that the half spaces $H_1^+, H_1^-, \ldots, H_n^+, H_n^-$ cover all of $\R^d$, except for the point $\bigcap_{i = 1}^n H_i$. 
Thus, the nerve of a central, essential arrangement in $\R^{d}$ has the homotopy type of a $d-1$ sphere, while the nerve of an affine arrangement is contractible. 
Notice that when the maximal value of $2^{n}$ chambers is achieved,  $\nerve(H_1^+, H_1^-, \ldots, H_n^+, H_n^-)$ is a $n$-dimensional cross polytope. 
This means that the maximal value can be achieved only when $d = n$ and $H_1,  \ldots, H_n$ is an essential arrangement. 

If  $H_1, \ldots, H_n$  are an essential hyperplane arrangement in $\R^{d}$, then there is some $\sigma \subseteq [n]$ such that the normal vectors to $\{H_i\}_{i\in \sigma}$ span $\R^{d}$. Then restricting the nerve to this subset produces a $d$-dimensional cross polytope, with $2^{d}$ facets. If we restrict to any larger set of hyperplanes, there must be some ``missing facet". The complex  $\nerve(H_1^+, H_1^-, \ldots, H_n^+, H_n^-)$ is studied in more detail in \cite{itskov2020hyperplane}, in the context of convex neural codes. 

The same numerical data used to define a hyperplane arrangement may also be taken to define a point arrangement $p_1, \ldots, p_n \in \R^d$. 
We can describe the combinatorial structure of a point arrangement in terms of which sets of points can be separated with hyperplanes. In particular, there is an affine hyperplane separating the points 
$\{p_i\}_{i\in \sigma}$, $\{p_j\}_{i\in \tau}$ if and only if there is some vector $h \in \R^{d+1}$ such that $h \cdot (p_i, 1) > 0$ for $i\in \tau$, $h \cdot (p_j, 1) < 0$ for $j\in \sigma$. Notice that this is the same condition for $(\sigma, \tau)$ to be the chamber of a hyperplane arrangement. 
Thus, we can also determine dimension of a point arrangement from the partitions of points which can be achieved with a hyperplane--for any set of at $n$ points in $\R^{d}$, if $n \geq d + 2$, there is some partition of the points which cannot be achieved with a hyperplane.

This observation is equivalent to \emph{Radon's theorem}. Notice that $\conv_{i\in \sigma}\{x_i\} \cap \conv_{j\notin\sigma}\{x_j\} = \emptyset$ if and only if there is an affine hyperplane $H$ which separates $\{x_i\}_{i\in \sigma}$ from $\{x_j\}_{j\notin\sigma}$.

\begin{thm}[Radon's Theorem \cite{radon1921mengen}]
If $p_1, \ldots, x_n$ are points in $\R^{d}$, and $n \geq d+1$, then there exists a \emph{Radon partition} $\sigma \cup \tau =[n]$ such that $\sigma\cup \tau = \emptyset$, but $\conv_{i\in \sigma}\{x_i\} \cap \conv_{j\notin\sigma}\{x_j\} \neq \emptyset$. 
\end{thm}

The bound provided by Radon's theorem is tight: $d+1$ affinely independent points in  $\R^{d}$ have no Radon partition. 
In Chapters \ref{chapter:urank} and \ref{chapter:urank_math}, we use Radon's theorem as a method for estimating underlying rank. 

\section{Oriented Matroids}
\label{sec:oriented_matroid_intro}

Oriented matroid theory is a powerful tool in discrete geometry which we use in Chapters \ref{chapter:matroids_codes} and \ref{chapter:urank}. Oriented matroids abstract and generalize the properties of hyperplane arrangements and point arrangements. Here, we provide a short overview of oriented matroid theory. See \cite{bjorner1999oriented} for a comprehensive reference.

\subsection{Covector Axioms}
Much like the code of a cover records the combinatorial information about how a family of convex sets overlap, an oriented matroid records combinatorial information about a hyperplane arrangement. In fact, we can see the oriented matroid of a hyperplane arrangement as a special case of a convex neural code, as illustrate in Figure \ref{fig:OMandCodeExample}.

\begin{figure}[ht!]
\begin{center}
  \includegraphics[width = 5 in]{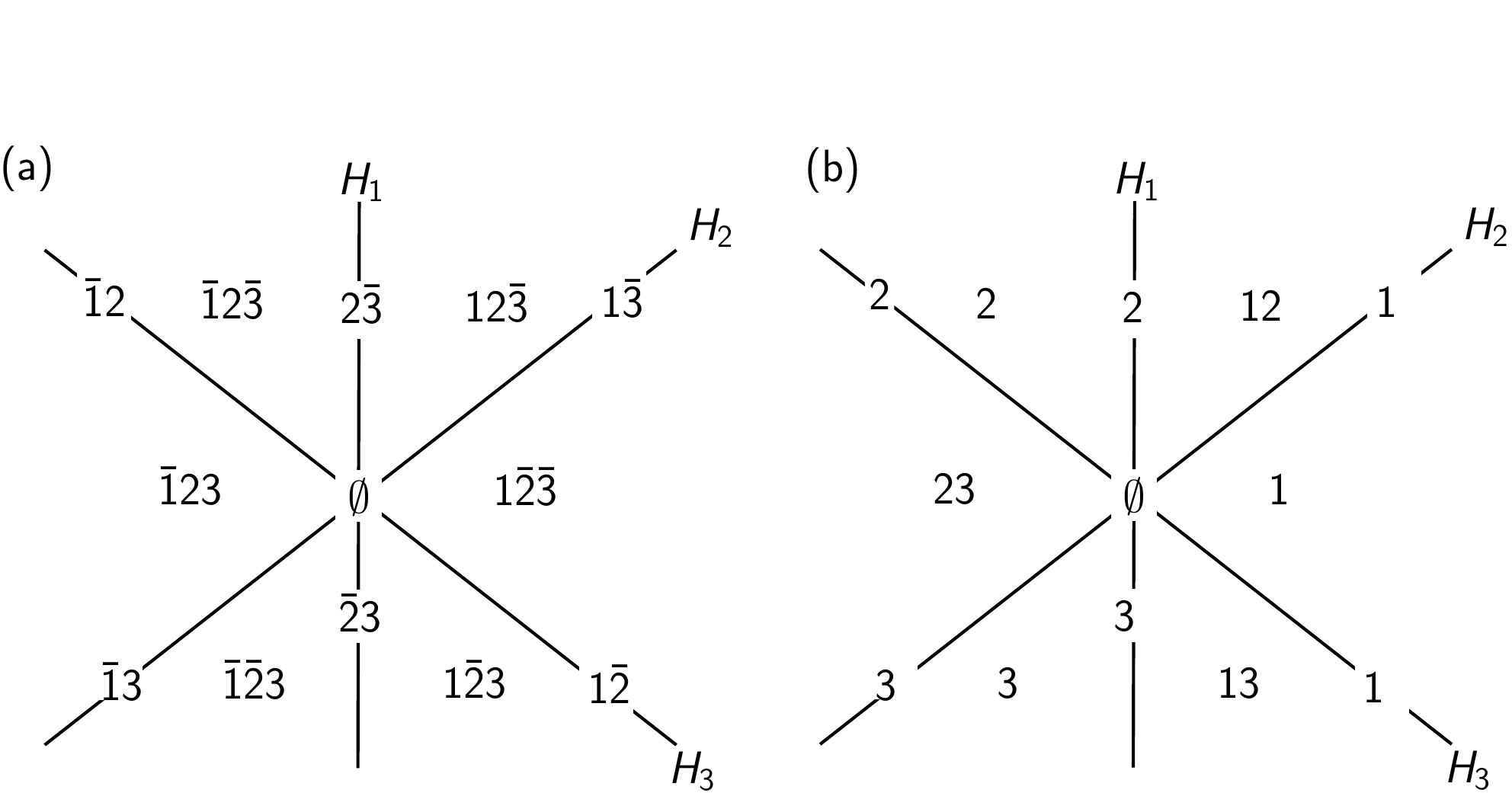}
  \caption[Oriented matroids and neural codes from hyperplane arrangements.]{(a)~The covectors of an oriented matroid arising from a central hyperplane arrangement. The topes are $\bar 12\bar 3, 12\bar 3, 1\bar2\bar3, 1\bar23, \bar1\bar23$ and $\bar123$ (b)~The combinatorial code of the cover given by the positive open half-spaces.}
  \label{fig:OMandCodeExample}
  \end{center}
\end{figure}
%Notice that an arrangement of $n$ hyperplanes in $\R^d$ divides space into a set of at most $2^{n}$ full-dimensional chambers. Each of these chambers is defined by restricting to either the positive or negative side of each hyperplane.  In particular, for a partition $\sigma, \tau$ of $[n]$, the chamber $\bigcap_{i \in \sigma} H_i ^+ \cap \bigcap_{j \notin \sigma} H_j ^-$ is non-empty if and only if there is a point $x$ such that $h^*_i(x) >0$ for $i\in \sigma$, $h_j^*(x) < 0$ for $j \notin \sigma$.  

A central hyperplane arrangement $\cH$ divides $\R^d$ into a set of polyhedral chambers. The natural labels assigned to these chambers form the \emph{covectors} of a \emph{representable oriented matroid}. These labels can be written as sign vectors, i.e. elements of $\{+, -, 0\}^n$.  
We can assign each point $x\in \R^d$ to a sign vector by 
  \begin{align*}
    X_i = \begin{cases}
      ~ + &\mbox{ if }  h^*_i(x)  > 0\\
     ~ - &\mbox{ if }  h^*_i(x)  < 0\\
    ~  0 &\mbox{ if }  h^*_i(x)  = 0.
    \end{cases}
  \end{align*}
The family $\cL(\cH)$ of sign vectors which arise in this way is known as the set of \emph{covectors} of the oriented matroid $\cM(\cH) = ([n], \cL(\cH))$. The covectors of top-dimensional cells are called \emph{topes} of $\cM$. 

Notice that $\cL(\cH)$ records the same information as $\code(H_1^+, \ldots, H_n^+, H_1^-, \ldots, H_n^-)$. We can clarify this with alternate notation for sign vectors: defining $\pm[n] = [n]\cup \bar[n] := \{1, \ldots, n\} \cup \{\bar 1, \ldots, \bar n\}$, we can write the sign vector $X$ as the set $X := \{i\mid X_i = +\} \cup \{\bar i\mid X_i = -\}$. 
In this notation, 
\begin{align*}
\cL(\cH) = \code(H_1^+, \ldots, H_n^+, H_1^-, \ldots, H_n^-) \subseteq 2^{\pm[n]}.
\end{align*}

\begin{figure}[ht!]
\begin{center}
\includegraphics[scale = 0.75]{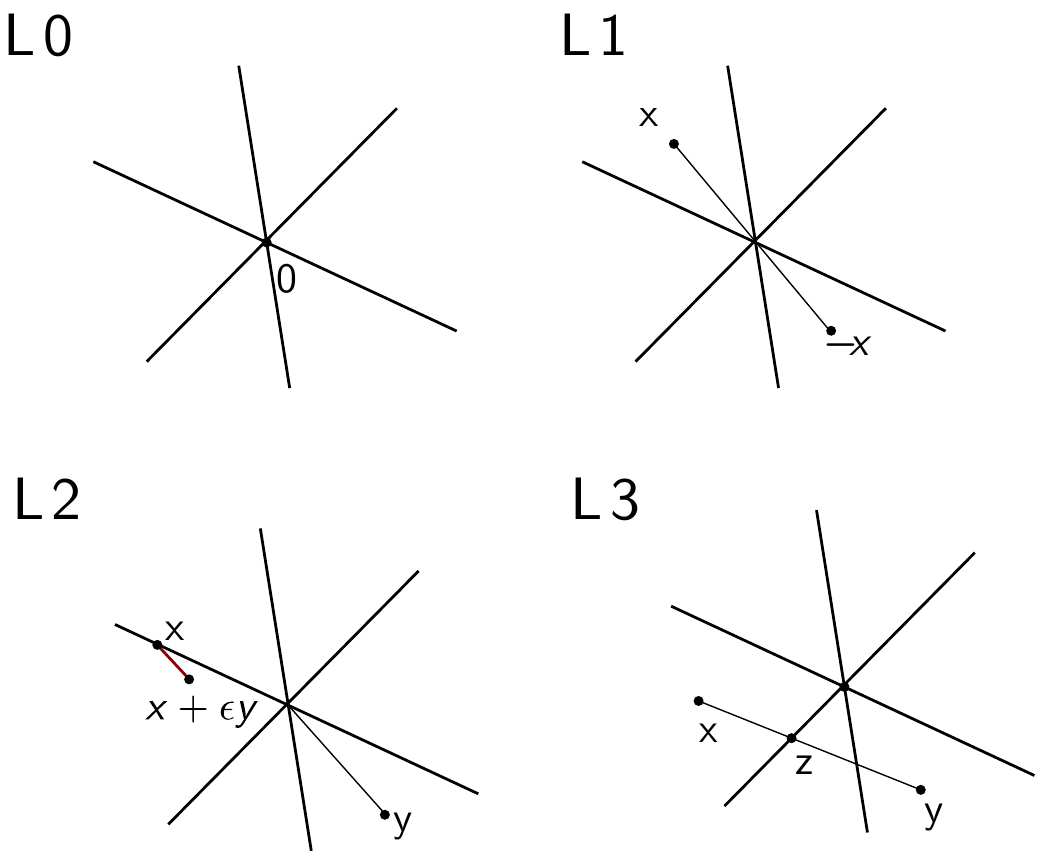}
\caption[Geometric interpretation of the covector axioms.]{Geometric interpretation of the covector axioms.}
\label{fig:axioms}
\end{center}
\end{figure}

The set $\cL(\cH)$ satisfies a list of axioms know as the \emph{covector axioms for oriented matroids}. Oriented matroids are defined in general via these axioms. In order to state them, we introduce some more notation. The \emph{support} of a sign vector $X$ is the set
$\underline X := \{ i \mid X_i \neq 0\}$. 
The \emph{positive part} of $X$ is $X^+ := \{i \mid X_i = +\}$ and the
\emph{negative part} is $X^- := \{i \mid X_i = -\}$. 
The \emph{composition} of sign vectors $X$ and $Y$ is defined component-wise by
\begin{align*}
  (X\circ Y)_i := 
  \begin{cases}
    X_i \mbox{ if } X_i\neq 0\\
    Y_i \mbox{ otherwise}.
  \end{cases}
\end{align*}
The \emph{separator} of $X$ and $Y$ is the unsigned set $\sep(X, Y) := \{ i\mid X_i = -Y_i\neq 0\}$. 
%	\subsubsection{Covectors}
\begin{defn}
  Let $E$ be a finite set, and $\cL \subseteq 2^{\pm E}$ a collection of sign vectors satisfying the following
  \emph{covector axioms}:
\begin{enumerate}[(L1)]
\item\label{axiom:emptysetV} $\emptyset \in \mathcal L$
\item\label{axiom:symmetryV} $X\in \mathcal L $ implies $-X\in \mathcal L$. 
\item\label{axiom:compositionV} $X, Y\in \mathcal L$ implies $X\circ Y\in \mathcal L$. 
\item\label{axiom:crossingV} If $X, Y\in \mathcal L$ and $e \in \sep(X,Y)$, then there exists $Z\in \mathcal L$ such that $Z_e = 0$ and $Z_f = (X\circ Y)_f = (Y\circ X)_f$ for all $f\notin \sep(X, Y)$. 
\end{enumerate}
Then, the pair $\cM = (E, \cL)$ is called an \emph{oriented matroid}, and $\cL$ its set of covectors.
\end{defn}

For any hyperplane arrangement,  $\cL(\cH)$  must satisfy all of the covector axioms, thus the oriented matroid of a hyperplane arrangement is, in fact, an oriented matroid. We give a geometric interpretation of each covector axiom in Figure \ref{fig:axioms}. An oriented matroid $\cM$ is \emph{representable} if there exists a hyperplane arrangement $\cH$ such that $\cM = \cM(\cH)$. We discuss representability in more detail in Section \ref{sec:rep}. 

We can view $\cL$ as a poset, with the covectors partially ordered by inclusion. By adjoining a top element $\hat 1$ with $X \leq \hat 1$ for all $X\in \cL$, we can construct the \emph{face lattice} of the oriented matroid, $\cL \cup \{\hat 1\}$. Notice that in the realizable case, traversing upwards in this face lattice corresponds to moving from a lower dimensional cell to an adjacent cell of one dimension higher. Thus, the height of this poset tracks the dimension of the space. This lets us define the \emph{rank} of a matroid as 

\[\rank(\cL) = \mathrm{height}\, (\cL \cup \hat 1) -2.\]

In the realizable case, $\rank(\cL(\cH))$ recovers the rank of the matrix whose columns are the normal vectors to the hyperplanes in $\cH$. 

We can use oriented matroids to describe affine hyperplane arrangements as well. In particular, an \emph{affine oriented matroid} $(\cM, g)$ is an oriented matroid together with a distinguished ground set element $g$. We can define the \emph{positive covectors} of $(\cM, g)$ as the set $\cL_+(\cM, g) := \{X \in \cL(\cM) | X_g = +\}$. Notice that if $\cM$ is the oriented matroid of $H_1, \ldots, H_n, H_g$ where $H_g$ is the hyperplane $x_{d+1} = 0$, and $H_1, \ldots, H_n$ are centralized versions of affine hyperplanes as above, then the positive covectors  $\cL_+(\cM, g)$ correspond to cells of the affine hyperplane arrangement. 

%%% ABK - do we want this to get a definition number?
%	\begin{defn}
%		An \emph{oriented matroid} $\mathcal M$ is a tuple $\mathcal M = (E, \mathcal L)$, where $E$ is a finite set and $\mathcal L$ is a family of signed subsets of $\pm E$ satisfying (V\ref{axiom:emptysetV})-(V\ref{axiom:crossingV}).
%		Elements of $\cL$ are \emph{covectors} of $\cM$. 
%	\end{defn}

%	An oriented matroid is \emph{loopless} if \textcolor{red}{it has no loops.}

%\subsubsection{Equivalent Formulations}

\subsection{Circuit axioms}
There are many equivalent axiomatizations of oriented matroids.
%\footnote{Indeed, somewhat surprising equivalence of the many different axiomatizations gave rise to the term ``cryptomorphism'' to describe the situation}
%, reflecting the diversity of situations described by oriented matroids.
The two formulations we use most often throughout this work are the covector axioms (L\ref{axiom:emptysetV})-(L\ref{axiom:crossingV}), stated above, and the circuit axioms (C\ref{axiom:emptysetC})-(C\ref{axiom:weakelimC}), which we state here. The circuit axioms most naturally arise when we consider the oriented matroid of a point arrangement, $p_1, \ldots, p_n \in \R^{d}$. 

\begin{defn}
Let $\cP= \{p_1, \ldots, p_n\}$ be a point configuration in $\R^{d}$. The sign vector $X$ is a \emph{circuit} of the oriented matroid $\cM(\cP)$ of $\cP$ if $X^+$, $X^-$ is a \emph{minimal Radon partition} of $V$. That is, $$\conv(\{p_i \mid i \in X^+\})\cap \conv(\{p_j \mid j \in X^-\}) \neq \emptyset,$$ and for all  $Y < X$, 
$$\conv(\{p_i \mid i \in Y^+\})\cap \conv(\{p_j \mid j \in Y^-\}) = \emptyset.$$
\end{defn}

The minimal Radon partitions of a point arrangement follow a list of rules known as \emph{the circuit axioms for oriented matroids.} Oriented matroids are defined via these axioms.

\begin{defn}\label{D:circuitaxioms}
  Let $E$ be a finite set, and $\cC\subseteq 2^{\pm E}$ a collection of signed subsets satisfying the following \emph{circuit axioms}: 
\begin{enumerate}[(C1)]
	\item\label{axiom:emptysetC} $\varnothing \notin \cC$.
	\item\label{axiom:symmetryC} $X \in \cC$ implies $-X \in \cC$.
	\item\label{axiom:incomparableC} $X,Y \in \cC$ and $\underline X \subseteq \underline Y$ implies $X = Y$ or $X = -Y$.
	\item\label{axiom:weakelimC} For all $X,Y \in \cC$ with $X \neq -Y$ and an element $e \in X^+ \cap Y^-$, there is a $Z \in \cC$ such that $Z^+ \subseteq (X^+ \cup Y^+) \setminus e$ and $Z^- \subseteq (X^- \cup Y^-) \setminus e$.
\end{enumerate}
Then the pair $\cM = (E, \cC)$ is an oriented matroid, and $\cC$ is its set of circuits. 

%In some contexts, we admit the sets $\{e, \bar{e}\}$ as \emph{improper circuits}.  
%We will call a circuit a \emph{proper circuit} when we wish to emphasize that it is a signed set, i.e.\ its positive and negative parts are disjoint. 
%An element $e \in E$ is a \emph{loop} of $\cM$ if $\{e\} \in \cC(\cM)$. An oriented matroid is \emph{loopless} if no element is a loop.
\end{defn}

Note that it is possible to recover the dimension of the affine span of the point arrangement of $\cC(\cP)$ via Radon's theorem.

 In particular, if $\cC(\cP)$ is a point configuration in $\R^{d}$, then every set of at least $d+2$ points contains the support of circuit. Further, as long as all points are not contained in a lower-dimensional subspace, there is at least one set of $d+1$ points which does not contain the support of circuit. 
Motivated by this, the \emph{rank} of an oriented matroid is defined as the maximum size of a set $X$ which does not contain the support of a circuit. Note that this means a point arrangement in $\R^{d}$ corresponds to a rank $d+1	$ matroid.  An oriented matroid is \emph{uniform} if all of its circuits have the same cardinality. Uniform oriented matroids correspond to point arrangements which are in general position.

\subsection{Duality}

We can translate between the circuit and covector descriptions of an oriented matroid, illustrated in the case of point arrangements in Figure \ref{fig:radon_matroid}. 
Circuits are related to covectors as follows: %\abk{QUESTION: how do improper circuits play with this definition of orthogonality? That is, ``signed set'' specifically means that the positive and negative parts are disjoint.}
Two signed sets $X$ and $Y$ are called \emph{orthogonal} if either $\underline X \cap \underline Y = \emptyset $ or if there exist $e, f\in \underline X \cap \underline Y $ such that $X_eX_f  = - Y_e Y_f$. 
A signed set is called a \emph{vector} of $\cM$ if and only if it is orthogonal to every covector.
Equivalently, a signed set is a vector of $\cM$ if and only if it is orthogonal to every tope.
The circuits are the minimal vectors of $\cM$, while minimal covectors are called \emph{cocircuits}. 
The vectors of an oriented matroid $\cM$  are the covectors of its dual oriented matroid $\cM^*$.
Thus, vectors and covectors satisfy the same set of axioms, as do circuits and cocircuits.

\begin{figure}[ht!]
\begin{center}
\includegraphics[width = 3 in]{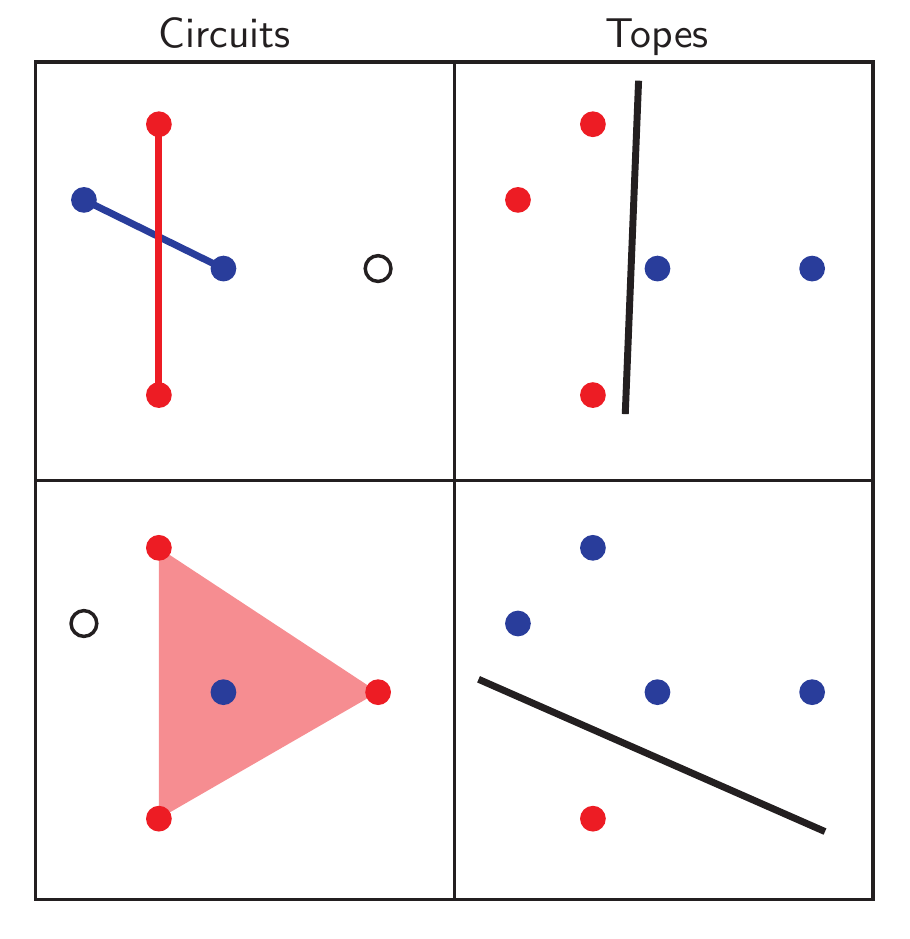}
\caption[Circuits and topes of oriented matroids]{Circuits of an oriented matroid correspond to minimal Radon partitions, while topes correspond to maximal hyperplane partitions. 
\label{fig:radon_matroid}
}
\end{center}
\end{figure}
  
For a given oriented matroid $\cM$, each one of the set of covectors $\cL(\cM)$, the set of topes $\cW(\cM)$, the set of vectors $\cV(\cM)$, and the set of circuits $\cC(\cM)$ is sufficient to recover all of the others.
%when necessary, we denote the covectors, topes, etc.\ of a fixed oriented matroid $\cM$ by $\cL(\cM)$, $\cW(\cM)$, etc.
%		Thus, an oriented matroid can be equivalently characterized as $\cM = (E, \cL)$, $\cM = (E, \cT)$, $\cM = (E, \cV)$, or $\cM = (E, \cC)$.
%		In addition to the vector and covector axioms given above, there are axiom systems characterizing the sets $\cT$ and $\cC$. 

We can build geometric intuition around duality by considering the oriented matroid of a point arrangement. Let $\cP$ be a point arrangement, and note that a signed set $X$ with $\underline X = [n]$ is a tope of $\cM$ if and only if it is orthogonal to every circuit of $\cM$. Then there is no circuit $Y$ of $\cM$ such that $Y \subseteq X$.  Then $\conv(\{p_i\}_{X_i = +}) \cap \conv(\{p_j\}_{X_j = -}) = \emptyset$. Thus, the topes of the oriented matroid of a point arrangement correspond to the partitions of the set of points which can be achieved with a hyperplane. In general, $X$ is a covector of $\cM(\cP)$ if there exists a hyperplane $H$ such that $X^+ = \{i \mid p_i \in H^+\}, X^- =  \{j \mid p_j \in H^-\}, X^0 = \{k \mid p_k \in H\}$. 

We can also see duality in the case of hyperplane arrangements through receptive field relationships, similar (but not identical) to those defined via the neural ring in \cite{curto2013neural}.  In particular, suppose $X$ is a circuit of $\cC(\cH)$. Then  $X$ is orthogonal to each covector of $\cL(\cH)$. For each point $p$, let $Y(p)$ be the sign vector at the point $p$. Then either $\underline {Y(p)} \cap \underline X = \emptyset$, or there exists $j \in Y(P)^+\cap X^-$, $k \in Y(P)^-  \cap X^+$. In the first case, we have $p\in H_i$ for all $i\in \underline Y$. In the second case, we have $p \in H_j^+$ for some $j \in X^-$, $p \in H_k ^-$ for some $k \in X^+$. Thus, the sets $\{\bar H_j^+\}_{j \in X^-} \cup \{\bar H_k^-\}_{j \in X^+}$ cover $\R^d$.  Equivalently, $\bigcap_{j\in X^+} H_j^+ \subseteq \bigcup_{k \in X-} \bar H_k^+$.

\subsection{Representability}\label{sec:rep}
An oriented matroid $\cM$ is \emph{representable} if $\cM = \cM(\cH)$ for some hyperplane arrangement $\cH$, or, equivalently, $\cM = \cM(\cP)$ for some point arrangement $\cP$.  
Figure \ref{F:OMandCodeExample}(a) illustrates an example in $\R^2$.
Not every oriented matroid is representable--see Chapter 8 of \cite{bjorner1999oriented} for more information on represnentability. 
However, we are able to take this hyperplane picture as paradigmatic. The topological representation theorem guarantees that every oriented matroid has a representation by a pseudosphere arrangement: a collection of centrally symmetric topological $d-2$ spheres embedded in $\bS^{d-1}$ whose intersections are also spheres of the appropriate dimension \cite{folkman1978oriented}.  See Chapter 5 of \cite{bjorner1999oriented} for  more information on  the topological representation theorem. See Figure \ref{fig:pseudosphere} for an illustration.  For a representable oriented matroid, we can obtain a representation with a sphere arrangement by intersecting each hyperplane with a sphere containing the origin. 

\begin{figure}[ht!]
\begin{center}
  \includegraphics[width = 4 in]{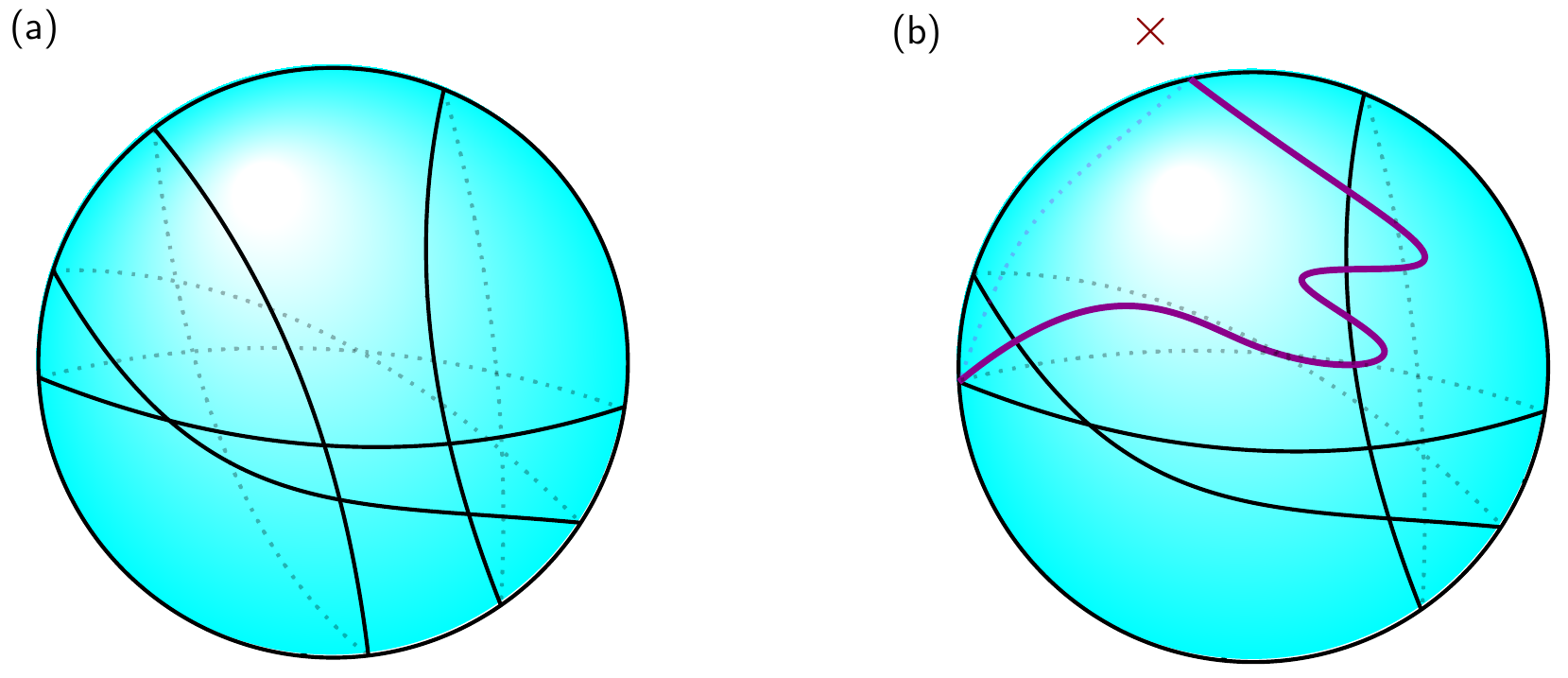}
  \caption[Pseudosphere arrangements.]{(a)~A pseudosphere arrangement (b)~Not a pseudosphere arrangement.}
  \label{fig:pseudosphere}
  \end{center}
\end{figure}

In general, it is difficult to determine whether an oriented matroid is representable. In particular, by \cite{sturmfels1987decidability, mnev1988universality, shor1991stretchability} determining representability is NP-hard. In fact, something stronger is true: determining representability is complete for the \emph{existential theory of the reals}. This is the complexity class of decision problems of the form
$$\exists(x_1 \in \R)\cdots \exists(x_n\in \R)P(x_1, \ldots , x_n),$$
where P is a quantifier-free formula whose atomic formulas are polynomial equations and inequalities in the $x_i$ \cite{broglia1996lectures}. 
Problems which are $\exists \R$-complete are not believed to be computationally tractable. In particular, they must be NP-hard. Many classic problems in computational geometry fall into $\exists\R$ \cite{schaefer2009complexity}. In particular, in Chapter \ref{chapter:matroids_codes}, we will show that determining whether a code is convex is $\exists\R$ complete. In Chapter \ref{chapter:urank_math}, we will show that determining the underlying rank of a matrix is $\exists\R$ complete. 

There is likely no combinatorial characterization of representability. In particular, focusing on  unoriented matroids but with results which apply to oriented matroids as well, a series of papers makes the claim that the ``missing axiom of oriented matroid theory is lost forever" \cite{vamos1978missing, mayhew2014missing, mayhew2018yes}. That is, there is no statement in the language of the original matroid axioms which characterizes representable matroids.

%% file: CombinatoricsBackground/ConvexCodesBackground/ConvexCodesBackground.tex
%auto-ignore 
\chapter{Introduction to Convex Neural Codes} \label{sec:codes}

How does the brain keep track of the body's position in space? 
In 1948, based on observations that rats in mazes learn the broader geography of the maze, rather than just the correct sequence of turns to reach the goal, Tolman \cite{tolman1948cognitive} speculated that the brains of rodents (and humans) create and maintain maps of their environments. 
In 1971 \cite{o1971hippocampus}, O'Keefe and Dostrovsky recorded the activity of neurons in the hippocampus of a rat which they held on a platform, and found cells which appeared to form part of the cognitive map Tolman posited. 
In particular, they discovered some neurons which were more active when the rat was at one particular location on the platform, and other neurons which were more active when the rat was oriented in one particular direction. 
In \cite{o1976place}, O'Keefe mapped the receptive fields of cells they named place cells, and found them to be contiguous, roughly convex subsets of the rat's environment.

From subsequent work, we now know that place cells determine location by integrating input from multiple sensory systems, and by performing path integration based on the animal's motion \cite{knierim1995place}.  
While place cells have one convex firing field within a small environment, place cells have multiple fields in larger environments, with no apparent relationship between the different fields \cite{yim2021place}. 
Over time, place fields within one environment are remapped, i.e. some cells stop being active, others become active, and some have their place fields move. The relationships between place fields change over time. 
Place cells are part of a larger navigational system, involving grid cells and head direction cells.

Early on, it was observed that it is possible to decode location from the collective activity of place cells \cite{o1976place}. 
Once technology made it possible to record enough place cells to be possible, this was demonstrated in \cite{brown1998statistical}. 
However, these decoding mechanisms make use of the encoding map: observations of the locations of the place fields. 
This is not information that the brain has access to. 
 
In \cite{curto2008cell}, Curto and Itskov asked to what information about the environment can be decoded from place cell activity alone, without information about the locations of place fields. 
In particular, they use the \emph{nerve theorem} to show that it is possible to recover the topology of the environment using the sets of place cells which fire together, on the assumption that place cells have convex receptive fields.

A key observation of  \cite{curto2008cell} is that if $U_1, \ldots, U_n$ are interpreted as place fields, $\nerve(\cU)$ consists of the sets of place cells which fire together. Thus, we can use the nerve theorem to recover the topology of $\bigcup_{i = 1}^n U_i$ from neural activity, even if all we know about the receptive fields is that they are convex.

\section{Convex and non-convex codes}\label{sec:cvx_or_not}
In \cite{curto2013neural}, Curto, Itskov, Veliz-Cuba, and Youngs go beyond the nerve, studying the relationships between receptive fields which are implied by neural activity. In particular, they focus on the the combinatorial neural codes, which record more detail about how receptive fields interact than the nerve of a cover.

\begin{defn}
A combinatorial neural code is a subset $\cC\subseteq 2^{[n]}$. Elements of the code are called codewords. 
\end{defn}

We interpret the codewords as sets of neurons which fire at roughly the same time, and the neural code as the collection of all sets of neurons which are observed to fire together over some time. While codewords are often written as binary vectors, we use more compact subset notation. For instance, if at some time neuron 1 fires alone, at another time neurons 1 and 2 fire together, and at a third time, neurons 2 and 3 fire together, and at a fourth time no neurons fire, we denote this with the neural code $\cC = \{\{1\}, \{1,2\}, \{2,3\}, \emptyset\}$. For compactness, we will omit brackets and commas on the inner sets, abbreviating this as $\cC = \{1, 12, 23, \emptyset\}$. 	

Given a family of sets $\cU$, we can define a combinatorial neural code.

\begin{defn}
Let $\cU = \{U_1, \ldots, U_n\}$ be a collection of subsets of a set $X$. The code of $\cU$, written  $\code(\cU)$, is the set 
\begin{align*}
\code(\cU) := \{\sigma \mid U_\sigma \setminus \bigcup_{j\notin\sigma} U_j \neq \emptyset\}. 
\end{align*}
We define the \emph{atom} of a codeword as 
\[\cU^\sigma := \{U_\sigma \setminus \bigcup_{j\notin \sigma} U_j\}\]
In cases where the universe $X$ is not clear from context, we will write $\code(\cU, X)$.

Equivalently, $\code(\cU)$ is the set of labels which arise when we label each point $p\in X$ with the set of $i$ such that $p\in U_i$. See Figure \ref{fig:cvx} for the relationship between neural activity, combinatorial neural codes, and receptive fields. 
\end{defn}

We can recover $\nerve(\cU)$ by completing $\code(\cU)$ to a simplicial complex. 

\begin{defn}
Let $\cC$ be a combinatorial neural code. Define a simplicial complex $\Delta(\cC)$ as the smallest abstract simplicial complex containing $\cC$, i.e. as 
\begin{align*}
\Delta(\cC) := \{\tau\mid \tau \subseteq \sigma \mbox{ for some }\sigma \in \cC\}
\end{align*}
\end{defn}
Notice that 
\begin{align*}
\Delta(\cC(\cU)) = \nerve(\cU).
\end{align*}

To what extent does the constraint of having convex receptive fields show up in the structure of the combinatorial code itself?  That is, can we characterize which combinatorial codes arise from the activity of neurons with convex receptive fields? 

\begin{defn}
A neural code $\cC $ is \emph{convex open} (resp. closed) if there exists a family of convex open (resp. closed) sets $\cU = \{U_1, \ldots, U_n\}$  in $\R^d$ such that $\cC = \code (\cU)$. The if $\cC$ is convex, the minimum value of $d$ for which this is possible is referred to as the \emph{minimal embedding dimension}.
\end{defn}

\begin{question}\label{quest:convex}
Can we give an intrinsic characterization of which neural codes are convex? Is there an algorithm to determine whether a code is convex? Can we compute or estimate the minimal embedding dimension of a code? 
\end{question}

An answer to Question \ref{quest:convex} would make it possible to search for convex receptive field geometry in regions of the brain where the receptive fields are less straightforward than those of hippocampal place cells. Additionally, such a characterization would help us to characterize the connectivity of neural circuits which give rise to neurons with convex receptive fields. Finally, Question \ref{quest:convex} turns out to be mathematically rich, with connections to other work in discrete geometry.

Not every code is convex. For instance, the code $\cC = \{12, 13\}$ is neither convex open nor convex closed. To see this, suppose to the contrary that there are convex sets $U_1, U_2, U_3$, either all open or all closed, such that $\cC = \code(U_1, U_2, U_3)$.  Since neuron $1$ never fires alone, we have $U_1 = U_2 \cup U_3$. However, neurons 2 and 3 never fire together, so $U_2 \cap U_3 = \emptyset$. Thus, $U_2$, $U_3$ gives a disconnection of $U_1$. Since convex sets must be connected, this is a contradiction. This is an example of a \emph{local obstruction}. Without the assumption that our sets are all open or all closed, it is true that all codes are convex \cite{franke2017every}, though these constructions are often highly degenerate. 

Giusti and Itskov make a first step towards answering Question  \ref{quest:convex} in  \cite{giusti2014no}, which characterizes non-convex codes via local obstructions. We use the characterization of local obstructions provided in \cite[Theorem 1.3]{curto2017makes}. We first define the link of a simplex in a simplicial complex. 

\begin{defn}
Let $\Delta$ be a simplicial complex. Then the \emph{link} of a simplex $\sigma \in \Delta$ is the set 
\begin{align*}
\link_\Delta(\sigma)= \{\tau\in \Delta \mid \sigma \cap \tau = \varnothing, \sigma \cup \tau \in \Delta\}
\end{align*}
\end{defn}

\begin{defn}
Let $\cC$ be a combinatorial neural code. Then $\cC$ has a \emph{local obstruction} if there is a $\sigma \in \Delta(\cC)$ such that  $\sigma \notin \cC$ and $\link_{\Delta(\cC)}(\sigma)$ is not contractible. 
\end{defn}

Notice that for each simplicial complex $\Delta$, this defines a minimal code with no local obstructions 
\begin{align*}
\cC_{\min}(\Delta) = \{ \sigma \in \Delta \mid  \link_{\Delta(\cC)}(\sigma) \mbox{ is not contractible }\}
\end{align*}

If $\sigma$ is not the intersection of facets of $\Delta(\cC)$, then $\link_{\Delta(\cC)}(\sigma)$ is automatically contractible. Thus, when checking for local obstructions, we can restrict to codewords $\sigma$ which are intersections of facets, also called \emph{maximal codewords}. We often write the maximal codewords in bold. 

\begin{thm}[Theorem 3, \cite{giusti2014no}]
If $\cC$ is convex (either open or closed), it has no local obstructions. 
\label{thm:no_local}
\end{thm}
To see this, notice that if $\sigma \in \Delta(\cC)\setminus \cC$, then  $U_\sigma = \bigcup_{j\notin \sigma} U_\sigma \cap U_j$. On the assumption that $U_1, \ldots, U_n$ are convex and either all open or all closed, this is a good cover. Also notice that $\nerve(\{U_\sigma \cap U_j\}_{j\notin \sigma}) = \link_\sigma(\Delta(\cC)$. Thus,  $\link_{\Delta(\cC)}(\sigma)$ is homotopy equivalent to  $\bigcap_{i\in \sigma} U_i$ by the nerve theorem. Thus, if $\cC$ is convex, $\sigma\in \Delta(\cC)\setminus \cC$,   $\link_{\Delta(\cC)}(\sigma)$  must be contractible. Notice that we can weaken the requirement that $U_1, \ldots, U_n$ be convex here to a requirement that $U_1, \ldots, U_n$ forms a good cover. 

One might hope that the converse of this theorem holds, that a code is convex if and only if it has no local obstructions. Unfortunately, while this is true for codes on at most four neurons, this is not the case in general. In particular, the code 
\begin{align*}
\cC_1 = \{{\bf 2345}, {\bf 123}, {\bf 134}, {\bf 145}, 13, 14, 23, 34, 45, 3, 4, \varnothing\} 
\end{align*}
 has no local obstructions, and is closed-convex (Figure \ref{fig:non_convex} (a)), but not open-convex \cite{lienkaemper2017obstructions}. 
On the other hand, the code 
\begin{align*}
\cD_1 = \{{\bf 123},{\bf 234},{\bf 345},{\bf 145},{\bf 125},12,23,34,45,15,\emptyset\},
\end{align*} the has no local obstructions and is open convex  (Figure \ref{fig:non_convex} (b)), but is not closed convex. This code first appears in \cite{goldrup2020classification}, and is a simplified version of a code introduced in \cite{cruz2019open}. We will give a proof that this code is not convex in Example \ref{ex:wheel} in Chapter \ref{chapter:order_forcing}.  By combining these two codes in a clever way, \cite{gambacini2021non} gives an eight neuron code which has no local obstructions, but is neither open nor closed convex
\begin{align*}
\cC = \{\mathbf {2345}, \mathbf {123}, \mathbf {124},\mathbf { 145}, 12, 14, 23, 24, 45, 2, 4,\varnothing\}\cup\{\mathbf {237}, \mathbf {238},\mathbf { 367},\mathbf { 678}, \mathbf{26}, 37, 67, 6, 8\}. 
\end{align*}
It is true, however, that $\cC$ is a good cover code if and only if it has no local obstructions \cite{chen2019neural}. 

\begin{figure}
\begin{center}
\includegraphics[width = 4.5 in]{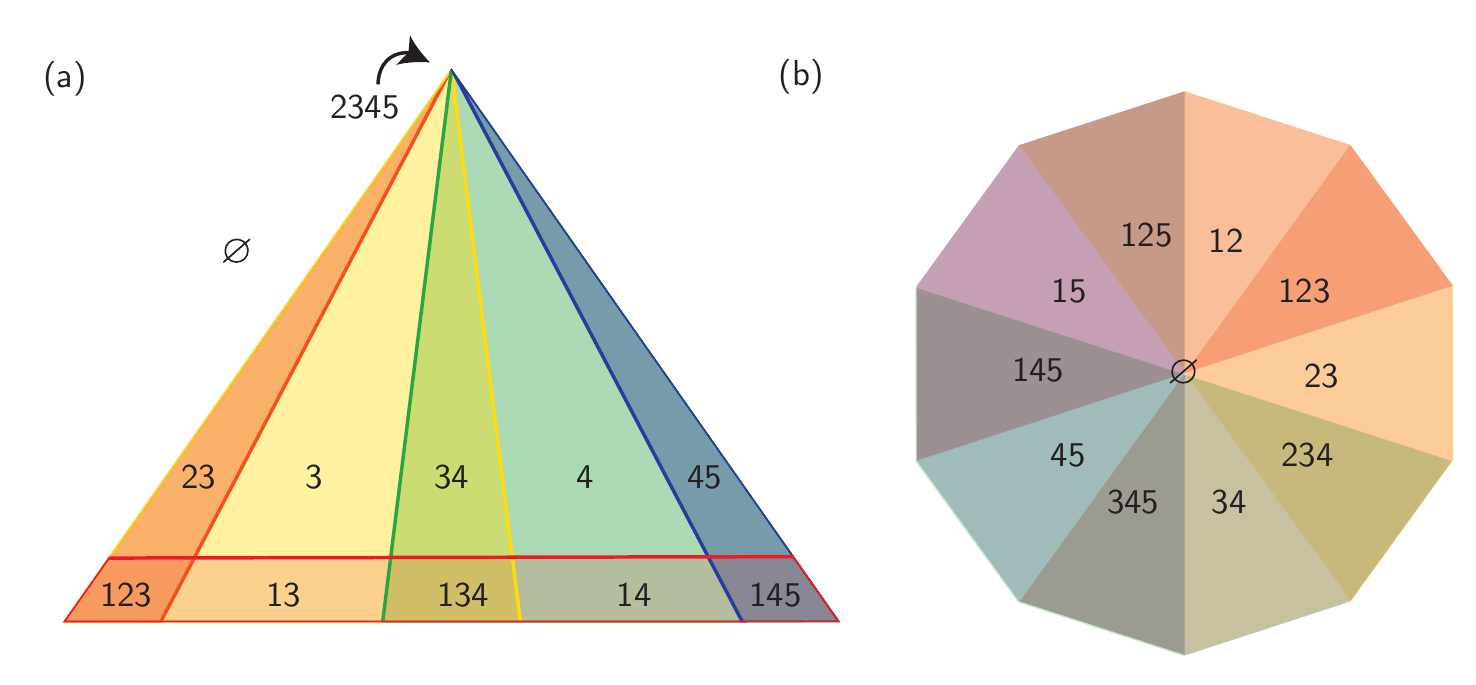}
\end{center}
\caption[Examples of non-convex codes]{(a) A closed convex realization of $\cC_1$. (b) An open convex realization of $\cD_1$. \label{fig:non_convex}}
\end{figure}

On the other hand, there are large families of codes which we can guarantee are convex. In particular, we say a code $\cC$ is max-intersection complete if whenever $\sigma\in \Delta(\cC)$ is an intersection of maximal codewords, $\sigma\in \cC$. By \cite{cruz2019open}, all max-intersection complete codes are both open and closed convex. Thus, the codes where convexity is an interesting problem are the codes which are not max-intersection complete, but have no local obstructions. 
We can partially order the set of codes $\cC$ with the same simplicial complex $\Delta(\cC) = \Delta$ by inclusion. Cruz et al. prove that open convexity is monotone increasing under this order \cite{cruz2019open}. On the other hand, closed convexity is \emph{not} monotone increasing under this order \cite{gambacini2021non}. 

On up to four neurons, a code has no local obstructions if and only if it is max intersection complete--thus on up to four neurons, a code is convex if and only if it has no local obstructions. A classification of all codes on at most three neurons appears in \cite{curto2013neural}, while a classification of codes on four neurons appears in \cite{curto2017makes}. A complete classification of codes on five neurons appears in \cite{goldrup2020classification}. In particular, $\cC_1$ is the only code on five neurons which has no local obstructions, but is not open convex. In addition to $\cD_1$, there are two other five neurons which are closed, but not open convex:
\begin{align*}
\cC_6 &= \{{\bf 125}, {\bf 234}, {\bf 145}, {\bf 123}, 4, 23, 15, 12, \varnothing\}\\
\cC_{10} &= \{{\bf 134}, {\bf 245}, {\bf 234}, {\bf 135}, 12, 1, 5, 34, 13, 2, 24, \varnothing\}
\end{align*}
Codes with at most three maximal codewords are convex if and only if they have no local obstructions, by \cite{johnston2020neural}. 

By a stronger version of Theorem \ref{thm:no_local}  proved in \cite{chen2019neural}, if a code is convex, the simplicial complexes which occur as links of missing codewords must be collapsible, not just contractible. In fact, they must satisfy even stronger properties established in \cite{jeffs2021convex}. The complete classification of which simplicial complexes arise as links in convex codes remains open. 

\section{Embedding dimensions of open and closed convex codes}

Beyond determining whether or not a code is convex, we can characterize its minimal embedding dimension. It turns out that we get different answers asking this question for open, closed, and \emph{non-degenerate} convex codes. 

\begin{defn}
A collection $\cU = \{U_1, \ldots, U_n\}$ of open convex sets is \emph{non-degenerate} if the collection  of their closures $\cl(\cU) =  \{\cl(U_1), \ldots, \cl(U_n) \}$ has code $\code(\cU) = \code(\cl(\cU))$. Likewise, a collection of closed sets $\cU = \{U_1, \ldots, U_n\}$ is non-degenerate if the collection of their interiors  $\Int(\cU) =  \{\Int(U_1), \ldots, \Int(U_n) \}$  has $\code(\cU) = \code(\Int(\cU))$. 
\end{defn} 

The open, closed, and non-degenerate embedding dimensions of a code are defined as follows: 

\begin{defn}
Let $\cC$ be a neural code. Then the open, closed, and non-degenerate embedding dimensions of $\cC$ are defined, respectively, as 
\begin{align*}
\odim(\cC) &= \min\{d \mid \cC \mbox{ has an open realziation in } \R^d\}\\
\cdim(\cC) &= \min\{d \mid \cC \mbox{ has an closed realziation in } \R^d\}\\
\ndim(\cC) &= \min\{d \mid \cC \mbox{ has an non-degenerate realziation in } \R^d\}\\
\end{align*}
\end{defn}

First, we note that if $\cC$ is convex, then $\odim(\cC) \leq |\cC| -1$, since we can obtain a lower-dimensional realization of $\cC$ by intersecting our realization of $\cC$ with the affine hull of a set of points $\{p_\sigma\}_{\sigma \in \cC}$, where $\{p_\sigma\}$ is taken to be in the atom of $\sigma$. This result holds for $\cdim(\cC)$ and $\ndim(\cC)$ as well. A slightly better bound exists when $\cC$ is max-intersection complete: by Theorem  1.2 of \cite{cruz2019open}, 
$\odim(\cC), \cdim(\cC) \leq \max\{2, k-1\}$ where $k$ is the number of maximal codewords of $\cC$.  If $\cC$ is intersection complete, then $ \cdim(\cC) \leq \min\{2d+1, n-1\}$, where  $d$ is the dimension of $\Delta(\cC)$. 

These bounds allow for the possibility that the minimal embedding dimension is exponential in the number of neurons, even for max-intersection complete codes. In fact, this can occur, at least for open embedding dimension: Jeffs \cite{jeffs2022embedding} gives an infinite family of intersection-complete codes $\cE_n$, such that $\odim(\cE_n)$ grows as fast as $\frac{2^{n-1}}{n}$. These codes have closed embedding dimension at most $n$, since they are intersection complete. There are no know examples of closed-convex codes on $n$ neurons such that $n <\cdim(\cC) <\infty$.  

As the previous example demonstrates, open and closed embedding dimension can be wildly different. In some cases, we can still control the relationship between $\odim(\cC), \cdim(\cC)$, and $\ndim(\cC)$.  For instance, if any one of $\odim(\cC), \cdim(\cC)$,or $ \ndim(\cC)$ is equal to 1, then the other two embedding dimensions must be 1. If $\cC$ is a simplicial complex, then $\odim(\cC) = \cdim(\cC)$, and if $\cC$ is intersection complete, then $\cdim(\cC) \leq \odim(\cC)$. Other than this, the only constraint on the open, closed, and non-degenerate embedding dimensions of a code is the clear constraint that the non-degenerate embedding dimension must be at least the maximum of the closed and open embedding dimensions. That is, any triple $(a, b, c)$ such that $a, b \leq c$ and $2\leq a, b, c $, there is a code $\cC$ with $\odim(\cC) = a, \cdim(\cC) = b$, and $\ndim(\cC) = c$ \cite{jeffs2021open}.

\section{Morphisms of neural codes}
Across all of mathematics, objects make more sense when we can relate them to one another. Combinatorial codes are no exception. In order to describe the relationships between codes, Jeffs introduces \emph{neural code morphisms} in \cite{jeffs2020morphisms}. These maps between codes allow us to relate the convexity of one code to the convexity of another, or even the convexity of one class of codes to another class of codes. In particular, they allow give us a formal way to talk about whether non-convex code is novel, rather than a trivial modification of a previous code. 

Morphisms of neural codes are defined in terms of \emph{trunks}. 

\begin{defn}
Let $\cC$ be a neural code. The trunk of $\sigma \in \cC$ is the set 
\begin{align*}
\trunk_{\cC}(\sigma) := \{\tau\in \cC\mid \sigma \subseteq \tau\}
\end{align*}
A subset of $\cC$ is a trunk if it is empty, or if it is equal to $\trunk_{\cC}(\sigma)$ for some $\sigma \subseteq [n]$. 
\end{defn}

\begin{defn}
Let $\cC$, $\cD$ be combinatorial neural codes. A map $f: \cC \to \cD$ is a \emph{morphism} of neural codes if the preimage of every trunk of $\cD$ is a trunk of $\cC$. Two codes $\cC$ and $\cD$ are isomorphic if there is a bijective code morphism $f:\cC \to \cD$
whose inverse is also a code morphism.
\end{defn}

Note that while this definition feels reminiscent of the definition of a continuous map between topological spaces, the trunks of a code need not form the open sets of a topology on $\cC$. In particular, the union of trunks need not be a trunk. 

A key property of code morphisms is that the preserve convexity.
\begin{thm}[Theorem 1.3, \cite{jeffs2020morphisms}]
If $\cC$ is a convex code, and $f: \cC \to \cD$ is a surjective map, then $\cD$ is a convex code.\label{thm:morph}
\end{thm}

This fact motivates Jeffs to define a partial order on codes such that convex codes are a down-set. 
If there is a sequence of codes $\cC = \cC_0,\cC_1,\dots,\cC_k=\cD$ such that each successive code is either the image of a morphism from or a trunk of the preceding code, we say $\cD$ is a \emph{minor} of $\cC$.
Codes are then quasi-ordered by setting $\cD \leq \cC$ if $\cD$ is a minor of $\cC$.
The poset of isomorphism classes of codes induced by this order is denoted $\pcode$.
We can then rephrase Theorem \ref{thm:morph} as the statement that the set of convex codes in $\pcode$ is downward closed. 

The proof of Theorem \ref{thm:morph} is constructive, allowing us to build a realization of $\cD$ out of the realization of $\cC$. To see this, we first use Proposition 2.11 of \cite{jeffs2020morphisms}, which says that all code morphisms take a certain form. 

\begin{prop}[Proposition 2.11, \cite{jeffs2020morphisms}]
\label{prop:trunks_morph}
Let $\cC$ be a neural code, $S = \{T_1, \ldots, T_m\}$ a finite collection of trunks of $\cC$. Define the function $f_S: \cC \to 2^{[m]}$ by
\begin{align*}
f_S(\sigma) = \{i \mid \sigma\in T_i\}.
\end{align*}
We say that $f_S$ is the \emph{morphism determined by the trunks} in $S$. The map $f_S$ is indeed a code morphism.  Further, every code morphism is of this form. 
\end{prop}

Jeffs uses this fact prove Theorem 1.3 constructively. In fact, the only property of convex sets his proof uses is that the intersection of convex sets is convex. Thus, in \cite{kunin2020oriented}, Kunin, Rosen and I generalize this result to intersection-closed families, of which  the family of open convex subsets of $\R^d$ is one example. In particular, this allows us to show that codes with no local obstructions form a down-set in $\pcode$. We include this argument here. 

A family $\cF$ of subsets of a topological space is called {\em intersection-closed} if it is closed under finite intersections and contains the empty set. 
%Examples of intersection-closed families include the open convex subsets of $\R^n$, the closed convex subsets of $\R^n$, the open or closed convex polyhedra, the compact subsets of $\R^n$, and the bounded subsets of $\R^n$. 
We say that a neural code $\cC$ is {\em $\cF$-realizable} if $\cC= \code(\cU, X)$ for some $\cU \subseteq \cF$ and $X\in \cF$.
For instance, a neural code is convex if and only if it is $\cF$-realizable for the set $\cF$ of convex open subsets of some $\R^d$.

\begin{lem} \label{lem:int_closed}
	For any intersection closed family $\cF$,
%	 the $\cF$-realizable codes form a down-set in $\pcode$. That is, 
	if $\cC$ is $\cF$-realizable and  $\cD \leq \cC$, then $\cD$ is $\cF$-realizable. 
\end{lem}

\begin{proof}
%Recall that $\cD \leq \cC$ if 
%%we can obtain $\cD$ from $\cC$ by taking surjective morphisms, or by taking trunks. Thus, we can reduce to two cases:
%$\cD = f(\cC)$ or $\cD = \trunk_{\cC} (\s)$ for some $\s \subseteq [n]$. 
 This closely follows the proof of Theorem 1.4 in \cite{jeffs2020morphisms},  since the only property of convex sets this proof uses is that the family of open convex subsets of $\R^d$ is closed under finite intersection.  . We repeat the details here. Let $\cC\subseteq 2^{[n]}$, $\cD\subseteq 2^{[m]}$. Since $\cD \leq \cC$, we have $\cD = f(\cC)$.  By Proposition \ref{prop:trunks_morph}, there are trunks  $T_1, \ldots, T_m$  in $\cC$ that define the morphism $f: \cC \to \cD$. Let $\{U_1, \ldots, U_n\} \subseteq \cF$ be an $\cF$-realization of $\cC$. 

If $T_j$ is nonempty, let $\sigma_j$ be the unique largest subset of $[n]$ such that $T_j = \trunk_\cC(\sigma_j)$. In particular, $\sigma_j$ will be the intersection of all elements of $T_j$. Then, for $j\in [m]$, define 
\begin{align*}
V_j = \begin{cases}
\emptyset & T_j = \emptyset\\
\bigcap_{i\in \sigma_j} U_i & T_j \neq \emptyset
\end{cases}
\end{align*}
Since $\cF$ is closed under finite intersection and contains the empty set, $V_j\in \cF$ for all $j\in [m]$. Thus, it suffices to show that the code $\cE$ that they realize is $\cD$. To see this, note that we can associate each point $p\in X$ to a codeword in $\cC$ or $\cE$ by $p\mapsto \{i\in [n] \mid p\in U_i\}$ and $p\mapsto\{j\in [m] \mid p\in V_j\}$. Then let $p\in X$ be arbitrary, and let $c$ and $e$ be the associated codewords in $\cC$ and $\cE$ respectively. Observe that by the definition of the $V_j$, we have that $c\in T_j$ if and only if $j\in e$. But this is equivalent to $e = f(c)$. Since $p$ was arbitrary and every codeword arises at some point, we conclude that $\cE = f(\cC) = \cD$, as desired. 

An example of this construction is illustrated in Figure  \ref{fig:int_closed}. 

%Next, we check the case $\cD = \trunk_{\cC} (\sigma)$.  In this case, let $\cC \subseteq 2^{[n]}$, $\sigma \subseteq [n]$, $\cU = \{U_1, \ldots, U_n\}$ be a $\cF$-realization of $\cC$. Then for $i\in [n]$, define $\cV = \{V_1, \ldots, V_n\}$ where 
%$$V_i = U_i \cap \left( \bigcap_{j\in \sigma} U_j\right),$$
%and $$Y = X \cap \left( \bigcap_{j\in \sigma} U_j\right).$$
%Then $\cD = \code(\cV, Y)$. To check this, as above, we can associate each point $p\in Y$ to a codeword by $p\mapsto\{j\in [n] \mid p\in V_j\}$. Since $Y =  X \cap \left( \bigcap_{j\in \sigma} U_j\right)$, each of these codewords will contain $\sigma$, and we will obtain every codeword of $\cC$ containing $\sigma$ in this way. 
\end{proof}

\begin{figure}[h]
\begin{center}
 \includegraphics[width = 5 in]{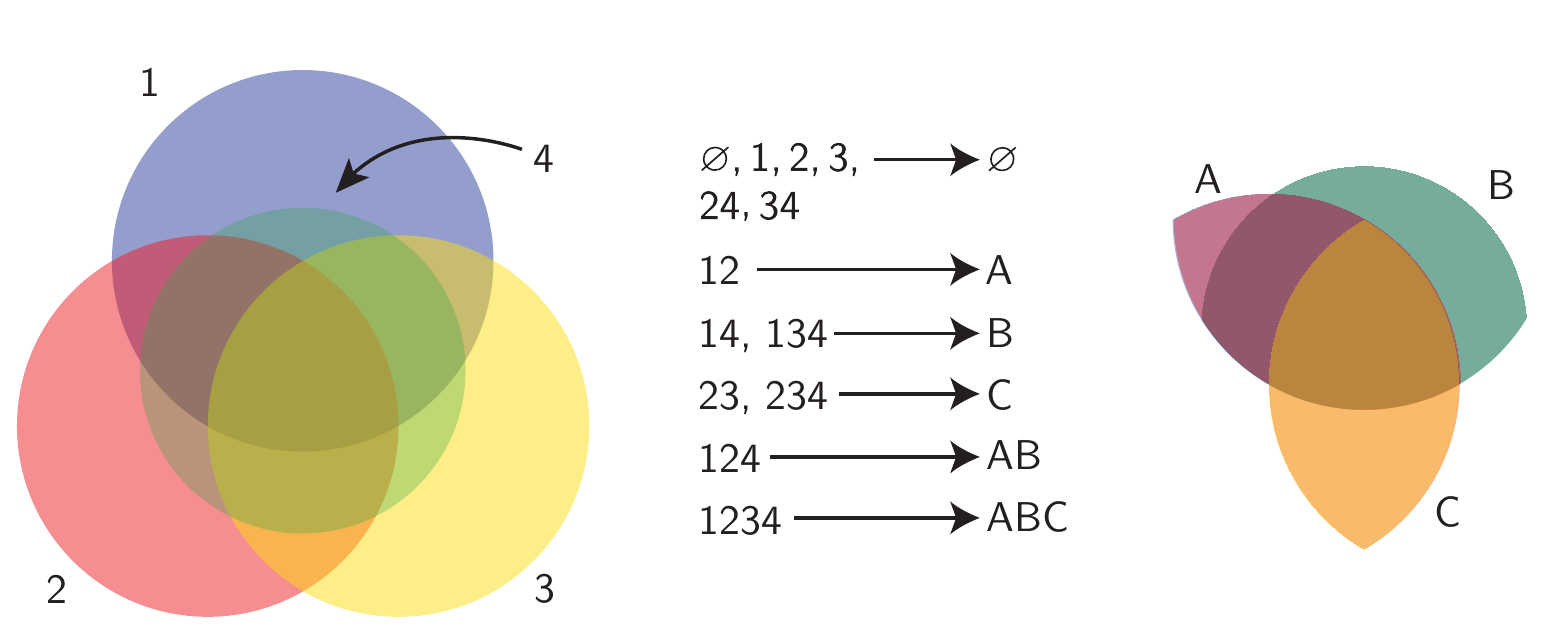}
 \end{center}
  \caption[Geometric interpretation of code morphisms.]{Example of construction from proof of Lemma \ref{lem:int_closed}.
  \label{fig:int_closed}
  }
\end{figure}

To prove Theorem \ref{thm:morph}, we apply Lemma \ref{lem:int_closed} to the intersection closed family of open or closed convex sets in $\R^d$.  We also show that codes with no local obstructions form a down-set in $\pcode$. %, which was stated as a conjecture in \cite{jeffs2019morphisms}.  
The only requirement to be an open set in some good cover is contractibility, and the family of contractible sets is not intersection-closed.
Instead, we consider the sets $U_1, \ldots, U_n$ in one particular good cover and their intersections as our intersection-closed family.

\begin{cor} 
The set of codes with no local obstructions is a down-set in $\pcode$. \label{thm:good_cover}
\end{cor}

\begin{proof}
	Let $\cC$ be a code with no local obstructions, $\cD \leq \cC$.
	By \cite[Theorem 3.13]{chen2019neural},  $\cC$ is a good cover code.
	Fix a good cover $\cU = \{U_1, \ldots, U_n\}$ realizing $\cC$.
%	The family $\cF_{\cU}$ is intersection-closed by \ref{lem:good_cover}.
	Let $\cF_\cU$ denote the family of sets obtained by arbitrary intersections of sets in $\cU$, together with the empty set.
	This family still forms a good cover.
%	By \ref{lem:int_closed}, the $\cF_\cU$-realizable codes form a down-set. Thus $\cD$ is $\cF_\cU$ realizable. Thus $\cD$ is a good cover code, and therefore has no local obstructions. 
	$\cD$ lies below $\cC$ and is therefore $\cF_\cU$-realizable by \ref{lem:int_closed}; it is therefore a good cover code and thus has no local obstructions.
\end{proof}

Using the structure provided by $\pcode$, Jeffs defines \emph{minimally non-convex codes}: a code $\cC$ is \emph{minimally non-convex} if it is not convex, but all codes $\cC'$ such that $\cC' \leq \cC$ are convex. Using this framework, Jeffs constructs a minimally non-convex code $\cC'$ on six neurons by taking images of the code $\cC_1$ from \cite{lienkaemper2017obstructions}.  

We can produce infinitely many distinct non-convex codes by (for instance) adding new neurons to the code $\cC$ from \cite{lienkaemper2017obstructions}. A more interesting question is whether or not there are infinitely many minimal non-convex codes. In contrast to case for graphs, where any minor-closed family has finitely many excluded minors by the famous result of Robertson and Seymour, there are infinitely many minimally non-convex codes. In particular, there is a minimal non-convex code corresponding to each non-collapsible simplicial complex by Proposition 5.8 of \cite{jeffs2020morphisms}. 

A more explicit infinite family of minimally non-convex codes generalizing $\cC_1$ is given in \cite{jeffs2019sunflowers}. This family depends on the following fact about sunflowers of convex open sets, proved in the same paper. Say $\cU = \{U_1, \ldots, U_n\}$ is a \emph{sunflower} if for any $i, j$, $U_i \cap U_j = \bigcap_{i = 1}^n U_i$. The intersection $U = \bigcap_{i = 1}^n U_i$ is referred to as the \emph{center} of the sunflower. 

\begin{thm}[Theorem 1.1, \cite{jeffs2019sunflowers}]
Let $\cU = \{U_1, \ldots, U_n, U_{d+1}\}$ be a sunflower of $d+1$ convex opens sets in $\R^d$. Then any hyperplane which intersects each $U_i$ must also intersect the center $\bigcap_{i = 1}^n U_i$. 
\label{thm:sunflower}
\end{thm}

We can construct non-convex codes by forcing a hyperplane to intersect each of the $U_i$, but not the center. In Chapter \ref{chapter:order_forcing}, we provide an alternate family of minimally non-convex codes with no local obstructions generalizing $\cC_1$ which uses only the $n = 3$ case of the sunflower theorem.

%% file: OrderForcing/OrderForcing.tex
%auto-ignore 

% !TEX root = ../YourName-Dissertation.tex

\chapter{Order Forcing in Neural Codes } \label{chapter:order_forcing}

This chapter is adapted from the paper ``Order Forcing in Neural Codes", which is joint work with Amzi Jeffs and Nora Youngs \cite{jeffs2020order}, and is included here with their permission. 

\section{Introduction}

The arguments that the codes $\cC_1$ and $\cD_1$ in Section \ref{sec:cvx_or_not} are not convex share a common feature: at some step, they derive a contradiction by showing that any convex realization would have a straight line path passing through a certain sequence of atoms in a certain order. 
In this chapter, we introduce a combinatorial concept that we call \emph{order-forcing} which allows us to generalize these arguments.  Order-forcing provides an elementary connection between the combinatorics of a code and the geometric arrangement of atoms in its open or closed realizations. In particular, our main result is the following: %we will prove that if a sequence $\sigma_1, \sigma_2, \ldots, \sigma_n$ of codewords of a code $\cC$ is order forced, then in any open or closed convex realization of $\cC$, any straight line from the atom of $\sigma_1$ to the atom of $\sigma_n$ must pass though the atoms of $\sigma_1, \sigma_2, \ldots, \sigma_k$, in this order. 

\begin{ithm}\label{thm:order-forcing}
Let $\sigma_1, \sigma_2, \ldots, \sigma_k$ be an order-forced sequence of codewords in a code $\cC\subseteq 2^{[n]}$. Let $\cU = \{U_1, \ldots ,U_n\}$ be a (closed or open) convex realization of $\cC$, and let $x\in A^{\cU}_{\sigma_1}$, and $y\in A^{\cU}_{\sigma_k}$. Then the line segment $\overline{xy}$ must pass through the atoms of $\sigma_1, \sigma_2, \ldots, \sigma_k$, in this order.
\end{ithm} 

\setcounter{ithm}{0}

We will use order-forcing to construct new examples of non-convex codes. 
In Section \ref{sec:order-forcing}, we define an order-forced sequence of codewords and prove Theorem \ref{thm:order-forcing}.  In Section~\ref{sec:of_examples} we use order-forcing to describe new good cover codes that are not convex:
\begin{itemize}
    \item We generalize a minimally non-convex code from \cite{jeffs2020morphisms} based on sunflowers of convex open sets to a family $\{\cL_n\mid n\ge 0\}$ of minimally non-convex  good cover codes (Proposition \ref{prop:stretch_sun}).
    \item We build a good cover code $\cR$ that is neither open nor closed convex by using order-forcing to guarantee that two disjoint sets would cross one another in a convex realization of $\cR$ (Proposition \ref{prop:rowboat}). This example is notable in that it relies on the order that codewords appear along line segments, rather than just certain codewords being ``between" one another.
    \item We build a good cover code $\cT$ that is neither open nor closed convex by using order-forcing to guarantee a non-convex ``twisting" in every realization of $\cT$ (Proposition \ref{prop:braid}).

\end{itemize}

These examples illustrate the utility of order-forcing. The codes $\cT$ and $\cR$ have the advantage that they require only elementary geometric techniques (i.e. order-forcing) to analyze. The codes $\cT$ and $\cR$ are also the first ``natural" examples we know of of good cover codes which are not produced by combining a non-open-convex code and a non-closed-convex code.

\section{Order-Forcing \label{sec:order-forcing}}

When we constrain ourselves to realizations that use only open (or only closed) convex regions $U_i$, we restrict not only which codes may be realized, but how regions in these realizations can be arranged. %{\color{red}\sout{As we move along straight-line paths through such a realization, the possible transitions from one atom to the next are restricted.}}
In particular, when we move along continuous paths through realizations composed of open (or closed) sets $U_i$, we are limited in the transitions we can make from one atom to the next.  

\begin{lem}\label{lem:edges} Suppose $\cC\subseteq 2^{[n]}$ is a neural code with a good cover realization $\cU$, and let $\sigma$ and $\tau$ be codewords of $\cC$. If there are points $p_\sigma\in A^\cU_\sigma$ and $p_\tau \in A^\cU_\tau$ and a continuous path from $p_\sigma$ to $p_\tau$ that is contained in $A^\cU_\sigma\cup A^\cU_\tau$ (that is, if the atoms are adjacent in the realization), then either $\sigma\subseteq\tau$ or $\tau\subseteq\sigma$.
\end{lem}

\begin{proof} Let $P$ be the image of a continuous path from $p_\sigma$ to $p_\tau$ with $P\subseteq A^\cU_\sigma\cup A^\cU_\tau$. Suppose for contradiction that $\sigma\not\subset\tau$ and $\tau\not\subset \sigma$. Then there exist elements $i\in \sigma\setminus \tau$ and $j\in \tau\setminus\sigma$. But then $U_i\cap P$ and $U_j\cap P$ partition $P$ (every point in $P$ is in exactly one of $A_\sigma^\cU$ or $A_\tau^\cU$ and thus in exactly one of $U_i$ or $U_j$). Since our good cover consists of sets that are all open or all closed, the sets $U_i\cap P$ and $U_j\cap P$ are both relatively open or both relatively closed in $P$. This contradicts the fact that $P$ is connected, so $\sigma\subseteq\tau$ or $\tau\subseteq \sigma$ as desired.
\end{proof}

Thus, as we move continuously through any good cover realization of a code, we are moving along  edges in the following graph $G_\cC$:

\begin{defn}\label{def:codewordgraph}
Let $\cC\subseteq 2^{[n]}$ be a neural code. The \emph{codeword containment graph} of $\cC$ is the graph $G_{\cC}$ whose vertices are codewords of $\cC$, with edges $\{\sigma,\tau\}$ when either $\sigma\subsetneq\tau$ or $\tau\subsetneq \sigma$. Note that this graph is also defined in \cite{chan2020nondegenerate}. 
\end{defn}

\begin{ex}\label{ex:manypaths}
Consider the code $\cC = \{{\bf 1235}, {\bf 1245}, {\bf 1256}, 125, 13, 14, 15, \emptyset\}$. The graph $G_\cC$ for this code is shown in Figure \ref{fig:manypaths}.

\begin{center}
\begin{figure}[h!]
\begin{center}
\includegraphics[width=2in]{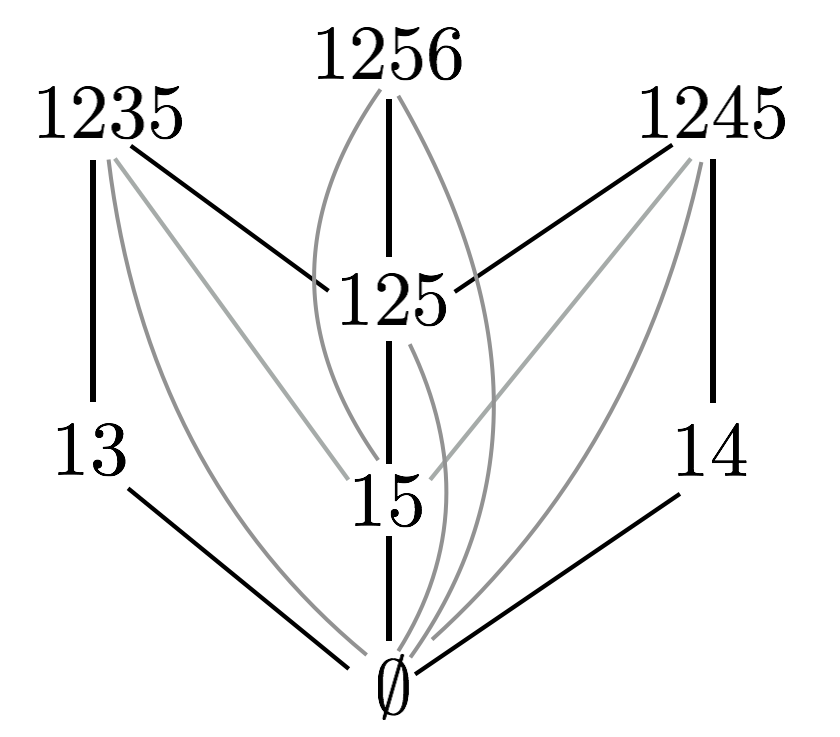}
\caption[The codeword containment graph.]{The codeword containment graph $G_\cC$ for the code in Example \ref{ex:manypaths}. \label{fig:manypaths}}
\end{center}
\end{figure} 
\end{center}
\end{ex}

Lemma \ref{lem:edges} implies that any continuous path from one codeword region to another codeword region in an open (or closed) realization of the code $\cC$ must correspond to a walk in the graph $G_\cC$. For straight-line paths within a convex realization, this walk must respect convexity, a property we call {\it feasibility} (see Lemma \ref{lem:convexfeasible}). 
    
\begin{defn}\label{def:feasible}
Let $\cC$ be a neural code and $G_\cC$ its codeword containment graph. A $\sigma,\tau$ walk  $\sigma=v_1,v_2,...,v_k=\tau$ in $G_\cC$ is called \emph{feasible} if $v_i\cap v_j\subseteq v_m$ for all  $1\leq i<m<j\leq k$.
\end{defn}

In general, if there exists a feasible $\sigma,\tau$ walk, then by removing portions of the walk between repeated vertices, we can obtain a feasible $\sigma, \tau$ path. 
 This does not, however, mean that there is a corresponding straight line path in the realization which would follow precisely this sequence of codewords. For example, one could form a closed realization of the code $\{\mathbf{1},\mathbf{2},\mathbf{3},\emptyset\}$ 
 in which $U_2$ is a hyperplane, and $U_1$ and $U_3$ are contained in its positive and negative side respectively. 
 Then any straight line from the atom of $1$ to the atom of $3$ must pass through the atom of 2, but the path $1, \emptyset, 3$ is a feasible path in $G_{\cC}$ regardless.

\begin{ex}[Example \ref{ex:manypaths} continued] Consider the codewords $\sigma = 13$ and $\tau = 14$ in $\cC$ from the code in Example \ref{ex:manypaths}. There are many $\sigma,\tau$ walks; however, not all are feasible. For example, the walk $13, 1235, 15, 1245, 14$ would not be feasible; however, the walk $13, 1235, 125, 1275, 125, 1245$ is a feasible $\sigma,\tau$ walk. This walk contains the feasible path $13, 1235, 125, 1245, 14$, which in this case is the unique  feasible $\sigma,\tau$ path. 

\end{ex}

\begin{lem}\label{lem:convexfeasible} Suppose $\cC\subseteq 2^{[n]}$ is a neural code with a convex realization $\cU = \{U_1,\ldots, U_n\}$, and let $\sigma$ and $\tau$ be codewords of $\cC$. If there are points $p_\sigma\in A^\cU_\sigma$ and $p_\tau \in A^\cU_\tau$, then the sequence of atoms along the line $\overline{p_\sigma p_\tau}$ forms a feasible  walk in $G_\cC$.
\end{lem}

\begin{proof}   Select points $p_\sigma\in A^\cU_\sigma$ and $p_\tau\in A^\cU_\tau$, and let the sequence of atoms along the line be given by $\sigma=\tau_1,\tau_2,...,\tau_\ell = \tau$. 
By Lemma \ref{lem:edges}, if we cross directly from $A_{\tau_i}^\cU$ to $A_{\tau_{i+1}}^\cU$ along the path $\overline{xy}$, then either $\tau_i\subseteq \tau_{i+1}$ or $\tau_{i+1}\subseteq \tau_i$, thus $(\tau_i, \tau_{i+1})$ is an edge of $G_\cC$. 
Thus, $\tau_1, \ldots, \tau_\ell$ is a walk in $G_\cC$. 
To check feasibility, we need to show that for all $i \leq j \leq k$,  $\tau_i \cap\tau_k  \subseteq \tau_j$. 
We can choose points $x_i, x_j$, and $x_j$ in this order along $\overline{p_\sigma p_\tau}$ such that $x_i \in A_{\tau_i}^\cU$, $x_j \in A_{\tau_j}^\cU$, $x_k \in A_{\tau_k}^\cU$. 
Since $\{U_1, \ldots, U_n\}$ is a convex realization and intersections of convex sets are convex, $U_{\tau_i\cap \tau_k}$ is a convex set.
By the definition of a convex set, the line segment $\overline{x_i x_k}$ is contained in $U_{\tau_i\cap \tau_k}$. 
Thus, $x_j\in U_{\tau_i\cap \tau_k}$, so $\tau_i \cap \tau_k \subseteq \tau_j$. Thus, $\tau_1, \ldots, \tau_\ell$ is a feasible walk. 

\end{proof}

The idea of feasibility gives us a new tool for finding possible obstructions to convexity. In any convex realization of a code, straight line paths between points in the same set $U_i$ must correspond to feasible walks in the graph, and so codes where feasible walks are rare or nonexistent can force us into contradictions. To that end, we define a few particular restrictions we will encounter.

\begin{defn}\label{def:forcedbetween} Let $\cC$ be a neural code and $G_\cC$ its codeword containment graph. We say a vertex $v$ of $G_\cC$ is \emph{forced between} vertices $\sigma$ and $\tau$ if every feasible $\sigma,\tau$ path passes through $v$. 
\end{defn}

\begin{ex}[Example \ref{ex:manypaths} continued]
In the codeword containment graph $G_\cC$, we see that $1245$ is forced between $14$ and $15$. There are many possible feasible paths from $14$ to $15$ (for example (14, 1245, 15) or (14, 1245, 125, 15) or (14, 1245,125, 1235, 15), but all these paths must use $1245$. 
\end{ex}

In cases where there are multiple codewords forced between two vertices of our graph, we often find that these vertices are also forced into a particular order, a situation we call {\it order-forcing}.

\begin{defn} Let $\cC$ be a neural code and $G_\cC$ its corresponding graph. A sequence of codewords $\sigma_1,...,\sigma_k$ is \emph{order-forced} if every feasible $\sigma_1,\sigma_k$ path contains these codewords as a subsequence. 
\end{defn}

%{\color{blue}[I was able to show that if $\sigma_1....\sigma_k$ is order-forced then everything between them is forced between them; however, showing that you have order-forced between other elements keeps bringing me back to the stupid ``existence of feasible paths" argument. However, the following definition may be the one we really want to use:]}

%I suspect this may be the same as saying there's a unique feasible $\sigma_1,\sigma_k$ path which uses these vertices in this sequence - i.e., we may be able to complicate things into walks by taking unnecessary detours, but there's only one true path. So philosophical!

%[Note that having a unique feasible path is not the same thing as having a unique feasible walk; the latter condition implies you couldn't take a detour even if you wanted to. An example may be helpful here.]

%{\color{blue} Still not clear that we need this definition/fact?

\begin{defn}\label{def:orderforced} Let $\cC$ be a neural code and  $G_\cC$ its corresponding graph.  A feasible path $\sigma_1,...,\sigma_k$ is \emph{strongly order-forced} if $\sigma_1,...,\sigma_k$ is the unique feasible $\sigma_1,\sigma_k$ walk in  $G_\cC$. 
\end{defn} 

%{\color{cyan}Notice that if $\sigma_1, \ldots, \sigma_n$ is strongly order-forced, then any subsequence of $\sigma_1, \ldots, \sigma_n$ is also strongly order-forced, thus if $i<j<l$ then $\sigma_j$ is forced between $\sigma_i$ and $\sigma_l$. }

%{\color{blue} Lemma? If the sequence $\sigma_1,...,\sigma_k$ is strongly order-forced, then any sub-path is also strongly order-forced and so we can say if $i<j<l$ then $\sigma_j$ is forced between $\sigma_i$ and $\sigma_l$? }

\begin{ex}
Consider the code $$\cC = \{{\bf 2456},{\bf 123},{\bf 145},{\bf 437}, {\bf467},45,46,47,1,2,3,\emptyset\}$$

In this code, the sequence $145,45,2456,46,467,47,473$ is strongly order-forced.  In order to have a path in $G_\cC$ from $145$ to $437$ which is feasible, we can certainly only use codewords which contain $4$. If we restrict to the portion of  $G_\cC$ which contains $4$, we have that this subgraph is a path with endpoints $145$ and $437$. Thus, there is a unique path from $145$ to $437$, and we can check that this path is feasible. 

\end{ex}

\begin{thm}\label{thm:order-forcing}
Let $\sigma_1, \sigma_2, \ldots, \sigma_k$ be an order-forced sequence of codewords in a code $\cC\subseteq 2^{[n]}$. Let $\cU = \{U_1, \ldots ,U_n\}$ be a (closed or open) convex realization of $\cC$, and let $x\in A^{\cU}_{\sigma_1}$, and $y\in A^{\cU}_{\sigma_k}$. Then the line segment $\overline{xy}$ must pass through the atoms of $\sigma_1, \sigma_2, \ldots, \sigma_k$, in this order.
\end{thm} 

\begin{proof}
Let $\sigma_1, \ldots, \sigma_k$ be an order-forced sequence in a code $\cC.$ Let $\cU = \{U_1, \ldots, U_n\}$ be a (closed or open) convex realization of $\cC$. 
Let $x\in A_{\sigma_1}^{\cU}$ and $y\in A_{\sigma_k}^\cU,$ and let $\overline{xy}$ be the line segment from $x$ to $y$. Let $\tau_1 = \sigma_1, \tau_2, \ldots, \tau_\ell = \sigma_k$ be the sequence of atoms along $\overline{xy}$.
By Lemma \ref{lem:convexfeasible}, we have that $\tau_1, \ldots, \tau_\ell$ is a feasible  walk from $\sigma_1$ to $\sigma_k$ in $G_\cC.$
Since every feasible walk from $\sigma_1$ to $\sigma_k$ contains a feasible path from $\sigma_1$ to $\sigma_k$, and every feasible path from $\sigma_1$ to $\sigma_k$ contains  $\sigma_1, \sigma_2, \ldots, \sigma_k$ as a subsequence, this suffices to prove Theorem \ref{thm:order-forcing}.

%By Lemma \ref{lem:edges}, if we cross directly from $A_{\tau_i}^\cU$ to $A_{\tau_{i+1}}^\cU$ along the path $\overline{xy}$, then either $\tau_i\subseteq \tau_{i+1}$ or $\tau_{i+1}\subseteq \tau_i$, thus $(\tau_i, \tau_{i+1})$ is an edge of $G_\cC$. 
%Thus, $\tau_1, \ldots, \tau_\ell$ is a path in $G_\cC$. 
%To check feasibility, we need to show that for all $i \leq j \leq k$,  $\tau_i \cap\tau_k  \subseteq \tau_j$. 
%We can choose points $x_i, x_j$, and $x_j$ in this order along $\overline{xy}$ such that $x_i \in A_{\tau_i}^\cU$, $x_j \in A_{\tau_j}^\cU$, $x_k \in A_{\tau_k}^\cU$. 
%Since $U_1, \ldots, U_n$ is a convex realization and intersections of convex sets are convex, $U_{\tau_i\cap \tau_k}$ is a convex set.
%By the definition of a convex set, the line segment $\overline{x_i x_k}$ is contained in $U_{\tau_i\cap \tau_k}$. 
%Thus, $x_j\in U_{\tau_i\cap \tau_k}$, so $\tau_i \cap \tau_k \subseteq \tau_j$. Thus, $\tau_1, \ldots, \tau_\ell$ is a feasible path. 
\end{proof}

 The situation where a codeword $v$ is forced between $\sigma$ and $\tau$ is a special case of order-forcing, and in this case we obtain the following result.  Once we know that a sequence is order-forced in a code $\cC$, we are often able to obtain several instances of order-forcing.
\begin{cor}\label{cor:forcedbetween}  %{\color{red}\sout{(of Theorem \ref{thm:order-forcing}):}} %{\color{blue}\sout{If $\sigma_1,...,\sigma_k$ is order-forced, then for all $i$ with $1<i<k$, we have $\sigma_i$  forced between $\sigma_1$ and $\sigma_k$.} 
 Let $\cC\subseteq 2^{[n]}$ and suppose $\cU = \{U_1, \ldots ,U_n\}$ is a (closed or open) convex realization of a code $\cC$. If  $v\in \cC$ is forced between $\sigma$ and $\tau$, then for any $x\in A^{\cU}_{\sigma}$, and $y\in A^{\cU}_{\tau}$, the line segment $\overline{xy}$ must pass through the atom of $v$. 
\end{cor}

%{\color{blue} [This Corollary is what I have ended up using most.]}

%{\color{red} Need  a long list of codewords, order-forced}

In the following example, we illustrate the value of these ideas by showing a proof that a relatively small code is open-convex but not closed-convex.

 \begin{ex}\label{ex:wheel} We revisit the code $\cD_1$ from Chapter \ref{sec:codes}, which was first introduced in  \cite{cruz2019open}.  This code is open convex, but not closed convex. This example (in particular the proof that it has no closed convex realization) is an instance of order-forcing, though it was not described by that name in \cite{cruz2019open}. A slightly smaller example of a similar code which is closed convex, but not open convex appears as code C15 in \cite{goldrup2020classification}. In this example, we give a proof of this result which resembles the proof in \cite[Lemma 2.9]{cruz2019open}, but is written to make the use of order-forcing explicit. 

 The code $$\cD_1 =\{{\bf 123},{\bf 234},{\bf 345},{\bf 145},{\bf 125},12,23,34,45,15,\emptyset\}$$
has an open convex realization, but does not have a closed convex realization.\\

We have already provided an open-convex realization of $\cD_1$ in  Figure \ref{fig:non_convex} (b). 
%\begin{figure}[ht!]
%\begin{center}
%
%\includegraphics[width=1.5in]{OrderForcing/Figures/5circle2.png}
%\caption{A realization of $\cC$ using open sets. The regions $U_1,...,U_5$ are open half discs. }\label{fig:5circle}
%\end{center}
%\end{figure} 

To show that no closed convex realization may exist, we proceed by contradiction.  Suppose that for some $d\geq 1$, there exists a closed convex realization $\{U_1,\ldots, U_5\}$ of $\cC$ in $\R^d$.  Select points $p_{125}\in U_{125}$ and $p_{345}\in U_{345}$. %{\color{blue} may need to note $U_\sigma=V_\sigma$ and is closed for max CWs} 
Since both points are within the convex set $U_5$, the line segment $L_1$ from $p_{125}$ to $p_{345}$ is contained within $U_5$. Thus, it cannot pass through $U_{123}$. Note that $U_{123}$ is a closed set, as $123$ is a maximal codeword. Pick a point $p_{123}\in U_{123}$ which minimizes the distance to the set $L_1$; this is possible as these sets are disjoint and $L_1$ is compact. 

Now, consider the line segment $L_2$ from $p_{125}$ to $p_{123}$; note that $L_2\subset U_{12}$. In this code, $12$ is forced between $125$ and $123$, so by Corollary \ref{cor:forcedbetween} there exists a point $p_{12}$, between $p_{125}$ and $p_{123}$ along this line, which is in $A^{\mathcal U}_{12}$. Likewise, if we consider the line segment $L_3$ from $p_{123}$ to $p_{345}$, the order-forced sequence $123,23,234,34,345$ implies there is a point $p_{234}\in U_{234}$ on $L_3$ which is between $p_{123}$ and $p_{345}$. 

Finally, consider the line segment $L_4$ between $p_{12}$ and $p_{234}$. $L_4$ must pass through $U_{123}$ somewhere between these points because $123$ is forced between $12$ and ${234}$. Select a point $q_{123}$ on this line and within $U_{123}$; then, $q_{123}$ will be closer to $L_1$ than $p_{123}$, a contradiction. %{\color{blue} probably would be useful to include a picture of the relevant triangle here as they do in [1].}
\end{ex}

\section{New Examples of Non-Convex Codes \label{sec:of_examples}}
In this section, we demonstrate the power of order-forcing by using order-forcing to construct a  new infinite family of minimally non-convex codes and two new non-convex codes. 

\subsection{Stretching sunflowers}\label{sec:stretch}

Early examples of good cover codes which are not convex come from the sunflower theorem, Theorem \ref{thm:sunflower}.  The $d = 2$ case of this theorem was used as a lemma to give the first example of a non-convex good cover code in \cite[Theorem 3.1]{lienkaemper2017obstructions}. 
In this section, we give a new infinite family of non-convex codes generalizing this code.
In order to produce further examples of non-convex codes, we need a notion of what it means for a new code to be genuinely different from an old one. 
For instance, it is easy to produce ``new" non-convex codes by relabeling neurons, or by adding more neurons in some trivial way.

In this subsection, we introduce a family of codes $\{\cL_n\mid n\ge 0\} $ which generalize $\cC_0$ to an infinite family of minimally non-convex codes. Geometrically, each of these codes is only a small modification of the code $\cC_0$, and has the same basic obstruction to convexity. This lies in contrast to  \cite[Theorem 4.2]{jeffs2019sunflowers}, which generalizes the non-convex code $\cC_0$ in  \cite[Theorem 5.10]{jeffs2020morphisms} to an infinite family of minimally non-convex codes by using higher-dimensional versions of the sunflower theorem.  Thus the family $\{\cL_n\mid n\ge 0\}$ demonstrates that the intuition that each minimally non-convex code should result from a ``new" obstruction to convexity does not hold.

\begin{center}
\begin{figure}[ht!]
\begin{center}
\includegraphics[width = 4 in]{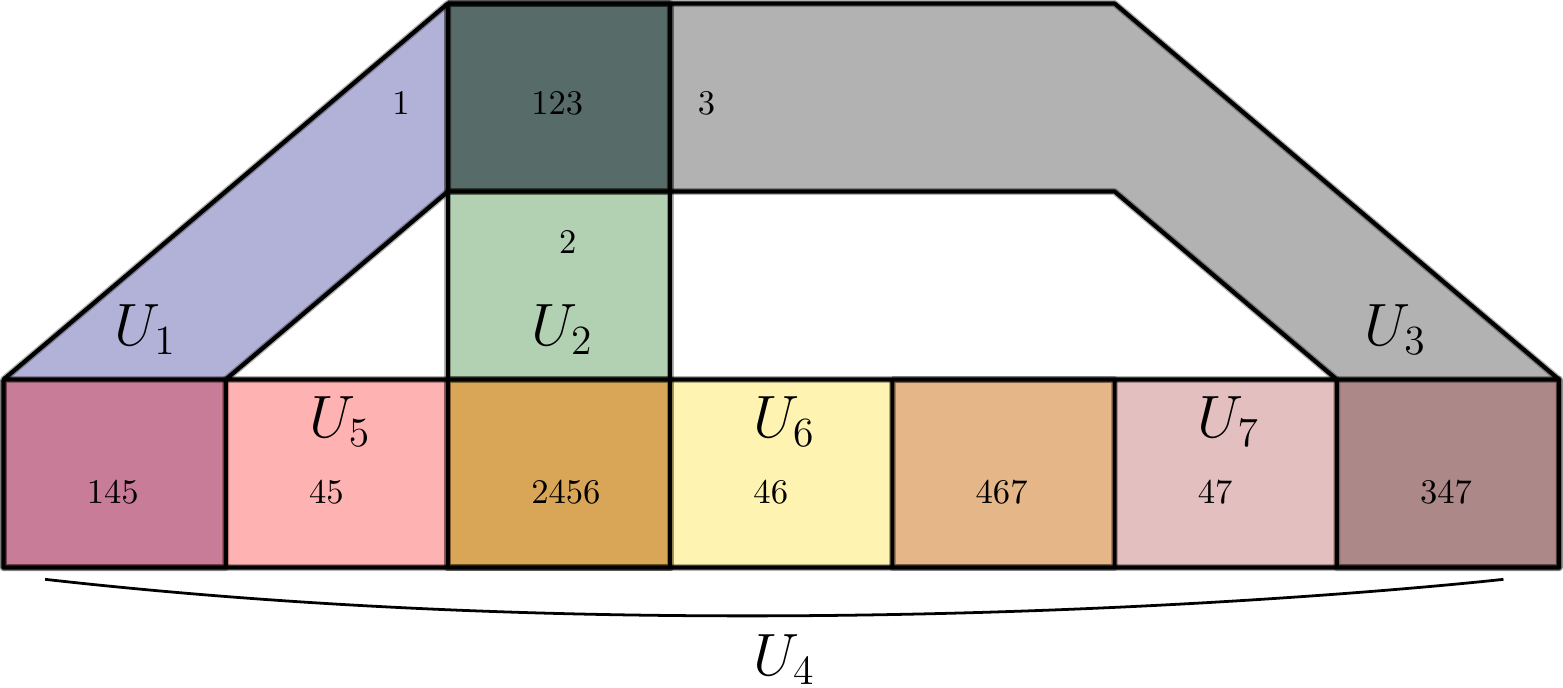}
\end{center}
\caption[Realization of the code $\cL_1$]{$\cL_1 = \{\mathbf{2456}, \mathbf{123}, \mathbf{145}, \mathbf{437}, \mathbf{467}, 45, 46, 47, 1, 2, 3, \emptyset\}$}
\label{fig:L1}
\end{figure} 
\end{center}

\begin{defn}For $n\geq 0$, define the code 
$$\cL_{n} = \{\emptyset, 1, 2, 3, \mathbf{123}, \mathbf{145}, 45, \mathbf{2456}, 46, \mathbf{467}, 47, \mathbf{478}, \ldots, 4(n+6),34(n+6)\}.$$
\end{defn}

% \begin{ex} 
% The code $\cL_1 = \{\mathbf{2456}, \mathbf{123}, \mathbf{145}, \mathbf{437}, \mathbf{467}, 45, 46, 47, 1, 2, 3, \emptyset\}$ is minimally non-convex. {\color{red}[This example is really short now, should we add a sentence or two of commentary and justification, or delete it? ---Amzi]} {\color{blue} I say commentary, because we have the figure and should reference it here. Maybe something like: ``A (non-convex) realization of this code is seen in Figure 3. Proposition 3.6 will use order-forcing to show that no convex realization of $L_1$ is possible."?  }
% \end{ex}

For instance, $\cL_1 = \{\mathbf{2456}, \mathbf{123}, \mathbf{145}, \mathbf{437}, \mathbf{467}, 45, 46, 47, 1, 2, 3, \emptyset\}$. A good cover realization of $\cL_1$ is given in Figure \ref{fig:L1}.

Notice below that $\cL_0$ is equal to $\cC_0$ under the permutation of the neurons $2 \lra 3$ and $4\lra 5$. 
Thus, the family $\cL_n$ generalizes $\cC_0$. 
Even though each $\cL_n$ is minimally non-convex, the non-convexity of $\cL_0$ directly implies the non-convexity of each $\cL_n$ for $n > 0$.

\begin{prop}\label{prop:stretch_sun}
For $n\ge 0$, the code $\cL_n$ is a good cover code, but is minimally non-convex.
\end{prop}

\begin{proof}
We first show that $\cL_n$ is non-convex by induction on $n$. 
The base case, that $\cL_0$ is non-convex, is proven by \cite[Theorem 5.10]{jeffs2020morphisms} since $\cL_0$ is permutation equivalent to the code $\cC_0$ in this paper. Now, we show that if $\cL_{n-1}$ is not convex, then neither is $\cL_n$.
We do this by proving the contrapositive: in any convex realization of $\cL_n$, we can merge the sets  $U_{n+5}$ and $U_{n+6}$ in a convex realization of $\cL_n$ to produce a convex realization of $\cL_{n-1}$.
That is, if $\{U_1, \ldots, U_{n+5}, U_{n+6}\}$ is a convex realization of $\cL_n$, then $\{V_1,\ldots, V_{n+5}\}$ is a convex realization of $\cL_{n-1}$ where $V_1 = U_1, \ldots, V_{n+4} = U_{n+4}, V_{n+5} = U_{n+5}\cup U_{n+6}$. 

This gives us two things to check. First, we must check that $\code(\{V_1, \ldots, V_{n+5}\}) = \cL_{n-1}$. If $\sigma$ is a codeword of $\cL_n = \code(\{U_1, \ldots, U_{n+6}\})$ which does not contain the neuron $n+6$, then $\sigma$ is still a codeword of $\code(\{V_1, \ldots, V_{n+5}\})$. The three codewords of $\cL_n$ which contain $n+6$ are $4(n+5)(n+6)$, $4(n+6)$, and $34(n+6)$. If we pick a point $p$ in the atom of  $4(n+5)(n+6)$ or $4(n+6)$ with respect to $U_1, \ldots, U_n$, it is now in the atom of $4(n+5)$. If we pick a point in the atom of $34(n+6)$ with respect to $U_1, \ldots, U_{n+6}$, it is in the atom of  $34(n+5)$ with respect to $\{V_1, \ldots, V_{n+6}\}$.

Next, we must check that $V_5 = U_{n+5} \cup U_{n+6}$ is convex.
That is, we must check that for each pair of points $x, y\in U_{n+5}\cup U_{n+6}$, the line segment from $x$ to $y$ is contained in $U_{n+5}\cup U_{n+6}$. 
Without loss of generality, let $x\in U_{n+5}\setminus U_{n+6}$, $y\in U_{n+6}\setminus U_{n+5}$. 
The point $x$ must be contained in the atom of $4(n+5)$ or $4(n+5)(n+4)$. (If $n = 1$,  $4(n+5)(n+4)$ replaces with $24(n+5)(n+4)$.) The point $y$ must be contained in the atom of $4(n+6)$ or $34(n+6)$. In all of these cases, the only feasible path from $x$ to $y$ in $G_{\cL_n}$ includes only codewords containing $n+5$ or $n+6$:
$$ 4(n+5) \lra 4(n+5)(n+6) \lra 4(n+6) \lra 34(n+6).$$ 
Thus the line segment from $x$ to $y$ is contained in $U_{n+6}\cup U_{n+6}$.  See Figure \ref{fig:sun_proof} for an illustration of this argument.
%Thus, for any $p\in U_145$, $q\in U_{34(n+6)}$, the line segment $\overline{pq}$ must pass through $U_{2456}$, and by convexity  $\overline{pq} \subseteq U_4$. 
%However, $U_1, U_2, U_3$ form a sunflower. Thus, since $\overline{pq}$ intersects with $U_1, U_2$, and $U_3$, it must intersect with $U_{123}$. However, this forces $U_4\cap U_1 \cap U_2 \cap U_3 \neq \emptyset$, a contradiction, since no codeword of $\cL$ contains $1234$. Thus $\cL_n$ is not convex. 

%Now, to show that $\cL_n$ is minimally nonconvex, we must show that the covered code $\cL_n^{(i)}$ is convex for $i\in[n+6]$. 
%For $i = 1, 2, 3$, we can construct a realization of $\cL_n^{(i)}$ in $\R^2$, since we no longer have a three-petal sunflower. 
%For $i \ge 4$, we can construct a realization of $\cL_n^{(i)}$ in $\R^3$, since by breaking order-forcing, we no longer force any line to pass though all petals of the sunflower. 

To show that $\cL_n$ is minimally nonconvex, we must show that all codes covered by $\cL_n$ in the poset $\pcode$ are convex. We give a proof of this in Appendix \ref{sec:constructions}, Construction \ref{const:minimal}. 
\end{proof}

Our proof uses ideas similar to the idea of a \emph{rigid structure} in Section 4 of \cite{chan2020nondegenerate}. In particular, our argument that $U_{n+5}\cup U_{n+6}$ must be convex is essentially an open-convex version of a rigid structure, which is a subset of neurons whose union must be convex in any closed-convex realization of a code. 

\begin{figure}[h]\label{fig:sun_proof}
    \centering
    \includegraphics[width = 4 in]{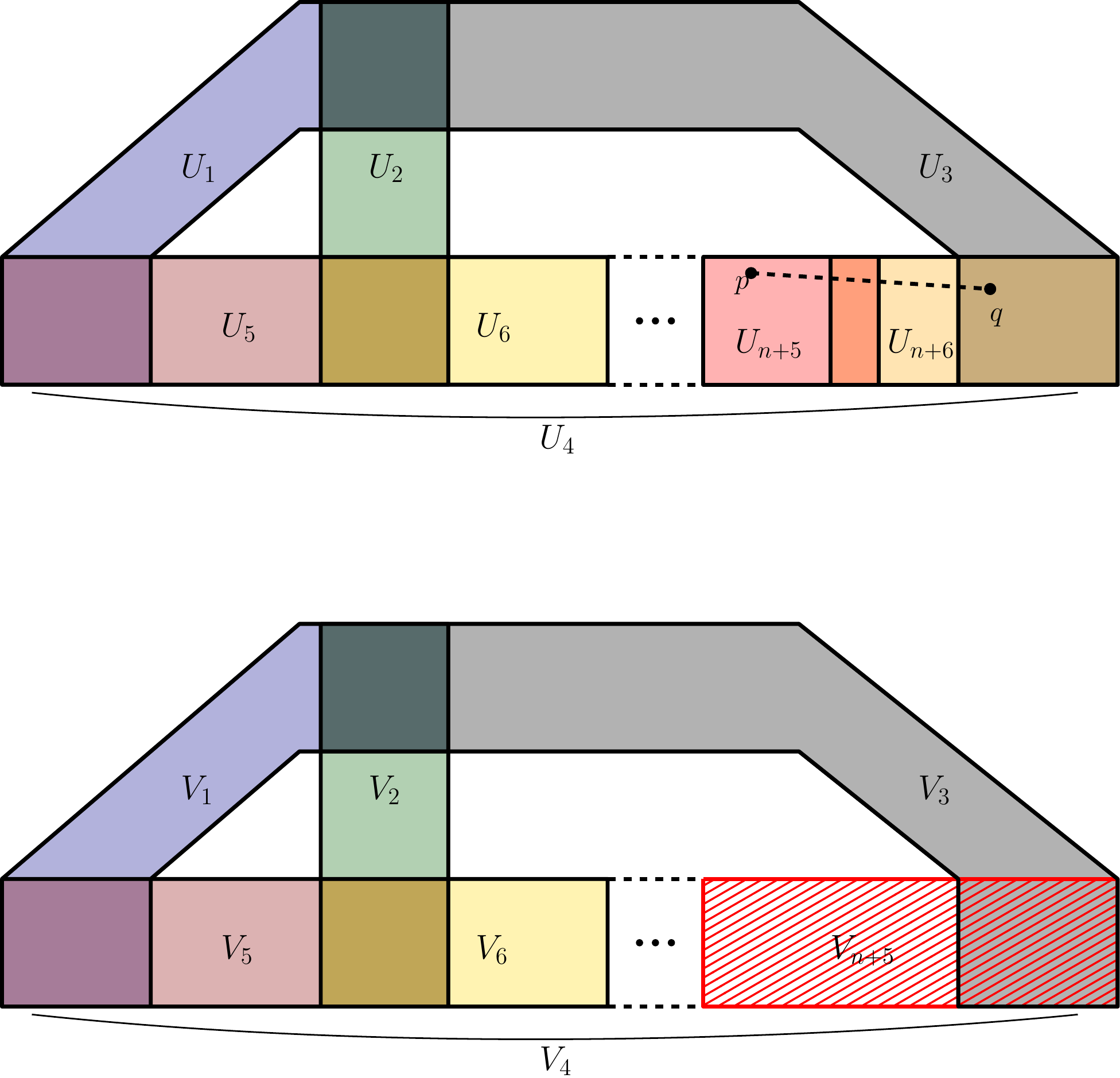}
    \caption[Proof sketch for Proposition \ref{prop:stretch_sun}. ]{A sketch of the proof of Proposition \ref{prop:stretch_sun}. Since the union of $U_{n+5}$ and $U_{n+6}$ is forced to be convex, we can use a realization of $\cL_{n}$ to construct a realization of $\cL_{n-1}$.}
    \label{fig:sun_proof}
\end{figure}

\subsection{Simple proofs of nonconvexity}\label{sec:braid}
In this section, we give two new examples of good cover codes which are neither open nor closed convex. The proofs that these codes are not convex depend only on order-forcing and elementary geometric arguments. Below, we use lowercase letters for neurons where it would be cumbersome to use only integers.

\begin{center}
\begin{figure}[ht!]
\begin{center}
\includegraphics[width = 4 in]{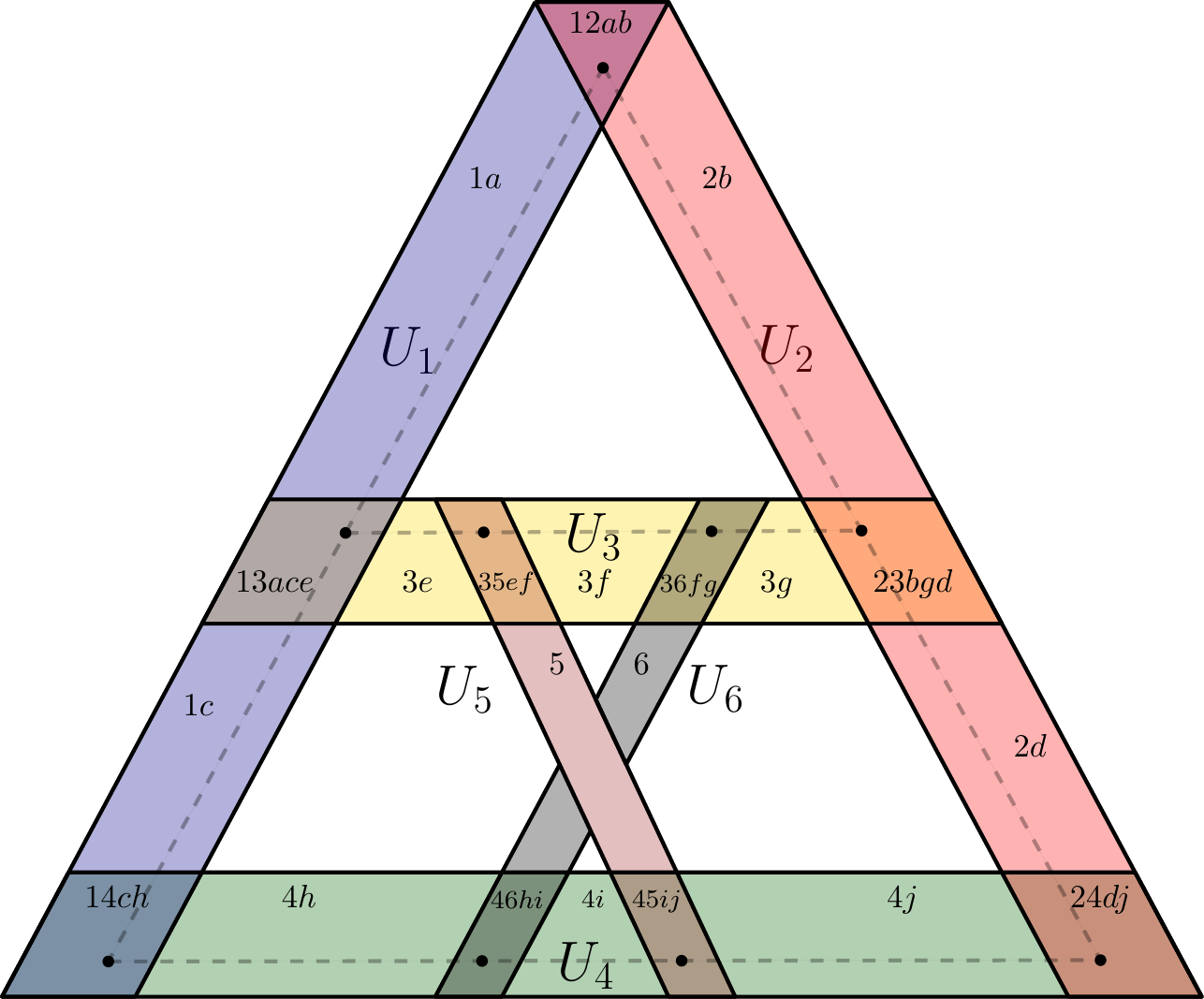}
\end{center}
\caption[A good-cover realization of the non-convex code 
$\cR$]{A good-cover realization of the non-convex code 
$\cR$ in $\R^3$. The open sets $U_a, U_b, U_c, U_d, U_e, U_f, U_g, U_h, U_i, U_j$ are not shown. Instead, maximal order-forced codewords are noted with vertices, and sets of order-forced vertices are indicated with dashed lines. \label{fig:rowboat}}
\end{figure} 
\end{center}

\begin{prop}\label{prop:rowboat}
The code \begin{align*}
\cR &= \{\mathbf{12ab}, \mathbf{13ace}, \mathbf{14ch}, \mathbf{23bgd},   \mathbf{24dj}, \mathbf{35ef},  \mathbf{36fg},   \mathbf{46hi},  \mathbf{45ij}, \\ &\quad\quad 1a,1c, 2b,2d, 3e,3f,3g,4i,4h, 4j, 5, 6, \emptyset\}
\end{align*} is a good cover code, but is neither open nor closed convex.
\end{prop}

\begin{proof}
We first show that if $\cR$ is convex, then it has a convex realization in the plane. We then show that it does not. Choose points $p_{12}\in A_{12ab}^\cU$, $p_{14}\in A_{14ch}^\cU$, and $p_{24}\in A_{24dj}^\cU$. 
We will use order-forcing to show that each atom of any realization of $\cR$ must have nonempty intersection with $A = \conv(p_{12}, p_{14}, p_{24})$, so that $\{U_i \cap A \mid i\in \{1, \ldots, 6, a, \ldots, h\}\}$ is a convex realization of $\cR$ in $\aff(p_{12}, p_{14}, p_{24})\cong \R^2$. 

First, notice the following order-forced sequences: 
\begin{enumerate} 
\item the only feasible path from $12ab$ to $14ch$ %in $G_{\tk_\cR(1)}$
is 
$$12ab \leftrightarrow 1a \leftrightarrow 13ace  \leftrightarrow  1c  \leftrightarrow 14ch$$
\item the  only feasible path from $12ab$ to $24dj$ %in $G_{\tk_\cR(2)}$
is 
$$12ab \lra 2b \lra 23bdg \lra 2d\lra 24dj$$
\item 
the only feasible path from $14ch$ to $24dj$% in $G_{\tk_\cR(4)}$ 
is 
$$14ch \lra 4h\lra 46hi\lra 45ij \lra 4j \lra 24dj$$
\item 
the only feasible path from $13ace$ to $23bgd$ %in $G_{\tk_\cR(3)}$ 
is 
$$13ace \lra 3e \lra 35ef \lra 3f \lra 36fg \lra 3g \lra 23bgd$$
\item the only feasible  path from $35ef$ to $45ij$ %in  $G_{\tk_\cR(5)}$
is 
$$35ef \lra 5 \lra 45ij$$
\item the only feasible path from $36fg$ to $46hi$ %in $G_{\tk_\cR(6)}$
is 
$$36fg \lra 6\lra 46hi.$$
\end{enumerate}
Now, by Theorem \ref{thm:order-forcing} and order-forcings (1), (2), and (4), the atoms corresponding to codewords $$\{12ab, 1a, 13ace, 1c,  14ch, 2b, 23bdg,  2d, 24dj, 4h, 46hi, 45ij, 4j \}$$ have nonempty intersection with $A$. Thus, we can pick $p_{13}\in A\cap A_{13ace}^\cU$, $p_{23}\in A \cap A_{23bdg}^\cU$, $p_{45} \in A \cap A_{45ij}^\cU$, and $p_{46}\in A\cap A_{46hi}^\cU$. Applying order-forcing (3) to $p_{13}$ and $p_{23}$, we deduce that the atoms corresponding to codewords $$\{3e , 35ef, 3f,  36fg, 3g\}$$ have nonempty intersection with $A$. Thus, we can pick $p_{35} \in A \cap A_{35ef}^\cU$ and $p_{36}\in A\cap A_{36fg}^\cU$. Finally, applying order-forcings (5) and (6), we deduce that the atoms corresponding to codewords $\{5, 6\}$ have nonempty intersection with $A$. This accounts for all codewords of $\cR$. 

Next, we show that $\cR$ cannot have a realization in the plane. Note that by applying an appropriate affine transformation, we can assume that $p_{12}$ is above $p_{14}$ and $p_{24}$, with $p_{14}$ to the left of $p_{24}$, as pictured in Figure \ref{fig:rowboat}. Then by order-forcings (3) and (4), $p_{35}$ must be to the left of $p_{36}$, while $p_{45}$ must be to the right of $p_{46}$. This implies the line segments $\overline{p_{35}p_{45}}$ and $\overline{p_{36}p_{46}}$ must intersect. But if $ p\in \overline{p_{35}p_{45}}\cap \overline{p_{36}p_{46}}$, then $p \in U_5 \cap U_6$. But, since $U_5$ and $U_6$ must be disjoint in any realization of $\cR$, this is not possible. 
\end{proof}

\begin{prop}\label{prop:braid} The code
$$\cT = \{\mathbf{14a}, \mathbf{15ab},  \mathbf{16bg},\mathbf{25c},  \mathbf{24cd}, \mathbf{26dgh},  \mathbf{34e},  \mathbf{35ef},  \mathbf{36fh},$$ $$1a, 1b,2c, 2d,3e,3f, 6g,6h, 4, 5,\emptyset\} $$  is a good cover code, but is neither closed nor open convex.
\end{prop} 

\begin{center}

\begin{figure}[ht!]
\begin{center}
\includegraphics[width = 3.5 in]{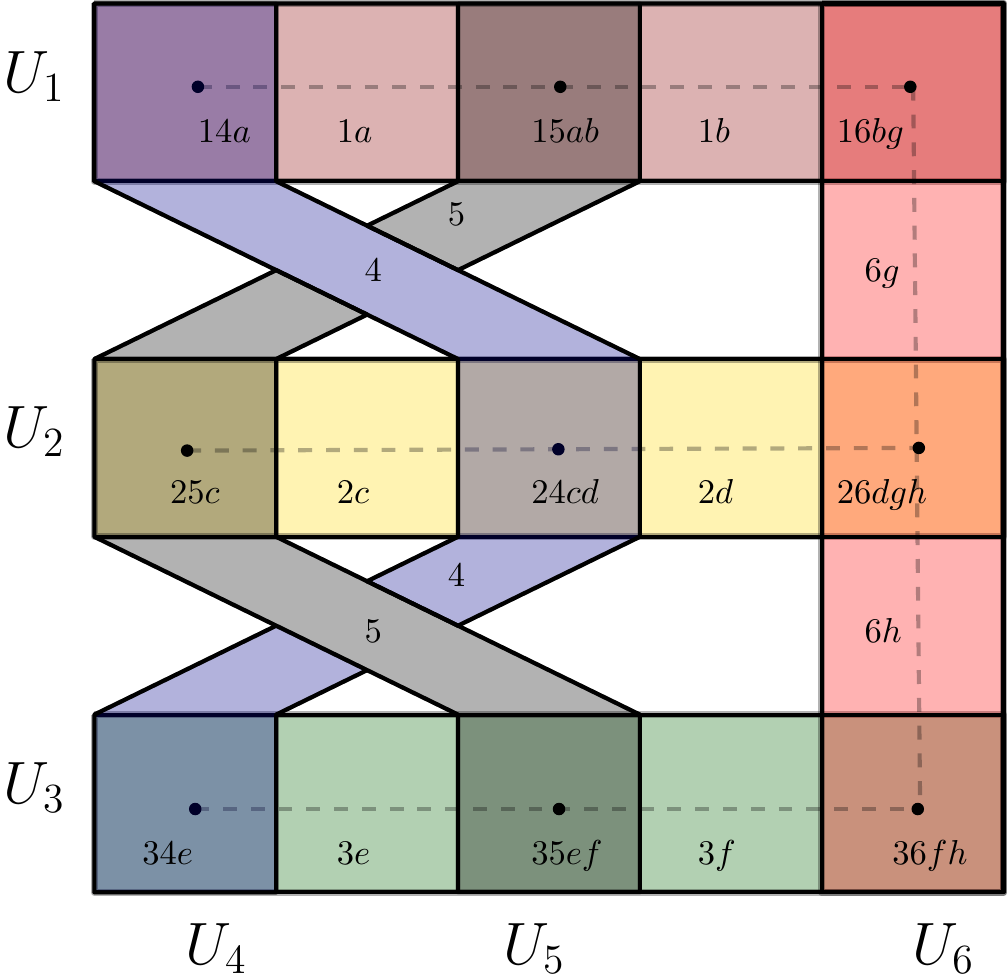}
\end{center}
\caption
{A good cover realization of the code $\cT$.
%\begin{align*}
%\protect
% \cT =  \{\mathbf{14a}, \mathbf{15ab},  \mathbf{16bg},\mathbf{25c},  \mathbf{24cd}, \mathbf{26dgh},  \mathbf{34e},  \mathbf{35ef},  \mathbf{36fh},\\ 1a, 1b,2c, 2d,3e,3f, 6g,6h, 4, 5,\emptyset\}  
% \end{align*}
%in $\protect\R^3$. The sets \(\protect U_1,\ldots, U_6\) are highlighted with various colors while \(\protect U_a,\ldots, U_h\) are not highlighted. 
\label{fig:twist}} 
\end{figure} 
\end{center}

\begin{proof}
Suppose to the contrary that $\mathcal T$ has a convex realization $\{U_1,\ldots, U_6, U_a,\ldots, U_h\}$. Since the sets $U_4$ and $U_5$ must be disjoint convex sets which are either both open or both closed,  there exists a hyperplane $H$ separating them. In particular, if $U_4$ and $U_5$ are both open, then by the open-set version of the hyperplane separation theorem there is a hyperplane strictly separates them.  That is, $H$ separates $\R^n$ into open half spaces $H^+$ and $H^-$  with $U_4\subseteq H^+$ and $U_5\subseteq H^-$. This also holds if $U_4$ and $U_5$ are both closed. In this case, then without loss of generality, we can choose both sets to be compact. Thus by the compact-set version of the separating hyperplane theorem, there exists a hyperplane $H$ strictly separating them.  We will use order-forcing to exhibit a line segment which crosses  $H$ twice, a contradiction. 

We show that the triples of codewords corresponding to marked points in Figure \ref{fig:twist} are  order-forced. 
More specifically,  we have that:
\begin{enumerate} 
\item the only feasible path from $14a$ to $16bg$ %in $G_{\tk_{\mathcal T}(1)}$
is 
$$14a \lra 1a \lra 15ab  \lra 1b  \lra 16bg$$
\item the  only feasible path from $25c$ to $26dgh$ %in $G_{\tk_{\mathcal T}(2)}$
is 
$$25c\lra 2c \lra 24cd \lra 2d\lra 26dgh$$
\item  the only feasible path from $34e$ to $36fh$ %in $G_{\tk_{\mathcal T}(3)}$ 
is 
$$34e \lra 3e \lra 35ef \lra 3f \lra 36fh $$

\item the only feasible path from $16bg$ to $36fh$ %in $G_{\tk_{\mathcal T}(6)}$
is 
$$16bg \lra 6g\lra 26dgh \lra 6h \lra 36fh.$$
\end{enumerate}

Choose points $p_{14}\in U_{14a}$, $p_{16}\in U_{16bg}$,$p_{25}\in U_{25c}$, $p_{34}\in U_{34e}$, and $p_{36}\in U_{36fh}$. Define line segments $L_1 = \overline{p_{14}p_{16}}$ and $L_3 = \overline{p_{34}p_{36}}$. Notice that by order-forcing (1) we may choose $p_{15}\in L_1\cap U_{15ab}$. Similarly by order-forcing (3) we may choose $p_{35}\in L_3\cap U_{35ef}$.

By ordering forcing (4) we may choose a point $p_{26}\in U_{26dgh}$ on the line segment $\overline{p_{16}p_{36}}$. Lastly, order-forcing (2) allows us to choose a point $p_{24}\in U_{24cd}$ on the line segment $L_2 = \overline{p_{25}p_{26}}$.

Since each of $L_1$ and $L_3$ can only cross $H$ once, the fact that $p_{14}$ and $p_{34}$ are contained in $U_4$, and thus in $H^+$ implies that the points $p_{16}$ and $p_{36}$ are contained in $H^-$. Likewise, the fact that $L_2$ crosses $H$ only once and $p_{25}$ is contained in $U_5$, and thus in $H^+$, implies that the point $p_{26}$ is contained in $H^-$. Thus, the line from $p_{16}$ to $p_{36}$ crosses $H$ twice, a contradiction. 
\end{proof}

Note that both of these codes can be used to generate infinite families of non-convex codes using the same trick we use to produce $\cL_{n}$ from $\cL_0$. %We have not shown that $\cT$ and $\cR$ are minimally non-convex, however, they do not lie above any previously known non-convex codes {\color{red}in $\pcode$}. {\color{red}[Is this easy to see? I'm not sure how I'd verify this in a non-tedious way.---Amzi]} {\color{cyan}[So I think $\cT$ and $\cR$ are both minimally nonconvex, but oof, I do not want do draw all of those pictures nicely on the computer. There might be a way to draw some of them and suggest how to do the rest? -Caitlin]}
The codes $\cT$ and $\cR$ do not lie above any previously known non-convex codes in $\pcode$, and in fact are minimally non-convex. This can be checked by exhaustive search of the codes that they cover in $\pcode$, as described in Definition \ref{def:covered}.

\section{Conclusion and Open Questions}\label{sec:conclusion}

Past work constructing non-convex codes has used notions that are similar to, but distinct from, order-forcing. For example, sunflower theorems such as \cite[Theorem 1.1]{jeffs2019sunflowers} and \cite[Theorem 1.11]{jeffs2022embedding} were used to show that the convex hull of points sampled from certain atoms in a convex realization must intersect another atom. Likewise, \cite{jeffs2021convex} used collapses of simplicial complexes to prove that in certain codes the convex hull of appropriately chosen points must intersect certain atoms.

Order-forcing brings a new perspective to this general approach: not only must certain atoms appear, but they must appear in a certain arrangement (i.e. in a particular order along a line segment). The order of points on a line may be generalized to higher dimensions by examining the ``order type" of a point configuration \cite{goodman1991complexity}. We thus ask the following.

\begin{question}\label{q:higherdimensional}
Does there exist a general result connecting the combinatorial structure of a code $\cC$ to the order type of points chosen from certain atoms in any convex realization of $\cC$? Can such a result be formulated so that the connections between convex codes and sunflower theorems \cite{jeffs2019sunflowers,jeffs2022embedding}, convex union representable complexes \cite{jeffs2021convex}, or oriented matroids \cite{kunin2020oriented} are special cases?
\end{question}

A cleanly formulated answer to Question \ref{q:higherdimensional} would allow us to create fundamentally new families of non-convex codes.

To connect the combinatorics of order-forcing with the geometry of convex realizations, we examined straight line segments between different atoms. One could try to replace convex realizations by good cover realizations, and straight lines by continuous paths, which leads to the following question.

\begin{question}
If $\cC$ is a good cover code, are there feasible paths between all pairs of codewords in $\cC$?
\end{question}

Our examples have used order-forcing to prove that codes are not convex. However, even if a code is convex, one might hope to use order-forcing to bound its open or closed embedding dimension. 

\begin{question}
Can one use order-forcing to provide new lower bounds on the open or closed embedding dimension of codes?
\end{question}

Morphisms and minors of codes have played a role in characterizing ``minimal" obstructions to convexity, contextualizing results, and systematizing the study of convex codes \cite{jeffs2020morphisms,jeffs2022embedding}. It would be interesting to phrase our results in this framework.

\begin{question}
How does order-forcing interact with code morphisms and minors? If $f:\cC\to \cD$ is a morphism, and $\sigma_1,\sigma_2,\ldots,\sigma_k$ is an order-forced sequence in $\cC$, under what conditions is $f(\sigma_1),f(\sigma_2),\ldots, f(\sigma_k)$ order-forced in $\cD$? Similarly, if $f$ is surjective and $\tau_1,\tau_2,\ldots, \tau_k$ is order-forced in $\cD$, when can we find $\sigma_1,\ldots, \sigma_k$ order-forced in $\cC$ with $f(\sigma_i) = \tau_i$ (i.e., when can we ``pull back" an order-forced sequence)?
\end{question}

Work in \cite{kunin2020oriented} used minors of codes to tie the study of convex codes to the study of oriented matroids, in particular showing that non-convex codes come in two types: those that are minors of non-representable oriented matroid codes, and those that are not minors of any oriented matroid code. Concretely, it would be useful to understand which of these classes our codes $\cT$ and $\cR$ fall into.

\begin{question}
Are the codes $\cT$ and $\cR$ from Section \ref{sec:of_examples} minors of oriented matroid codes?  
\end{question}

\section{Constructions of Various Realizations}\label{sec:constructions}
\begin{const}
\label{const:minimal}
In order to check that $\cL_n$ is minimally non-convex for all $n$, we must show that all codes covered by $\cL_n$ in $\pcode$ are convex. For this, we need the following characterization, from \cite{jeffs2019sunflowers}, of the covering relations in $\pcode$. 

\begin{defn}[Definition 3.9 of \cite{jeffs2019sunflowers}]\label{def:covered}
Let $\cC\subseteq 2^{[n]}$ be a code, let $i\in[n]$, and let $\sigma = [n]\setminus \{i\}$. Consider the morphism $f_i:\cC\to 2^{\sigma\cup\overline\sigma}$ defined by \[
f(c) = \begin{cases} c\cap\sigma & i\notin c,\\
c\cap\sigma \cup (\overline{c\cap\sigma}) & i\in c.\end{cases}
\]
The \emph{$i$-th covered code} of $\cC$ is the image of $\cC$ under $f_i$, and is denoted $\cC^{(i)}$.
\end{defn}

Importantly, if a code $\cD$ is covered by $\cC$ in $\pcode$, then $\cD$ must be one of the covered codes described above. Thus to prove that a non-convex code $\cC$ is minimally non-convex, it suffices to prove that all of its covered codes are convex. 

A useful geometric interpretation of covered codes is as follows. Suppose that $\cU = \{U_1,\ldots, U_n\}$ is a (possibly not convex) realization of $\cC$. Then we may obtain a realization of $\cC^{(i)}$ by deleting $U_i$ from $\cU$, and adding sets $U_{\overline j} = U_i\cap U_j$ for all $j\neq i$.

In some cases, there may be distinct neurons $j, k$ such that $U_{\bar j }= U_{\bar k}$. In this cases, one of the neurons $\bar j, \bar k$ is  redundant, and we can remove it from the code without discarding geometric information. More generally, a neuron $j$ is \emph{redundant}  to a set $\sigma \subseteq [n]\setminus \{j\}$ if $\tk_\cC(j) = \tk_\cC(\sigma)$, and a neuron is \emph{trivial} if it does not appear in any codeword \cite{jeffs2020morphisms}. A code is \emph{reduced} if it does not have any trivial or redundant neurons. Theorem 1.4 of \cite{jeffs2020morphisms} states that a code is always isomorphic to a reduced code. Thus, convexity of the reduced code is equivalent to convexity of the original code. Thus, we can ``clean up" $\cC^{(i)}$ by removing all trivial or redundant neurons. In what follows, we give realizations for reduced versions of all codes mentioned. 

Thus, to show that $\cL_n$ is minimal for all $n$, we need to construct realizations for each covered code $\cL_n^{(i)}$. In Figure \ref{fig:const123}, we construct realizations of $\cL_n^{(1)}$, $\cL_n^{(2)}$, and $\cL_n^{(3)}$ in $\R^2$. In Figure \ref{fig:const4}, we construct a realization of $\cL_n^{(4)}$ in $\R^3$. Finally, in Figure \ref{fig:const56}, we construct a convex realization of $\cL_n^{(7)}$ in $\R^3.$ An analogous process can be used to construct convex realizations of $\cL_n^{(8)}, \ldots,  \cL_n^{(n+6)}.$

\begin{figure}
    \centering
    \includegraphics[width = 4 in]{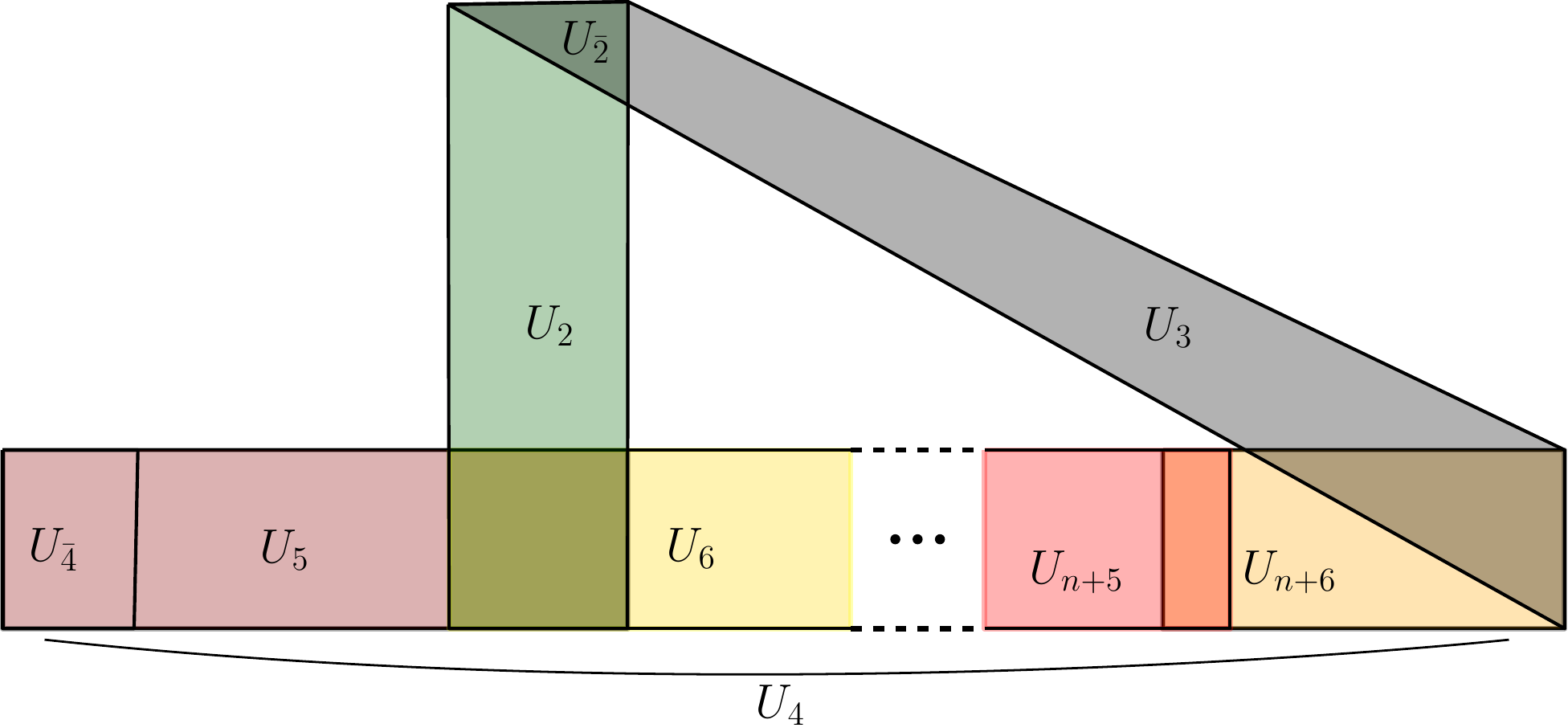}
    \includegraphics[width = 4 in]{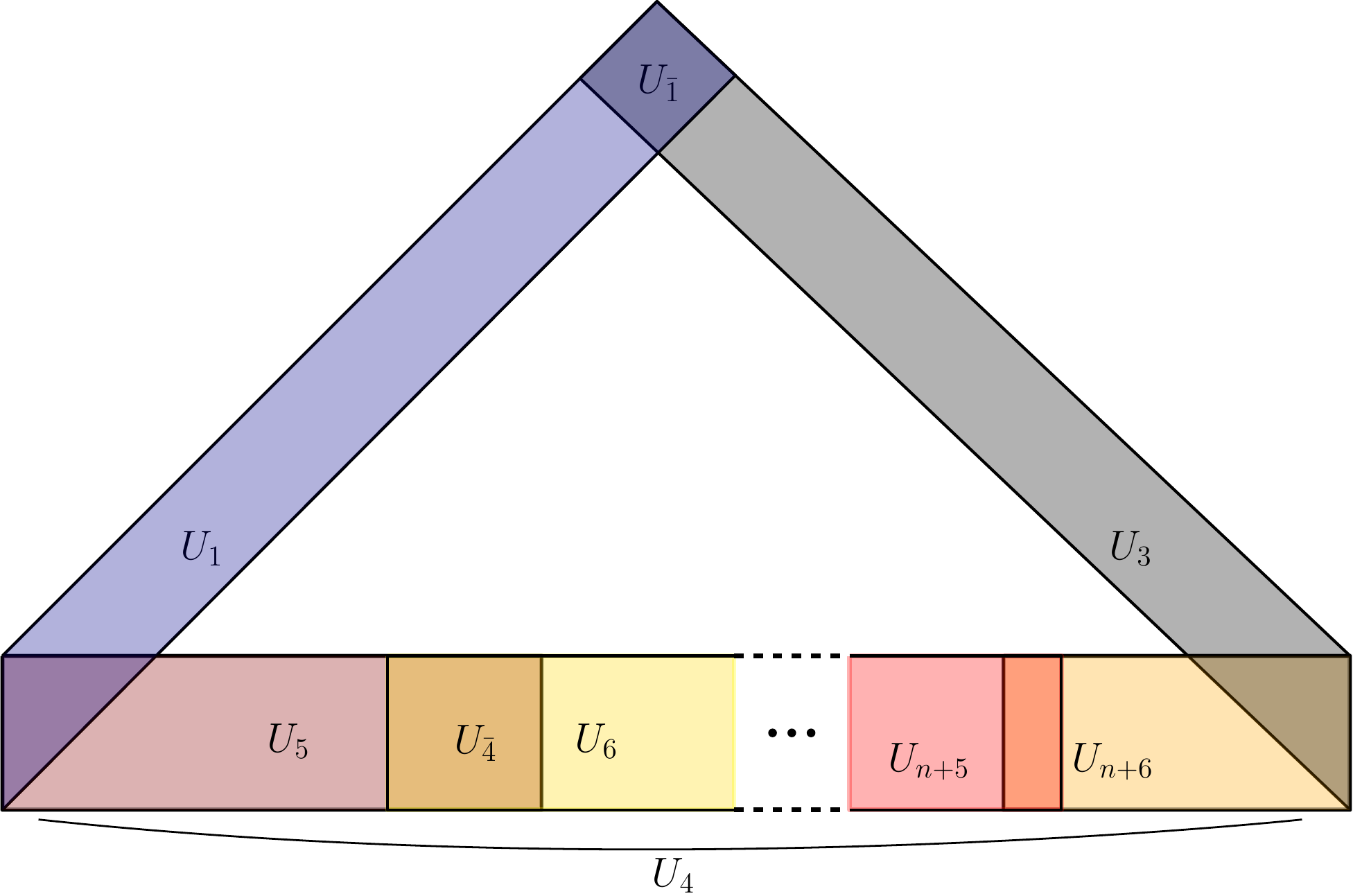}
    \includegraphics[width = 4 in]{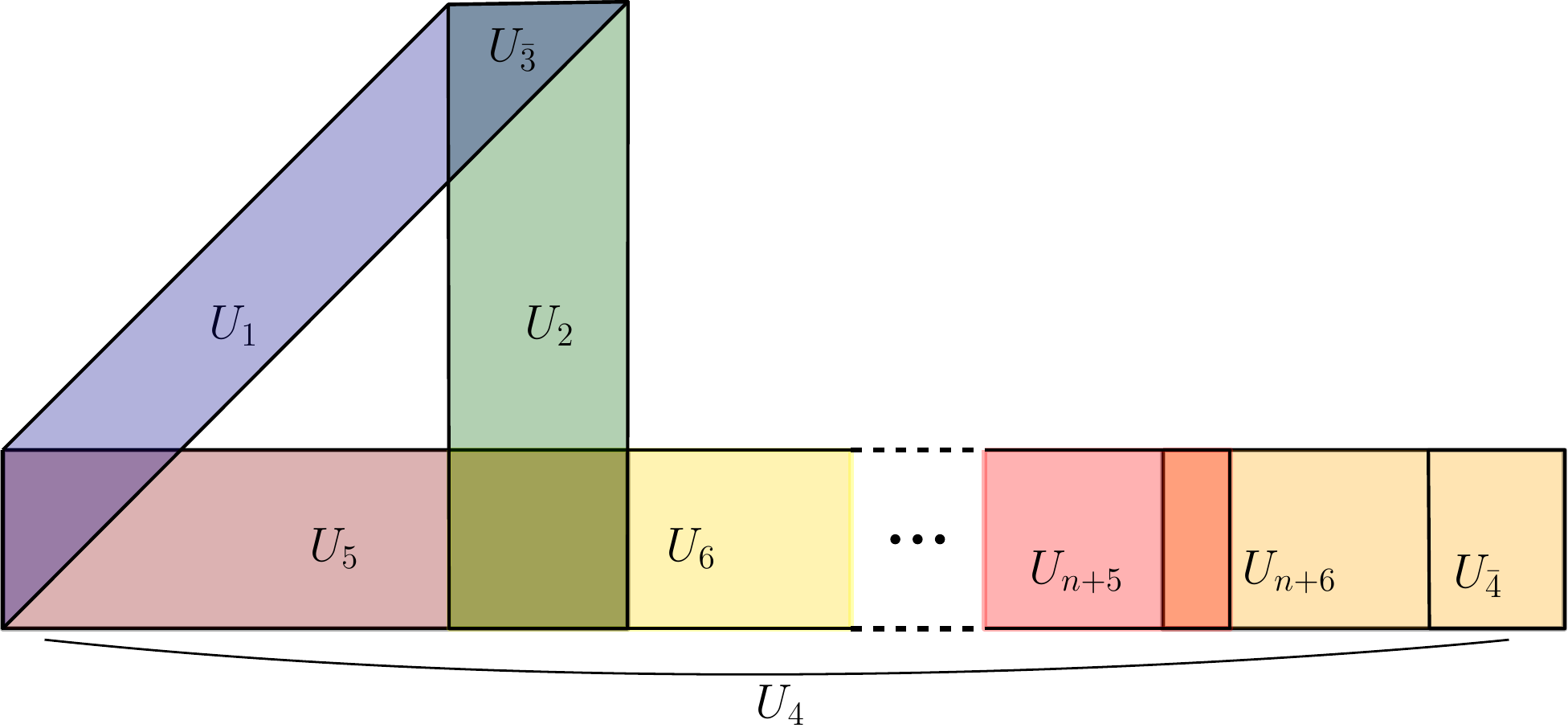}
    \caption[Convex realizations in $\R^2$ of the codes $\cL_{n}^{(1)}$, $\cL_{n}^{(2)}$, $\cL_{n}^{(3)}$.]{Convex realizations in $\R^2$ of the codes $\cL_{n}^{(1)}$, $\cL_{n}^{(2)}$, $\cL_{n}^{(3)}$.}
    \label{fig:const123}
\end{figure}

\begin{figure}
    \centering
    \includegraphics[width = 4 in]{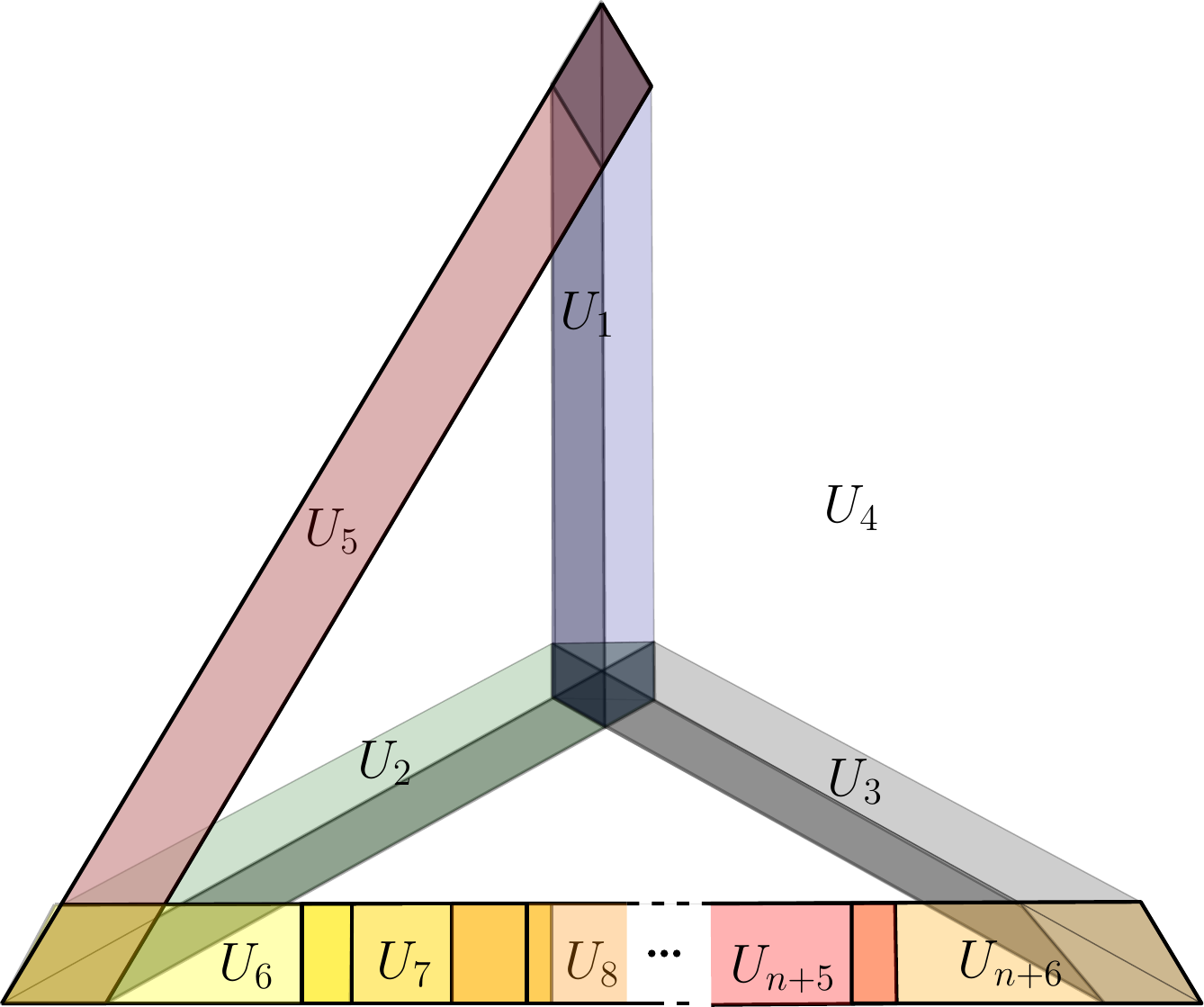}
    \caption[A convex realization in $\R^3$ of the code $\cL_n^{(7)}$.]{A convex realization in $\R^3$ of the code $\cL_n^{(7)}$. An analogous construction can be used to construct convex realizations for $\cL_n^{(5)}$, 
    $\cL_n^{(6)}$, and $\cL_{n}^{(8)}, \ldots, \cL_{n}^{(n+6)}$ in $\R^3$. }
    \label{fig:const4}
\end{figure}

\begin{figure}
    \centering
    \includegraphics[width = 4 in]{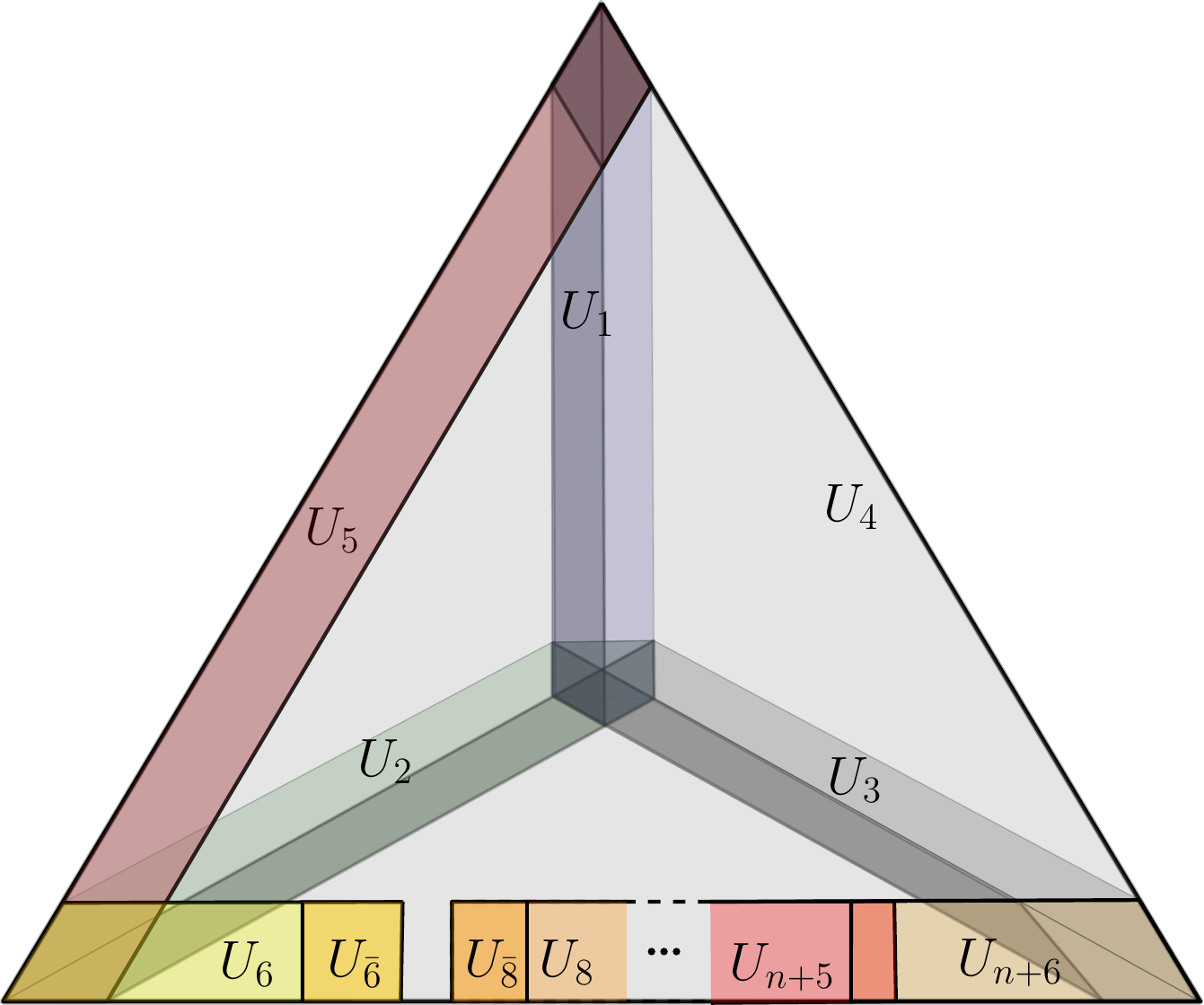}
    \caption[A convex realization in $\R^3$ of the code $\cL_n^{(7)}$.]{A convex realization in $\R^3$ of the code $\cL_n^{(7)}$. An analogous construction can be used to construct convex realizations for $\cL_n^{(5)}$, 
    $\cL_n^{(6)}$, and $\cL_{n}^{(8)}, \ldots, \cL_{n}^{(n+6)}$ in $\R^3$. }
    \label{fig:const56}
\end{figure}

\end{const}

%% file: MatroidsCodes/MatroidsCodes.tex
%auto-ignore 

% !TEX root = ../YourName-Dissertation.tex

\chapter{Oriented Matroids and Convex Neural Codes } \label{chapter:matroids_codes}

%Oriented matroids are like hyperplane arrangements, but worse. 
%%%%
This chapter is adapted from the paper ``Oriented Matroids and Convex Neural Codes", which is joint work with Alexander Kunin and Zvi Rosen \cite{kunin2020oriented}, and is included here with their permission. 

\section{Introduction}

A convex neural code records the same information about the intersection pattern of a family of convex sets as a representable oriented matroid records about a hyperplane arrangement. For instance, the hyperplane arrangement in Figure \ref{F:OMandCodeExample} gives rise to the covectors illustrated in Figure \ref{F:OMandCodeExample}(a), while the codewords of the associated combinatorial code of the positive half-spaces are shown in panel (b).
%Figure \ref{F:OMandCodeExample}(a) illustrates the collection of covectors arising from a hyperplane arrangement, while the codewords of the associated cover are shown in panel (b).
%The set of un-barred parts of covectors is precisely the code of the cover.  
In fact, as we noted in Section \ref{sec:oriented_matroid_intro}, the set of covectors of the oriented matroid of a hyperplane arrangement is the code of the positive and negative half-spaces. 
Thus, we can consider (representable) oriented matroids as a special case of convex neural codes. 
Because the study of oriented matroids long precedes the study of convex neural codes, making the connection between oriented matroids and convex codes explicit will allow us to leverage results about oriented matroids to prove new theorems about convex codes. 

\begin{figure}[ht!]
\begin{center}
 \includegraphics[width = 5 in]{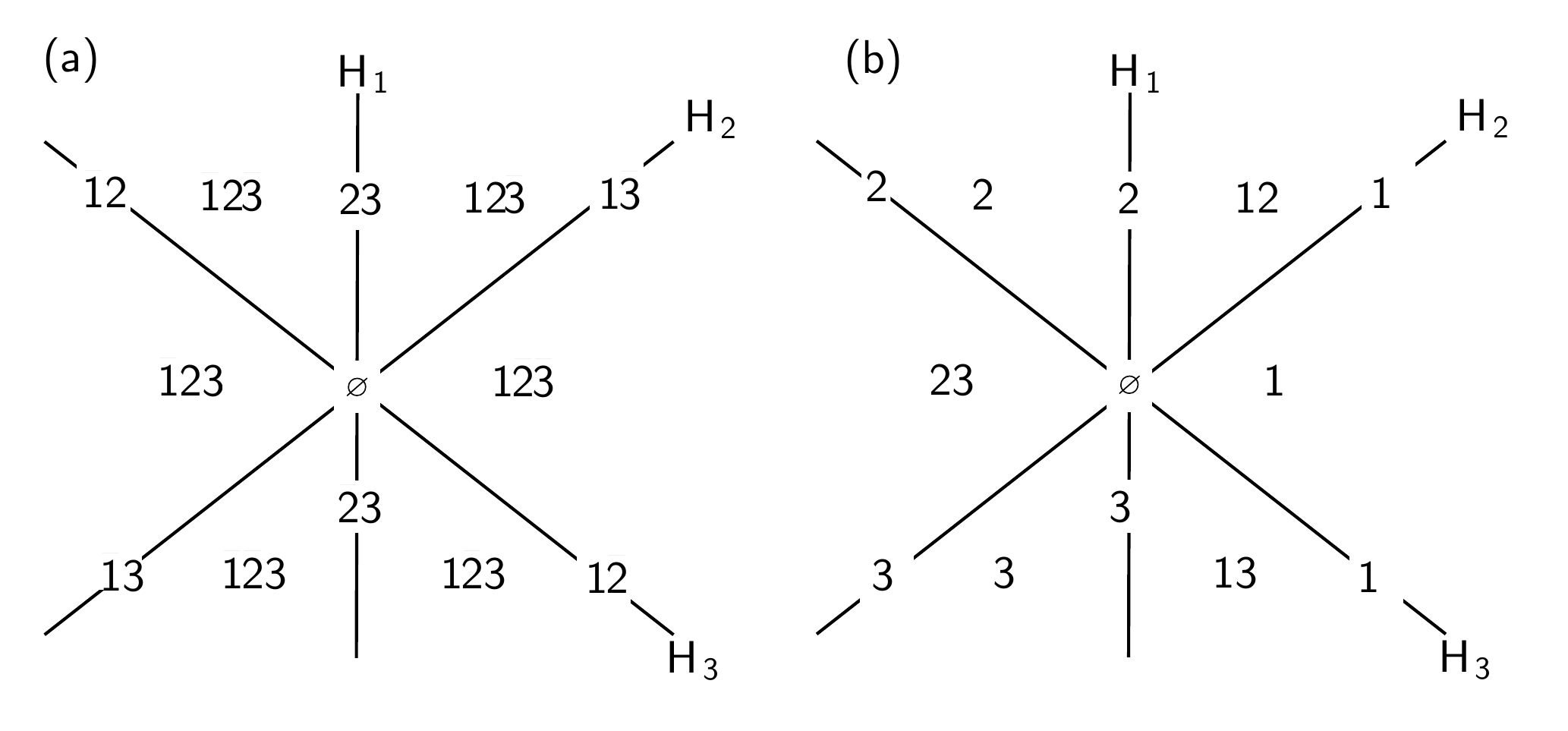}
 \end{center}
  \caption[Comparing covectors and codewords.]{(a)~The covectors of an oriented matroid arising from a central hyperplane arrangement. (b)~The combinatorial code of the cover given by the positive open half-spaces.}
  \label{F:OMandCodeExample}
\end{figure}

We define a map $\sfL^+$ which takes an oriented matroid to the set of positive parts of its covectors. 
Using this map, the connection between oriented matroids and convex codes comes primarily through the following theorem, which roughly holds that oriented matroids form the ``upper boundary"  of the set of convex neural codes in the poset $\pcode$. 
\begin{ithm}\label{thm:polytope_matroid}
A code $\cC$ has a realization with convex polytopes if and only if there is an oriented matroid $\cM$ such that $\cC$ lies below $\sfL^+( \cM)$ in the poset $\pcode$. 
\end{ithm}

This allows us to categorize non-convex codes: if a code is not convex, then either it does not lie below any oriented matroid in $\pcode$, or it lies below non-representable matroids only. However, it is not yet known whether every convex code has a realization with convex polytopes. If this does hold, then  Theorem \ref{thm:polytope_matroid} would give a full characterization of convex codes in terms of representable oriented matroids. %The literature on neural codes contains many interesting examples of non-convex codes \cite{curto2017makes, lienkaemper2017obstructions,jeffs2019morphisms, jeffs2019sunflowers, chen2019neural}. 

There are many known examples of non-convex codes \cite{curto2017makes, lienkaemper2017obstructions,jeffs2020morphisms, jeffs2019sunflowers, chen2019neural},
and we show that many of these fall into the first category: they are non-convex because they are not below any oriented matroids in $\pcode$.
For instance, codes with topological local obstructions do not lie below oriented matroids. 
Furthermore, well known examples of non-convex codes with no local obstructions also do not lie below oriented matroids. 
	\begin{ithm}\label{thm:bad_sunflower}
		The non-convex codes with no local obstructions introduced in \cite{jeffs2020morphisms, jeffs2019sunflowers} and \cite{lienkaemper2017obstructions} do not lie below the set of covectors of an oriented matroid in $\pcode$. 
	\end{ithm}

We are also able to generate an infinite family of non-convex codes of the second kind, those which lie below non-representable matroids only.
In order to obtain this family, we establish a relationship between representability and convexity. We do this for the special case of uniform oriented matroids of rank 3, which correspond to non-degenerate pseudoline arrangements in the plane. This construction makes use of order-forcing results from the previous chapter. 
\begin{ithm}\label{thm:witness}
%	A uniform oriented matroid is representable if and only if it is convex when viewed as a neural code. 
\
	 Let $\cM$ be a uniform, rank 3 oriented matroid. Then we can construct a code which is convex if and only if $\cM$ is representable. 
\end{ithm}
Using this last result, we are able to compare two fundamental decision problems: (1) is a given oriented matroid representable, and (2) is a given neural code realizable by convex sets. 
We demonstrate that deciding convexity for arbitrary neural codes is {\em at least} as hard as deciding representability of an oriented matroid. 
The latter problem is known to be NP-hard and $\exists\R$-hard, leading to the following theorem:

\begin{ithm}\label{thm:hard}
%	Any problem in the existential theory of the reals can be reduced to the problem of determining whether a neural code is convex.
	The convex code decision problem is NP-hard and $\exists\R$-hard.
\end{ithm}

%%% ABK -- do we want to state explicitly which theorems are proven in which section?
The paper is organized as follows:  
In Section \ref{sec:intersection}, we define the map $\sfL^+$ and prove Theorem \ref{thm:polytope_matroid}. 
%including polytope-convex codes. 
In Section \ref{sec:non_convex}, we discuss classes of non-convex codes and their relationships to oriented matroids, proving Theorems \ref{thm:bad_sunflower}, \ref{thm:witness}, and \ref{thm:hard}. 
%In Section 6, we show that the decision problem of neural code convexity is at least as hard as matroid representability. 
Finally, in Section \ref{S:questions}, we present open questions related to each area discussed in the paper.

\section{Relating convex codes to oriented matroids}
\label{sec:intersection}

In this section, we establish the relationship between representable oriented matroids and convex neural codes, as well as between oriented matroids and good cover codes. 

While the set of covectors of an oriented matroid, viewed as a subset of $2^{\pm [n]}$, is a combinatorial code, it often makes sense to consider a more compact code using only the positive parts of covectors. 
We define this code as 
\begin{align*}
\sfL^+(\cM) = \{ X^+\subseteq [n] \mid X \in \cL(\cM)\}
\end{align*}
In the case that $\cM$ is realized by a hyperplane arrangement $H_1, \ldots, H_n$, $\sfL^+(\cM)$ corresponds to the code of the positive open half-spaces $H_1 ^+, \ldots, H_n^+$. Thus, we refer to  $\sfL^+(\cM)$  as the \emph{open code} of $\cM$. 

We can also define the  \emph{closed code} of $\cM$ via the map $\sfL^\geq(\cM)$  which takes a representable oriented matroid to the code of its closed positive half-spaces. 
We do this by taking the complement of the negative part of each covector. 
\begin{align*}
\sfL^\geq(\cM) = \{ [n]\setminus X^- \mid X \in \cL(\cM)\}
\end{align*}

Notice that there is not a one-to-one relationship between covectors of an oriented matroid and codewords of $\sfL^+(\cM)$ or $\sfL^\geq(\cM)$: multiple covectors may have the same positive part or the same negative part. 
For instance, in Figure \ref{F:OMandCodeExample}, the covectors $\bar 1 2 \bar 3$ and $\bar 1 2$ both have the same positive part, $2$. 
More significantly, neither $\sfL^+(\cM)$ nor $\sfL^-(\cM)$  is  an injective map from the set of oriented matroids to the set of convex codes. 
For example, see Figure \ref{fig:not_inj} for an example of two oriented matroids which map to the same code under $\sfL^+$. 

\begin{figure}[ht!]
\begin{center}
 \includegraphics[width = 5 in]{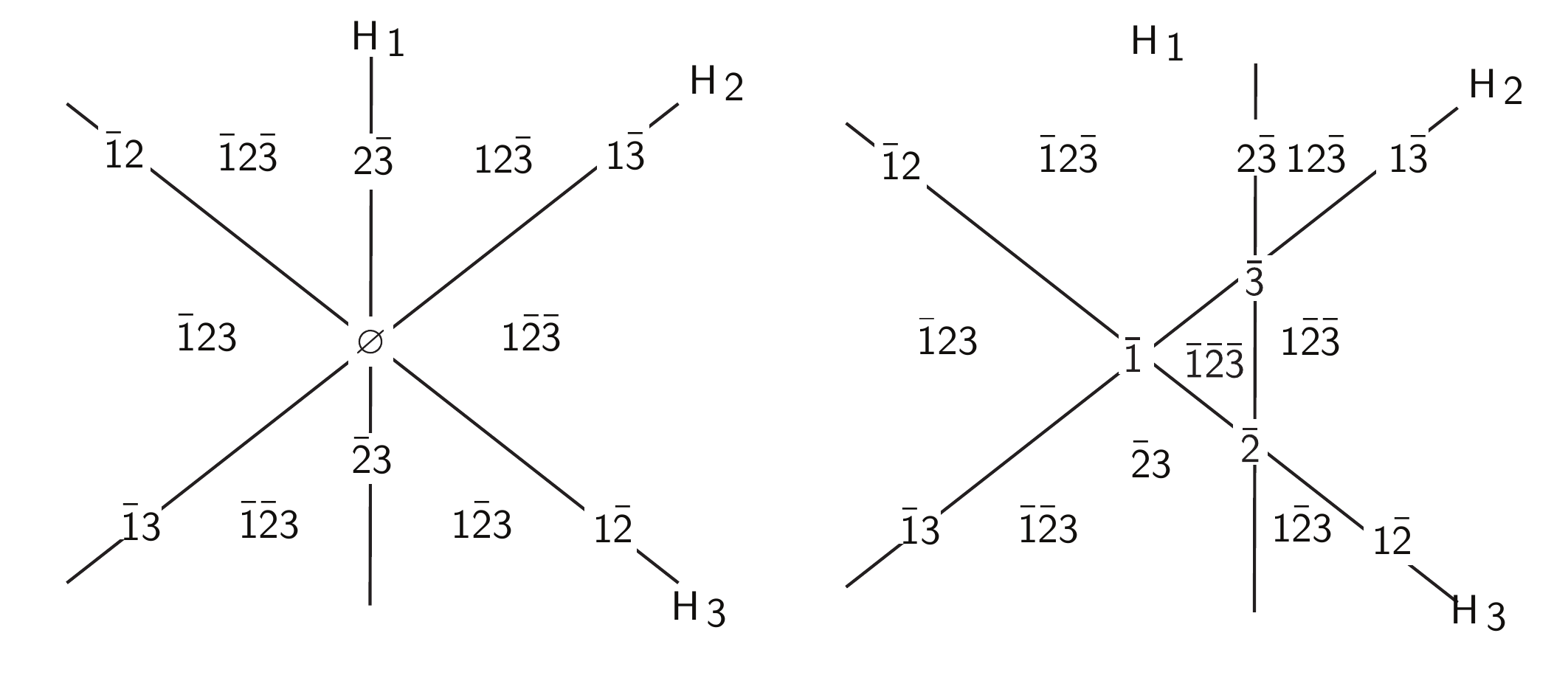}
 \end{center}
  \caption[Affine pieces of two oriented matroids which map to the same code]{Affine pieces of two oriented matroids which map to the same code. These arrangement can be centralized by adding a fourth hyperplane.	}
  \label{fig:not_inj}
\end{figure}

We can also apply $\sfL^+$  and $\sfL^{\geq}$ to an affine oriented matroid $(\cM, g)$. In this case, we have 
\begin{align*}
\sfL^+(\cM, g) &= \{X^+ \mid X \in \cL_+(\cM,g)\}\\
\sfL^\geq(\cM, g) &= \{ [n]\setminus X^- \mid X \in \cL_+(\cM,g)\}.
\end{align*}

In the representable case, $\sfL^+(\cM, g)$ is the code of the open half spaces of the affine hyperplane arrangement realizing $(\cM, g)$ and  $\sfL^\geq(\cM, g)$  is the code of the closed half spaces of the affine hyperplane arrangement realizing $(\cM, g)$.

We can relate the open code of an affine oriented matroid to the open code of its oriented matroid via trunks. Notice that   
\begin{align*}
\sfL^+(\cM, g) = \tk_{\sfL^+(\cM) }(g).
\end{align*}

However, no such relationship holds in the closed case, since $\tk_{\sfL^\geq(\cM) }(g)$ contains codewords whose atoms lie on $H_g$, while $\sfL^\geq(\cM, g)$ does not. 
In fact, notice that for any oriented matroid, $\sfL^\geq(\cM)$ contains the full support codeword $[n]$, which is contained in every trunk.

Using Theorem \ref{lem:int_closed}, we can prove relationships between open convex codes and codes of the form $\sfL^+(\cM)$, and between closed convex codes and codes of the form  $\sfL^\geq(\cM, g)$.  
We say that a code $\cC$ is open polytope convex if there exists a collection of interiors of convex polytopes $\cP = \{P_1, \ldots, P_n\}$ and a bounding convex polytope $X$ such that $\cC = \code(\cP, X)$. 
Likewise, we say that a code is closed polytope convex if there exists a collection of closed of convex polytopes $\cP = \{P_1, \ldots, P_n\}$ and a bounding convex polytope $X$ such that $\cC = \code(\cP, X)$. 

Notice that both the set of interiors of convex polytopes and the set of convex closed polytopes  are intersection-closed. 
Thus, Theorem \ref{lem:int_closed} implies that the image of any open (closed) polytope code under a surjective morphism is also an open (closed) polytope code. 
Thus, a code $\cC$ is open polytope convex if $\cC  \leq \sfL^+(\cM)$ and is closed polytope convex if  $\cC  \leq \sfL^\geq (\cM, g)$ for some representable oriented matroid $\cM$. 
We prove the converse, showing that every polytope code is itself the image of the code of an oriented matroid under some surjective morphism. 
This demonstrates that polytope codes are a down-set whose ``upper boundary" is the set of representable oriented matroid codes. 
%More formally, the set of representable oriented matroid codes is \emph{cofinal} in the set of of convex polytope codes \cite{fraisse2000theory}. 

%TODO: CLOSED SET VERSION IS DIFFERENT, AND MUCH SADDER. 
%THESE SHOULD BE FULLY SEPARATED. 

%tbh, the affine to central thing being a trunk is actually nice and useful. 
 
%I think a structure that works out best is to make the trunk thing a lemma, and note that it only works for L^+

%And make the convex iff below affine thing a thing that splits for both open and closed pictures. 
\setcounter{ithm}{0}
\begin{ithm}\label{thm:polytope_matroid}
A code $\scC$ is open polytope convex if and only if there exists a representable oriented matroid $\cM$ such that 
$\scC \leq\sfL^+ (\cM)$. 
  
A code $\scC$ is closed polytope convex if and only if there exists a representable affine oriented matroid $(\cM, g)$ such that $\scC \leq\sfL^\geq (\cM, g)$. 
\end{ithm}

\begin{proof}

%	By \ref{prop:covectors_match_cover}, a code $\scH$ has a realization with a central hyperplane arrangement if and only if $\scH = \sfL^+( \cM)$, where $\cM$ is a representable oriented matroid. Thus, it remains to be shown that $\scC$ is a convex polytope code if and only if $\scC \leq \scH$, where $\scH$ is the code of a central hyperplane arrangement. We show this. 

	($\Rightarrow$)
	Note that if $\cM$ is a representable (affine) oriented matroid, then $\sfL^+ (\cM)$ and  $\sfL^+ (\cM, g)$ can be realized with open half-spaces, and  $\sfL^{\geq} (\cM)$  and  $\sfL^{\geq} (\cM, g)$ can be realized with closed half-spaces. Then by Theorem \ref{lem:int_closed}, any code $\cC \leq \sfL^+(\cM)$ or $ \cC \leq \sfL^+(\cM, g)$ has a realization with intersections of open half-spaces, i.e. open convex polytopes, and any code $\cC \leq \sfL^{\geq}(\cM)$ or $\cC \leq \sfL^{\geq}(\cM, g)$ has a realization with intersections of closed half-spaces, i.e closed convex polytopes. 
%	Notice that an open half-space is the interior of a convex polyhedron, and that the set of interiors of convex polyhedra is an intersection closed family. By \ref{lem:int_closed}, the set of polyhedral codes is a down-set. Thus if $\scC \leq \scH$, $\scC$ is a polyhedral code. 
	
($\Leftarrow$)
Let $\scC$ be a polytope convex code with $(\cV,X)$ a realization of $\scC$ with open (closed) convex polytopes $V_i$ and bounding convex set $X$.
Without loss of generality, we can choose $X$ to be a convex polytope (the convex hull of one point in each atom).
Then each $V_i$ is the intersection of a collection of open (closed) half spaces $H_{i1}^+, \ldots, H_{ik_i}^+$, and $X$ is the intersection of open half spaces $K_1^+, \ldots, K_k^+$.
Now, let $\scH= \code(H_{i1}^+, \ldots, H_{ik_i}^+, \ldots, K_1^+, \ldots, K_k^+.\}, \R^d)$. 
Notice that in the open case, $\scH = \sfL^+(\cM, g)$ and in the closed case, $\scH = \sfL^\geq(\cM, g)$ for some representable affine oriented $(\cM, g)$.   
Let $\scH'$ be the trunk of the neurons associated to $K_1^+, \ldots, K_k^+$.

	Now, we define a surjective morphism $f: \scH' \to \scC$ as follows.
	Choose trunks $T_1, \ldots, T_n$ of $\scH'$ by $T_i = \tk_{\scH'}(\{i1,\dots,ik_i\})$.
%	$T_i = T_{i1, \ldots, i k_i}(\scH')$. 
	Let $f$ be the morphism defined by the trunks $T_1, \ldots, T_n$.
	We now show that its image is $\scC$.  
	
	To do this, construct the realization of $f(\scH')$ given in the proof of \ref{lem:int_closed}. This construction gives the realization
	\begin{align*}
	V_j' = \bigcap_{i = 1}^{i = k_j}U_{ji}
	\end{align*}
	relative to the convex set $X = \bigcap_{i= 1}^k X_i$. 
	Thus, $f(\scH') = \code(\{V_1, \ldots, V_n\}, X) = \scC$.
	Thus, we have shown that any open (closed) polytope  convex code lies below the open (closed) polytope convex code of an affine oriented matroid in $\pcode$.
Further, in the open case, the code of an affine oriented matroid  $(\cM, g)$ is a trunk of the code of the oriented matroid $\cM$. Thus, any convex open polytope code lies below the open code of an oriented matroid in $\pcode$. 
\end{proof}

For the remained of this chapter, we focus on open convex codes and $\sfL^+(\cM)$. 
We begin by noting that codes below oriented matroids have no local obstructions. This result appears in different language in \cite{edelman2002convex}. We flesh this out. 

\begin{prop} \label{prop:no_local}
  Let $\cM$ be an oriented matroid.
  If $\scC \leq \sfL^+( \cM)$ in $\pcode$, $\scC$ is a good cover code, and thus has no local obstructions. 
\end{prop}

\begin{proof}
%The paper \cite{edelman2002convex} defines 
Edelman, Reiner, and Welkder define a simplicial complex $\Delta_{\mathrm{acyclic}}(\cM)$ which is identical to  $\Delta(\sfL^+( \cM))$ \cite{edelman2002convex}. Proposition 11 and Lemma 13 of  \cite{edelman2002convex}  establish that if $\sigma \in  \Delta(\sfL^+( \cM)) \setminus \sfL^+( \cM)$, then $\link_{\sigma}  \Delta(\sfL^+( \cM)) $ is contractible. Thus,  $\sfL^+( \cM)$ has no local obstructions, and is thus a good cover code. 
  %Let $\scC = \sfL^+( \cM)$, and take $S_1, \ldots, S_n\subset \S^{r(\cM) -1}$ to be a pseudosphere arrangement representing $\cM$.  Let $U_1, \ldots, U_n$ be the positive hemispheres of $S_1, \ldots, S_n$. By \ref{prop:covectors_match_cover}, $\scC = \code(U_1, \ldots, U_n)$.  By Lemma 5.1.8 of \cite{bjorner1999oriented}, $U_1, \ldots, U_n$ is a good cover.
By Theorem \ref{thm:good_cover}, good cover codes form a down-set in $\pcode$, so if $\scC \leq \sfL^+( \cM)$ in $\pcode$, then  $\scC$ has no local obstructions. 
\end{proof}

This result suggests an analogy: the relationship between good cover codes and convex codes mirrors that between oriented matroids and representable oriented matroids.  
This analogy suggests that any good cover code is the image of an oriented matroid. 
However, this is not true: in the next chapter, we give an example of non-convex good cover code which is not the  image of any oriented matroid, representable or otherwise.

%!TEX root = CT-submission.tex

\section{Non-convex codes} %Codes -> codes per Y 10
\label{sec:non_convex}
Though it is unknown whether every convex code has a realization with convex polytopes, the contrapositive to Theorem \ref{thm:polytope_matroid} helps us characterize non-convex codes. 
If $\scC$ is not convex, one of two possibilities hold: either $\scC$ does not lie below any oriented matroid,
or $\scC$ lies below only non-representable oriented matroids in $\pcode$.
In this section, we prove that codes with local obstructions as well as
``sunflower codes''
do not lie below {\em any} oriented matroids. We also construct a new class of
non-convex codes which lie below non-representable oriented matroids.

%The contrapositive of Proposition \ref{prop:no_local} implies the following result:
%
%\begin{prop}
%  If $\scC$ is a code with a local obstruction, then $\scC$ does not lie
%  below any oriented matroid code.
%\end{prop}

\subsection{Sunflower codes do not lie below oriented matroids}

%The first example of a non-convex code with no local obstructions,
%\[ \scC_1 = \{\emptyset, 123, 13, 134, 14, 145, 23, 2345, 3, 34, 4, 45\}, \]
%appeared in \cite{lienkaemper2017obstructions}. In \cite{jeffs2020morphisms}, Jeffs uses this code to construct a smaller non-convex code $\scC_2 \leq \scC$ with no local obstructions,
%	\[ \scC_2 = \{\emptyset, 1236, 13, 135, 23,  234, 4, 456, 5, 6  \}. \]
%This code  is minimally non-convex, in the sense that any code $\scC' \leq \scC_2$ in $\pcode$ is convex. The proofs that $\scC$ and $\scC_2$ are not convex depend on
%% the following lemma: 
%the $n=3$ case of the following theorem:
%
%%\begin{lem}[\cite{lienkaemper2017obstructions}, Lemma 3.2] Let $U_1, U_2, U_3$ be convex open sets in $\R^2$ such that for $U_1\cap U_2 = U_2\cap U_3 = U_1\cap U_3 = U_1\cap U_2 \cap U_3$. Then any line which passes through $U_1, U_2$, and $U_3$ passes through $U_1\cap U_2 \cap U_3$. 
%%\end{lem}
%%
%%
%%
%%
%%In \cite{jeffs2019sunflowers}, this result is generalized to the following theorem:
%
%\begin{thm}[\cite{jeffs2019sunflowers}, Theorem 1.1] \label{thm:sunflower} Let $U_1, \ldots, U_n$ be convex open sets in $\R^{n-1}$ such that for all $i, j\in [n]$, $U_i \cap U_j = \bigcap_{k\in [n]} U_k$. Then any hyperplane which passes through $U_1, \ldots, U_n$  passes through $\bigcap_{k\in [n]} U_k$. 
%\end{thm}
Recall from Chapter \ref{sec:codes} the non convex code \[ \scC_1 = \{\emptyset, 123, 13, 134, 14, 145, 23, 2345, 3, 34, 4, 45\}. \]
In \cite{jeffs2020morphisms}, Jeffs uses this code to construct a smaller non-convex code $\scC_2 \leq \scC_1$ with no local obstructions,
\[ \scC_2 = \{\emptyset, 1236, 13, 135, 23,  234, 4, 456, 5, 6  \}. \]
This code  is minimally non-convex, in the sense that any code $\scC' \leq \scC_2$ in $\pcode$ is convex.
In  \cite{jeffs2019sunflowers}, 
Jeffs uses the sunflower theorem, our Theorem \ref{thm:sunflower} to construct an infinite family $\{\scC_n\}_{n\geq 2}$ of minimally non-convex codes with no local obstructions generalizing $\scC_2$; we refer to these as ``sunflower codes.'' In the rest of this subsection, we define the code $\scC_n$ for $n\geq 2$ and give a proof that for all $n\geq 2$, the code $\scC_n$ does not lie below any oriented matroid, representable or otherwise. 

\begin{defn}[\cite{jeffs2019sunflowers}, Definition 4.1]
  Let $n \geq 2$, $P = \{p_1,\ldots,p_{n+1}\}$ and $S = \{s_1,\ldots,s_{n+1}\}$ be sets of size $n+1$.
Denote by $\scC_n \subseteq 2^{P\cup S}$ the code that consists of the following codewords: 
\begin{itemize}
	\item $\varnothing$; 
	\item $S \cup \{p_{n+1}\}$; 
	\item $P$;
	\item the codeword $X \cup \{s_{n+1}\}$  for each $\varnothing \subsetneq X \subsetneq \{s_1,\dots,s_n\}$;
	\item the codewords $\{p_i\}$ for each $1 \leq i \leq n+1$;
	\item and $(S \setminus \{s_i\}) \cup \{p_i\}$ for each $1 \leq i \leq n$. 
  \end{itemize}
  A good-cover realization of $\cC_2$ is given in Figure \ref{fig:c2}. 
\end{defn}

We will refer to the regions indexed by $P$ as {\em petals},
and the regions indexed by $S$ as {\em simplices}.
\begin{figure} [ht!]
\begin{center}
	\includegraphics[width = 3 in]{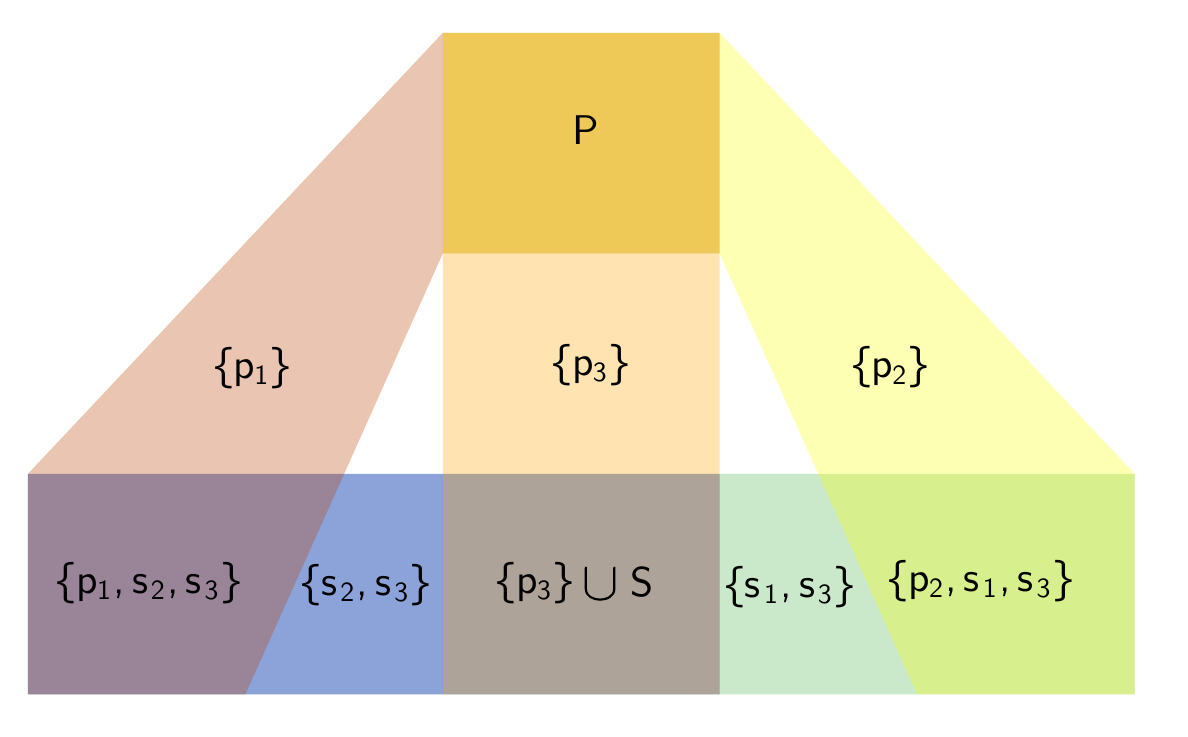}
	\end{center}
	\caption[A good cover realization of $\scC_2$]{A good cover realization of $\scC_2~=~\{\emptyset, 23, 13, 4, 5, 6, 234, 135,  1236, 456\}.$ Here $P = \{1,2,3\}$ and $S = \{4,5,6\}$. \label{fig:c2}}
\end{figure}

The proof of  Theorem \ref{thm:bad_sunflower} depends on some basic facts about tope graphs of oriented matroids. The \emph{tope graph} $\cT$ of an oriented matroid $\cM$ is a graph whose vertices are the topes of $\cM$, and whose edges connect pairs of topes which differ by one sign. A subgraph $\mathcal Q \subseteq \cT$ is called $T$-convex if it contains the shortest path between any two of its members. Any $e\in E$ divides the tope graph into two \emph{half-spaces} $\cT_e^+ = \{W\in \cW \mid  e\in W^+\}$ and $\cT_e^- = \{W\in \cW \mid  e\in W^-\}$.
A subgraph $\cQ \subseteq \cT$ is $T$-convex if and only if it is an intersection of half-spaces \cite[Proposition 4.2.6]{bjorner1999oriented}.
%Our proof uses Proposition 4.2.6 of \cite{bjorner1999oriented}: a subgraph $\mathcal Q \subseteq \cT$ is $T$-convex if and only if it is the intersection of half-spaces. 

\begin{ithm}
	For each $n \geq 2$, the code $\scC_n \not \leq \sfL^+( \cM)$ for any oriented matroid $\cM$.
\end{ithm}

\begin{proof}
Fix $n \geq 2$.
Suppose to the contrary that there is an oriented matroid $\cM$ such that $\scC_n \leq \sfL^+( \cM)$.
For ease of notation, let $\scM$ denote the code $\sfL^+( \cM)$.
Since $\emptyset \in \scC_n$, we can assume without loss of generality that  $\scC_n = f(\scM)$  for some code morphism $f$.

Denote the ground set of $\cM$ by $E$. The map $f$ must be defined by
trunks
	 \[\tk_{\scM}(\pi_1), \dots, \tk_{\scM}(\pi_{n+1}), \tk_{\scM}(\sigma_{1}), \ldots,\tk_{\scM}(\sigma_{n+1}),\]
with $\pi_i,\sigma_i \subseteq E$ corresponding to $p_i$ and $s_i$ respectively. 

\vspace{3mm}
\noindent  {\bf Claim 1:} There is a tope $T$ of $\cM$ such that
    $\left(\bigcup_{i=1}^{n+1} \sigma_i\right)\cup \left(\bigcap_{j=1}^{n} \pi_j\right)\cup \pi_{n+1} \subseteq T^+$.\\
  Roughly speaking, we are producing a codeword in the intersection of the
  last petal and all simplices, which also lies in the convex hull of the other petals.
  
Define a morphism $g:\scM \to 2^{[n+1]}$ by the trunks $T_i = \tk_{\scM}(\tau_i)$, with 
$\tau_i = \sigma_i \cup \left(\,\,\bigcap_{j = 1}^{n}{ \pi_j}\,\right)$ for $i = 1, \ldots, n+1$.
Let $\scD = g(\scM)$.

Since $(S \setminus \{s_i\} )\cup \{p_i\} \in \scC_n$ for each $i \in [n]$, we deduce that $[n+1]\setminus i$ is a codeword of $\scD$ for each $i\in [n]$. Thus, $\link_{\{n+1\}}(\Delta(\scD))$ is either a hollow $(n-1)$-simplex or a solid $(n-1)$- simplex. Since we have defined $\scD$ as the image of an oriented matroid code, it cannot have local obstructions.
The codeword $\{n+1\}$ is not in $\scD$; if it were, then $f(g^{-1}(\{n+1\}))$
would be a codeword of $\scC$ including $s_{n+1}$ without any other $s_i$. No such codeword
exists in $\scC$.
Thus $\link_{\{n+1\}}(\Delta(\scD))$ must be contractible. Because $\{n+1\}$ is not a codeword of $\scD$,  the $\link_{\{n+1\}}(\Delta(\scD))$ must be a solid $(n-1)$-simplex; therefore, $[n+1]$ is a codeword of $\scD$. 

Based on the trunks defining $g$, we know that
$\left(\bigcup_{i =1}^{n+1}\sigma_i\right) \cup \left( \bigcap_{j = 1}^{n} \pi_j \right)\subseteq g^{-1}([n+1])$.
By definition of $f$, we must also have $S \subseteq f(g^{-1}([n+1]))$; however, the only codeword of $\scC_n$ which contains $S$ is $S \cup \{p_{n+1}\} $.
Thus, there is a codeword of $\scM$ containing $\left(\bigcup_{i =1}^{n+1}\sigma_i\right)
\cup \left(\bigcap_{j = 1}^{n} \pi_j \right)\cup \pi_{n+1}$. This implies that
$\cM$ has a covector $X$ such that 
 $\left(\bigcup_{i =1}^{n+1}\sigma_i\right)
\cup \left(\bigcap_{j = 1}^{n} \pi_j \right) \cup \pi_{n+1} \subseteq X^+$. To produce a tope satisfying the condition, take $T = X\circ W$ for any tope $W$ of $\cM$.

\vspace{3mm}
  
\noindent  {\bf Claim 2:}
$ \pi_{n+1} \cup \left( \,\,\bigcap_{ j=1}^{n}\,\pi_j\right) \subseteq T^+ $   implies $\,\,\bigcup_{ j=1}^{n+1}\,\pi_j \subseteq T^+ $ for any tope $T$ of $\cM$.\\
The intuition here is that the last petal must intersect the convex
hull of the other petals {\em only} in the common intersection of all petals, as illustrated in Figure \ref{fig:crossing }. 

\begin{figure}
\begin{center}
\includegraphics[width = 2 in]{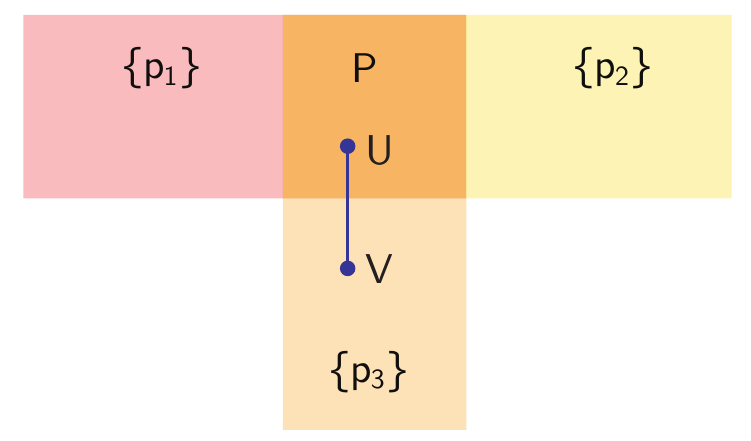}
\caption[A step from the proof of Theorem \ref{thm:bad_sunflower}]{Any path from a tope $U$ with $\left(\bigcup_{ j=1}^{n+1}\,\pi_j\right) \subseteq U^+$ to a tope $V$ with $\left(\bigcup_{ j=1}^{n+1}\,\pi_j\right) \not\subseteq V^+$ must cross an edge in  $\left(\bigcap_{ j=1}^{n+1}\,\pi_j\right)$. Analogously, a path from a point in the atom $P = \{p_1, p_2, \ldots, p_{n+1}\}$ to the atom $\{p_{n+1}\}$ must cross the boundaries of $p_1, p_2, \ldots, p_{n}$ all at one time.  \label{fig:crossing }}
\end{center}
\end{figure}

Let $U$ be a tope with $\left(\bigcup_{ j=1}^{n+1}\,\pi_j\right) \subseteq U^+$. Such a tope must exist, since $P \in \scC_n$.   Suppose for the sake of contradiction that there exists a tope $V$ such that 
\[
\pi_{n+1} \cup \left( \,\,\bigcap_{ j=1}^{n+1}\,\pi_j\right) \subseteq V^+, \text{   but   }\bigcup_{ j=1}^{n+1}\,\pi_j \not \subseteq  V^+.
\]
Consider a shortest path from $U$ to $V$ in the tope graph of $\cM$.  Each edge of the tope graph is naturally labeled by the ground set element $e$ by which the two incident topes differ.  By the $T$-convexity of intersections of half-spaces in the tope graph, each tope along this path has
$\pi_{n+1} \cup \left( \,\bigcap_{ j= 1}^{n}\,\pi_j\right) $ in its positive part, so no edge is labeled with an element of $\bigcap_{ j=1}^{n}\,\pi_j$. 

Thus at some point along the path from $U$ to $V$, we must cross an edge $(T, W)$ labeled by a ground set element 
$e \in
\left(\,\,\bigcup_{ j=1}^{n+1}\,\pi_j \right) \setminus
\left(\,\,\bigcap_{ j=1}^{n+1}\,\pi_j \right)$.
Choose the first such edge  $(T, W)$ labeled with ground set element $e$. By our choice of $e$, there exist
$k, \ell \in [n]$ such that $e\in \pi_k$, and $e\notin \pi_\ell$.
This means $\pi_k \not\subseteq W^+$, whereas $\pi_\ell \subseteq W^+$. Then
$\{p_\ell, p_{n+1}\} \subseteq f(W^+)$, but $f(W^+) \neq P$. However, the only codeword of $\scC_n$ containing $\{p_\ell, p_{n+1}\}$ is $P$, so we have reached a contradiction.
Therefore, no such tope $V$ may exist.

\vspace{6mm}

\noindent By Claim 1, $\cM$ must have a tope $T$ which has
$\left(\bigcup_{i=1}^{n+1} \sigma_i\right)\cup \left(\bigcup_{j=1}^{n} \pi_1\right)  \cup \pi_{n+1} \subseteq T^+$.
Because $T$ satisfies
$\left(\bigcap_{i=1}^{n+1} \pi_i\right) \cup \pi_{n+1} \subseteq T^+$, Claim 2 implies
that $\bigcup_{i=1}^{n+1} \pi_i \subseteq T^+$.
Therefore,
$\left(\bigcup_{i=1}^{n+1} \pi_i\right) \cup \left(\bigcup_{i=1}^{n+1} \sigma_i\right) \subseteq T^+$,
but this implies $f(T) = P\cup S \in \scC_n$, a contradiction. 
\end{proof}

By showing that the family of codes $\{\scC_n\}_{n\geq 2}$ do not lie below oriented matroids, we have given an alternate proof that these codes do not have realizations with interiors of convex polytopes.
 This proof is significantly different in structure than the original proof that these codes are not convex using Theorem \ref{thm:sunflower}, which is in turn proved by induction on dimension. 
 In contrast, our proof makes no reference to rank or dimension, and does not use induction. 
 While the codes $\{\scC_n\}_n$ are not open convex, they do have realizations with \emph{closed} convex sets, which can even be chosen to be (non-full dimensional) closed convex polytopes. 
 Notice that Theorem  \ref{thm:polytope_matroid}  establishes that if $\scC$ has a  realization with \emph{interiors} of convex polytopes, then  $\scC \leq \sfL^+( \cM)$. 
 However, the fact that a code has a realization with closed convex polytopes does not guarantee this. 
 Further, in showing that these codes do not lie below any oriented matroids at all, we have established that, even while these codes are good cover codes, their obstructions to convexity are somehow still topological in nature. 

\subsection{Representability and convexity}
	Having exhibited that many well-known non-convex codes do not lie below any oriented matroids at all, we now exhibit a family of non-convex codes which lie below non-representable oriented matroids.
	For each uniform, rank 3 affine oriented matroid $(\cM, g)$, we construct a code $\scC(\cM, g)$ which is convex if and only if $\cM$ is representable (recall a uniform oriented matroid is one in which all circuits have the same cardinality).
%	We will show that for any uniform, oriented rank 3 matroid, 
	Moreover, this code is always the image of an oriented matroid under a code morphism.
%	Our construction makes use of the topological representation theorem for oriented matroids \cite{folkman1978oriented} (see also \cite[Section 1.3]{bjorner1999oriented}).

%	In order to construct $\scC(\cM)$, we make use of the topological representation theorem for oriented matroids. 
%	In particular, by the rank 3 case of the topological representation theorem, there is a one-to-one correspondence between (re-orientation classes of) oriented matroids and pseudoline arrangements in the plane. 
%	Here, a pseudoline arrangement is defined as any collection of simple curves in $\R^2$ such that any two curves cross in exactly one point, and such that the intersection of all curves is empty. 
%	If a rank 3 oriented matroid is uniform,  then no more than two pseudolines meet in one point. 
%	A rank 3 oriented matroid is representable if and only if its pseudoline arrangement is \emph{stretchable}, i.e. if all of the lines can be made straight. 
%	See \cite[Section 1.3]{bjorner1999oriented} for details. 
	
	Consider a uniform, affine oriented matroid $(\cM, g)$ of rank 3.
	A \emph{pseudoline} is a simple unbounded curve $L$ in $\R^2$, which partitions the plane into pieces $\R^2 = L^+ \sqcup L \sqcup L^-$.
	By the topological representation theorem (\cite[Section 1.3]{bjorner1999oriented},\cite{folkman1978oriented}),
	%A uniform oriented matroid of rank 3 $\cM = (E,\cL)$ 
%	$(\cM, g)$ can be represented by a family $\cP$ of simple, piecewise-linear unbounded curves $\cP = \{L_i\}_{i \in [n]}$ in $\R^2$, called pseudolines, such that every pair of pseudolines intersects at exactly one point and no more than two pseudolines meet at any point.
	$(\cM, g)$ can be represented by a uniform arrangement of piecewise linear pseudolines, that is, a family $\cP = \{L_i\}_{i\in[n]}$ of pseudolines such that each pair intersects exactly once and no more than two meet at any point.
	The sign vectors of this arrangement are the covectors of $(\cM,g)$.
	An example is illustrated in \ref{F:pseudoline}.
%	Assigning an orientation to a pseudoline $L_i$ partitions $\R^2$ into three pieces $\R^2 = L_i^+ \sqcup L_i \sqcup L_i^-$, and in this way the sign vectors of a pseudoline arrangement are the covectors in the affine space of $\scC(\cM, g)$. 
%	, and the sign vectors  of this arrangement are the covectors of $\cM$. 
	%\textcolor{red}{ABK -- We need to be careful here. The set I have described here sounds like the covectors of an affine OM; is it enough to say "this set plus its negative is the covectors of an OM" or something like that?}

	\begin{figure}
		\includegraphics[width = 6 in] {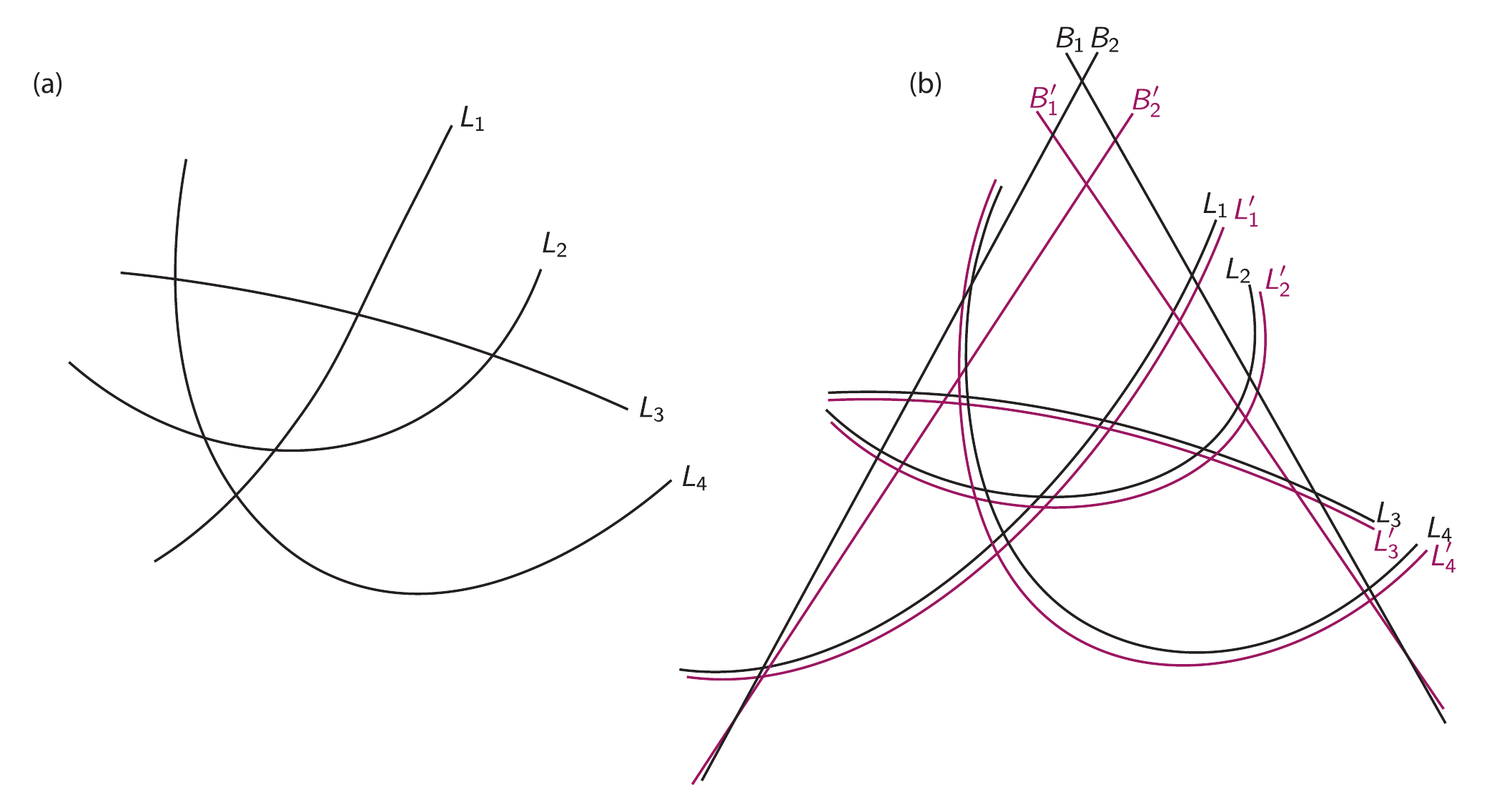}
		\caption[Pseudoline arrangements for the  proof of Proposition \ref{prop:img_of_matroid}. ]{(a) An arrangement of four pseudolines corresponding to an affine oriented matroid $(\cM, g)$. (b) The augmented pseudoline arrangement $\cM'$ used in the proof of Proposition \ref{prop:img_of_matroid}.  }
		\label{F:pseudoline}
	\end{figure}

	Note that the oriented matroid of a pseudoline arrangement is completely determined by the order in which each line meets all of the other lines. We can record this information as follows: 
	Let $L_1, \ldots, L_n$ be a pseudoline arrangement.
	For each pseudoline, fix one end of the pseudoline as the ``head".
	Let $\pi_i(j)$ denote the index $k$ such that $L_j$ is the $k^{\mathrm{th}}$ pseudoline we encounter as we follow $L_i$ from the head to the tail. 
	
	We use this order to define a code $\scC(\cM, g)$.
	We will then use the concept of order-forcing, introduced in the previous chapter, to prove that this code is convex if and only if $\cM$ is representable. 
	
	Now, we construct the code $\scC(\cM, g)$
%	for any rank 3 oriented matroid $\cM$
	so that a sequence of codewords along each pseudoline is order-forced. 
%	More exactly,  $\scC(\cM, g)$  corresponds to a particular pseudoline arrangement representing $\cM$, and is defined up to relabeling of the neurons. 

	\begin{defn}
		Let $(\cM, g)$ be a uniform, affine oriented matroid of rank 3 with pseudoline arrangement $L_1, \ldots, L_n$.
%		Then there exists a pseudoline arrangement of $n$ lines corresponding to $\cM$. 
%		Let $\cP = \{L_e\}_{e \in E}$ be a pseudoline arrangement representing $\cM$.
		Relabel the pseudolines $L_1, \ldots, L_n$, with their heads in clockwise order around the outside of the plane.
%		Augment the arrangement with pseudolines $B_\ell,B_r$ so that all bounded cells of $\cP$ are contained in $B_\ell^+ \cap B_r^+$.
		An example is illustrated in Figure \ref{F:pseudoline}. \\
		
		$\scC(\cM, g)$ is a code on $n+ 2 + n^2 +2n = (n+1)(n+2)$ neurons, labeled: \\
		
		\begin{center}
		\begin{longtable}{@{} p{0.3\textwidth}  p{0.65\textwidth} @{}} 
			$a_1,\ldots, a_n$: & Strips corresponding to each pseudoline of $\cM$.\\ 
			$b_\ell, b_r$: & Strips corresponding to two new ``boundary'' pseudolines whose positive quadrant includes all pseudoline intersections.\\ 
			$c_{1,1},c_{1,2},\ldots,c_{n,n}$: & $n$ neurons along each $a_i$ to apply order-forcing.\\ 
			$d_{\ell,1},d_{\ell,2},\ldots,d_{\ell,n},\ldots,d_{r,n}$: & $n$ neurons along $b_r$ and $b_\ell$ to apply order-forcing.\\ 
		\end{longtable}
		\end{center}
		
		\vspace{5mm}
		
		The codewords of $\scC(\cM, g)$  are as follows: \\
		
		\begin{center}
		\begin{longtable}{@{} p{0.3\textwidth} p{0.65\textwidth} @{}} 
			$b_\ell b_r d_{\ell,1}d_{r,1} $: & Intersection of the two boundary strips.\\ 
			$b_s d_{s,j}$: & Order-forcing along each boundary strip \\ & ($s = \ell,r$, $j = 1, \ldots, n$.)\\ 
			$b_r a_id_{r,i}d_{r,i+1} c_{i,1}$: & Intersection of each pseudoline with {\bf right}
			boundary strip, with order-forcing neurons. ($i = 1, \ldots, n-1$)\\ 
			$b_r a_n d_{r,n}c_{n,1}$: & Intersection of final pseudoline with {\bf right} boundary strip (one less
			order-forcing neuron is required.) \\ 
			$b_\ell  a_{n + 1 -i} d_{\ell,i}d_{\ell,i+1} c_{n+1-i,n}$:  & Intersection of each pseudoline with {\bf left}
			boundary strip plus order-forcing neurons. ($i = 1, \ldots, n-1$.) \\ 
			$b_\ell a_1d_{\ell ,n}c_{1,n}$: & Intersection of first pseudoline with {\bf left} boundary strip.  \\ 
			$a_i c_{ij}$: & Order-forcing along each pseudoline \\ & ($i = 1, \ldots, n$,
                        $j = 1, \ldots, n$) \\ 
			$a_i a_{j} c_{i,\pi_i(j)}c_{i,\pi_i(j)+1}$  & Pairwise intersections of pseudolines plus order-forcing \\ \hfill $c_{j,\pi_j(i)}c_{j,\pi_j(i)+1} $: & ($i = 1, \ldots, n$, $j = 1, \ldots, n-1$.) \\ 
		\end{longtable}
		\end{center}
		
		\vspace{3mm}
		
		We include an example of a good cover realization of this code in Figure \ref{F:pseudoline_code}. Note that this code resembles the code $\cR$ from the previous chapter. 
	\end{defn}
	
%	\begin{figure}
%	%\includegraphics[width = 5 in]{pseudolines_code.pdf}
%	\caption{A good-cover realization of the code $\scC(\cM, g)$ for the oriented matroid of the pseudoline arrangement $L_1, L_2, L_3, L_4$. For clarity, only the codewords arising from  $U_{ a _1}$ and $U_{ b_r }$  are labeled. }
%	\label{F:pseudoline_code}
%	\end{figure}
	\begin{figure}
		\includegraphics[width=0.45\textwidth]{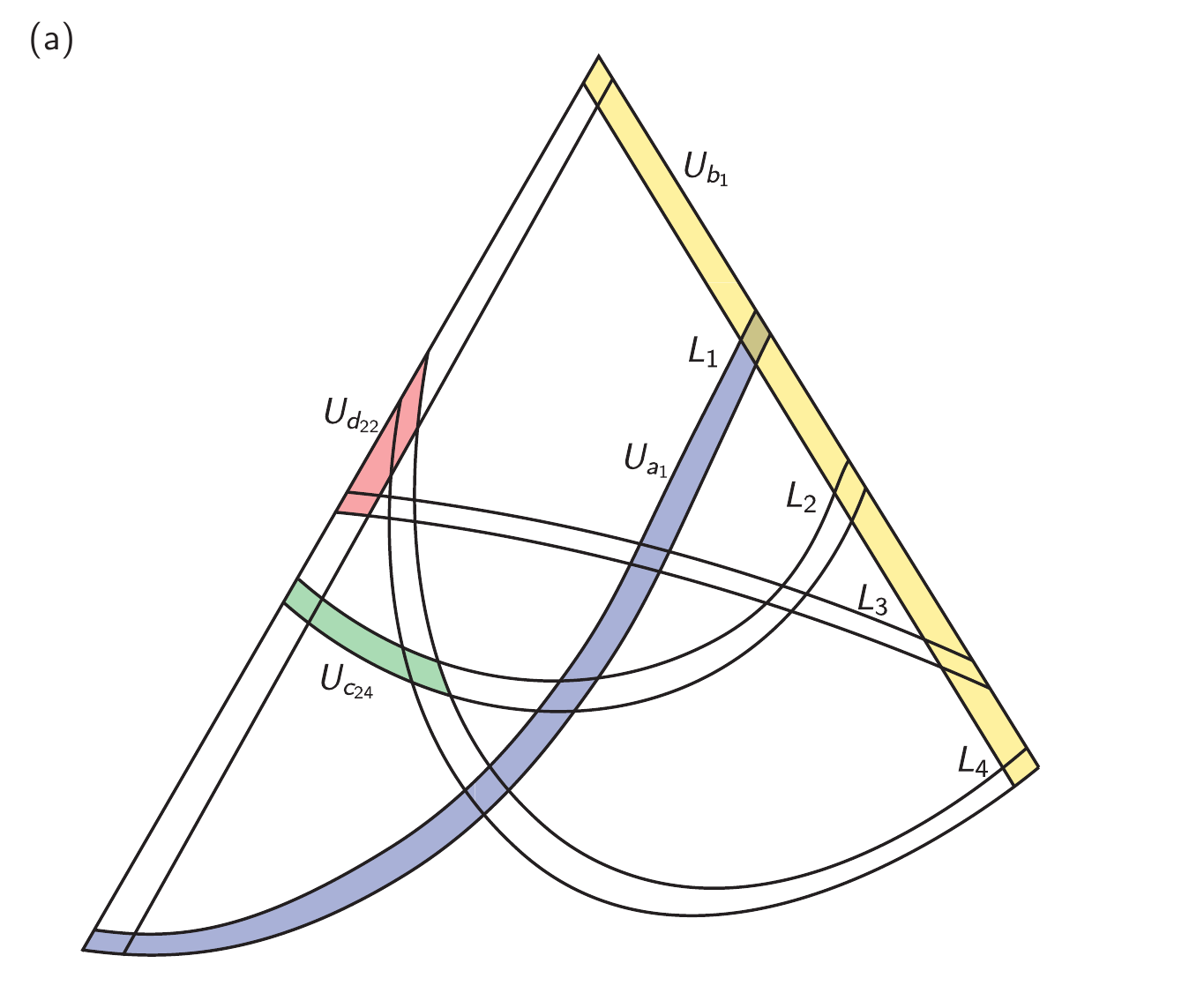}
		\includegraphics[width=0.45\textwidth]{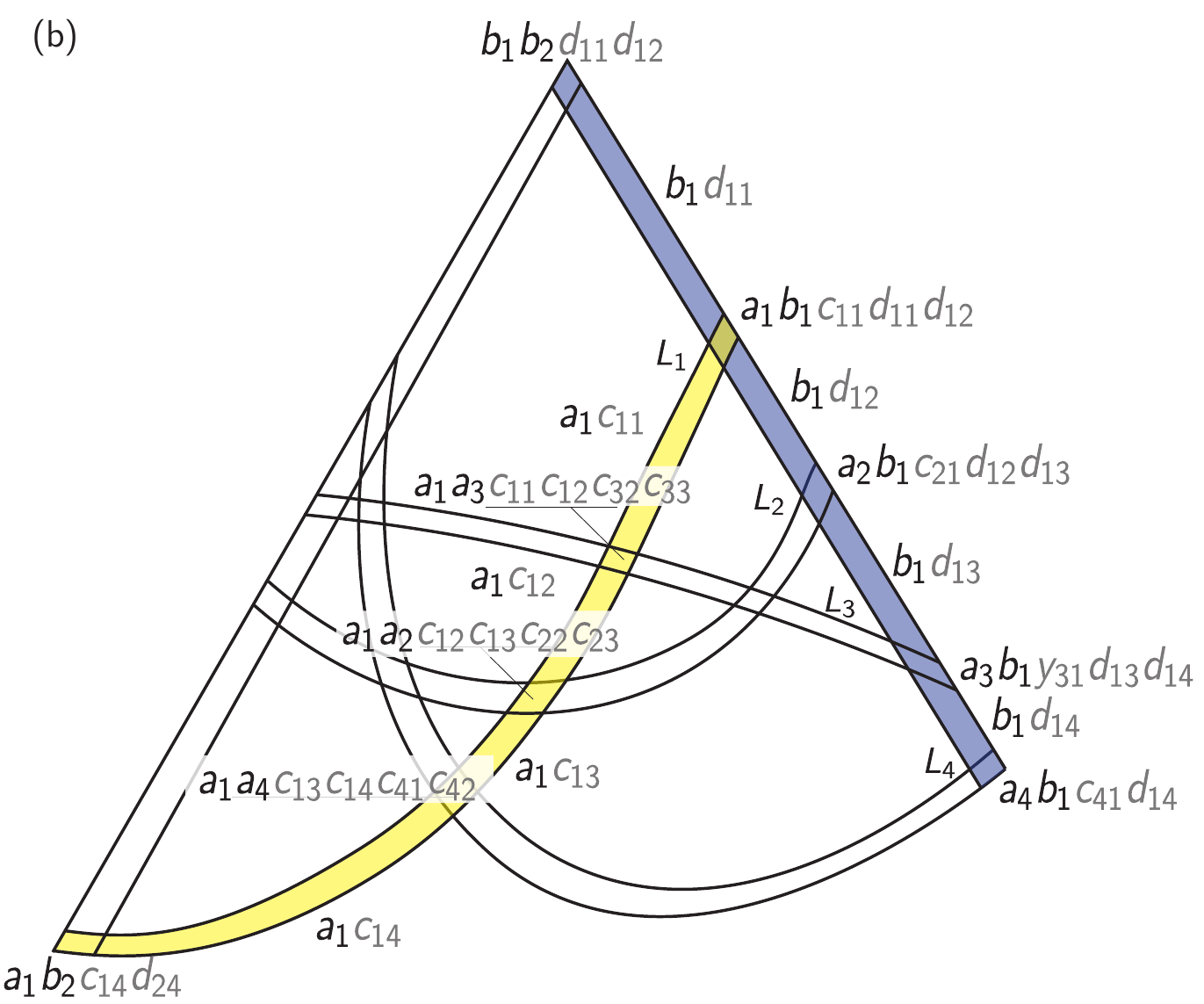}
		\caption[A good-cover realization of the code $\scC(\cM, g)$.]{A good-cover realization of the code $\scC(\cM, g)$ for the oriented matroid of the pseudoline arrangement $L_1, L_2, L_3, L_4$. (a) The sets $U_x$  corresponding to neurons $x$ of $\scC(\cM, g)$. (b) The codewords of $\scC(\cM, g)$. For clarity, only the codewords arising from  $U_{a_1}$ and $U_{b_r}$  are labeled. }
		\label{F:pseudoline_code}
	\end{figure}

	\begin{prop}\label{prop:img_of_matroid}
	For any uniform, rank 3  affine oriented matroid $(\cM, g)$, there exists a rank 3 oriented matroid $\widehat\cM$ such that $\scC(\cM, g) \leq \sfL^+\widehat\cM$.  
	\end{prop}
	
%	We describe the pseudoline arrangement associated to $\widehat\cM$ by first thickening a pseudoline arrangement $\cP$ associated to $\cM$.

	\begin{proof}
%		The ground set of $\widehat\cM$ is $\widehat E = E \cup \{\ell, r\} \cup E' \cup \{\ell', r'\}$ where $E' = \{e'\}_{e\in E}$.
		We describe the pseudoline arrangement associated to $\widehat\cM$. 
		Fix a piecewise-linear pseudoline arrangement $L_1, \ldots, L_n$ representing $(\cM, g)$ consistent with the labeling in $\scC(\cM, g)$. 
		Let $B_r$ be a line which meets $L_1, L_2, \ldots, L_n$ in the clockwise order consistent with the labeling. 
		Let $B_\ell$ be a line which meets $B_r$ and then $L_n, L_{n-1}, \ldots, L_{1}$ in the opposite of this clockwise order.  
		Orient $B_r$ and $B_\ell$ such that $B_r^+$ and $B_\ell^+$ are the half-spaces containing all bounded cells of the pseudoline arrangement.
		Orient each $L_i$ such that $B_r\cap B_\ell$ lies in $L_i^-$. 
		
%		For each $i\in E$, let $L_i'$ be a translate of $L_i$ in the positive direction (i.e.\ $L_i' \subseteq L_i^+) by a displacement small enough that $L_i'$ meets $B_r, B_\ell$, and all of the other $L_j, L_j'$ in the same order as $L_i$.
%%		Ensure that within $B_r^+\cap B_\ell^+$, $L_i'$ is on $L_i^+$. 
%		Rotate $L_i'$ to ensure that the intersection of $L_i, L_i'$ occurs outside $B_r^+\cap B_\ell^+$. Finally, orient $L_i'$ in the opposite direction of $L_i$, with $B_r\cap B_\ell \in L_i'^+$. 
		%Let $\delta$ denote the minimum distance between pairs of points where pseudolines intersect, and fix a positive $\varepsilon < d/2$.
		%For each $i\in [n]$, define the pseudoline
		%	\[ L_i' = \{x \in L_i^+ \mid \inf_{y \in L_i} \|x - y\| = \varepsilon\}. \]
		Now, for each $i\in [n]$, we define a pseudoline $L'$ which  acts as a translation of $L_i$ into its positive half-space. 
		That is, we let $L_i'$ be a pseudoline which intersects $B_\ell, B_r$  each $L_j, j\neq i$ in the same order as $L_i$, and such that for each other pseudoline $L$, the intersections of $L_i$ and $L_i'$ are adjacent along $L$. 
		Further, we ensure that $L_i, L_i'$ do not intersect. 
		Orient $L_i'$ so that $L_i \subseteq L_i'^+$.
		Define $B_\ell', B_r'$ and orient them analogously.
		%Because $\varepsilon < \delta/2$, the pseudoline arrangement $\widehat\cP = \{L_i, L_i'\}_{i\in [n]} \cup \{B_\ell,B_\ell', B_r,B_r'\}$ is uniform, that is, at most two pseudolines meet at any point. {\color{red} technically no, because of parallel things}

%		Now, let $B_r'$ be a small translate of $B_r$ inwards which meets all of the $L_i, L_i'$ in the same order as $B_r$, and let $B_\ell'$ be small translate of $B_\ell$ inwards which meets all of the $L_i, L_i'$ in the same order as $B_\ell$. 
%		Orient $B_r'$ and $B_\ell'$ towards $B_r$ and $B_\ell'$ respectively, so that $B_r^+\cap B_r'^+$ and  $B_\ell^+\cap B_\ell'^+$ bound narrow strips near the boundary of   $B_r^+\cap B_\ell^+$. 
%		Consider the oriented matroid of this pseudoline arrangement, with these orientations. 
		
%		We now describe the morphism $f$ such that $f(\sfL^+( \cM)') = \scC(\cM, g) $. 
%		This will be a morphism determined by trunks corresponding to the neurons of $\scC(\cM, g)$:
		Finally, we produce an oriented matroid $\widehat \cM$ from this pseudoline arrangement by fixing a ground set element $h$ such that the set of covectors of the pseudoline arrangement is the affine space of $(\widehat \cM, h)$.  We claim the oriented matroid of this arrangement, $\widehat\cM$, lies above $\scC(\cM, g)$.
		The morphism $f$ such that $f(\tk_{\sfL^+\widehat\cM}(h)) = \scC(\cM, g)$ is defined by the the trunks
		\begin{align*}
		\{T_{a_i}\}_{i = 1, \ldots, n} \cup \{T_{b_r}, T_{b_\ell}\} \cup \{T_{c_{i,j}}\}_{i = 1, \ldots, n, j = 1, \ldots, n} \cup \{T_{d_{s,j}}\}_{s=r,\ell, j = 1, \ldots, n}.
		\end{align*}
		corresponding to the neurons of $\scC(\cM, g)$.
		These trunks are defined as follows:
		\begin{align*}
		T_{a_i} &:= \{i, i', r, \ell\} \mbox{ for  } i = 1, \ldots, n\\
		T_{b_r} &:= \{r, r', \ell, n'\}\\
		T_{b_\ell} &:= \{\ell, \ell', r, n'\}\\
		T_{d_{r,1}} &:= \{r, r', \ell, 1'\}\\
		T_{d_{r,i}} & := \{r, r', {i-1}, {i}'\} \mbox{ for } i = 2, \ldots, n\\
		T_{d_{\ell,1}} &:= \{\ell, \ell', r, n'\}\\
		T_{d_{\ell,i}} &:= \{\ell, \ell', {n-i+2}, {n-i+1}'\} \mbox{ for } i = 2, \ldots, n.
		\end{align*}
		
		In order to define $T_{c_{i,j}}$, we introduce some notation. Let 
		\begin{align*}
		L_1(i, j)  = 
		\begin{cases}
		 j \mbox{ if } j < i\\
		 j'  \mbox{ if } j > i
		\end{cases} \qquad
		L_2(i, j)  = 
		\begin{cases}
		 j' \mbox{ if } j < i\\
		 j  \mbox{ if } j > i\\
		\end{cases}.
		\end{align*}
		That is, $L_1(i, j)$ is whichever of $j, j'$ the line $L_1$ meets first as we follow it from its intersection with $B_r$ to its intersection with $B_\ell$, and $L_2(i,j)$ is whichever it hits second.
%		Likewise,  $L_2(i, j)$ is whichever of $j, j'$ the line $L_1$ meets second as we follow it in this direction.
		Now, we define
		\begin{align*}
		T_{ c_{i1}} &:= \{i, i', r, L_2(i, \pi(i, 1))\} \mbox{ for } i = 1, \ldots, n\\ 
		T_{ c_{ij}} &:= \{i, i',  L_1(i, \pi(i, j-1)),  L_2(i, \pi(i, j-1))\} \mbox{ for } i = 1, \ldots, n, j = 2, \ldots, n-1\\ 
		T_{ c_{in}} &:= \{i, i',  L_1(i, \pi(i, n-1)), \ell \} \mbox{ for } i = 1, \ldots, n, j = 2, \ldots, n-1.\\ 
		\end{align*}
		Finally, we verify that the map $f(\sigma) = \{ s \mid \sigma\in T_s\}$ has image $\scC(\cM, g)$.
%		We do this by observing that
		This follows from the fact that the good cover arising from $\{B_s^+ \cap B_s'^+\}_{s = \ell,r} \cup \{L_i^+ \cap L_i'^+\}_{i = 1, \ldots, n}$ gives rise to a good cover realization of $\scC(\cM, g)$.
		This completes the proof.
	\end{proof}

	\begin{ithm}
	 Let $\cM = (E, \cL)$ be a uniform, rank 3 oriented matroid. Then for $g\in E$, the code $\scC(\cM, g)$ is convex if and only if $\cM$ is representable. \label{prop:cvx_iff_rep}
	 \end{ithm}
	 
	\begin{proof}
	First, we show that if $\cM$ is representable, $\scC(\cM, g)$ is convex. Note that by Proposition \ref{prop:img_of_matroid}, we have that $\scC(\cM, g) \leq \sfL^+\widehat\cM$. Also note that by construction, if $\cM$ is representable, then so is $\widehat\cM$. Therefore, by Theorem \ref{thm:polytope_matroid}, if $\cM$ is representable, $\scC(\cM, g)$ is convex.

Next, we show that if $\scC(\cM, g)$ is convex, then $\cM$ is representable. Note that the following sequences are order-forced in $\scC(\cM, g)$. 
 
\begin{enumerate}

 \item The only feasible path from $  b_r   b_\ell   d_{r1} d_{\ell 1}$ to  $ b_{r} a _n d_{rn}$ in $G_{\scC(\cM, g)}$ is 
 \begin{align*}  b_r   b_\ell   d_{r1} d_{\ell 1} \lra     b_r   d_{r1} \lra   b_r   a _1  d_{r1} d_{r2} \lra   b_r   d_{r2} \lra \cdots \lra  b_{r} d_{rn} \lra  b_{r} a _n d_{rn}
 \end{align*}
 
  \item The only feasible path from $  b_r   b_\ell   d_{r1} d_{\ell 1}$ to  $  b_\ell  a _1 d_{\ell n}$ in $G_{\scC(\cM, g)}$ is 
  \begin{align*}
    b_r   b_\ell   d_{r1} d_{\ell 1} \lra     b_\ell   d_{\ell 1} \lra   b_\ell   a _n  d_{\ell 1} d_{\ell 2} \lra   b_r   d_{\ell 2} \lra \cdots \lra   b_\ell  d_{\ell n} \lra   b_\ell  a _1 d_{\ell n}
  \end{align*}

  \item For each $i$, the only feasible path from $ b_r  a_i  c_{i1}  d_{ri} d_{ri+1}$ to $ a _i   b_\ell   c_{in} d_{\ell (n-i+1)}$ in \\ $G_{\scC(\cM, g)}$ is 
  \begin{align*} 
    b_r  a _i  c_{i1}  d_{ri} d_{ri+1} 
    &\lra  a_i c_{i1}
    \lra  a_i  a_{\pi\inv_i(1)}  c_{i1} c_{i2}c_{\pi\inv_i(1), \pi_{\pi\inv_i(1)}(i)}c_{\pi\inv_i(1), \pi_{\pi\inv_i(1)}(i)+1} \\
   \cdots &
    \lra  a _i  c_{i2} 
   \lra  a _i  c_{i(n-1)} \\
   \cdots &
   \lra  a_i  a_{\pi\inv_i(n-1)}  c_{i(n-1)} c_{in}c_{\pi\inv_i((n-1)), \pi_{\pi\inv_i(n-1)}(i)}c_{\pi\inv_i(1), \pi_{\pi\inv_i(n-1)}(i)+1}\\ 
   \cdots& \lra  a _i  c_{in}
   \lra  a _i   b_\ell   c_{in} d_{\ell (n-i+1)}
  \end{align*}
 \end{enumerate}

We claim if $\scC(\cM, g)$ is convex, then it has a realization in the plane.
Suppose that  $\scC(\cM, g)$ is convex, and fix a realization $\cU$ in $\R^d$.  
Choose points $p_1, p_2, p_3$ in the atoms $A_{  b_r   b_\ell   d_{r1} d_{\ell 1}}$, $A_{ b_{r} a _n d_{rn}}$, and $A_{  b_\ell  a _1 d_{\ell n}}$ respectively. 
We will show that each atom in this realization has a nonempty intersection with $\conv(p_1, p_2, p_3)$. 
By order forcing (1), the line  from $p_1$ to $p_2$ must pass through the atoms of all codewords containing $b_r$ in the listed order. 
Likewise, by order forcing (2), the line from $p_1$ to $p_3$ must pass through the atoms of %the codewords $$  b_r   b_\ell   d_{r1} d_{\ell 1},    b_\ell   d_{\ell 1} ,    b_\ell   a _n  d_{\ell 1} d_{\ell 2} ,    b_r   d_{\ell 2} , \ldots ,   b_\ell  d_{\ell n} ,    b_\ell  a _1 d_{\ell n}$$ in this order. 
all codewords containing $b_\ell$ in the listed order. 

%This accounts for all codewords containing $  b_r $ or $  b_\ell $. 
 
In particular, we have shown that for each $i$, the atoms of $\sigma = b_ra _id_{ri}d_{r(i+1)}c_{i1}$ and $\tau = b_\ell  a _i d_{\ell (n+1-i)} d_{\ell (n+2-i)} c_{i1}$ have a nonempty intersection with $\conv(p_1, p_2, p_3)$. For each $i$, pick a point $q_i \in \conv(p_1, p_2, p_3) \cap A_{\sigma}$ and a point $r_i \in \conv(p_1, p_2, p_3) \cap A_{\tau}$. Applying order forcing (3) for each $i$, we have that the line from $r_i$ to $q_i$ passes through the atoms of  % the codewords 
%\begin{align*} 
% b_r  a _i  c_{i1}  d_{ri} d_{ri+1} ,  a _i c_{i1},  a _i  a _{\pi(i,1)}  c_{i1} c_{i2} , a _i  c_{i2} , \ldots  a _i  c_{i(n-1)} ,  a _i a _{\pi(i,n)}  c_{i(n-1)}  c_{in},  a _i  c_{in}, a _i  b_\ell   c_{in} d_{\ell (n-i+1)}
%\end{align*}
%in this order. 
all codewords containing $a_i$, in the listed order. 
This accounts for every codeword of $\scC(\cM, g)$. Thus, intersecting the open  sets in $\cU$ with the plane $\aff(p_1, p_2, p_3)$ produces a two-dimensional convex realization of $\scC(\cM, g)$. 

Now, we obtain a straight line arrangement for $(\cM, g)$ in this plane by extending the line segment from $q_i$ to $r_i$ to be a line. Notice that by order forcing (3), this line meets the sets $U_{ a _1}, \ldots, U_{ a _n}$ in the order consistent with the pseudoline arrangement. Thus, if this code is convex, then the pseudoline arrangement is stretchable, and thus $\cM$ is representable. 
\end{proof}

Proposition \ref{prop:cvx_iff_rep} demonstrates that matroid representability and
convex code realizability are intertwined. One consequence is that
non-representable oriented matroids are a new source
for constructing non-realizable codes:

\begin{cor}
  There is an infinite family of non-convex codes which lie below
  oriented matroids in $\pcode$.
 \end{cor}

\begin{proof}
There are infinitely many non-representable uniform oriented matroids of rank 3 \cite[Proposition 8.3.1]{bjorner1999oriented}. 
By Proposition~\ref{prop:cvx_iff_rep}, $\scC(\cM, g)$ is non-convex for each of these.
By Proposition \ref{prop:img_of_matroid}, $\scC(\cM, g) \leq \sfL^+\widehat\cM$. 
\end{proof}

\begin{ex} Let $(\cM, g)$ be the uniform non-Pappus matroid from \cite{shor1991stretchability}, whose pseudoline arrangement appears in Figure \ref{fig:non-papp}. This matroid is non-representable, since a realization of it would violate Pappus's hexagon thoerem. Then $\scC(\cM, g)$ is a non-convex code with no local obstructions. 

\begin{figure}[ht!]
\begin{center}
\includegraphics[width = 3.5 in]{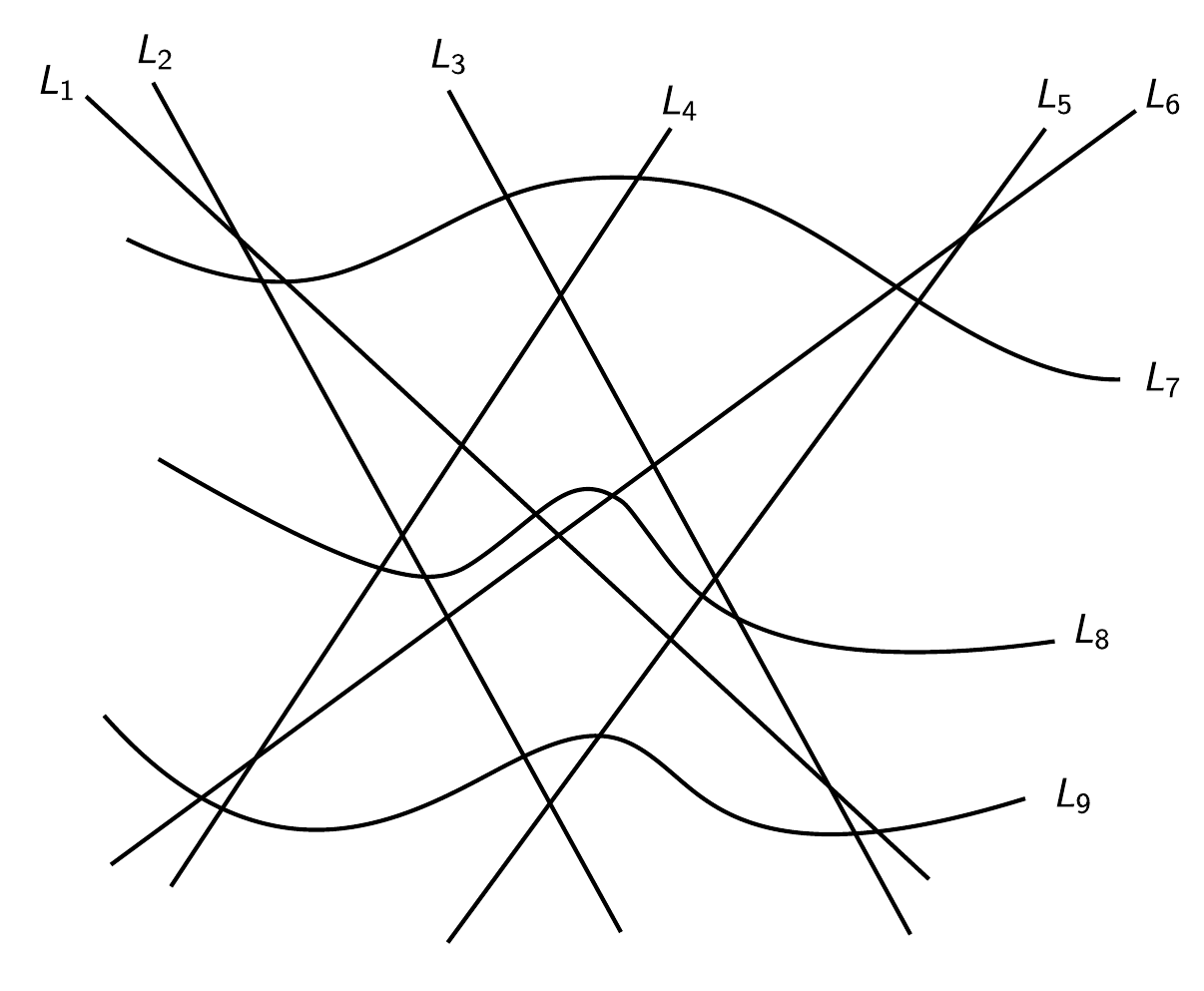}
\end{center}
\caption{The pseudoline arrangement for the uniform non-Pappus matroid. \label{fig:non-papp} }
\end{figure}

\end{ex}

\subsection{The convex code decision problem is NP-hard}

	We now turn to the computational aspects of convex codes.
	Using the relationship between convex codes and representable oriented matroids (Theorem~\ref{thm:witness}), we demonstrate the convex code decision problem is NP-hard and $\exists\R$-hard, though it remains open whether the convex code decision problem lies in either of these classes, or is even decidable.
	The complexity class $\exists \R$, read as \emph{the existential theory of the reals}, is the class of decision problems 
%	which ask if a sentence 
	of the form 
		\[ \exists (x_1\in \R) \ldots \exists(x_n\in \R) P(x_1, \ldots, x_n), \]
%	is true of false. In this definition, $P$ is a quantifier free formula whose 
	where $P$ is a quantifier-free formula whose atomic formulas are polynomial equations, inequations, and inequalities in the $x_i$.  In other words, a problem in $\exists \R$ defines a semialgebraic set over the real numbers and asks whether or not it contains any points \cite{broglia2011lectures}. Many well known problems in computational geometry lie in $\exists \R$, including some problems very similar to determining whether a code is convex. For instance, determining whether a graph is the intersection graph of convex sets in the plane is $\exists \R$ complete \cite{schaefer2009complexity}. 
	
Theorem~\ref{thm:witness} implies the convex code decision problem is at least as difficult as deciding if an oriented matroid is representable.
This decision problem is $\exists\R$-complete \cite{mnev1988universality, shor1991stretchability, sturmfels1987decidability} and therefore the convex code decision problem is $\exists\R$-hard.
	
\begin{ithm}
	Any problem in $\exists \R$ can be reduced in polynomial time to the problem of determining whether a neural code is convex. 
\end{ithm}
	
\begin{proof}

By the Mn\"ev-Sturmfels universality theorem (see \cite{mnev1988universality, sturmfels1987decidability, shor1991stretchability, bjorner1999oriented}), determining whether a rank 3 uniform oriented matroid is representable is complete for the existential theory of the reals.
By Proposition~\ref{prop:cvx_iff_rep}, a rank 3 uniform oriented matroid $\cM$ is representable if and only if $\scC(\cM, g)$ is a convex neural code. Further, the number of neurons in $\scC(\cM, g)$  is quadratic in the size of the ground set of $\cM$, and the number of codewords of $\scC(\cM, g)$ is less than the number of covectors of $\cM$. Any problem in $\exists\R$ can be reduced in polynomial time to deciding representability of a uniform oriented matroid and thus convexity of the corresponding code.
\end{proof}

\noindent 
Since any $\exists \R$ complete problem is also $\mathrm{NP}$-hard, we have as a corollary that determining whether a code is convex is $\mathrm{NP}$-hard.
 \begin{cor}
The problem of determining whether a code is convex is $\mathrm{NP}$-hard, where the problem size is measured in the number of codewords. 
\end{cor}
Notice that because we can perform this reduction of a problem in $\exists \R$ to a neural code in polynomial time, this result holds even when we measure the problem size in terms of the number of codewords, which may be exponentially large in the number of neurons.  Again, this NP hardness result is not surprising. For instance, it parallels the result that recognizing whether a simplicial complex is the nerve of convex sets in $\R^d$ is $\mathrm{NP}$-hard for $d\geq 2$ \cite{tancer2010d}.

\section{Open questions}\label{S:questions}

The preceding sections have presented our case for employing
oriented matroid theory in the study of neural codes.
However, we stand at the very beginning of exploring this connection. In this
section, we outline some directions for future work. 

\subsection{Concrete Examples}

\begin{question}
Are the non convex codes $\cT$ and $\cR$ from the previous chapter images of oriented matroids? 
\end{question}

\subsection{Is the missing axiom of convex codes also lost forever?}
 While general neural codes are not required to satisfy any axioms,
the codes below oriented matroids may be more tractable to combinatorial description.
\begin{question}\label{q:axioms}
  Can the class of neural codes below oriented matroids be characterized
  by a set of combinatorial axioms?
\end{question}

If this question is answered in the affirmative, then these codes can be thought of as
``partial oriented matroids.''  
Suppose that $\scC\subseteq 2^{[n]}$ is a code and
$\cM$ is an oriented matroid on ground set $[N]$ such that $\scC = f(\sfL^+( \cM))$; then,
we obtain constraints on the set of covectors of $\cM$.  
Each included codeword $\sigma \in\scC$ implies existence of a preimage covector in $\cM$, and
each excluded codeword $\tau\notin\scC$ implies a set of forbidden covectors which may not be in $\cM$.
The oriented matroids satisfying these constraints can then be said to
be ``completions'' of the partial oriented matroid.

%Thus, we ask whether it is possible to build a useful theory of partially defined matroids. 
%In particular, are there ways to translate from ``partial covectors" to partial chirotopes, partial circ%uits, or other partial descriptions of a matroid?

Just as we wish to characterize codes lying below oriented matroids with a set of combinatorial axioms, we might also wish to characterize convex codes using a set of combinatorial axioms. However, this is likely not possible. 
In \cite{mayhew2018yes}, Mayhew, Newman, and Whittle show that ``the missing axiom of matroid theory is lost forever." Slightly more formally, they show that there is no sentence characterizing representability in the monadic second order language $MS_0$, which is strong enough to state the standard matroid axioms. Roughly, this means that there is no ``combinatorial" characterization of representability, or no characterization of representability in the language of the other matroid axioms. 

Because we have found strong connections between representability and convexity, it is natural to ask whether a similar statement can be proven for convex codes. 

\begin{question}
Is there a natural language in which we can state ``combinatorial" properties of neural codes, in analogy with the $MS_0$ for matroids? If so, is it possible to characterize convexity in this language?
\end{question}

\subsection{Computational questions}
While we have shown that the convex code decision problem is $\exists \R$-hard,  we have not actually shown that the convex code decision problem lies in $\exists \R$, or is even algorithmically decidable. 
A similar problem, that of determining whether a code has a good cover realization, is undecidable by \cite[Theorem 4.5]{chen2019neural}. 
Here, the distinction between codes with good cover realizations and convex realizations may be significant. 
For instance, while there is an algorithm to decide whether, for any given $d$, a simplicial complex is the nerve of convex open subsets of $\R^d$, for each $d \geq 5$, it is algorithmically undecidable whether a simplicial complex is the nerve of a good cover in $\R^d$ \cite{tancer2013nerves}. 

We outline a possible path towards resolving \cite[Question 4.5]{chen2019neural}, which asks whether there is an algorithm which decides whether a code is convex.	
Our approach hinges on \ref{thm:polytope_matroid}: a code is polytope convex if and only if it lies below a representable oriented matroid.
A first step towards solving the convex code decision problem is answering the following open question:

\begin{question}\label{Q:convexiffpolytope}
Can every convex code be realized with convex polytopes?
\end{question}
	
If this can be answered in the affirmative,
then our Theorem \ref{thm:polytope_matroid} becomes strengthened to the following: 

\begin{conj}
  A code $\scC$ is convex if and only if $\scC \leq \sf{L}^+\cM$ for $\cM$
  a representable oriented matroid. 
\end{conj}

If this conjecture holds, then we can replace the problem of
determining whether a code is convex with the problem of determining
whether a code lies below a representable matroid. 
We only need to enumerate matroids above the code, and then
check these matroids for representability. 

\begin{question}\label{Q:OMsabovecode}
  Given a code $\scC$, is there an algorithm to enumerate the set of oriented
  matroids $\cM$ which lie above $\scC$? 
\end{question}

One way to find oriented matroids above a code $\scC$ is to travel step-by-step up the poset $\pcode$. 
While there is a straightforward algorithm to enumerate the $O(n)$
codes which are covered by a code $\scC\subseteq 2^{[n]}$ in $\mathbf{P}_{\mathbf{Code}}$
\cite{jeffs2019sunflowers}, we do not know of a straightforward way to characterize the codes which cover $\scC$. 
If we can characterize these codes as well, we may be able to find a way to ``climb up''
towards an oriented matroid. Alternatively, we can use the ``partial oriented matroid''
perspective described above to obtain a set of constraints that must be obeyed by any oriented
matroid above this code. Then we can look for a matroid satisfying these constraints.

Both of these approaches depend on the minimal size of the
ground set of oriented matroids that lie above $\scC$ in $\pcode$. 
Let 
\[ M(n) = \max_{\substack{\scC \subseteq 2^{[n]}
 \\ \scC \mbox{   \scriptsize  below an}\\\mbox{  \scriptsize oriented matroid}}} \left[ \min_{\scC\leq \sf{L}^+(\cM)} |E(\cM)|\right] \]
 be the smallest $N$ such that any code $\cC$ on $n$ neurons which lies below an oriented matroid lies below an oriented matroid with ground set of size at most $N$. 
Similarly, let
\[ H(n) = \max_{\substack{\scC \subseteq 2^{[n]}
 \\ \scC \mbox{   \scriptsize  below a representable}\\\mbox{  \scriptsize oriented matroid}}} \left[ \min_{\substack{\scC\leq \sf{L}^+\cM \\ \cM \mbox{   \scriptsize  representable}}} |E(\cM)| \right] \] be the smallest $N$ such that any code $\scC$ on $n$ neurons below a representable oriented matroid lies below a representable oriented matroid with ground set of size at most $N$. 
Clearly, $M(n) \leq H(n)$, since any representable matroid is a matroid. 

\begin{question}
Describe the growth of $M(n)$ and $H(n)$ as functions of $n$. Are they equal?
\end{question}

Note that if $H(n)$ is a computable function of $n$, and Question \ref{Q:convexiffpolytope}
is answered in the affirmative, then the convex code decision problem is decidable. 

\subsection{Polyhedral approximation questions}

The first part of this program is to prove  Conjecture  ~\ref{Q:convexiffpolytope}. 
Standard theorems about approximating convex bodies with convex polytopes suffice to prove that any simplicial complex is the nerve of a collection of interiors of convex polytopes. However, the code of a cover is a more delicate object than the nerve, so approximation techniques may fail.

\subsection{Other questions in geometric combinatorics}

Many classic theorems about convex sets, such as Helly's theorem, Radon's theorem, and Caratheodory's theorem, have oriented matroid analogues. In some way, we can view our Theorem \ref{thm:bad_sunflower} as an oriented matroid version of Jeffs' sunflower theorem \cite[Theorem 1.1]{jeffs2019sunflowers}. The fact that the non-convex codes constructed from the sunflower theorem do not lie below oriented matroids shows us that there is some fact about oriented matroids underlying the sunflower theorem.

\begin{question}
Is there a natural oriented matroid version of Jeffs' sunflower theorem? 
\end{question}

Proposition~\ref{prop:no_local} stated that if $\cM$ is an oriented matroid,
the code $\sf{L}^+\cM$ has no local obstructions. That is, for any
$\sigma \in \Delta(\sf{L}^+\cM)\setminus \sf{L}^+\cM$,
$\link_\sigma( \Delta(\sf{L}^+\cM))$ is contractible.
This result can also be found in \cite{edelman2002convex},
where is is phrased as a result about the simplicial complex $\Delta_{\mathrm{acyclic}}(\cM)$.
Something stronger holds for representable oriented matroids: by \cite[Theorem 5.10]{chen2019neural},
if $\cM$ is a {\bf representable} oriented matroid, and
$\sigma \in \Delta(\sf{L}^+\cM)\setminus \sf{L}^+\cM$, then
$\link_\sigma(\sf{L}^+\cM)$ must be collapsible.
Expanding upon this work, \cite{jeffs2021convex} gives stronger conditions
that the link of a missing codeword in a convex code must satisfy. 

We ask whether this holds for all oriented matroids:

\begin{question}
  If $\cM$ is an oriented matroid, and $\sigma \in \Delta(\sf{L}^+\cM)\setminus \sf{L}^+\cM$,
  is $\link_\sigma( \Delta(\sf{L}^+\cM))$ collapsible? More generally, which simplicial complexes can arise as $\link_\sigma( \Delta(\sf{L}^+\cM))$ for $\sigma \in \Delta(\sf{L}^+\cM)\setminus \sf{L}^+\cM$? 
\end{question} 

If not, then the non-collapsibility of $\link_\sigma( \Delta(\sfL^+( \cM)))$ gives a new ``signature" of non-representability.

%% file: UnderlyingRank/UnderlyingRank.tex
%auto-ignore 

% !TEX root = ../YourName-Dissertation.tex

\chapter{A Novel Notion of Rank for Neural Data Analysis } \label{chapter:urank}

This chapter is adapted from an upcoming paper, which is joint work with Carina Curto, Juliana Londono Alvarez, and Hannah Rocio Santa Cruz \cite{curto22novel}. 
\section{Introduction}
Monotone nonlinear transformations present challenges for data analysis in neuroscience and beyond. For instance, calcium imaging is a widely used tool for recording the activity of large populations of neurons. However, calcium fluorescence has a nonlinear, but monotone, relationship with underlying spiking activity \cite{akerboom2012optimization, huang2021relationship, siegle2021reconciling}. This has consequences for analysis of population codes \cite{nauhaus2012nonlinearity}. The nonlinear relationship between membrane potential and spiking activity also effects analysis of neural coding \cite{hansel2002noise}. In psychology, signed difference analysis is used to evaluate models when there is an unknown monotone relationship between underlying psychological constructs and measured variables, such as task performance \cite{dunn2018signed,dunn2003signed}. In studies of gene interaction, there is often a monotone, nonlinear relationship between the trait genes directly act on and the trait we are are able to measure \cite{weinreich2006darwinian, husain2020physical, otwinowski2018biophysical}. 

Here, we focus on estimating the dimensionality of neural activity in the presence of monotone nonlinear transformations. The dimensionality of neural activity has emerged as a key variable describing how populations of neurons encode stimuli and perform tasks \cite{cunningham2014dimensionality, stringer2019high}. In many cases, neural activity has been observed to have dimensionality much lower than the number of neurons \cite{machens2010functional, chapin1999principal, churchland2012neural }.  Standard linear methods for computing the dimensionality of neural activity rely on singular values, which are not stable under monotone nonlinear transformation. In Figure \ref{fig:distortion}, we see the effect of the saturating monotone nonlinear transformation $f(x) = 1 - e^{-5x}$ on the singular vales of a $20\times 20$ matrix with rank 5. There is a sharp drop after the first five singular values, consistent with the low rank of the matrix.  This drop does not appear in the transformed data: instead, singular values decay smoothly. How can we get around this to recover the original rank?

\begin{figure}[ht!]
    \centering
    \includegraphics[width = 3 in]{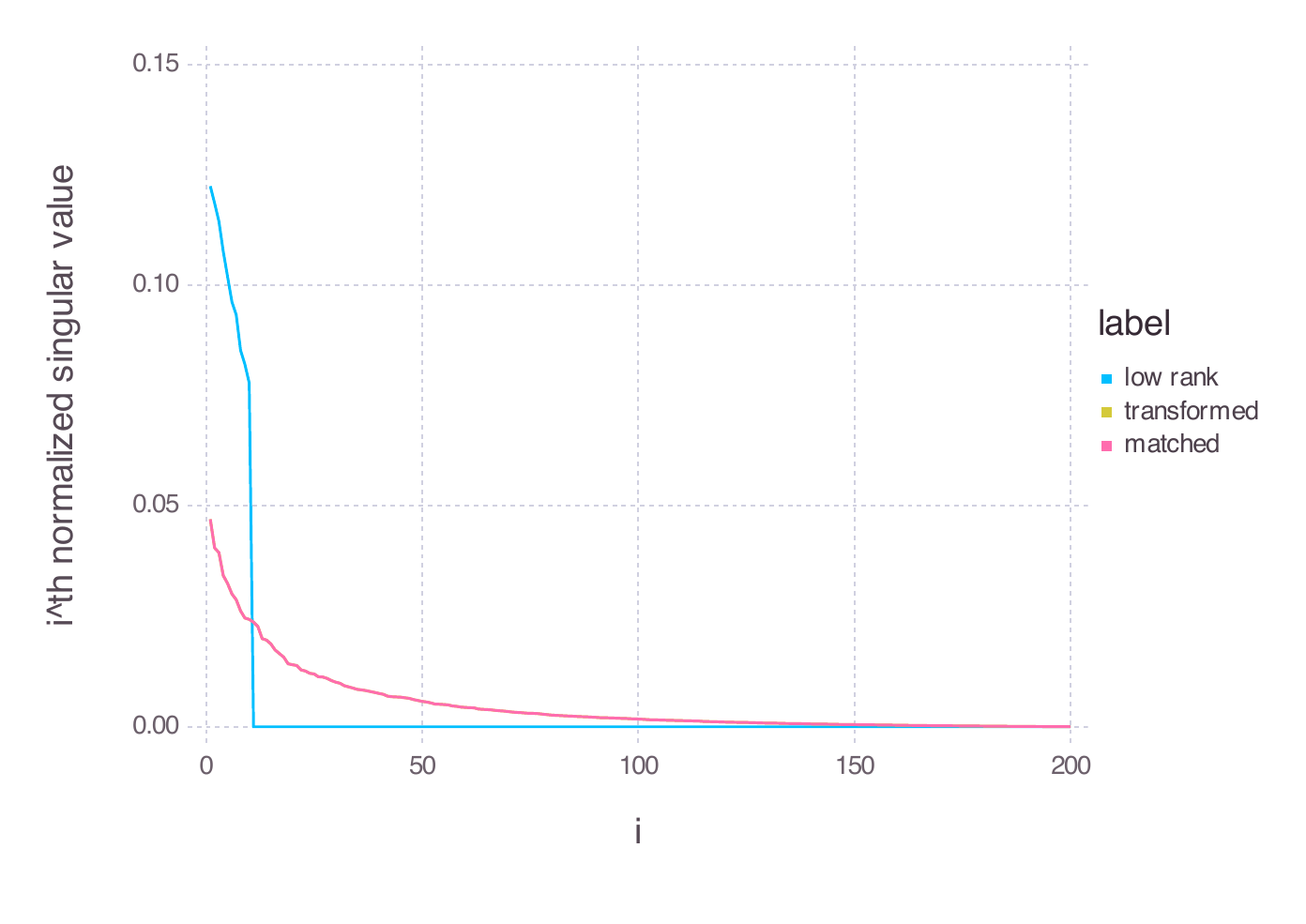}
    \includegraphics[width = 3 in]{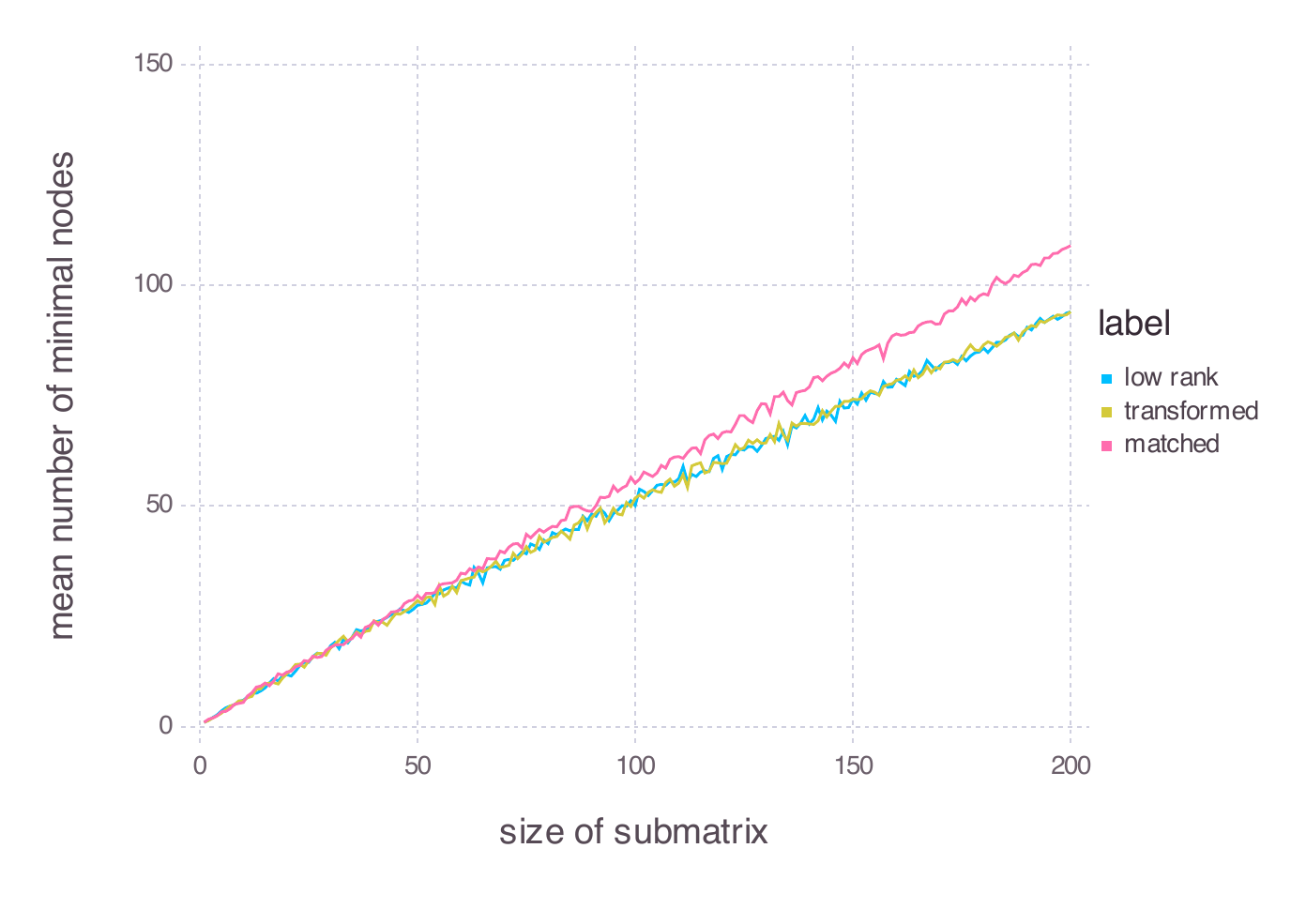}
    \caption[The effect of a monotone nonlinear transformation on singular values.]{Left: We plot the normalized singular values for three $20\times 20 $ matrices. The ``low rank" matrix a random matrix with rank 5. The ``transformed" matrix is  the result of applying the saturating function  $f(x) = 1 - e^{-5x}$ entrywise on the low rank matrix. Finally, the ``matched" matrix is a random matrix generated to have the same singular values as the transformed matrix. The ``transformed" and ``matched" curve are indistinguishable, by construction.
    Right: We are able to distinguish the ``transformed" and ``matched" matrices using the technique of counting \emph{minimal nodes}, introduced in Section \ref{sec:minimal} of this paper. Because minimal modes are stable under monotone nonlinear transformations, the minimal node curves of the low-rank and transformed matrices overlap.    \label{fig:distortion}}

\end{figure}

In this paper, we introduce the \emph{underlying rank} of a matrix, a notion of rank which is stable under monotone nonlinear transformations. The underlying rank of a matrix $A$ is the minimal rank $d$ such that there exists a rank $d$ matrix $B$ and a monotone function $f$ such that $A_{ij} = f(B_{ij})$. In other words, the underlying rank of a matrix is the minimal rank consistent with the ordering of matrix entries. When we interpret the matrix $A$ as our data and our matrix $B$ as the unobserved underlying values, the underlying rank $d$ is a lower bound for the true dimensionality of our data. Our goal in this paper is to put the concept of underlying rank into mathematical context and to present techniques for estimating the underlying rank of a given matrix.

In Section \ref{sec:geom}, we show that we can associate a rank $d$ matrix to a pair of point configurations in $\R^d$. We show that it is possible to recover features of these point arrangements from the ordering of matrix entries. This makes it possible to bring toolset of discrete geometry in to estimate underlying rank. In Section \ref{sec:minimal}, we introduce the \emph{minimal nodes} of a matrix as a tool for estimating the underlying rank of a random matrix. In In Section \ref{sec:radon}, we apply Radon's theorem to obtain lower bounds for underlying rank. Section \ref{sec:zebrafish}, we give a proof of concept by applying this analysis to calcium imaging data from the optic tectum in zebrafish larvae. Our discussion of underlying rank will continue into the next chapter, where we will relate underlying rank to oriented matroid theory and show that computing underlying rank is hard. 

\section{The geometry of underlying rank}
\label{sec:geom}

We can decompose a rank $d$ matrix $A$ as a pair of point arrangements in $\R^d$. The order of entries of $A$ carries information about these point arrangements which can help us estimate the underlying rank  of $A$. 

In this section, we introduce the key idea of this paper, that which underlies all of our results: rank $d$ matrices correspond to point arrangements in $\mathbb{R}^d$ (Figure \ref{fig:point_arr}), and the ordering of matrix entries encodes geometric information about this point arrangement (Figure \ref{fig:geom_pic}).

\begin{figure}[h]
\begin{center}
\includegraphics[width = 4 in]{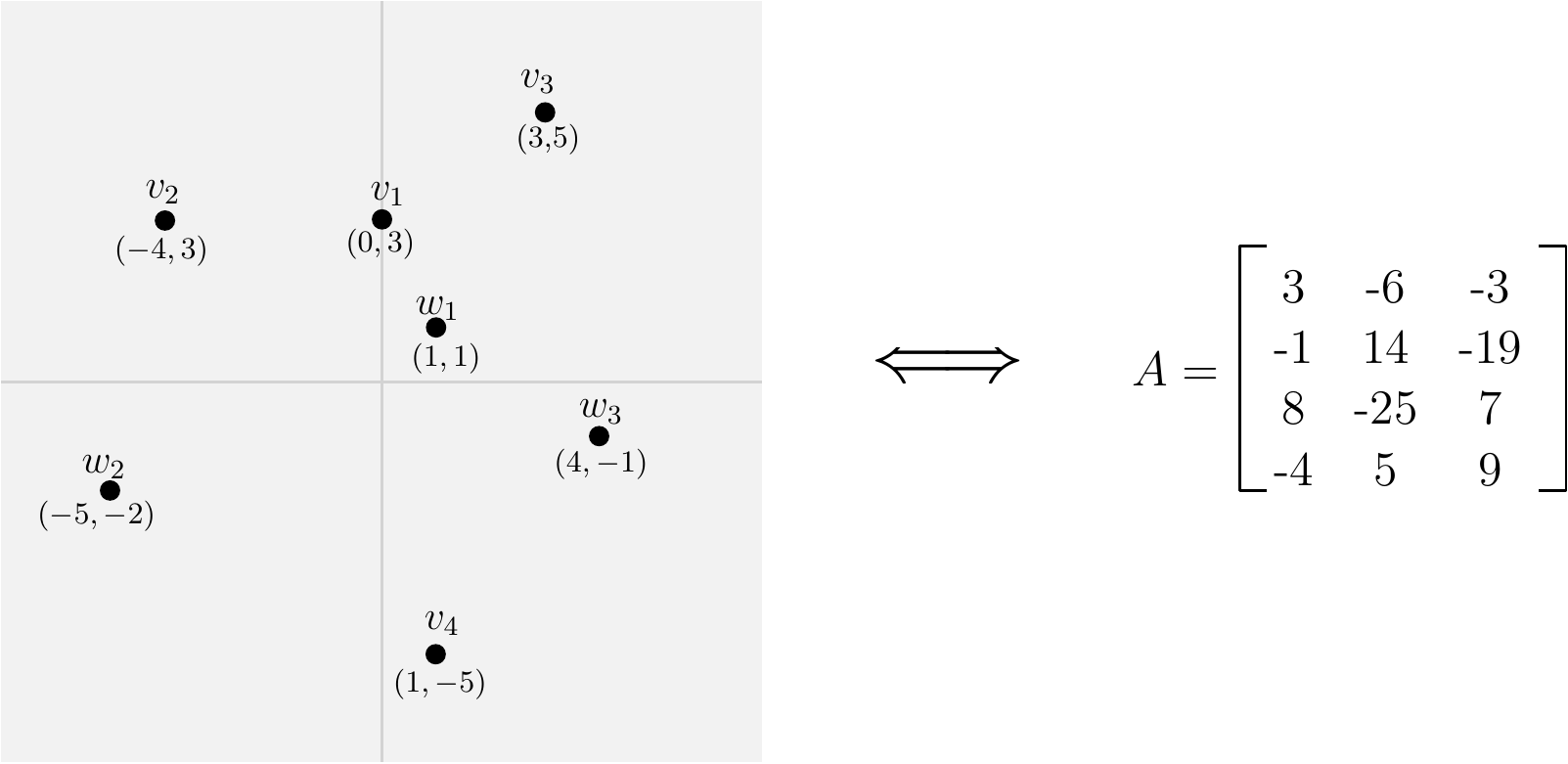}
\caption{\label{fig:point_arr}  A point arrangement in $\R^2$ generating the matrix $A$.} 
\end{center}
\end{figure}

\begin{prop} Let $A$ be a $m\times n$ matrix of rank $d$ with real entries. Then there exist point arrangements $ v_1,  v_2, \ldots,  v_m\in \R^d$, and $ w_1,  w_2, \ldots,  w_n \in \R^d$ such that $A_{ij} =  v_i\cdot  w_j$. Further, if $A$ is positive semidefinite, then we can take these two point arrangements to be the same. 
\end{prop}

\begin{proof}
Since $A$ is a $m\times n$ matrix of rank $d$, $A$ has a rank factorization 
$A = VW$
where $V$ is $m\times d$ and $W$ is $d\times n$. Further, if $A$ is positive semidefinite, we can take $W = V^T$. 
Let $v_i$ be the $i$-th row of $V$ and $w_j$ be the $j$-th column of W. By matrix multiplication, we have 
\begin{align*}
%A_{ij}&=  v_i \cdot  w_j\\
\begin{pmatrix}
A_{11}& A_{12}& \ldots &A_{1n}\\
A_{21}& A_{22}& \ldots &A_{2n}\\
\vdots &\vdots & \ddots & \vdots\\
A_{m1} & A_{m2}& \ldots & A_{mn}
\end{pmatrix} &= 
\begin{pmatrix}
v_1\cdot  w_1& v_1\cdot  w_2& \ldots &v_1\cdot  w_n\\
v_2\cdot  w_1& v_2\cdot  w_2& \ldots &v_2\cdot  w_n\\
\vdots &\vdots & \ddots & \vdots\\
v_m\cdot  w_1& v_m\cdot  w_2& \ldots & v_m\cdot  w_n
\end{pmatrix}
\end{align*}
\end{proof}

From this, we have a correspondence between matrices of underlying rank $d$ and point arrangements in $\mathbb{R}^d$. First, we restate the definition of underlying rank.

\begin{defn}
The \textit{underlying rank}  $\ur(A)$ of  $A$  is the minimum value of $d$ such that there exists a rank-$d$ matrix $B$ such that $A_{ij} \leq A_{kl}$ if and only if $B_{ij} \leq B_{kl}$:
\begin{align*}
   \ur(A) = 
    \min\{\rank(B) \mid  A_{ij} \leq A_{kl} 
    \Longleftrightarrow B_{ij} \leq B_{kl}\}
\end{align*}
In other words, it is the minimum value of $d$ such that there exists a rank-$d$ matrix $B$ and a strictly increasing function $f$ such that $A_{ij}=f(B_{ij})$.
\end{defn}

It is convenient to define a matrix with small integer entries which has the same order of $A$. Thus, we define the \emph{order matrix} of $A$ to be a matrix entries
\begin{align*}
\order(A)_{ij} = |\{a_{kl} \mid a_{kl} \leq a_{ij}\}|.
\end{align*}

\begin{cor}
If $A$ is a $m\times n$ matrix of \textbf{underlying} rank $d$ with real entries, then there exist points $$ v_1,  v_2, \ldots,  v_m,  w_1,  w_2, \ldots,  w_n \in \R^d$$ and a monotone function $f$ such that $A_{ij} = f( v_i\cdot  w_j)$. 
%If $A$ is symmetric, we can take $ v_i =  w_i$.
%If is a $A$ $n\times n$ symmetric matrix of underlying rank $d$ with real entries, then there exist points $ v_1,  v_2, \ldots,  v_n \in \R^d$ and a monotone function $f$ such that $A_{ij} = f( v_i\cdot  w_j)$. We refer to $ v_1,  v_2, \ldots,  v_n\in \R^d$ as a rank $d$ representation of $A$.
If $A$ is a positive semidefinite matrix, then we can choose the underlying point arrangements to be the same, with $v_i = w_i$. 
\end{cor}

\begin{figure}[h]
\begin{center}
\includegraphics[width = 4 in]{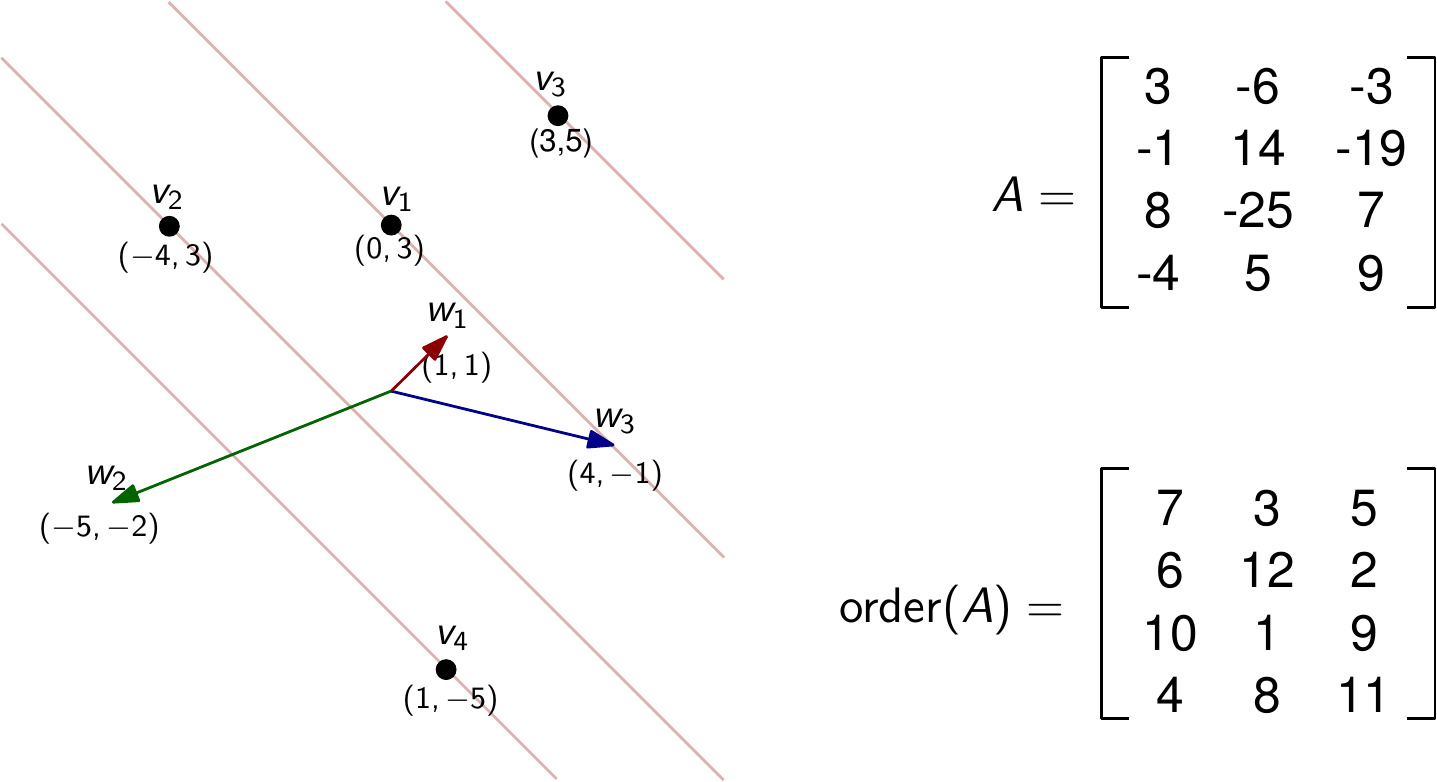}
\end{center}
\caption[ A rank-two representation of the matrix $A$.]{\label{fig:geom_pic}  A rank-two representation of the matrix $A$. The points of $V$  are shown as black disks and the points of $W$ are shown as colored vectors. In light red, a hyperplane with normal vector $w_1$ is shown sweeping across space. It encounters the points of $V$ in the order $v_4, v_2, v_1, v_3$.  }
\end{figure}

\begin{defn}
A \emph{rank-d representation} of a matrix $A$ is a set of points 
$$ v_1,  v_2, \ldots,  v_m,  w_1,  w_2, \ldots,  w_n \in \R^d$$ and a monotone function $f$ such that $A_{ij} = f( v_i\cdot  w_j)$. 
\end{defn}

See Figure \ref{fig:geom_pic} for an example of a rank-2 representation of a matrix.

%{\color{purple} Did we decide to leave Vladimir the non symmetric matrices? If so, we need to change this example}
%\caitlin{I hope not! I think the paper is a lot easier to understand if we don't force this assumption!}

\begin{ex}\label{ex:rank_rep}
The matrix $A$ has underlying rank at most two.
\begin{align*}
A = 
\begin{pmatrix}
7 & 3& 5\\
6 & 12& 2 \\
10& 1 & 9 \\
4 & 8& 11 \\
\end{pmatrix}
\end{align*}
This is because the order of entries of $A$ is the same as the order of entries of $B$, which has rank two:
\begin{align*}
B = 
\begin{pmatrix}
3 & -6 & -3 \\
-1 & 14 & -19 \\
8 & -25 & 7 \\
-4 & 5 & 9 \\
\end{pmatrix}
\end{align*}
A rank-two representation of $A$ is given in Figure \ref{fig:geom_pic}. We have $w_1, \ldots, w_4$ and $v_1, \ldots, v_4$ given by the rows of the matrices
\begin{align*} W= 
\begin{pmatrix}
 1&1 \\
-5 & -2 \\
4 & -1 \\
\end{pmatrix}
\qquad 
V =
\begin{pmatrix} 
0&3 \\
-4 & 3\\
3 & 5 \\
1 & -5  \\
\end{pmatrix}.
\end{align*}
Notice that the order of entries within the first column corresponds to the order in which a hyperplane normal to the vector $w_1$ sweeps past the points $v_1, v_2, v_3, v_4$. In Lemma \ref{lem:sweep_order}, we will show that this is always true. Using this result, together with Corollary \ref{cor:rank_one}, we will show that this matrix has underlying rank exactly two. 

\end{ex}

As we saw in Example \ref{ex:rank_rep}, we are able to recover a considerable amount of geometric information about the point arrangement $ v_1, \ldots,  v_m$ using the order of entries of $A$. We can use this information to bound the underlying rank  of $A$. 

We first notice that the order of entries in each column of a matrix $A$ corresponds to the order in which a sequence of  hyperplanes sweep past the points $ v_1, \ldots,  v_m$ in any rank-$d$ representation of $A$. This is illustrated in Figure \ref{fig:geom_pic}. 

\begin{lem}\label{lem:sweep_order} Let $ v_1, \ldots,  v_m,  w_1, \ldots,  w_n$, $f$ be a rank-$d$ realization of $A$.

Let $r_{j}(t)$ be the hyperplane defined by $r_j(t) = \{ w\mid  v_j \cdot  w = t\}$. The order of entries in the $j$-th row of $A$ is equal to the order in which the hyperplane $r_{j}(t)$ encounters the points $ w_1, \ldots,  w_m$ as we increase $t$. 

Likewise, let $c_{i}(t)$ be the hyperplane defined by $c_i(t) = \{ v\mid  w_i \cdot  v = t\}$. The order of entries in the $i$-th column of $A$ is equal to the order in which the hyperplane $c_{i}(t)$ encounters the points $ v_1, \ldots,  v_m$ as we increase $t$.

\end{lem}

\begin{proof}
Set $t_{ij}$ to be the ``time" that $r_j(t)$ crosses the point $ w_i$. That is, $t_{ij}$, $ w_i \cdot  v_j = t_{ij}$. The order in which the sweeping hyperplane $r_i(t)$ crosses the points $ w_1, \ldots,  w_n$ is the order of the values of $t_{1j}, \ldots, t_{nj}$. 

Thus, $v_j$ is the $k^{th}$ point the hyperplane $h_i(t)$ hits if $t_{ij}$ is the $k^{th}$ smallest value among $t_{1j}, \ldots, t_{nj}$. Now, by the definition of a rank $d$ realization, $A_{ij} = f( w_i\cdot  v_j) = f(t_{ij})$. 
Since $f$ is monotone, the values of $A_{ij} = f(t_{ij})$ are in the same order as the values of $t_{ij}$. Thus, the order within the $j^{th}$ row is the same as the order in which the hyperplane $r_{j}(t)$ encounters the points $ w_1, \ldots,  w_m$ as we increase $t$. 

By the same argument, we can show that the order of entries in the $i$-th column of $A$ is equal to the order in which the hyperplane $c_{i}(t)$ encounters the points $ v_1, \ldots,  v_m$ as we increase $t$.

 \end{proof}

\begin{cor}\label{cor:rank_one}
If $A$ is a matrix of underlying rank one, then the order of entries in each row of $A$ is either the same as or the reverse of the order of entries in the first row of $A$. Likewise, the order of entries in each column of $A$ is either the same as or the reverse of the order of entries in the first column of $A$.
\end{cor}

\begin{proof}
Let $v_1, \ldots, v_n, w_1, \ldots, w_n,f$ be a rank-one realization of $A$. By Lemma \ref{lem:sweep_order}, the order of entries in the $j$-th row of $A$ is the order in which the hyperplane perpendicular to $v_j$ sweeps past the points $w_1, \ldots, w_n$. Since the points $w_1, \ldots, w_n$ are on the real line, there are only two possible orders to sweep past them, one of which is the reverse of the other. 
\end{proof}

Using Corollary \ref{cor:rank_one}, we can show that the matrix $A$ in Example \ref{ex:rank_rep} has underlying rank exactly two. We notice that the order in the second column is neither the same as nor reversed from the order in the first column, thus the underlying rank is greater than one. Since we have already demonstrated that it is as most two by providing a rank two representation, this shows that it is exactly two. As it turns out, $A$ has the highest possible underlying rank for a $3\times 3$ matrix.

% \begin{figure}
%     \centering
%  \includegraphics[width = 5 in ]{point_arr_1.pdf}
%     \caption{Caption}
%     \label{fig:my_label}
% \end{figure}

\begin{prop}
If $A$ is a $n\times n$ matrix, then $\ur(A)\leq n-1$.
\end{prop}

\begin{proof}
 By the matrix determinant lemma \cite
 {ding2007eigenvalues}(Lemma 1.1), for any ${u}, {v}\in \R^n$,
\begin{align*}{\displaystyle \det \left(A + {uv} ^{\textsf {T}}\right)=\left(1+ {v} ^{\textsf {T}}A ^{-1} {u} \right)\,\det \left(A \right)\,.}
\end{align*}
Let $\mathbf  1$ be the all ones vector. Let $\lambda = \mathbf 1^{\textsf {T}}A\inv\mathbf 1$. Notice that we can perturb $A$ without changing the order of entries so that $\lambda \neq 0$. Now, let $ u =  \mathbf 1$, $ v = \frac{-1}{\lambda}\mathbf 1$,  $B = A +  u v^{\textsf {T}}$. Notice that since $ u v^{\textsf {T}}$ is the matrix whose entries are all $-\lambda$, the order of entries in $B$ matches the order of entries in $A$. By the matrix determinant lemma, 
$$\det(B ) = \left( 1+ {v} ^{\textsf {T}}A ^{-1} {u} \right)\,\det \left(A \right) = 0.$$
Thus, $B$ is a matrix of rank at most $n-1$ whose entries are in the same order as those of $A$.

\end{proof}

In small cases, it is possible for this bound to reach equality: for $n = 2, 3, 4$, we can give examples of $n\times n$ matrices of rank $n-1$. In particular, any $2\times 2$ order matrix has underlying rank at least one, since the only rank-zero matrix is the all zeros matrix, which is not compatible with any strict ordering. In Example  \ref{ex:rank_rep}, we gave an example of a $3\times 3$ matrix with underlying rank two. In Example \ref{ex:allowable}, we give an example of a $4\times 4$ matrix of underlying rank three. We do not know whether, for all $n$, it is possible for an $n\times n$ matrix to have rank $n-1$.

\section{Minimal nodes: a practical tool to estimate underlying rank}
\label{sec:minimal}

In this section, we introduce the minimal, maximal, and extremal nodes of a matrix. These are the first practical tools we give for estimating the underlying rank. 

\begin{figure}[ht!]
    \centering
    \includegraphics[width = 4 in]{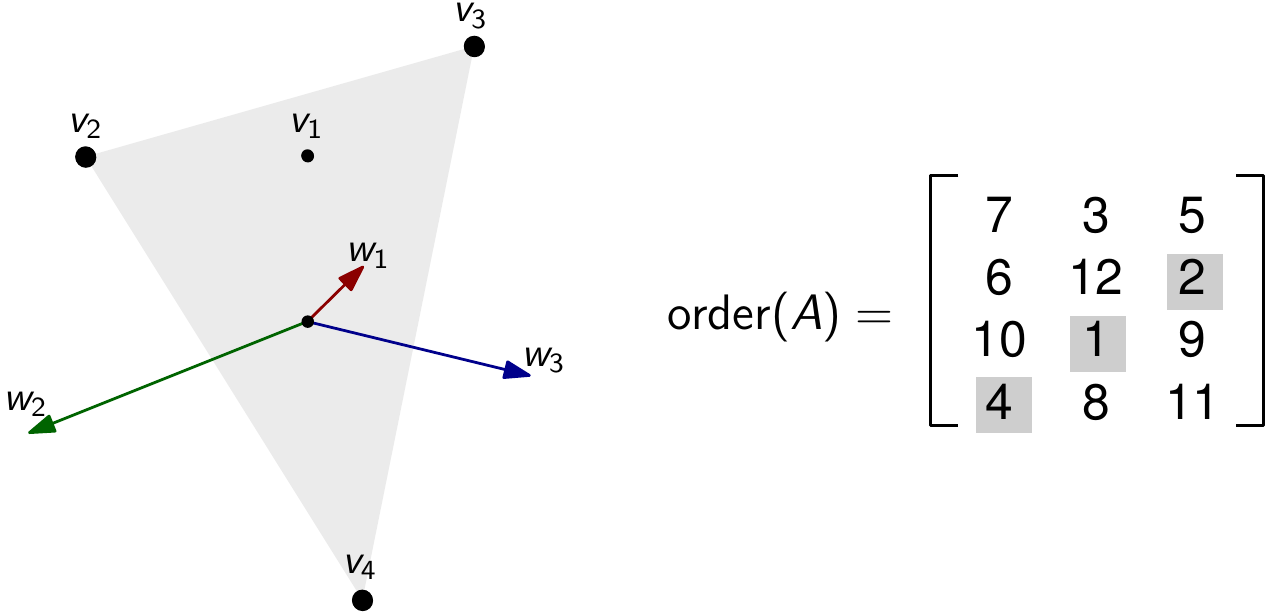}
    \caption[Minimal nodes]{The indices 2, 3 and 4 are minimal for the matrix $A$. These vertices are on the boundary of $\conv(v_1, v_2, v_3, v_4)$.  
    \label{fig:minimal}}
\end{figure}

\subsection{Minimal nodes}

\emph{Minimal}, \emph{maximal}, and \emph{extremal} nodes  are a useful feature for identifying low-rank matrices. 

\begin{defn}
Let $A$ be a $m\times n$ matrix. A row index $i\in [m]$ is a \emph{minimal node} if  $A$ has a column index $j\in [n]$ such that $A_{ij} < A_{kj}$ for all $k\in [m].$ That is, $i$ is minimal if there exists a column $j$ whose smallest entry is in the $i$-th row. 
An index $i\in [m]$ is a \emph{maximal node}  if  $A$ has a column index $j\in [n]$ such that $A_{ij} > A_{kj}$ for all $k\in [m].$ We say $i$ is an \emph{extremal node} if it is either maximal or minimal. 
\end{defn}

Extremal nodes capture a feature of the underlying point arrangement: in order to correspond to an extremal node, a point must be a vertex of the convex hull of the full set of points, illustrated in Figure \ref{fig:minimal}. 

\begin{defn}
 The \emph{convex hull} of a set of points  $v_1, \ldots, v_m \subset\R^d$ is 
\begin{align*}
    \conv\left(\{v_j\}_{j\in [m]}\right) = \left\{\sum_{i = 1}^m \lambda_j v_j\, \bigg\rvert\, \sum_{j=1}^m \lambda_j = 1, \lambda_j \geq 0\right\}
\end{align*}
\end{defn}

\begin{prop}
If $i$ is an extremal node of $A$, then $v_i$ is a vertex of $\conv(\{ v_j\}_{j\in [m]})$ for any $ v_1,  v_2, \ldots,  v_m\in \R^d$ rank-$d$ representation of $A$. \label{prop:convex_hull}
\end{prop}

\begin{proof}
If $i$ is an extremal node of $A$, then the linear functional $f(x) = w_j\cdot x$ is either maximized or minimized on $v_i$ among the points $v_1, \ldots, v_m.$ The maximum and minimum values of a linear function on the convex polytope $\conv(\{ v_j\}_{j\in [m]})$ occur on vertices. Thus, both the maximum and minimum value of   $f$ among  $v_1, \ldots, v_m$ must occur on a vertex. 
\end{proof}

%Maximal nodes can also be applied to the convex sensing problem we discuss in Subsection \ref{sec:cvx}. Since the maximum value of a quasi-convex function on a convex set must occur on the boundary, within the convex-sensing paradigm, maximal nodes must correspond to vertices of $\conv\{w_1, \ldots, w_n\}$. However, minimal nodes lose their meaning in this context. 

\subsection{Expected numbers of minimal nodes}

Because extremal nodes must correspond to vertices of $\conv( v_1, \ldots,  v_m)$, the expected number of extremal nodes is bounded by the expected number of vertices of a random polytope. As dimension increases, the expected number of vertices increases. Thus, we can use the number of extremal nodes to estimate the underlying rank. One complication is that this expected number of vertices on the convex hull of $m$ random points in $\R^d$ depends on the probability distribution used to choose the points. Another complication is that not all vertices are minimal nodes: in particular, if a point arrangement is chosen within the positive orthant, many vertices are prohibited from being minimal nodes. 

\subsubsection{Expected numbers of vertices}

Our first model for a random rank $d$ matrix chooses points $ v_1, \ldots,  v_m$ uniformly in a unit cube centered at the origin. 

\begin{prop}
\label{prop:expected_cube}
Let 
$A = f(v_i \cdot w_j)$ be an underlying rank-$d$ matrix whose underlying point arrangement $v_1, \ldots, v_m$ is drawn uniformly at random from the unit cube in $\R^d$. Then the expected number $n_e$ of extremal nodes is bounded above by 
$$\E(n_e) \leq \frac{2^d d}{(d+1)^{d-1}}\log(m)^{d-1} + O(\log(m)^{d-2} \log \log(m))$$
\end{prop}

\begin{proof} 

Let $ \vertices P_n $ be the number of vertices of the convex hull of $n$ points samples from the convex polytope $P$. Let $T(P)$ count the number of maximal chains in the face lattice of the polytope. Equation 1.6 of \cite{barany92} states that 

$$\E(\vertices P_n) = \frac{T(P)}{(d+1)^{d-1}(d-1)!}\log(m)^{d-1} + O(\log(m)^{d-2} \log \log(m))$$

Now, we can count maximal chains in the face lattice of the cube as follows. We an label each face $F$ of the $d$-dimensional cube with a string $X_f$ of $d$ zeros, ones, and stars. We define $X_F(i) = 0$ if $x_i$ is constrained to be zero on $F$, $X_F(i) = 1$ if  $x_i$ is constrained to be one on $F$, and $X_f(i) = *$ if the value of $x_i$ varies on $F$. Each maximal chain starts with one of the $2^d$ vertices of the cube, which has a string $X_v$ with no stars. Now, a chain which starts at $v$ corresponds to an order in which we choose coordinates to replace with stars. Thus, there are  $d!$ ways to do this. This comes out to a total of $2^d d!$ maximal chains. 

Thus, for $d$-dimensional cube $C^d$, 

$$\E(\vertices C^d_n) = \frac{2^d d}{(d+1)^{d-1}}\log(m)^{d-1} + O(\log(m)^{d-2} \log \log(m))$$
Since each minimal node corresponds to a vertex, if $A$ is a random rank $d$ matrix generated from a point arrangement uniformly sampled from the unit cube, the expected number $n_m$ of minimal nodes is bounded above by 
$$\E(n_m) \leq  \frac{2^d d}{(d+1)^{d-1}}\log(m)^{d-1} + O(\log(m)^{d-2} \log \log(m)).$$
\end{proof}

We can make similar estimates if the points $ v_1, \ldots,  v_m$ are chosen according to a Gaussian distribution, using results from \cite{baryshnikov1994regular}. 

\begin{prop}
\label{prop:expected_gauss}

Let 
$A = f(v_i \cdot w_j)$ be a random underlying rank-$d$ matrix whose underlying point arrangement $v_1, \ldots, v_m$ is drawn from a Gaussian distribution in $\R^d$. Then the expected number $n_e$ of extremal nodes is bounded above by 
$$\E(n_e) \leq\beta_d\frac{2^d}{\sqrt{d}}(\pi \log m)^{(d-1)/2},$$ where $\beta_d$ is a constant depending on $d$. 
\end{prop}

\begin{proof}
By \cite{baryshnikov1994regular}, the expected number of vertices of the convex hull of $m$ points drawn from a Gaussian distribution in $\R^d$ is asymptotic to 
$\beta_d\frac{2^d}{\sqrt{d}}(\pi \log m)^{(d-1)/2},$
where $\beta$ is a constant depending on $d$. By Proposition \ref{prop:convex_hull}, this is an upper bound on the number of minimal nodes. 
\end{proof}

\begin{prop}
\label{prop:expected_ball}
Let 
$A = f(v_i \cdot w_j)$ be an underlying rank-$d$ matrix whose underlying point arrangement $v_1, \ldots, v_m$ is drawn uniformly at random from the unit ball in $\R^d$. Then the expected number $n_e$ of extremal nodes is bounded above by 
$$\E(n_e) \leq O\left(m^{\frac{d-1}{d+1}}\right)$$
\end{prop}

\begin{proof}
By \cite{raynaud1970enveloppe}, the expected number of vertices of the convex hull of $m$ points sampled uniformly from the unit ball in $\R^d$ is 
$O\left(m^{\frac{d-1}{d+1}}\right)$. By Proposition \ref{prop:convex_hull}, this is an upper bound on the number of minimal nodes. 
\end{proof}

Notice that the exponent controlling the dependence on $n$ in the Gaussian case is half of that from the uniformly distributed case. Further, the function giving the expected number of vertices in the case of the uniform distribution on a ball is completely different from either of the other distributions.  Further, these are just three possible ways of choosing a point arrangement--there is not a finite list of distributions to check. Finally, not all vertices are picked up as minimal nodes. This means that the expected number of vertices gives an upper bound on the number of minimal nodes, based on the dimension. 
Thus, without information about the underlying probability distribution, we cannot reliably estimate underlying rank from the number of extremal nodes alone. However, with an appropriately chosen family of control distributions, we can estimate the underlying rank by computing the number of minimal nodes and comparing to a control distribution. This control distribution can be chosen based off of the scientific context. 

\subsubsection{Sign constraints}
While every extermal node of the matrix is a vertex of the convex hull, not every vertex of the convex hull actually is observed as an extremal node. This means that features of the distribution generating the underlying point arrangement which do not affect the expected number of vertices can nonetheless affect the number of minimal nodes. In particular, we consider sign constraints: what if we choose points within the positive orthant? We show that this can reduce the number of minimal nodes. 
More precisely, the number of minimal nodes of a point arrangement contained in the first quadrant is approximately $\frac{1}{2^d}$ that of a point arrangement centered at the origin, as exemplified in Figures \ref{fig:convhullposmix} and \ref{fig:nmin_empirical}

\begin{figure}[h!]
\begin{center}
\includegraphics[width = 5 in]{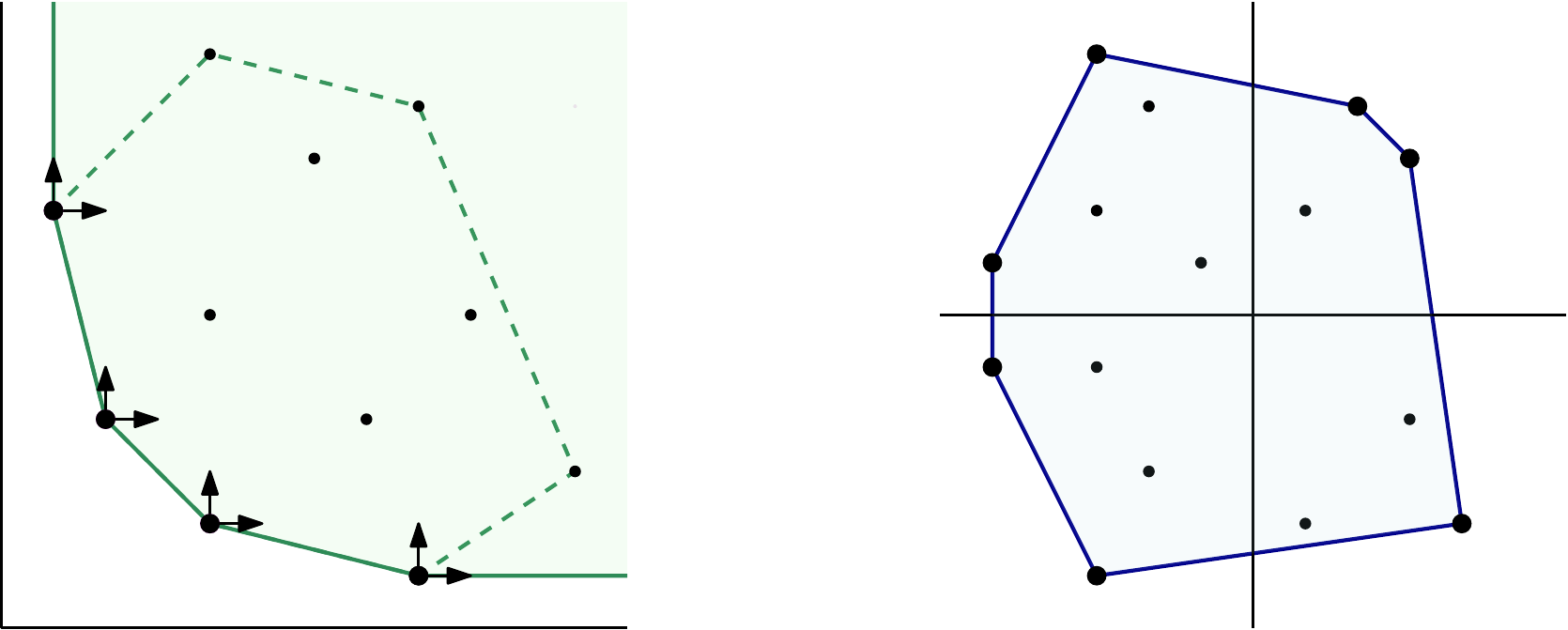}
\end{center}
\caption[Potential minimal modes for positive and mixed-sign arrangements]{Potential minimal modes for positive (A) and mixed sign (B) point arrangements in $\R^2$.  Potential minimal nodes are enlarged. On the left, the convex hull of the point arrangement is shown with a dashed line. }
\label{fig:convhullposmix}
\end{figure}

\begin{defn}
A \emph{nonnegative rank-$d$ realization} of a matrix $A$ is a rank $d$ representation of $A$ with all points in the positive orthant: $ v_1, \ldots,  v_m,  w_1, \ldots,  w_n \in \R^d_{\geq 0}$. The  \emph{nonnegative underlying rank} of a matrix is the smallest value of $r$ for which a nonnegative rank $r$ realization exists.
\end{defn}

The next result shows that the minimal and maximal nodes of a nonnegative rank-$d$ realization of a matrix $A$ must meet stricter conditions than merely being vertices. This implies that there are fewer extremal nodes in the mixed-sign case, and that the minimal and maximal nodes are (almost) disjoint. This fact can be used to determine whether a positive or mixed-sign model is more appropriate for a given dataset. 

\begin{defn}
Let $A, B\subset \R^d$. The Minkowski sum of $A$ and $B$ is the set 
\[A + B = \{a + b \mid a\in A, b\in B\}.\]
\end{defn}

\begin{prop}\label{prop:nonnegativeminimals}
If $i$ is a minimal node of $A$ and $ v_1, \ldots,  v_m,  w_1, \ldots,  w_n \in \R^d$ is a nonnegative rank-$d$ representation of $A$, then $ v_i$ is a vertex of the polytope 
$$\conv( v_1, \ldots,  v_m) + \R^d_{\geq}, $$
where the $+$ denotes the Minkowski sum. 

If $i$ is a maximal node of $A$ and $ v_1, \ldots,  v_m,  w_1, \ldots,  w_n \in \R^d$ is a nonnegative rank-$d$ representation of $A$, then $ v_i$ is a vertex of the polytope 
$$\conv( v_1, \ldots,  v_m) + \R^d_{\leq}. $$
\end{prop}

%{\color{purple} Added a proof of minimal node case of lemma \ref{lem:sweep_order}, which currently comes much later. If reorganized, could remove the addition}

\begin{proof}
We prove this statement for minimal nodes. The proof of the statement for maximal nodes is analogous. 

Let $i$ be a minimal node and $j\in [n]$ such that $A_{ij} < A_{kj}$ for all $k\in [m].$ 
%Notice $w_j$ is defining a sequence of hyperplanes, $h_j(t) = \{ v\mid  w_j \cdot  v = t\}$. Notice by hypothesis, we have $ w_j \cdot  v_i <  w_j \cdot  v_k$ for all $k\in [m]$. Thus $v_i$ is the first point encountered by a hyperplane $h_j(t)$ as $t$, the ``time", increases. 
Equivalently, $ v_i$ is the point of $ v_1, \ldots,  v_m$ that minimizes the inner product $ v \cdot  w_j$.
On the other hand, by the definitions of convex hull and Minkowski sum, we can write
$$\conv( v_1, \ldots,  v_m) + \R^d_{\geq}= \{\lambda_1  v_1 + \cdots + \lambda_m  v_m +  x\mid \lambda_1,\ldots \lambda_m \geq 0, \lambda_1 + \cdots + \lambda_m = 1,  x\in \R^d_{\geq 0}\}.$$
This tells us that $ v_i$ is a vertex of $\conv( v_1, \ldots,  v_m)$, and 
$$ w_j \cdot  v_i \leq  w_j\cdot(\lambda_1  v_1 + \cdots + \lambda_m  v_m),$$
 whenever $\lambda_1,\ldots \lambda_m \geq 0, \lambda_1 + \cdots + \lambda_m = 1$ and that this inequality is strict whenever $\lambda_i < 1$. 
 
Because $ w_j$ is in the positive orthant, $ w_j \cdot  x \geq 0$ for all $ x$ in the positive orthant. 
Thus, 
$$ w_j \cdot  v_i \leq  w_j\cdot(\lambda_1  v_1 + \cdots + \lambda_m  v_m +  x),$$
 for all $\lambda_1,\ldots \lambda_m \geq 0, \lambda_1 + \cdots + \lambda_m = 1,  x\in \R^d_{\geq 0},$  and this inequality is strict whenever $\lambda_i <1$. 
Thus, $ v_i$ minimizes the inner product $ w_j \cdot  v$ for all $ v \in \conv( v_1, \ldots,  v_m) + \R^d_{\geq}$. 
Thus, $ v_i$ is a vertex of $\conv( v_1, \ldots,  v_m) + \R^d_{\geq}$. 
\end{proof}

%The following result lets us distinguish matrices whose positive underlying rank is higher than their underlying rank. 

%\begin{cor}
%A matrix with nonnnegative underlying rank one has exactly one minimal node. 
%\end{cor}

Figure \ref{actualvsexpected} collects these calculations for varying $n$. Curves labeled ``random full rank'' correspond in both cases to randomly generated matrices (non-symmetric) with positive or mixed signed entries, respectively. Notice that in the mixed-sign case there is an overlap between the curves corresponding to (true) random full rank matrices and to the case $d=n$ (gold and aqua curves). 

\begin{figure}
\includegraphics[width = 6 in]{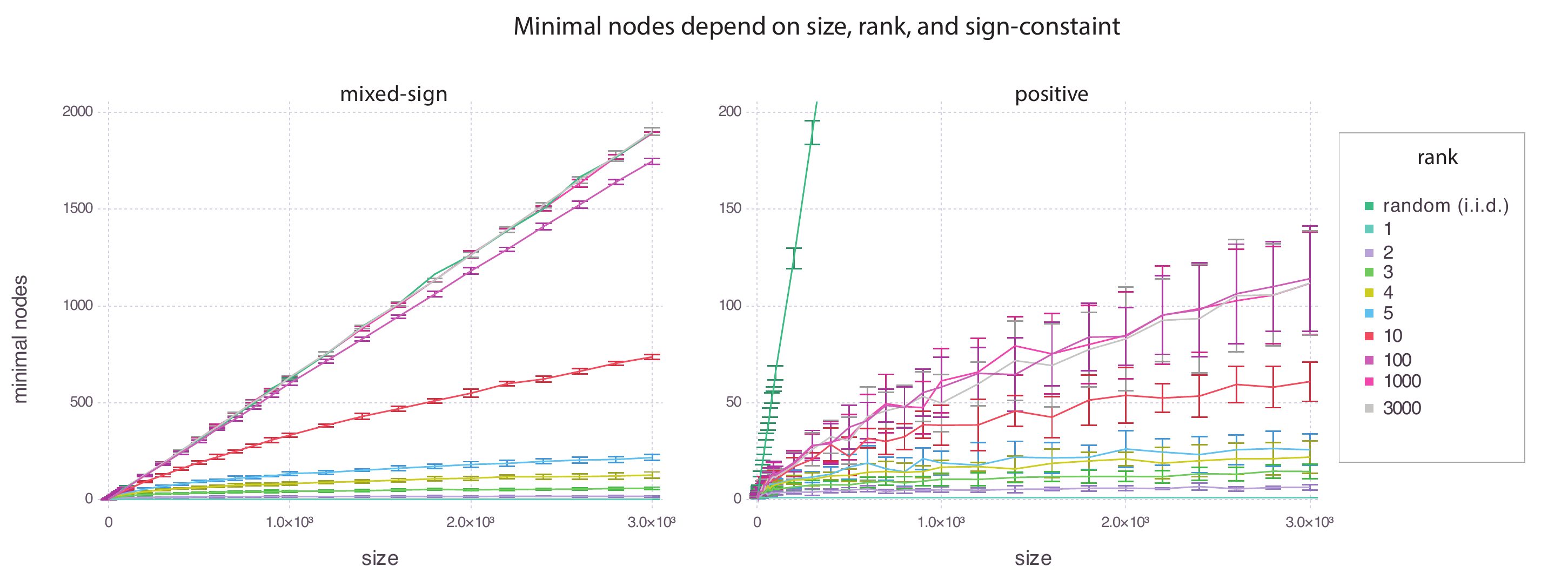}
\caption[Average number of minimal nodes]{Average number of minimal nodes for samples of 5 matrices, obtained as $BB'$, for $B$ an $n\times d$ matrix with random uniformly distributed entries in $[0,1]$ (left, positive) and  $[-1/2,1/2]$ (right, mixed-sign). In both cases $d$ is the rank of the resulting matrix $BB'$. Curves are color coded by rank. Error bars are given by standard deviations. 
\label{fig:nmin_empirical}}
\end{figure}

\begin{figure}
\includegraphics[width = 6 in]{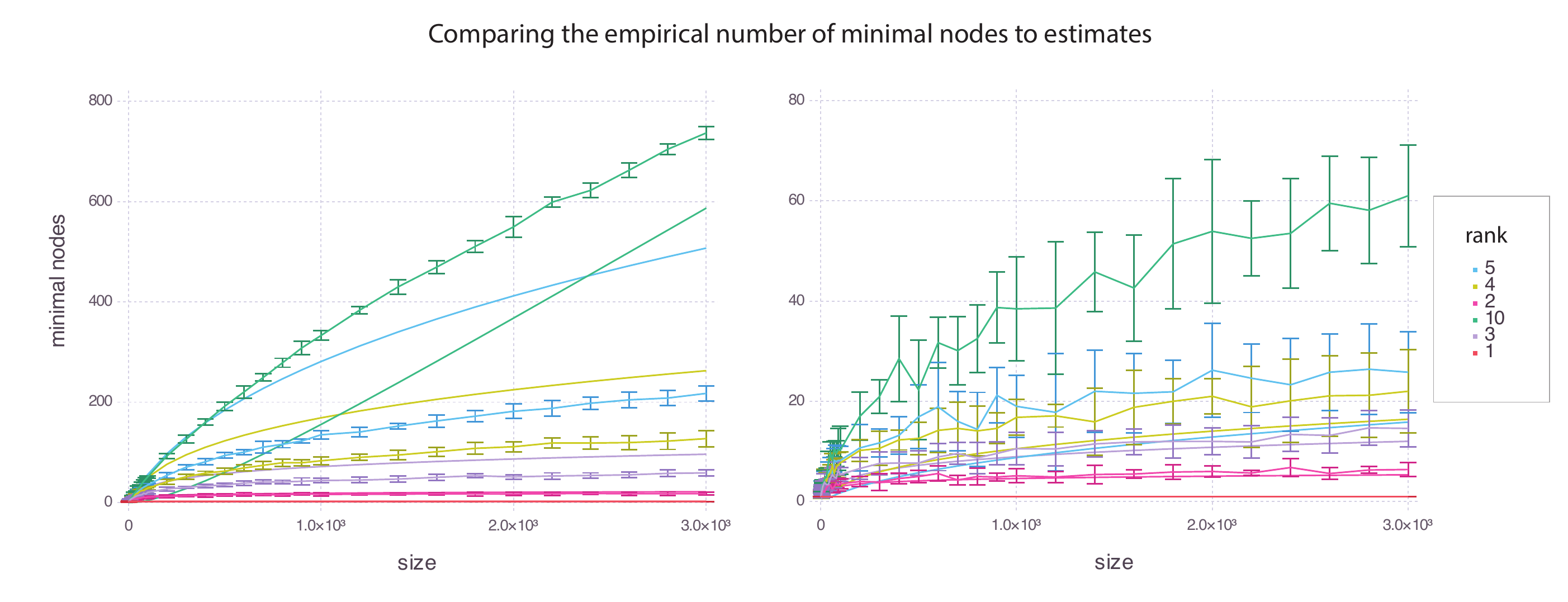}
\caption[Comparing the observed number of minimal nodes to the bounds]{We compare the observed number of minimal nodes to the asymptotic upper bounds appearing in Proposition \ref{prop:expected_cube} and Conjecture \ref{cong:expected_pos}. In particular, in the mixed sign case, we plot the quantity  $\frac{2^d d}{(d+1)^{d-1}}\log(m)^{d-1}$ and in the positive case, we plot the quantity $\frac{d}{(d+1)^{d-1}}\log(m)^{d-1}$. We notice that for higher ranks, we underestimate the number of minimal nodes, as the  $O(\log(m)^{d-2} \log \log(m))$ correction term matters more in this case. 
\label{actualvsexpected}}
\end{figure}

If $A$ has low non-negative rank, it will have fewer  minimal nodes. Roughly, if $A$ is a a random non-negative underlying rank-$d$ matrix whose underlying point arrangement $v_1, \ldots, v_m$ is drawn uniformly at random from the unit cube in the positive orthant,  the expected number of minimal nodes is bounded above by $ \E(n_v)/2^d$, which in the case of the unit cube is given by $\frac{d}{(d+1)^{d-1}}\log(n)^{d-1}.$ Observe in Figure \ref{actualvsexpected} that this matches up with computational experiments with random matrices. 
We give a heuristic argument for why this should be the case:
in order for $v_i$ to be a  vertex of $\conv(v_1, \ldots, v_n)+ \R^d_{\geq},$ $v_i$ must be a vertex of  $\conv(v_1, \ldots, v_n)$ whose normal cone intersects the positive orthant. Since the positive orthant in $\R^d$ takes up $\frac{1}{2^d}$ the volume of the unit cube centered at the origin in $\R^d$, the probability that  $v_i$ is a vertex of $\conv(v_1, \ldots, v_n)+ \R^d_{\geq}$ is roughly $\frac{1}{2^d}$ the probability that $v_i$ is a vertex of $\conv(v_1, \ldots, v_n)$. Thus, the probability that $v_i$ is a minimal node of a random positive rank-$d$ matrix is roughly $\frac{1}{2^d}$ the probability that $v_i$ is a minimal node of a random mixed-sign rank-$d$ matrix.  Note that this argument neglects the volume of the normal cone of $v_i$, making it a heuristic argument and not a proof. Thus, we leave this  as a conjecture:

\begin{conj}
If $v_1, \ldots, v_n$ are drawn uniformly at random from the unit cube $[0, 1]^d$, then the expected number of minimal nodes is 
$$\E(n_v)/2^d = \frac{d}{(d+1)^{d-1}}\log(n)^{d-1} +  O(\log(m)^{d-2} \log \log(m)).$$
\label{cong:expected_pos}
\end{conj}
In Figure \ref{actualvsexpected}, we compare the quantity $\frac{d}{(d+1)^{d-1}}\log(n)^{d-1}$ to the observed number of minimal nodes for random matrices of various ranks. We see that for low ranks, the estimate holds, but for high ranks, it does not, likely due to the factor  $O(\log(m)^{d-2} \log \log(m)).$

We can interpolate between the mixed-sign and positive case by constraining some, but not all, coordinates to be positive. We define the \emph{positivity} of a point arrangement as the number of coordinates which are required to be positive. For example, Figure \ref{fig:plotminmaxoverlapv2} shows how much the overlapping sets differ for $d=2$. As we increase the positivity, we decrease the number of minimal nodes. We also decrease the \emph{overlap} between minimal and maximal nodes, the number of nodes which are both minimal and maximal. We can use the overlap and positivity together to estimate underlying rank, assuming a distribution for the underlying point arrangement.

\begin{figure}[ht!]
\centering
  \includegraphics[width=\linewidth]{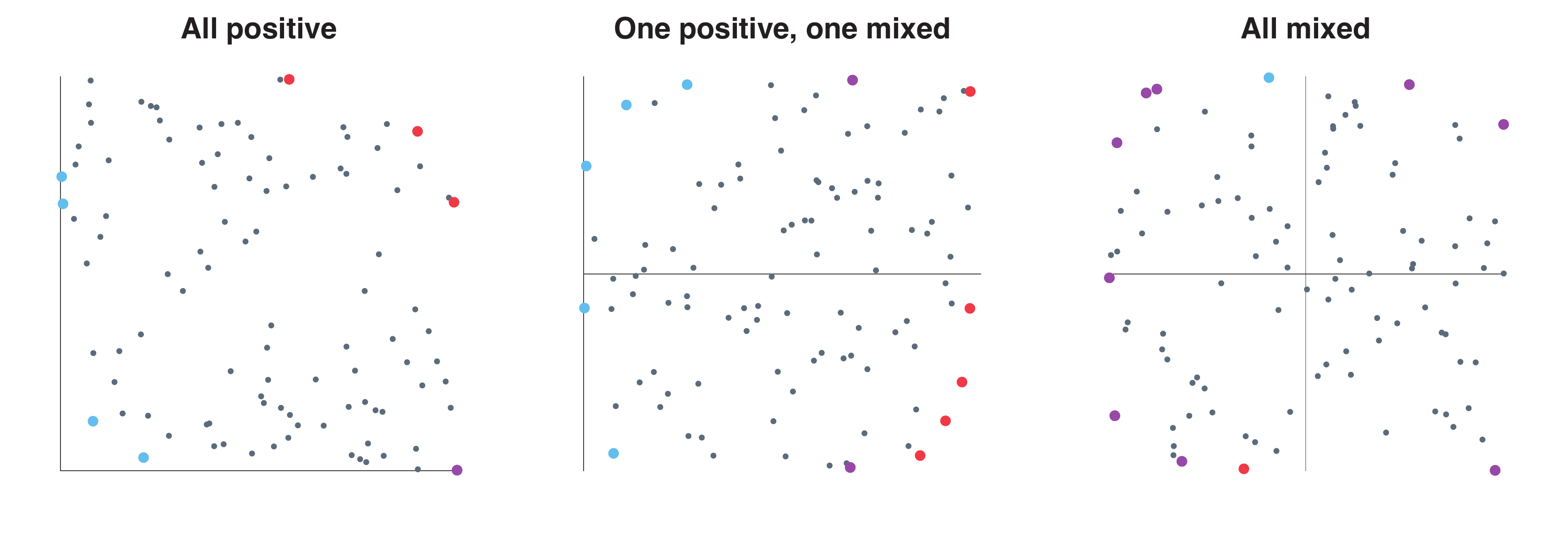}
\caption[The effect of positivity on the number of minimal nodes]{We show the effect of increasing positivity on the number of minimal, maximal, and overlap nodes. In each plot, the minimal nodes are colored blue, the maximal nodes red and the overlapping nodes (nodes that are both minimal and maximal) are colored purple. Notice that in the all-positive case and the one positive, one mixed case the overlap is minimal, whereas in the all-mixed case the overlap consists of most of the nodes.}
\label{fig:plotminmaxoverlapv2}
\end{figure}

% \begin{prop}
% Let 
% $A = f(v_i \cdot w_j)$  be a random non-negative underlying rank-$d$ matrix whose underlying point arrangement $v_1, \ldots, v_m$ is drawn uniformly at random from in the positive orthant unit cube.
% Then 
% $$\E(n_m) \leq \frac{d}{(d+1)^{d-1}}\log(n)^{d-1} = \E(n_v)/2^d.$$
% \end{prop}

% \begin{proof}\juliana{check this proof, please}
% Since A is non-negative monotone rank-$d$ matrix, by proposition \ref{prop:nonnegativeminimals} its minimal nodes are vertices of the polytope $$\conv( v_1, \ldots,  v_m) + \R^d_{\geq}, $$ and therefore we only need to count $\frac{1}{2^d}$-points in the boundary of this polytope.
% \end{proof}

\subsubsection{Matrices with i.i.d entries.}
The next proposition shows that we can characterize the expected number of minimal nodes for matrix with independent, identically distributed (i.i.d.) entries. In particular, we show that the expected number of minimal nodes of a square matrix with i.i.d. entries is linear in the matrix size. 
This allows us to distinguish matrices with some kind of low-rank structure (or other structure) from truly random matrices.

\begin{prop}
Let $A$ be a random $n\times n$ matrix with i.i.d. entries.  Then $$\lim_{n\to\infty}\frac{\E\left(|\min(A)|\right)}{n} = \left(1-\frac 1 e\right).$$
\end{prop}

\begin{proof}
Let $X_i$ be a random variable defined by 
\begin{align*}
X_i = \begin{cases} 
1 \qquad i\in \min(A)\\
 0 \qquad \mbox{otherwise}
\end{cases}
\end{align*}
Note that $$\E\left(|\min(A)|\right) = \E\left(\sum_{i = 1}^n X_i\right ) = \sum_{i = 1}^n \E (X_i),$$
since the sum of expected values is always the expected value of the sum. 
Thus, it is sufficient to calculate $\E (X_i)$. To do this, define another random variable 
\begin{align*}
X_{ij} = \begin{cases} 
1 \qquad i \mbox{ is \emph{not} minimal in row j}\\
 0 \qquad \mbox{otherwise}
\end{cases}
\end{align*}
By the assumption that the matrix entries are $i.i.d$, $\E(X_{ij}) = \frac {n-1} {n}$ for all $i, j$, since each entry in the $j^{th}$ row has an equal chance of being the smallest. 

Now, note that $i$ is minimal if $X_{ij} = 0$ for at least one $j$. Thus, $i$ is minimal if and only if $\prod_{j = 1}^nX_{ij} = 0$. Thus $X_i =1 - \prod_{j = 1}^nX_{ij}$, $$\E(X_i) =1 - \E\left(\prod_{j = 1}^nX_{ij}\right).$$
Again by assumption, $X_{ij}$ and $X_{ik}$ are independent for $j \neq k$. 
Thus $\E\left(\prod_{j = 1}^nX_{ij}\right)= \prod_{j = 1}^n \E(X_{ij})$, so 
$$\E(X_i) =1 - \left( \frac {n-1} {n}\right)^n = 1-\left(1-\frac 1 n\right)^n.$$
Thus, 
$$\E\left(|\min(A)|\right) = \E\left(\sum_{i = 1}^n X_i\right ) = n\left(1-\left(1-\frac 1 n\right)^n\right).$$
By a standard result of calculus, $\lim_{n\to \infty} \left(1-\frac 1 n\right)^n = \frac 1 e$. Thus 
$$\lim_{n\to\infty}\frac{\E\left(|\min(A)|\right)}{n} = \left(1-\frac 1 e\right).$$

\end{proof}

Notice that this proof does not hold for random symmetric matrices. Empirically, however, we have seen that it does hold in this case.

\begin{conj}
Let $A$ be a random $n\times n$ \emph{symmetric} matrix with the entries above the diagonal i.i.d.  Then $$\lim_{n\to\infty}\frac{\E\left(|\min(A)|\right)}{n} = \left(1-\frac 1 e\right).$$
\end{conj}

\section{Radon's theorem: a lower bound for underlying rank}
\label{sec:radon}
As we saw, extremal nodes can help us estimate the underlying rank of a matrix.  However, since a polygon can have an unlimited number of vertices, we cannot use extremal nodes to prove that underlying rank is greater than two.  Further, the minimal nodes only capture a very small part of the combinatorial structure of the point configuration $v_1, \ldots, v_n$. In this section, we leverage more of this structure in order to establish lower bounds on underlying rank. We do so by introducing two simplicial complexes, the sweep complex and the shatter complex, whose dimensions give lower bounds for underlying rank that we establish via Radon's theorem.

\subsection{The Sweep and Shatter Complexes}

\label{sec:shatter}
In this section, we consider the geometric constraints the ordering of entries in the matrix $A$ place upon rank-$d$ representations of $A$. By Lemma \ref{lem:sweep_order}, the order of entries in each row of $A$ corresponds to the order in which a sequence of hyperplanes sweeps past the points $v_1, \ldots, v_m$. We will use this to build two combinatorial objects, the sweep complex and the shatter complex, which will help us bound the underlying rank. While these invariants are more computationally expensive than minimal nodes, they are able to provide stronger lower bounds. 

\begin{figure}
    \centering
    \includegraphics[width = 4 in]{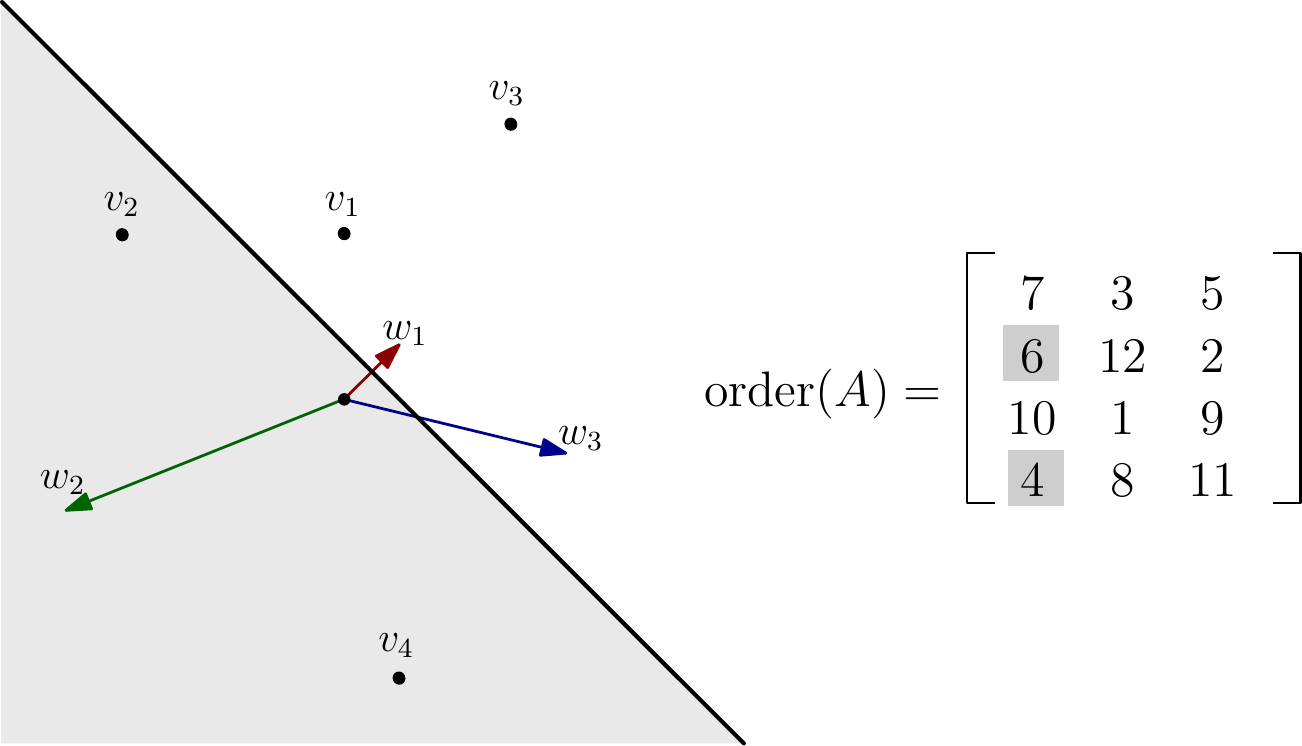}
    \caption[Swept sets.]{The set $\{2, 4\}$ is swept by the first column of $\order(A)$. Consequently, we can deduce that there is a hyperplane separating $\{v_2, v_4\}$ from  $\{v_1, v_3\}$. }
    \label{fig:sep}
\end{figure}

\begin{defn}
Let $A$ be a $m\times n$ matrix. 
A set of columns $\sigma\subseteq [m]$ is \emph{swept} by $A$ if there is some row $i$ such that:
\begin{enumerate}
\item 
$A_{ij} < A_{ik}$ for all $j\in \sigma, k\notin \sigma$
\item 
$A_{ij} > A_{ik}$ for all $j\in \sigma, k\notin \sigma $
\end{enumerate}

An example of a swept set appears in Figure \ref{fig:sep}. 
\end{defn}

\begin{lem}\label{lem:partition}
If a set $\sigma \subseteq [m]$ is swept by $A$, then in any rank-$d$ realization of $A$, there is a hyperplane separating the points $\{v_i\}_{i\in \sigma}$ from the points $\{v_j\}_{j\notin \sigma}$.
\end{lem}

\begin{proof}
Suppose  $\sigma$ is swept by the $k$-(th) column of $A$.
By Lemma \ref{lem:sweep_order}, the sweeping hyperplane $c_{k}(t)$ encounters all points  $\{v_i\}_{i\in \sigma}$ before it encounters any of the points $\{v_j\}_{j\in [n]\setminus \sigma}$. Let $t_i$ be the time such that $v_i \in  c_{k}(t_i)$. 
Thus we can choose $t'$ so that $\max_{i\in \sigma}{t_i} < t' <\min_{j\in [n]\setminus}{t_j}$.  The hyperplane $c_k(t')$ separates the points $\{v_i\}_{i\in \sigma}$ from the points $\{v_j\}_{[n]\setminus \sigma }$.

\end{proof}

We now describe two ways to build a simplicial complex out of the set system $$H_A := \{\sigma \subseteq [m] \mid  \sigma \mbox{ is swept by } A\}.$$

First, we define the \emph{sweep complex} $\Delta_{sw}(A, L)$ with respect to a set of landmarks $L\subseteq [m]$. This is the largest simplicial complex contained in the set 
$H_A\cap L := \{\sigma \cap L \mid \sigma \in H_A\}$. 
More explicitly, 
$$\Delta_{sw}(A, L) = \{\sigma \mid \tau \in H_A \cap L \mbox{ for all } \tau \subseteq \sigma \}$$
Notice that the vertex set of $\Delta_{sw}(A, L)$ is the set of minimal nodes contained in $L$. More generally, we will see that the faces of 
$\Delta_{sw}(A, L)$ are faces of the convex polytope $\conv(\{v_{i} \mid i\in L\})$. In this sense, the sweep complex is the most natural way to give a structure to the set of minimal nodes. We typically choose $L$ to be a small, proper subset of the minimal nodes of $A$. Picking $L$ to be smaller not only speeds up computation, but also loosens the conditions required to include a simplex $\sigma \in \Delta_{sw}(A, L)$, increasing the rank we detect.  The sweep complex is inspired by the witness complex of \cite{de2004topological}. 

Next, we define the \emph{shatter complex} $\Delta_{sh}(A)$. We say $\sigma\subseteq[m]$ is \emph{shattered} by $H_A$ if for all $\tau \subseteq \sigma$, there exists $\omega\in H_A$ such that $\tau = \sigma \cap \omega$. 
The shatter complex is 
$$ \Delta_{sh}(A) = \{\sigma \mid  \sigma \mbox{ is shattered by } H_A\}.$$
Notice that for any $L$, $\Delta_{sw}(A, L) \subseteq \Delta_{sh}(A)$. 
We will show that each simplex of the shatter complex corresponds to an affinely independent subset of the points $v_1, \ldots, v_m$. We will establish this using a version of Radon's theorem.

Notice that there is a hyperplane separating the points  $\{v_i \mid i \in \sigma\}$  from the $\{v_i \mid j \in \tau\}$ if and only if $\conv(\{v_i \mid i \in \sigma\} ) \cap \conv(\{v_i \mid j \in \tau\}) = \emptyset$. 

\begin{lem}\label{lem:radon}
\textbf{(Radon's theorem)} 
A set of points is affinely independent if for every partition $\sigma, \tau$ of the points, they points in $\sigma$ and the points in $\tau$ can be separated by a hyperplane. 
\end{lem}

\begin{prop} Let $A$ be a $m\times n$ matrix with underlying point arrangement $v_1, \ldots, v_m$. Then if $\omega \in \Delta_{sh}(A)$, the set $\{v_i \mid i \in \omega\}$ is affinely independent. 
\end{prop}

\begin{proof}
Suppose $\omega \in \Delta_{sh}(A)$. 
Then for each partition $\sigma \cup \tau = \omega$, $\sigma \cap \tau = \varnothing$, both $\sigma$ and $\tau$ are intersections $\sigma = h \cap \omega$, $\tau = h'\cap  \omega$ for some $h, h' \in H_A$. 
Thus there is a hyperplane separating the points $\{v_i \mid i \in \sigma\}$ from the points  $\{v_j \mid j \in \tau\}$. Thus, the set of points $\{v_i \mid i \in \sigma\}$ has no Radon partition. Thus, it is affinely independent.
\end{proof}

An affinely independent set in $\R^d$ has at most $d+1$ points.
Thus, the dimension of the sweep complex is a lower bound on underlying rank. 
Inspired by Radon's theorem, we term this the \emph{Radon rank} $\radr(A)$.

\begin{prop}\label{prop:radon_bound} The Radon rank of a matrix is a lower bound on its underlying rank. That is, $$\radr(A) \leq \ur (A).$$
\end{prop}

A similar bound can be derived from the sweep complex, however, this requires slightly more casework. Further, this cannot exceed the Radon rank. 

The \emph{Radon rank} closely resembles the concept of the \emph{Vapnik–Chervonenkis (VC) dimension} from statistical learning theory \cite{vapnik1971uniform}. The VC dimsnion of a set system $H$ is the size of the largest set shattered by $H$. Thus,  Radon rank of a matrix $A$ is the VC dimension of the set system $H_A$, minus one. Similar ideas have been used to estimate the dimensionality of neural activity in \cite{rigotti2013importance}. 

\subsection{How high can Radon rank be?}
The Radon rank is not equal to either the underlying rank or the monotone rank in general. In particular, we see in the next proposition that the the size of the smallest matrix with Radon rank $d$ is exponential in $d$. This is due to a combinatorial explosion in the number of column orders needed to shatter a simplex. 

\begin{prop}\label{prop:radon_limits}
The Radon rank $d = \radr(A)$ of a $m\times n$ matrix $A$ satisfies 
$$\frac 1 2 \binom{d+1}{\lfloor (d+1)/2\rfloor} \leq n.$$
\end{prop}

A plot of this bound appears in Figure \ref{fig:n_vs_d}.

\begin{proof}
If $\radr(A) = d$, $A$ shatters a set $\rho$ with $|\rho| = d+1$. Thus, for each $\sigma \subseteq \rho$, the set $\sigma$ is swept by $A$. 
Each column of $A$ sweeps two nested sequences of sets,  $\emptyset \subseteq \sigma_1 \subseteq \cdots \subseteq \\sigma_d \subseteq \rho$ and $\emptyset \subseteq \tau_1 \subseteq \cdots \subseteq \tau_d \subseteq\rho.$
These are \emph{chains} in the Boolean lattice, the partially ordered set whose elements are subsets of $\rho$ ordered by set-inclusion. 
In order to induce every partition of $\rho$, these chains must cover the Boolean lattice on $\rho$. 
Dilworth's theorem \cite{dilworth1950decomposition} states that the minimal number of chains needed to cover a poset $P$ is exactly equal to the length of the longest antichain of $P$. 
Sperner's theorem \cite{sperner1928satz} states that the longest antichain of the Boolean lattice on $k$ elements has length $\binom{k}{\lfloor k/2\rfloor}$. Thus, the number of permutations needed to induce all partitions of $\rho$ is  $\binom{d+1}{\lfloor (d+1)/2\rfloor}$. 
Since each column induces two permutations, we have $\frac 1 2 \binom{d+1}{\lfloor (d+1)/2\rfloor} \leq n$, as desired. 
\end{proof}

\begin{figure}
    \centering
       \includegraphics[width = 3	 in]{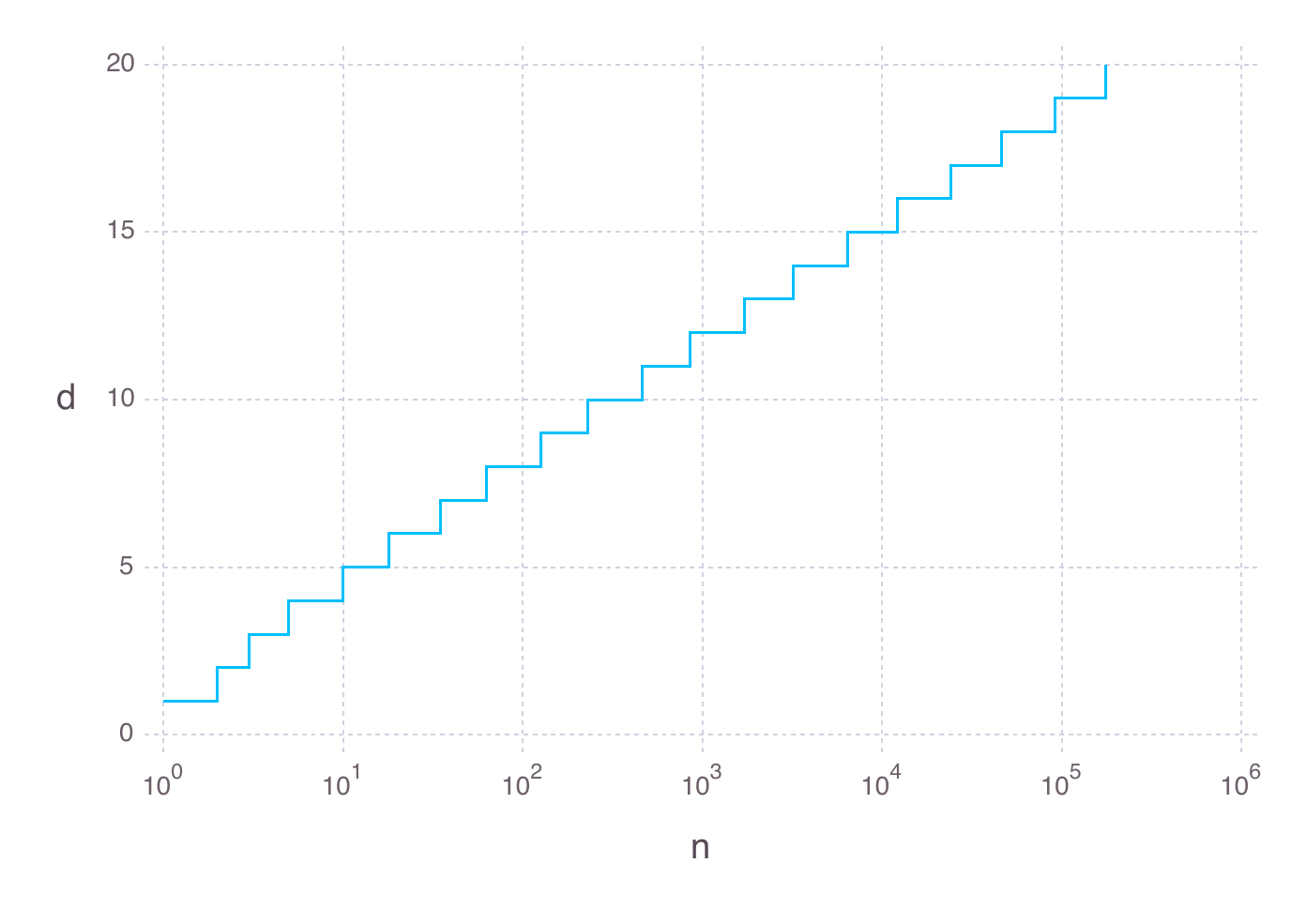}
    \caption{The maximum underlying rank $d$ detectable by an $n\times n$ matrix, plotted as a function of $n$. Note that the scale for $n$ is logarithmic.  }
    \label{fig:n_vs_d}
\end{figure}

\section{Application: dimensionality of neural activity in the larval zebrafish optic tectum}
\label{sec:zebrafish}

As a proof of concept, we apply the techniques developed here to estimate the underlying rank of a neural dataset.  
The data consists of calcium imaging of $N = 2042$ neurons in the optic tectum of a larval zebrafish obtained from the Sumbre lab at École Normale Supérieure. The recording is one hour long. 
The zebrafish is in the dark for the duration of the recording, thus the recorded activity is spontaneous, rather than stimulus driven. The sampling rate is 15 Hz. 
This means the dataset is a $54000 \times 2042 $ matrix $A$. 
To reduce noise, we smoothed the data by averaging the activity of each neuron over a sliding window of 2 seconds (30 time bins), producing a matrix $A_{smooth}$. 

To get an estimate of the linear dimensionality of our data, we compute the singular values of $A$ and $A_{smooth}$, plotted in Figure \ref{fig:data_svd}. Notice that for both $A$ and $A_{smooth}$, the singular values decay smoothly with dimension, following a roughly power law relationship between singular value and dimension. The singular values of $A_{smooth}$ decay faster. 

\begin{figure}[ht!]
\begin{center}
\includegraphics[width = 4 in]{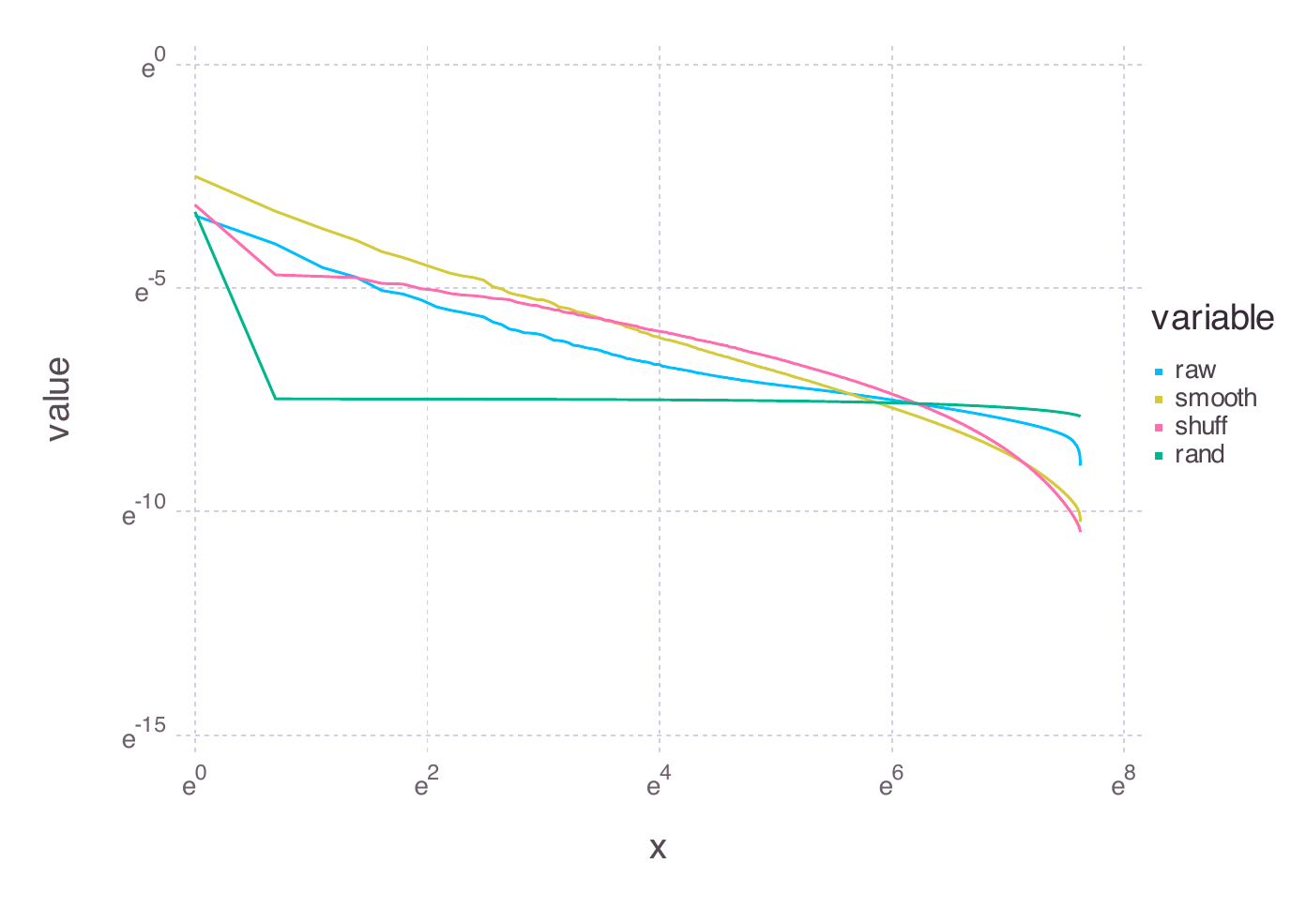}
\end{center}
\caption{Singular values of the data matrix, before and after smoothing. \label{fig:data_svd}}
\end{figure}

However, the monotone nonlinear relationship between firing rate and calcium fluorescence may affect this measure of dimensionality. Thus, we use the concepts we have introduced in this paper to determine whether the ordering of the entries of this matrix is consistent with the estimate of underlying rank based on singular values, or whether it is consistent with a lower rank. 

Because differences between low values of calcium fluorescence do not reflect anything meaningful, while differences between high values does, we focus on maximal nodes, rather than minimal nodes. 
We sample square submatrices of $A_{smooth}$ and count the number of maximal nodes that each one has. We then repeat this process with random matrices of ranks 1, 5, and 10. 
We plot the results of this in \ref{fig:min_nodes_data}. Notice that, in the real data, there is a distinction between maximal rows and maximal columns. In particular, a maximal row corresponds to a time point $t$ where some neuron $i$ achieves its maximal value.  On the other hand, a maximal column correspond to a neuron $i$ which is the most active neuron at some time point $t$.

\begin{figure}[ht!]
    \centering
    \includegraphics[width = 6 in]{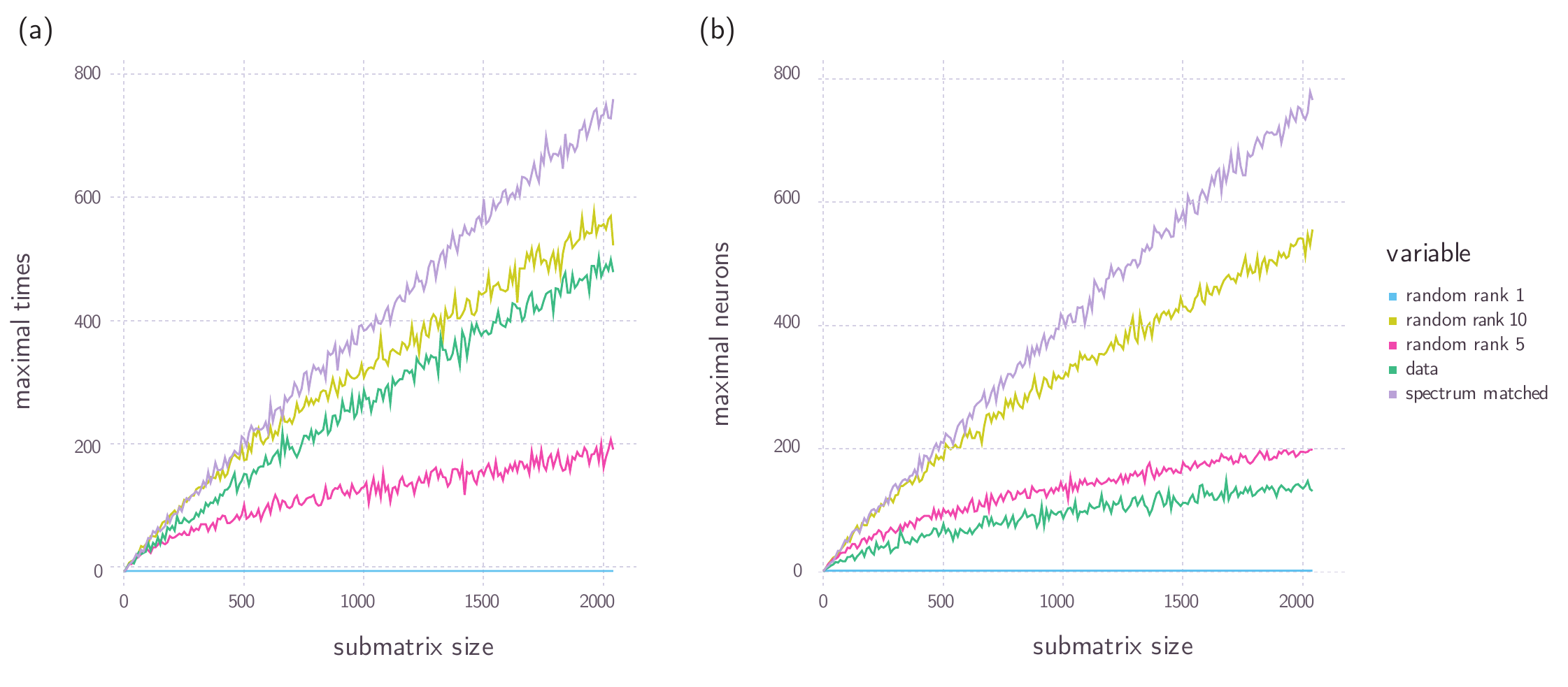}
    \caption[Applying minimal nodes to data.]{(a) maximal timepoints (b) maximal neurons}
    \label{fig:min_nodes_data}
\end{figure}

Notice in Figure \ref{fig:min_nodes_data}, based on the maximal rows (maximal times), the data is most consistent with underlying rank 10, while based on maximal columns (maximal neurons), the data is most consistent with underlying rank 5.  

Next, in Figure \ref{fig:shatter_data}, we use the concept of Radon rank to estimate the underlying rank of $A_{smooth}$. The exact Radon rank of a matrix of this size is prohibitively inefficient to compute, and subject to localized noise, since one set of size $k+1$ being shattered is enough to ensure Radon rank is at least $k$. Therefore, we estimate Radon rank through random sampling and comparison. Since there are more time points than neurons, it is possible to estimate Radon rank over a higher range by sampling sets of $k+1$ neurons and checking whether each possible partition of the neurons is achieved by a row. With the size of the matrix, the highest possible Radon rank which can be detected is 18. Therefore, for each  $k = 2, 3, \ldots, 19 $, we sampled 100 sets of $k$ neurons and checked which subsets were shattered. For the data matrix, as well as low-rank controls, we plot the fraction of sets which were shattered for each $k$ in Figure \ref{fig:shatter_data}. 

In contrast with minimal nodes, this result is consistent with the singular values of the matrix, and is not consistent with a lower rank. In particular, the data shatters some sets of 13 neurons, establishing a Radon rank of at least 12. The fraction of sets shattered as a function of $k$ most closely resembles a random rank 20 matrix, and is also very similar to that of a singular value-matched control.

\begin{figure}[ht!]
    \centering
    \includegraphics{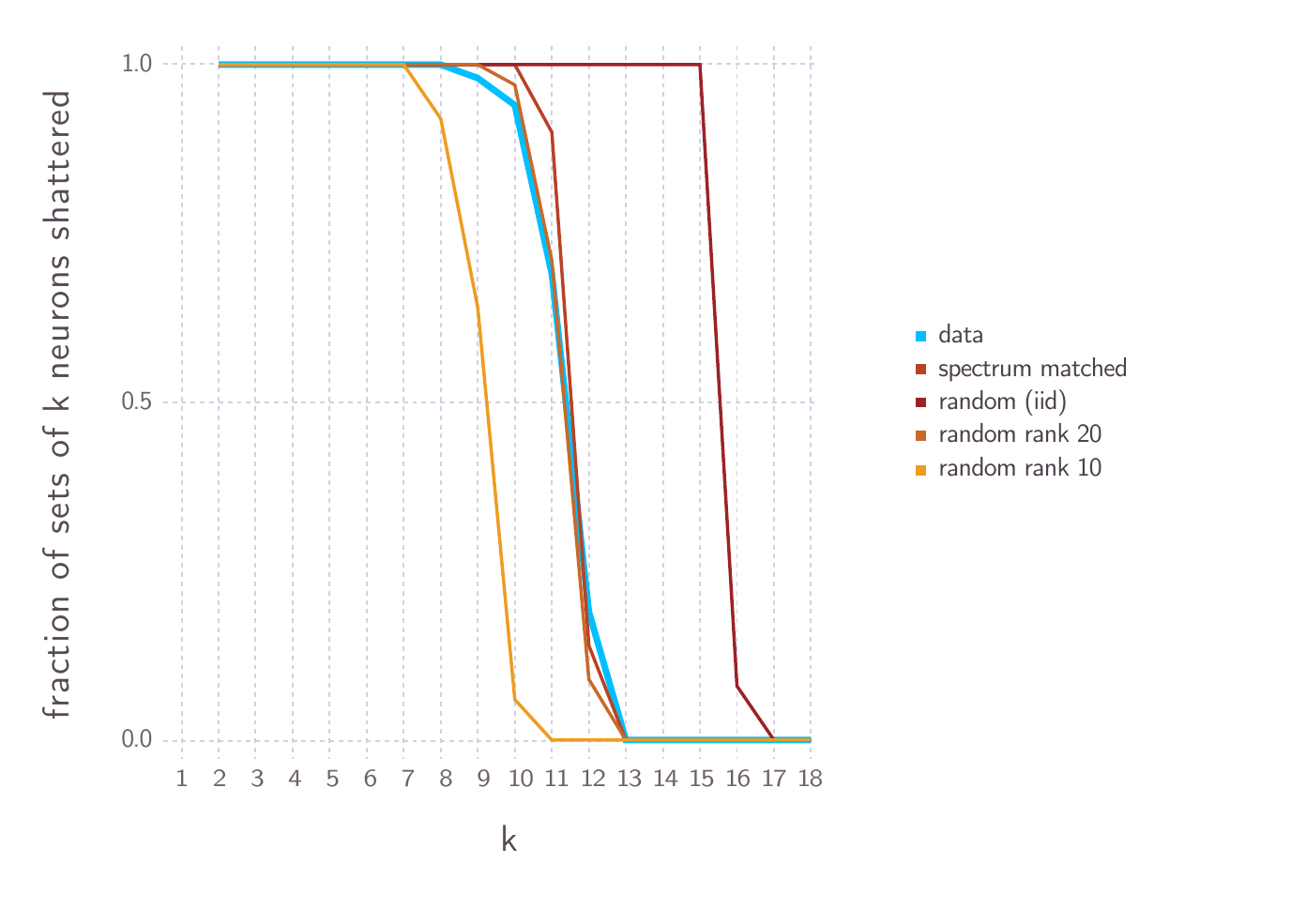}
    \caption[Applying Radon rank to data]{Fraction of sets of size $k$ shattered as a function of $k$. }
    \label{fig:shatter_data}
\end{figure}

\section{Conclusion and Open Questions}

In Section \ref{sec:minimal}, we show that the number of minimal nodes of a random matrix increases as we increase rank, thus that the number of minimal nodes can be used as an order invariant for estimating underlying rank. However, the precise dependence of the number of minimal nodes on rank depends on the probability distribution used to choose our random matrix. Thus, in order to obtain reliable estimates of the underlying rank of a matrix arising from neural data using minimal nodes, we need a biologically realistic way to simulate neural activity with a specified rank.
\begin{question}
Give a biologically plausible model for random matrices of each rank, and characterize the relationship between rank and minimal nodes from this model. 
\end{question} 
%Such a model could begin with the characterization of the relationship between network structure, input structure, and dimensionality given in \cite{}.
Further, often data does not have an idealized low rank + noise structure. Instead, singular values decay more smoothly, but it is still often reasonable to regard the data as low rank.
\begin{question}
Characterize the minimal nodes of matrices with a wider variety of spectra. 
\end{question}

%% file: UnderlyingRank/UnderlyingRankMatroids.tex
%auto-ignore 

\chapter{The Geometry of Underlying Rank }
\label{chapter:urank_math}

\section{Introduction}

In the previous chapter, we introduced the concept of underlying rank and introduced two tools for estimating it, minimal nodes and Radon rank. 
In this chapter, we explore the concept of underlying rank in more mathematical detail, show that underlying rank may exceed these estimates, and show that it is difficult to compute exactly. Each of these examples and results arises from the connection between underlying rank and oriented matroid theory. 

We begin by describing the relationship underlying rank and three closely related concepts: monotone rank, sign rank, and convex sensing. In particular, the connection to sign rank gives us the following result:

\begin{ithm}\label{cor:hadamard}
For each $N = 2^n$,  $O(H_n)$ is a $N\times N$ order matrix with $$\ur(O(H_n)) \geq  \sqrt{N}-1.$$
\end{ithm} 

Combined with Proposition \ref{prop:radon_limits}, this implies that underlying rank can exceed Radon rank. Motivated by this,  we give three explicit, small  examples of matrices whose underlying rank exceeds their Radon rank. These appear in Example \ref{ex:strict}, Proposition \ref{prop:rad_strict}, and Example \ref{ex:allowable}. In particular,  Proposition \ref{prop:rad_strict} demonstrates that the set of potential Radon partitions must have a certain structure. In order to generalize this, we relate underlying rank to oriented matroid theory via the following result. 

\begin{ithm}
Suppose $A$ has underlying rank $d$. Then the potential radon partitions of rank $d+1$ of $A$ contain the circuits of a representable oriented matroid of rank $d+1$.\label{thm:potential} 
\end{ithm}

Next, in Example \ref{ex:allowable}, we give an example of a matrix for which the detailed column orders contain information beyond the set of Radon partitions. This example arises from a connection underlying rank and allowable sequences. In particular,  we note in Observation \ref{obs:allowable} that if a matrix has underlying rank 2, then its columns are contained in a subsequence of an allowable sequence. We use this connection to give an approximate algorithm for determining whether a matrix has underlying rank 2, whose correctness is proven in Proposition \ref{prop:allow_alg}. Further, we exploit the connection between underlying rank and allowable sequences to prove that computing underlying rank is hard. 

\begin{ithm}\label{cor:urank_hard}
Checking whether a matrix has underlying rank two is $\exists \R$-complete and thus NP-hard. 
\end{ithm}

\setcounter{ithm}{0}

This chapter is organized as follows:
First, in Section \ref{sec:math_context}, we introduce the underlying rank in a mathematical context, comparing it to other similar notions.  
Next, in Section \ref{sec:counter}, we give an explicit example of a matrix whose underlying rank exceeds its Radon rank because the set of possible Radon partitions is not consistent. 
In Section \ref{sec:comb_geo}, we generalize this example using oriented matroid theory. 

\section{Mathematical context}
\label{sec:math_context}
In this section, we place the concept of underlying rank in its mathematical context though comparison to three related problems: monotone rank, convex sensing, and sign rank. 

\subsection{Monotone rank}
In Lemma \ref{lem:sweep_order}, we proved that the order of entries in the $j$-th column of a matrix $A$ corresponds to the order in which a hyperplane swept perpendicular to $w_j$  sweeps past the points $v_1, \ldots, v_m$ in a rank $d$ realization of $A$. 
Notice that this proof does not depend on the monotone function $f$ being the same for each column. That is, we can replace the global function $f$ with functions $f_1, \ldots, f_n$ such that $A_{ij} = f_j (v_i \cdot w_j)$. This fact is used in \cite{egger20xxtopological} to define the \emph{monotone rank} of a matrix.  Like underlying rank, monotone rank was introduced in the context of mathematical neuroscience. It makes sense to consider the monotone rank, rather than the underlying rank when the biological context means that each column is distorted in a different way.

\begin{defn}
The \emph{monotone rank} of $A$, written $\monr(A)$, is the smallest value of $d$ such that there exists a rank-$d$ matrix $B$ and a set of monotone functions $f_1, \ldots, f_n$ such that $A_{ij} = f_j(B_{ij})$. 
\end{defn}

Notice that the monotone rank depends only on the ordering of entries within columns of the matrix, while the underlying rank depends on the complete ordering of matrix entries. Further, notice that monotone rank of both $A$ and $A^T$ is a lower bound on underlying rank: for any matrix $A$, $\monr(A)\leq \ur(A)$ and  $\monr(A^T)\leq \ur(A)$.  However, we will see in Example \ref{ex:strict} that monotone rank of a matrix can be strictly lower than its underlying rank. To develop this example, we use a simple characterization of matrices which have monotone rank one noted in \cite{egger20xxtopological}: a matrix $A$ has monotone rank one if and only if there is some order such that every column is either in this order or the reverse order.

\begin{ex}\label{ex:strict}
Consider the matrix 
$$A = \begin{pmatrix}
1 & 3 & 4\\
2 & 5 & 8\\
6 & 7 & 9\\
\end{pmatrix}$$
We first note that  $\monr(A) = 1$ because each column of $A$ is in the same order, increasing from top to bottom. However, we prove that $\ur(A) = 2$.  

Suppose to the contrary that $\ur(A) = 1$. Then there exists a rank-one matrix $B$ whose entries are in the same order as the entries of $A$. In rank one, the inner product reduces to scalar multiplication, and thus $B$ is rank one if and only if there exist real numbers $v_1, v_2, v_3$ and  $w_1, w_2, w_3$ such that $B_{ij} = w_iv_j$ 

We can choose these real numbers to be positive. 
To see this, first notice that because all rows are in the same order, the points $w_1, w_2,$ and $w_3$ must all have the same sign. We can choose this sign to be positive, flipping the signs on $v_1, v_2, v_3$ to compensate if necessary. Now, once we have fixed  $w_1, w_2,w_3 > 0$, we notice that adding a constant shift to the $v_1, v_2, v_3$ to ensure that $v_1, v_2, v_3 > 0 $ does not change the ordering among the products $v_iw_j$. 
Thus, without loss of generality, we can choose $v_1, v_2, v_3 >0$. 

This allows us to replace each entry with its logarithm. Since the logarithm is a monotone increasing function, this does not change the ordering of matrix entries. Thus, the matrix $C$ defined by $C_{ij} = \log(w_iv_j) = \log(w_i)+ \log(v_j)$ has the same ordering as the original matrix $A$ and the rank-one matrix $B$. To simplify notation, denote $w_i' = \log(w_i)$, $v_i' = \log(v_i)$. Now, we derive some additional inequalities from the inequalities in the matrix:
\\
\\
From $A_{21}< A_{12}$ and $ A_{31} > A_{22} $, we have
\begin{align*}
w_2' + v_1' < w_1' + v_2'\\
w_3' + v_1' > w_2' + v_2'
\end{align*}
Subtracting, we have the inequality
\begin{align*}
w_2'-w_3' < w_1'-w_2'.
\end{align*}
On the other hand, from $A_{13}< A_{22}$ and $A_{23}> A_{32}$, we have
\begin{align*}
w_1' + v_3' < w_2' + v_2'\\
w_2' + v_3' > w_3' + v_2'
\end{align*}
Subtracting, we have the inequality. 
\begin{align*}
w_1'-w_2' < w_2'-w_3'.
\end{align*}
Thus, we have arrived at a contradiction.\\ 

Thus, the underlying rank of $A$ is at least two, even though the monotone rank is one. This means $\ur(A) > \monr(A)$. This argument can be generalized using Farkas lemma from linear programming. There is a general characterization of matrices of underlying rank one using this idea \cite{LauraAnderson}.  \end{ex}

Despite the simple characterization of matrices of monotone rank one, we show in Section \ref{sec:comb_geo} that characterizing matrices of monotone rank two or higher exactly is computationally intractable. On the other hand, if the rows of $A$ contain the complete set sweeping orders of $v_1, \ldots, v_m$, and we know that this is indeed the complete set, then it is in principle easy to recover the monotone rank of $A$. See  \cite{padrol2021sweeps} for an in depth exploration of this set of sweeping orders as a combinatorial object.

However, we do not expect our data sets to contain this complete set of permutations. The authors of \cite{egger20xxtopological} consider the problem of determining monotone rank in this more limited condition. They show that techniques using directed complexes are effective at estimating monotone rank when the number of sweep orders $n$ is much greater than $d$ and  techniques using the Dowker complexes are effective at estimating monotone rank when the number of points being swept past $m$ is much greater than $d$.

\subsection{Sign rank}\label{sec:sign_rank}
Finally, underlying rank is closely related to the concept of \emph{sign rank}, used in theoretical computer science \cite{alon2014sign, forster2002linear, basri2009visibility}. The relationship between sign rank and underlying rank is illustrated in Figure \ref{fig:order_vs_sign}. 

\begin{defn}
The \emph{sign rank} of a matrix $A$ is the minimum rank of a matrix with the same sign-pattern as $A$:
\begin{align*}
    \mathrm{sign\, rank}(A) = \min\{\rank(B): \sign(A_{ij}) = \sign(B_{ij})\mbox{ for all } i, j\}
\end{align*}
\end{defn}

Since the only information the sign rank preserves about a matrix is its sign pattern, we typically consider the sign ranks of matrices with entries $\pm 1$, which we term \emph{sign matrices}. 
We can relate each order matrix to a family of  sign matrices. This will allow us to bound underlying rank of a matrix in terms of the sign rank. In particular, this allows us to show that there is a family of $N\times N$ matrices with underlying rank on the order of $\sqrt{N}$ for all $N = 2^n$. For large enough $n$, this exceeds all other bounds we give for underlying rank in this paper. 

\begin{figure}
    \centering
    \includegraphics[width = 4.5 in]{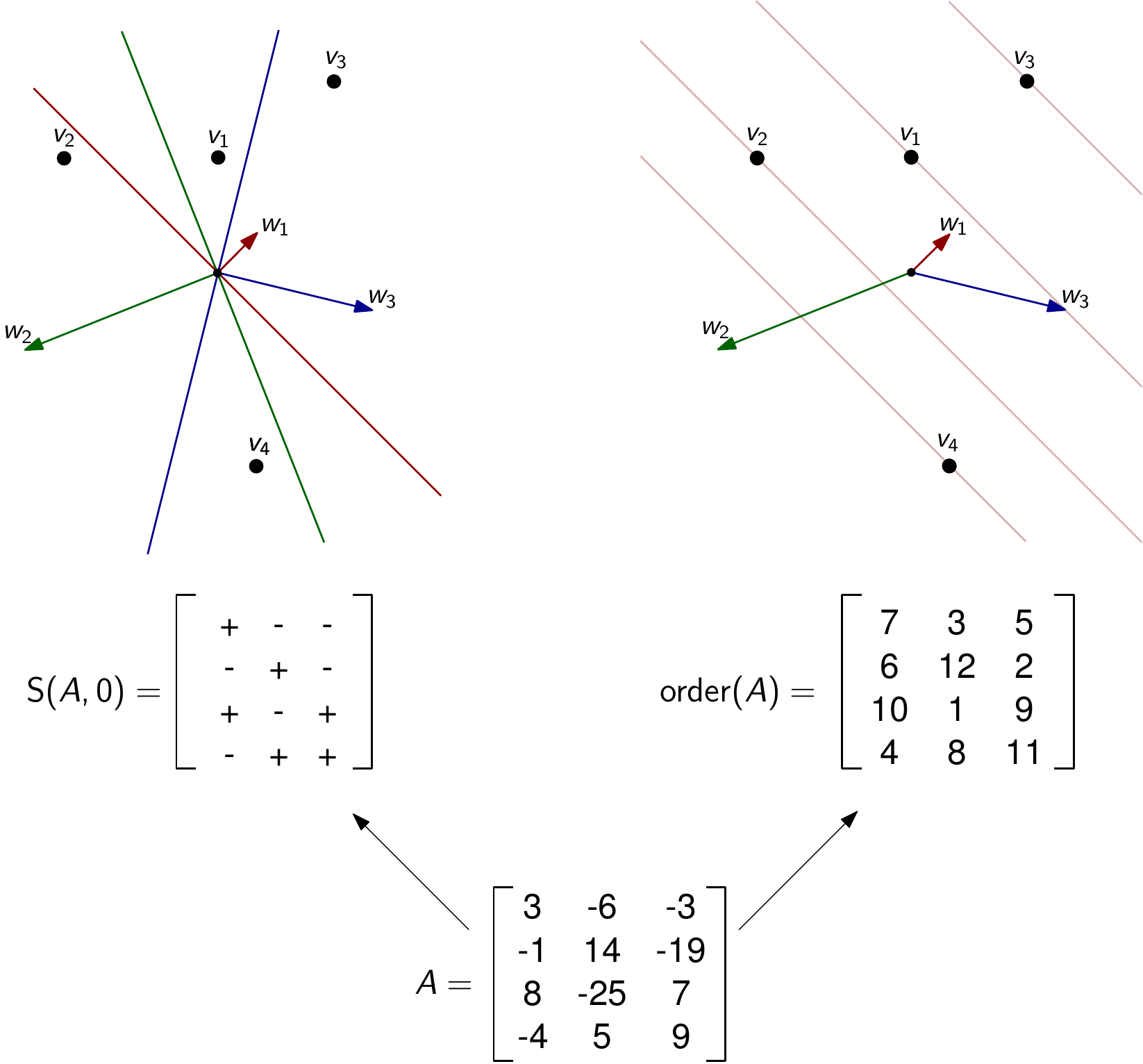}
    \caption[Comparison between sign rank and underlying rank.]{The information carried by the signs of a matrix as compared to the order. On the left, $S(A, 0)_{ij}$ indicates which side of the hyperplane perpendicular to $w_i$ through the origin the point $ v_j$ falls on. }
    \label{fig:order_vs_sign}
\end{figure}

\begin{defn}
Let $A$ be an $m\times n$  matrix, $\theta\in \R$. Then define the sign matrix of $A$ at threshold $\theta$ to be the matrix 
\begin{align*}
    S(A, {\theta})_{ij} =\sign(A_{ij} -\theta).
\end{align*}
We use the convention that if $A_{ij}-\theta = 0,$ then $S(A, {\theta})_{ij} = -1$. 
\end{defn}

Notice that the matrix  $S(A, {\theta})$ can only change when $\theta$ passes through some $A_{ij}$. Thus for each matrix $A$ there are at most $mn$ sign matrices which arise as  $S(A, {\theta})$. We can put all of these matrices together into one matrix. 

\begin{defn}\label{def:S(A)}
Define the matrix 
$$S(A) = 
\begin{pmatrix}
\sign(A,A_{11})\\
\sign(A,A_{12})\\
 \vdots\\
\sign(A, A_{nn})
\end{pmatrix}$$
That is, $S(A)$ is the matrix obtained by vertically stacking all of the matrices $\sign(A, A_{ij})$. 
\end{defn}

We show that sign rank is bounded in terms of underlying rank. 

\begin{prop}\label{prop:srur}
For all matrixes $A$, $\theta \in \R$, $\sign \rank (S(A, \theta)) \leq \sign\rank(S(A))\leq \ur(A)+1$
\end{prop}

\begin{proof}
We prove  the first inequality, $\sign \rank (S(A, \theta)) \leq \sign\rank(S(A)).$\\ First we look at the case where $\theta > \max_{i, j} A_{ij}$. Here, $S(A,\theta)$ is the all ones matrix, which has sign-rank one. Thus, because $S(A)$ has at least sign rank one, we have $\sign \rank (S(A, \theta)) \leq \sign\rank(S(A)).$
Now, for all $\theta \leq \max_{i, j} A_{ij}$, $S(A,\theta)$ is a submatrix of $S(A)$, hence $\sign \rank (S(A,\theta)) \leq \sign\rank(S(A)).$

Now, we show the second inequality, $\sign\rank(S(A))\leq \ur(A)+1.$ 
Let the underlying rank of $A$ be $r$. Let $B$ be a rank-$r$ matrix and $f$ be a monotone function such that $$A_{ij} = f(B_{ij}).$$ Now, for each $\theta$, observe that $B_\theta := B - f\inv(\theta)J$, where $J$ is the all-ones matrix, matches the sign pattern of  $S(A, \theta)$. Thus, the matrix 
$$ \begin{pmatrix}
B_{A_{11}}\\
B_{A_{12}}\\
\vdots\\ 
B_{A_{nn}}
\end{pmatrix} = 
\begin{pmatrix}B\\ B\\\vdots\\B
\end{pmatrix}
-  
\begin{pmatrix}f\inv(A_{11})J\\
f\inv(A_{12})J\\
\vdots\\
f\inv(A_{nn})J
\end{pmatrix}
$$
matches the sign pattern of $S(A).$
The matrix $[B\enspace B\enspace \cdots  \enspace B]$ is rank $R$, and can thus be written as a sum of $r$ rank one matrices. The matrix $[f\inv(A_{11})J\enspace f\inv(A_{12})J\enspace \cdots  \enspace f\inv(A_{nn})J]$ is rank one. Thus we have shown that $B_\theta$ can be written as a sum of $r+1$ rank-$r$ matrices, and thus has rank at most $r+1$. 
Thus, $\sign\rank(S(A))\leq \ur(A)+1.$ 
\end{proof}

Thus we can apply known results about sign rank to establish lower bounds on underlying rank. We focus on results which bound the sign rank of a sign matrix $M$ in terms of $||M||$. Recall that for any matrix $M$, $||M|| = \sigma_1(M)$, where $\sigma_1(M)$ is the greatest singular value of $M$. This is equal to the square root of the top eigenvalue of $M^TM.$ 

\begin{thm}\cite{forster2002linear}
Let $M\in \{\pm 1\}^{m\times n}$ be a sign matrix of sign rank $d$. Then $$d\geq \frac{mn}{||M||}.$$ 
\end{thm}

Using this together with Proposition \ref{prop:srur}, we can bound the underlying rank of $A$ in terms of the singular values of $S(A,\theta)$. 
\begin{cor}
Let $A\in \R^{m\times n}$. Then 
$$\ur(A)\geq \frac{m n}{||S(A)||}-1.$$
\label{cor:sign_rank_bound}
\end{cor}

The cited theorem, and related results, make it possible to construct explicit examples of matrices with high sign rank relative to their size. 
The primary example given in \cite{forster2002linear} is the family of Hadamard matrices $H_n$. The Hadamard matrices are a family of symmetric matrices with entries $\pm 1$ whose rows are pairwise orthogonal. They are defined recursively, with $$H_0= (1),$$ and $$H_{n+1} = \begin{pmatrix}
H_n & H_n \\ H_n & -H_{n}
\end{pmatrix}.$$ 
Notice that $H_{n}$ is $2^n\times 2^n$. Let $N:=2^n$, so that $H_n$ is $N\times N$.  
Now, we find $||H_n||$.
Now, because the columns of $H_n$ are pairwise orthogonal and the entries are $\pm 1$, we note that 
$$H_n^T H_n = NI.$$
This has eigenvalues $N, \ldots, N$. Thus, $||H_n|| = \sqrt{N}$. Thus, we have 
$$\sign\rank (H_n )\geq \frac{N}{\sqrt N} = \sqrt{N}.$$

We can also use the Hadamard matrices to construct examples of matrices with high underlying rank. To do this, we notice that given any sign matrix $S$, we can produce an order matrix $O(S)$, such that $\sign\rank(S) \leq \ur(O(S))+1$. Let $S$ be a $m\times n$ matrix with $N$ negative entries. To produce $O(S)$, we loop through the entries of $S$. If $S_{ij} = -1$, we set $O(S)_{ij}$ to be an element of $1, \ldots, N$. If $S_{ij} = 1$, we set $O(S)_{ij}$ to be an element of $N+1, \ldots, m n$. We make sure each value is only used once. Notice that $S = \sign(O(S), N)$, thus $\sign\rank(S) \leq \ur(O(S))+1$. Using this construction, we have 

\begin{ithm}\label{cor:hadamard}
For each $N = 2^n$,  $O(H_n)$ is a $N\times N$ order matrix with $$\ur(O(H_n)) \geq  \sqrt{N}-1.$$
\end{ithm} 

In general, sign rank is difficult to compute: it is complete for \emph{the existential theory of the reals}, written $\exists \R$. This is the complexity class of decision problems of the form
$$\exists(x_1 \in \R)\cdots \exists(x_n\in \R)P(x_1, \ldots , x_n),$$
where P is a quantifier-free formula whose atomic formulas are polynomial equations and inequalities in the $x_i$ \cite{broglia1996lectures}. Many classic problems in computational geometry fall into $\exists\R$ \cite{schaefer2009complexity}. %In particular, determining whether an oriented matroid is representable is $\exists \R$ complete. 
Problems which are $\exists \R$-complete are not believed to be computationally tractable. In particular, they must be NP-hard. 
The fact that sign rank is $\exists \R$ complete is established in  \cite{basri2009visibility} and independently in \cite{bhangale2015complexity}.  More precisely, for $r\geq 3$, the decision problem of determining whether a matrix $A$ has sign rank $r$ is $\exists \R$-complete. 
%This follows from the $\exists \R$-completeness of determining whether an oriented is representable. 

Note that the fact that computing sign rank is $\exists \R$ complete does not immediately establish that underlying rank is $\exists \R$ complete. However, we will show in Section \ref{sec:comb_geo} that underlying rank and monotone rank are also $\exists \R$ complete problems. 

\subsection{Convex sensing}\label{sec:cvx}
Another related problem is the convex sensing problem introduced in \cite{wu2021topological}. This problem considers a point arrangement 
$w_1, \ldots, w_n \subseteq \R^d$. The only information we have about this point arrangement is the matrix 
$$M_{ij} = f_i(w_j),$$
where each $f_i$ is a \emph{quasi-convex function}. That is, $f_i$ is a function which has convex sublevel sets, i.e. for each threshold $\theta$, 
$f\inv(-\infty, \theta)$ is either convex or empty. 

In essence, in the monotone rank and underlying rank problems, we learn about a point configuration by sweeping hyperplanes past it and noticing when the hyperplane passes each point. In the convex sensing problem, we learn about a point configuration by growing convex sets and noticing the order in which each set envelops each point. To obtain the monotone rank problem from the convex sensing problem, we restrict from quasi-convex functions to ``quasi-linear functions": functions whose sublevel sets are half-spaces. 

This restriction fundamentally changes the problem: while we will show that there exist matrices with arbitrarily high underlying rank, by Corollary 1.5 of \cite{wu2021topological}, any convex sensing problem has a degenerate solution in $\R^2$. Thus, \cite{wu2021topological} takes a probabilistic perspective, focusing on estimating dimension under assumptions about the probability distribution generating the points $w_1, \ldots, w_n$. Some techniques introduced in this paper apply to the convex sensing problem as well. In particular, since a quasi-convex function takes its maximum value on a convex set on the boundary, the maximal nodes we introduced in Section \ref{sec:minimal} are also respected by quasi-convex functions. This means that the estimation techniques introduced in this section apply to this problem as well.

\section{Improving upon Radon rank using the structure of Radon partitions}
\label{sec:counter}

Combining Proposition \ref{prop:radon_limits} with Theorem \ref{cor:hadamard}, we see that underlying rank can exceed the Radon rank. 
Theorem \ref{cor:hadamard} gives a family of matrices whose underlying rank scales with the square root of their size, while Proposition \ref{prop:radon_limits} shows that Radon rank cannot grow this fast. We next consider an explicit example where underlying rank exceeds Radon rank because the possible Radon partitions are not consistent with one another.

A matrix $A$  has Radon rank $d$ if there is a set of size $d+1$ with no possible Radon partition, i.e. partition $(\sigma, \tau)$ not induced by $A$.  However, it is possible for every set of size $d+2$ to have a possible Radon partition, but for the underlying rank to still be greater than $d$. 

\begin{prop}\label{prop:rad_strict}
The inequality in Proposition \ref{prop:radon_bound} may be strict. In particular, the order matrix $$ A = \begin{pmatrix}
12 & 13 & 3& 10 & 6 \\
13 & 14 & 4 & 9 & 5 \\
3 & 4 & 15 & 11 & 1 \\
10 & 9 & 11 & 8 & 2 \\
6 & 5 & 1 & 2 & 7 \\
\end{pmatrix}
$$
has Radon rank two and monotone rank three. 
\end{prop}

\begin{proof}
To see this, we first compute the shatter complex of $A$:  
\begin{align*}
\shatter(A) = \Delta(123, 124, 125, 134, 135, 145, 234, 235, 245, 345).
\end{align*}                                                            Since the largest sets of points we can shatter have size $3$, $\radr(A) = 2$. Now, we show that the monotone rank is, in fact, 3. 

We suppose for the sake of contradiction that $\ur (A) = 2$. Let $v_1,\ldots,v_m, w_1, \ldots, w_n\subset \R^2$ be a rank-2 representation of $A$.  
By Radon's theorem, every subset of $v_1,\ldots,v_m,$ of size 4 must has a Radon partition. If $(\sigma, \tau)$ is a Radon partition of $\rho$, then there is no hyperplane separating the sets $\{v_i\}_{i \in \sigma} $ and $\{v_j\}_{j \in \tau} $. Thus, the partition $(\sigma, \tau)$ must not be induced by $A$. 

For each subset of size 4, we compute the set of potential Radon partitions, which are the sets $\sigma, \tau$ such that $\sigma$ and $\tau$ are not swept by $A$. We find that there is exactly one potential Radon partition for each subset: 
\begin{align*}
1234: (14, 23)\\
1235: (1, 235)\\
1425: (14, 25)\\
1345: (135, 4)\\
2345: (235, 4)
\end{align*} 
Now, suppose $V$ is an arrangement of points with these Radon partitions, as illustrated in Figure \ref{fig:ur>rr}. 
Then the partition $(1, 235)$ implies that $v_1 \in \conv(v_2, v_3, v_5)$. 
The partition $(14, 25)$ implies that the line segment from $v_1$ to $v_4$ crosses out of the triangle $\conv(v_2, v_3, v_5)$ by crossing the line segment from $v_2$ to $v_5$. However, the partition $(14, 23)$ implies that the line segment from $v_1$ to $v_4$ crosses out of the triangle $\conv(v_2, v_3, v_5)$ by crossing the line segment from $v_2$ to $v_3$.  
Thus, we have reached a contradiction. 
\begin{figure}[h!]
\begin{center}
\includegraphics[width = 2 in]{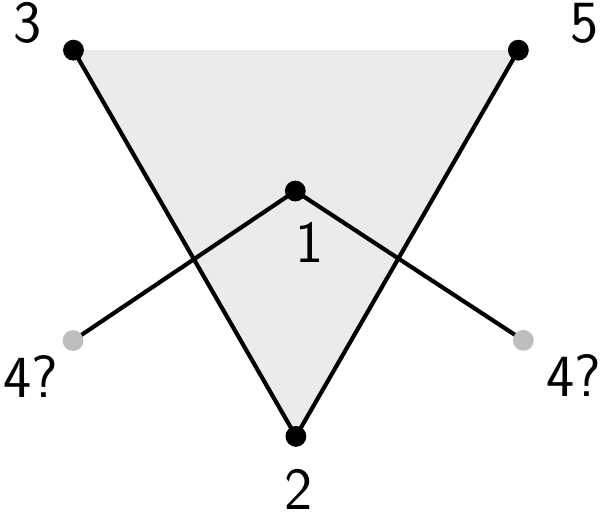}
\end{center}
\caption{We reach a contradiction when we try to construct a rank two realization of $A$. \label{fig:ur>rr}}
\end{figure}
In the next section, we generalize this obstruction by exploring the connection between underlying rank and oriented matroid theory.

\end{proof}

\section{The combinatorial geometry of underlying rank}
\label{sec:comb_geo}

In this section, we see how two combinatorial abstractions of point arrangements, oriented matroids and allowable sequences,  can inform our understanding of underlying rank. This will get us three main benefits: 
first, we will generalize the argument made in Proposition \ref{prop:rad_strict} by describing the structure the minimal Radon partitions of a set of points must have. Second, we will show that for $d\geq 2$, deciding whether a matrix has monotone rank $d$ is complete for the existential theory of the reals, and is therefore NP-hard. Contrary to this, we give a necessary condition for a matrix to have monotone rank 2 which can be checked in $O(mn^2)$ for a $m\times n$ matrix. 

\subsection{Oriented matroids}

We can recast the example in the proof of Proposition \ref{prop:rad_strict} in terms of oriented matroids, introduced in \ref{chapter:combo_background}. In order to do so, we introduce \emph{known topes} and \emph{potential circuits}. 

\begin{defn}
Let $A$ be a $m\times n$ matrix. A sign vector $X$ is a \emph{known tope} of $A$ if the sets $X^+$ and $X^-$ are swept by a column of $A$. A sign vector $Y$ is a \emph{potential circuit of rank $d+1$} of $A$ if  $|\underline{Y}| = d+2$ and $Y$ is orthogonal to every known tope. 
\end{defn}

In other words, the known topes of $A$ are the sign vectors which we know to be topes of $A$, and the potential circuits are the sign vectors of a particular size which we have not ruled out as circuits. 

\begin{ithm}
Suppose $A$ has underlying rank $d$. Then the potential radon partitions of rank $d+1$ of $A$ contain the circuits of a representable oriented matroid of rank $d+1$.\label{thm:potential} 
\end{ithm}

\begin{proof}
Let $\cT$ be the set of known topes of $A$. Let $v_1, \ldots, v_m \subseteq \R^d$ be a point arrangement realizing $A$. Since each $X\in \cT$ corresponds to a hyperplane which properly separates the points indexed by its positive and negative parts, we can perturb $v_1, \ldots, v_m$ to be in general position without changing the set of known topes. Then the set of known topes $\cT$ is contained in the set of covectors $\cL$ of a uniform oriented matroid of rank $d+1$. We consider the set of circuits $\cC$ of this oriented matroid. Since each circuit $Y \in \cC$ is orthogonal to every tope, the set $\cC$ is contained in the set of potential circuits of rank $d$ of $A$. 
\end{proof}

We reexamine Example \ref{ex:strict} in light of this result. The matrix $$ A = \begin{pmatrix}
12 & 13 & 3& 10 & 6 \\
13 & 14 & 4 & 9 & 5 \\
3 & 4 & 15 & 11 & 1 \\
10 & 9 & 11 & 8 & 2 \\
6 & 5 & 1 & 2 & 7 \\
\end{pmatrix}
$$
has known topes 
\begin{align*}
\cT = \{&+++++, ++-++, ++-+-, ++---, -+---, \\
        &-----, ++++-, -+++-, --++-, --+--,\\
        &+++--, +-+--, ++--+, +---+, ----+\}.
\end{align*}
This allows us to compute the potential circuits of rank $2$ as 
\begin{align*}
\cC = \{&+--+0, -++-0, +--0-, -++0+, +-0+-,\\ &-+0-+, +0+-+, -0-+-, 0++-+, 0--+-\}.
\end{align*}
Notice that there is exactly one pair of potential circuits on each support. Thus, any set of circuits contained in this set of potential circuits must actually be the full set of potential circuits. Thus, if $A$ has underlying rank 2, $\cC$ must follow the circuit axioms. We apply axiom $C3$ to $X = +--0-$, $Y = -+0-+$, $e = 5$. Then there must exist $Z \in  \cC$ with 
$Z^+ \subseteq X^+ \cup Y^+ \setminus \{e\} = \{1, 2\}$, 
$Z^- \subseteq X^- \cup Y^- \setminus \{e\} = \{1, 2, 3, 4\}$. 
No such potential circuit is present: the potential circuits on support $\{1, 2, 3, 4\}$ are $+--+0$ and $-++-0$, which do not conform to this pattern. Thus, by Theorem \ref{thm:potential}, the underlying rank of $A$ is at least 3. 

In fact, Example \ref{ex:strict} was generated using Theorem \ref{thm:potential}: we generated $A$ by generating random matrices of rank 3, computing the potential circuits of rank 2, and checking whether the potential circuits had a subset which satisfied the circuit axioms for oriented matroids. 

\subsection{Allowable Sequences}
\label{sec:allowable}
While the oriented matroid of a point configuration describes a point configuration in terms of its minimal Radon partitions, it does not capture all combinatorial information about the point configuration. \emph{Allowable sequences} fill in some of this further information for point configurations in $\R^2$. 
In this section, we discuss allowable sequences, and how they can be used to determine whether or not matrices have underling rank two. 

Let  $v_1, \ldots, v_n\in \R^2$ be a point arrangement in the plane. As discussed in Section \ref{sec:geom}, sweeping a sequence of hyperplanes normal to a vector $w$ produces a permutation of $[n]$ resulting from the order in which the hyperplane encounters the points.  As we rotate the normal vector to $w$ around the circle and record each permutation we encounter, we obtain a sequence of permutations. This sequence must be an allowable sequence:

\begin{defn}
Define $-\pi$ to be the permutation in the reverse order, i.e. if $\pi = 1342, -\pi = 2413.$ 

An \emph{allowable sequence} is a circular sequence of permutations $ S = \pi_1,\pi_2, \ldots, \pi_{2m}$ of $[n]$ such that: 
\begin{enumerate}
    \item $\pi \in S \Leftrightarrow -\pi\in S$ 
    \item $\pi_{i+1}$ is obtained from $\pi_i$ by reversing the order of one or more substrings of $\pi_i$
    \item The order of each pair $i, j$ is reversed exactly once between each $\pi_{k}, -\pi_{k}$
\end{enumerate}
An allowable sequence is called \emph{simple} if each pair $\pi_i, \pi_{i+1}$ differ by reversing a pair of adjacent entries. We are primarily concerned with simple allowable sequences because they correspond to point arrangements in general position. 
\end{defn}

\begin{ex}
We return to the point arrangement $v_1, v_2, v_3, v_4$ from our running example, illustrated in Figure \ref{fig:allowable}.  The allowable sequence associated to this point configuration is 2134, 2143, 2413, 4213, 4123, 4132, 4312, 3412, 3412, 3142, 3142, 3124. These are the orders in which a hyperplane swept perpendicular to the vector $w$ sweeps past $v_1, v_2, v_3, v_4$ as we rotate it counterclockwise.
\end{ex}

\begin{figure}[ht!]
    \centering
    \includegraphics[width = 3 in]{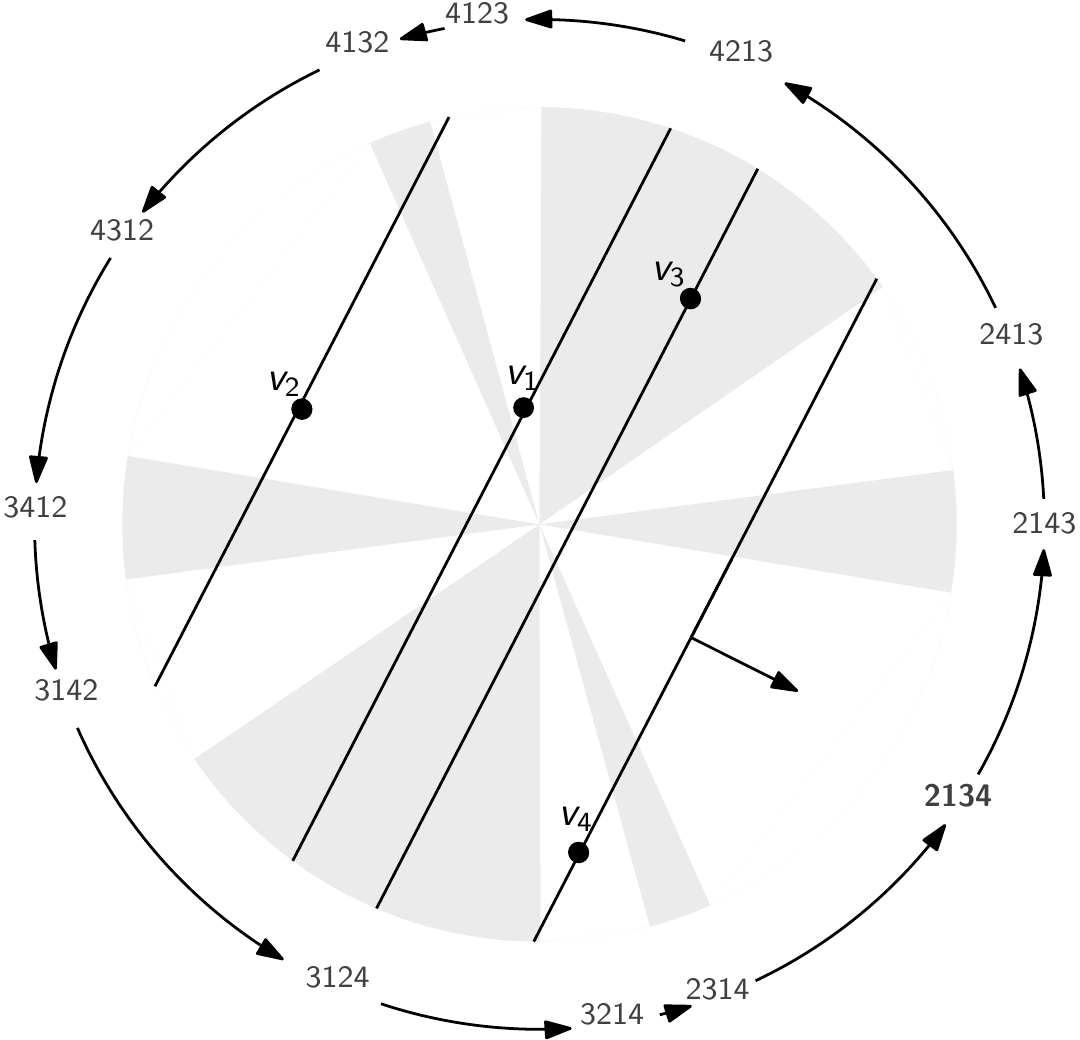}
    \caption[An example of an allowable sequence.]{The allowable sequence arising from the point configuration $v_1, v_2, v_3, v_4$.  Alternating shaded and un-shaded regions correspond to cones of sweep directions which yield particular orders.    \label{fig:allowable}  }
\end{figure}

An allowable sequence is \emph{realizable} if it arises from a planar point configuration. As with oriented matroids, not every allowable sequence is realizable, and realizablity is difficult to determine. By \cite{hoffmann2018universality}, checking whether an allowable sequence is realizable is complete for the existential theory of the reals.

\subsubsection{Using allowable sequences to estimate underlying rank}

Now, we notice the connection between monotone rank and allowable sequences. If $A$ is a matrix of monotone rank two, each column of $A$ corresponds to a sweep order a point arrangement in the plane. This means that there is a realizable allowable sequence which contains the column permutations of $A$ in some order.  In order to translate notation between notation used for allowable sequences and notation used for monotone rank, we introduce the \emph{sort permutation} of a list of numbers. 

\begin{defn}
Let $ a = [a_1, \ldots, a_n]$ be a vector with real entries. The \emph{sort permutation} $\pi( a)$ of $a$ is the order of entries in $ a$: that is, $\pi( a)_i$ lists the index of thesm $i$-th smallest entry of $ a$. 
\end{defn}

\begin{obs}\label{obs:allowable}
A matrix $A$ has monotone rank two if and only if the set of sort permutations of its columns can be ordered such that they form a subsequence of a realizable allowable sequence. 
\end{obs}
Because realizability is difficult to check, a weaker version of this observation is more useful.

\begin{obs}\label{obs:allowable}
If a matrix $A$ has monotone rank two, then the sort permutations of its columns can be ordered such that they form a subsequence of an allowable sequence. 
\end{obs}

Notice this condition is also satisfied if a matrix has underlying rank two. We show that this condition is relatively easy to check. 

\begin{prop} \label{prop:subsequence}
 A circular sequence of permutations $S = \pi_1, \pi_2, \ldots, \pi_{2m}$ such that $\pi \in S \Leftrightarrow -\pi\in S$ is a subsequence of an allowable sequence if and only if satisfies condition  (3) to be an allowable sequence: 
the order of each pair $i, j$ is reversed exactly once between each $\pi_{k}, -\pi_{k} = \pi_{k+m}$
\end{prop}

\begin{proof}
To prove the ``only if" part of this proposition, we note that if we fail to meet condition (3), we cannot fix it by adding more permutations. To prove the other direction, we construct a simple allowable sequence $\sigma_1, \ldots, \sigma_{2M}$ containing $ \pi_1, \pi_2, \ldots, \pi_{2m}$ as a subsequence. We need to show that for each  $\pi_{i}, \pi_{i+1}$, we can fill in $$\pi_{i} = \sigma_j, \sigma_{j+1}, \ldots, \sigma_{j+k} = \pi_{i+1}$$ such that each ${\sigma_j}, \sigma_{j+1}$ differs by reversing a pair of adjacent entries. 

We are able to do this iteratively for each pair: to produce $\sigma_{j+1}$ from $\sigma_{j}$, we look for a pair of adjacent elements of $\sigma_j$ which are in the reverse order of their order in $\pi_{i+1}.$ We reverse this pair to produce $\sigma_{j+1}$ from $\sigma_j.$ Once there is nothing left to reverse, we have arrived at $\pi_{i+1}$. 

We do this for each pair $\pi_i, \pi_{i+1}$ until we reach $\sigma_k = -\pi_{i}$. From here, we fill in the rest of the allowable sequence using the rule $\sigma_{i+k} = -\sigma_i$. By construction, we have clearly satisfied conditions (1) and (2) to be an allowable sequence. Now, notice that since we only reverse a pair of elements between $\sigma_{j}, \sigma_{j+1}$ if $i$ and $j$ were reversed between some $\pi_i$ and $\pi_{i+1}$, we do not introduce any extra swaps of each pair between any $\pi_i$, $-\pi_i$. Thus, each pair $i,j$ is reversed exactly once between each $\sigma_i, -\sigma_i = \sigma_{i+k}$. Thus, $\sigma_1, \sigma_2, \ldots, \sigma_{2k}$ is an allowable sequence. 
\end{proof}

\begin{ex}\label{ex:allowable}
The matrix 
\[A = \begin{pmatrix}
1&2&1&1\\
2&1&3&2\\
3&3&2&4\\
4&4&4&3
\end{pmatrix}\]
has Radon rank two and monotone rank three. To see that the Radon rank is two, note that no column of the matrix induces the partition $(14, 23)$.  To see that the monotone rank is three, we consider its set of column permutations and reversed column permutations 
\begin{align*}
1234, 2134, 1324, 1243, 4321, 4312, 4213, 3412    
\end{align*}
By brute force, we can see that there is no ordering which satisfies Proposition \ref{prop:subsequence}. Thus, by Observation  \ref{obs:allowable}, $$\ur(A) \geq \monr(A) \geq 3.$$ We will later show how to determine this much more efficiently. 

\end{ex}

Next, we give an $O(mn^2)$ algorithm to determine whether the column permutations of a matrix can be placed in an order which satisfies the conditions of Proposition \ref{prop:subsequence}. This algorithm makes use of a definition of distance between two permutations. 

\begin{defn}
Given two permutations $\pi, \sigma$, define their  distance by
$$d(\pi, \sigma) = |\{\{i, j\} \mid i \mbox{ and } j \mbox{ are swapped between } \pi \mbox{ and } \sigma \}|.$$
\end{defn}

\begin{lem}
Let $\pi_i, \pi_j, \pi_k$ be permutations of $[n]$. Then $\pi_j$ is allowed between $\pi_i$ and $\pi_k$ if and only if $d(\pi_i, \pi_j) + d(\pi_j, \pi_k) = d(\pi_i, \pi_k)$. 
\end{lem}

\noindent \textbf{Algorithm}\\
\textbf{Input}: A sequence of permutations $\pi_1, \ldots, \pi_{2n}$\\
\textbf{Output}: The same sequence of permutations, sorted in an order which satisfies Proposition \ref{prop:subsequence}, if possible. 
Otherwise, the algorithm will report failure. \\
\\
\textbf{Step one}:
For each $j\geq 1$ compute and record $d(\pi_i, \pi_j)$.  
Sort the permutations in the order of increasing $d(\pi_i, \pi_j)$, resolving ties arbitrarily. From here on out, we will assume that $\pi_1, \ldots, \pi_n$ is written in this order. \\
\textbf{Step two}: Now, we begin placing the permutations into a new order satisfying Proposition \ref{prop:subsequence}. 
We first construct only the first half of this order, between $\pi_1$ and $-\pi_1$, placing $\pi_1$ at the first spot in this order and $-\pi_1$ at the last spot. 
Now, beginning with $\pi_2$ and continuing in our sorted order until we reach $\pi_{2n}$, we check whether each new entry to be added $\pi_i$ is allowed between the last permutation added to the list, $\pi_j$, and $-\pi_1.$  
If yes, we insert it there. If not, we do not. \\
\textbf{Step three}: Now, we check whether, for each pair $\pi_i, -\pi_i$, we have inserted at least one member of the pair into the sequence. If we have not, we report failure. If we have, we then complete the rest of the sequence using the rule $\pi_{i+n} = -\pi_i$. We return this sequence. 

\begin{ex}
We apply our algorithm to the matrix 
\[A = \begin{pmatrix}
1&2&1&1\\
2&1&3&2\\
3&3&2&4\\
4&4&4&3
\end{pmatrix}\]
from Example \ref{ex:allowable}, trying to fit the set of column permutations and their reverses
\begin{align*}
1234, 2134, 1324, 1243, 4321, 4312, 4231, 3421    
\end{align*}
 into a sequence satisfying Proposition \ref{prop:subsequence}. \\
\textbf{Step one:} 
Sorting by distance from $1234$, we have 
 \begin{align*}
1234, 2134, 1324, 1243, 4312, 4231, 3421  , 4321
\end{align*}
\textbf{Step two:}
Now, we begin a new sequence with 
\begin{align*}
 1234, 4321
\end{align*}
\begin{itemize}
\item Now, we check whether $2134$ is allowed between $1234$ and $4321$. It is, so we insert it
\begin{align*}
 1234, 2134, 4321
\end{align*}
\item Next, we check whether $1324$ is allowed between $2134$ and $4321$. It is not, because the pair $12$ gets reversed twice along the sequence. 

\item Next, we check whether $1243$ is allowed between $2134$ and $4321$. Again, it is not, because the pair $12$ gets reversed twice along the sequence. 

\item Next, we check whether $4312$ is allowed between $2134$ and $4321$. Again, it is not, because the pair $12$ gets reversed twice along the sequence. 

\item Next, we check whether $4231$ is allowed between $2134$ and $4321$. It is, so we insert it
\begin{align*}
 1234, 2134, 4231, 4321
\end{align*}

\item Finally, we check whether $3421$ is allowed between $4231$ and $4321.$ It is not, because the pair $34$ is reversed twice along the sequence. 

\end{itemize}

Thus, the final sequence we have constructed is 
\begin{align*}
 1234, 2134, 4231, 4321
\end{align*}

\textbf{Step three}
We report failure because we have not inserted either member of the pair $1243, 3412$. 
\end{ex}

\begin{prop}
The sequence of permutations returned by this algorithm is a subsequence of an allowable sequence. 
If the algorithm reports failure, then these is no way to reorder $\pi_1, \cdots, \pi_n$ so that they form a subsequence of an allowable sequence. The worst case runtime of this algorithm is $O(mn^2)$. 

\label{prop:allow_alg}
\end{prop}

\begin{proof}
We first prove that any sequence of permutations returned by this algorithm is a subsequence of an allowable sequence using Proposition \ref{prop:subsequence}. By construction, we have that $\pi \in S \Leftrightarrow -\pi\in S$. Now, we check that each pair of points is reversed exactly along the sequence between $\pi_1$ and  $-\pi_1$. Note that this property is present at the beginning of Step 2, when the sequence is $\pi_1, \pi_2, -\pi_1$. Now, we show that it is maintained every time we insert a permutation.   Every time we insert a permutation $\pi_i$ between $\pi_j$ and $-\pi_1$, we check whether $\pi_i$ is allowed between $\pi_j$ and $-\pi_1$. This means that if a pair of points is reversed between $\pi_1$ and $\pi_j$, then it is not reversed again between $\pi_j$ and $\pi_i$. Thus, we satisfy condition (3) as well. 

Now, we prove that if the algorithm reports failure, then there is no allowable sequence containing all of the permutations. Suppose there exists such an allowable sequence. Without loss of generality, we can assume that the permutation $\pi_2$ comes after $\pi_1$ and before $-\pi_1$ in this sequence, since otherwise we can reverse the order of the sequence to get another allowable sequence. Now, notice that as we move along the sequence from $\pi_1$ to $\pi_n$, the distance of each permutation from $\pi_1$ must increase monotonically. Thus, a permutation $\pi_j$ must be inserted into the sequence at the position chosen by the algorithm if it can be inserted into the sequence at all. Further, since every pair is reversed between permutations $\pi_i$ and $-\pi_i$, for all $i\neq 1$, exactly one of $\pi_i$ and $-\pi_i$ appears between $\pi_1$ and $-\pi_1$ in a sequence. Thus, if we apply this algorithm to a set of permutations which can be arranged to form a subsequence of an allowable sequence, for each $i$, we can fit either $\pi_i$ or $-\pi_i$ into the sequence, so we do not report failure. 

Finally, we prove that the running time is $O(mn^2)$ for a set of $m$ sequences of $n$ points. In step one,  we can compute the distance between two permutations in time $O(n^2)$. Thus, we can compute all of the distances from $\pi_1$ to $\pi_i$ in time $O(mn^2)$. Next, we can sort the sequences by distance in time $O(m\log m).$ Asymptotically, we must have $\log m < n^2$, thus the overall runtime scales as $O(mn^s)$. In step two, we look at each permutation $\pi_j$ once and compute two  distances $d(\pi_{i}, \pi_j)$ and $d(\pi_{j}, -\pi_1)$ to determine whether $\pi_j$ can be placed into the sequence. Thus, we can do step two in time $O(mn^2)$ as well. Finally, we check whether for each $\pi_i$, we have placed either $\pi_i$ or $-\pi_i$ into the sequence in time $O(n)$, since the fact that the list is sorted by distance from $\pi_1$ means that we know where to look for $-\pi_i$ given $\pi_i$. 
 
\end{proof}

\subsubsection{Computing underlying rank is hard}
Now, we apply the result of \cite{hoffmann2018universality} that the problem of checking realizability for allowable sequences is complete in the existential theory of the reals $(\exists \R)$ in order to prove that determining whether a matrix has monotone rank two is $\exists \R$-complete. We then apply this result to show that determining whether a matrix has underlying rank two is $\exists \R$-complete as well. 

\begin{prop} 
Checking whether a matrix has monotone rank  two is complete  $\exists \R$-complete and therefore NP-hard.
\end{prop}

\begin{proof}
We can reduce the problem of determining whether an allowable sequence is realizable to the problem of determining whether a matrix has monotone rank two. 
Given an allowable sequence $\pi_1, \ldots, \pi_m$, we can construct a matrix $A$ whose columns have sort permutations $\pi_1, \ldots, \pi_m$. By Observation  \ref{obs:allowable}, $A$ has monotone rank two if and only if $\pi_1, \ldots, \pi_m$ form a subsequence of a realizable allowable sequence.
\end{proof}

Thus, we have shown that the problem of computing monotone rank is computationally intractable.

Finally, we show that computing underlying rank is intractable. We do this by reducing the problem of determining whether an allowable sequence is realizable to the problem of determining whether a matrix has underlying rank two.

To to this, we construct a matrix $A(\pi_1, \ldots, \pi_{2n})$ which has rank two if and only if $\pi_1, \ldots, \pi_{2n}$ is realizable. Ensure that the $i^{th}$ column of $A$ has the sort permutation $\pi_i$. Now, scale each  column that $A_{ij} > A_{k\ell}$ whenever $j > \ell$. We prove the following:

\begin{lem}\label{lem:underlying_reduction}
The matrix $A(\pi_1, \ldots, \pi_{2n})$ has underlying rank two if and only if $\pi_1, \ldots, \pi_{2n}$ is realizable. 
\end{lem}

\begin{proof}

First, notice that if $A$ has underlying rank two, then it also has monotone rank two. Thus by Observation \ref{obs:allowable}, $\pi_1, \ldots, \pi_{2n}$ is realizable.

Now, suppose $\pi_1, \ldots, \pi_{2n}$ is realizable by the point configuration $v_1, \ldots, v_m$. Let $i, j$ be the pair of indices whose order is switched between $\pi_{2n}$ and $\pi_1$.  Let $ x_{ij}$ be the vector which points from $v_i$ to $v_j$. We produce a new point configuration $v_1, \ldots, v_m$ by letting $v'_k = v_k + C  x_{ij}$ for some constant $C$. We can chose the constant $C$ and sweep vectors $w_1, \ldots, w_m$ such that the matrix $B_{ij} = w_iv_j$ is a rank two realization of $A$.

Notice that for any $\epsilon > 0$, we can choose a value of $C$ such that the angle between the vectors $v_k'$ and $v_j'-v_i' = v_j-v_i$ is less than $\epsilon.$ Thus, we choose $C$ such that the angle between $v_j'-v_i'$ and $v_k'$ is smaller than the angle between $v_j'-v_i'$ and any $v_\ell'-v_p'$.

Thus when we cyclically order the vectors 
$\{v_k'\}_{k \in [m]} \cup \{v_k'-v_\ell'\}_{k, \ell \in [m]} $, all of the vectors $v_k'$ are closer to $v_j-v_i$ then they are to any other $v_k'-v_\ell'$. This means that when we consider the allowable sequence of the larger point arrangement 
$\{v_k'\}_{k\in [m]} \cup \{\mathbf 0\}$, each permutation which has $i$ before $j$ occurs with $\mathbf 0$ at the beginning, and each permutation with $j$ before $i$ occurs with $\mathbf 0$ at the end. 

Thus for each permutation $\pi_k$ in the allowable sequence, we can choose a sweep direction $w_k$ which sweeps past $v_1, \ldots, v_m$ in the order $\pi_k$, with $v_\ell \cdot w_k > 0$ for all permutations with $i$ before $j$ and $v_\ell \cdot w_k < 0$ for all permutations with $j$ before $i$. Now, we rescale each $w_k$ to get a new vector $w_k'$ so as to satisfy the condition that $A_{ij} > A_{k\ell}$ whenever $j > \ell$. For the first $n$ permutations, $i$ comes before $j$, so $v_\ell \cdot w_k' > 0$. As we increase $k$ we can shrink the magnitude of $w_k'$ to satisfy $v_i' \cdot w_j' \leq v_k'w_{k+1}'$.  Once we pass permutation $\pi_{n+1}$, $j$ comes before $i$, and so the inner products satisfy $v_\ell \cdot w_k < 0$. Now, as we increase $k$, we can increase the magnitude of $w_k'$ as we increase $k$ in order to satisfy $v_i' \cdot w_j' \leq v_k'w_{k+1}'$. 

Thus, $v_1, \ldots, v_m, w_1, \ldots, w_{2n}$ is a rank two representation of $A$. 

\end{proof}

%\emph{Stuff to fix here because of row to column switch.}

\begin{ithm}\label{cor:urank_hard}
Checking whether a matrix has underlying rank two is $\exists \R$-complete and thus NP-hard. 
\end{ithm}

\begin{proof}
We show that there is a reduction from the problem of determining whether an allowable sequence is realizable to the problem of determining whether a matrix has underlying rank two. Let $\pi_1, \ldots, \pi_{2n}$ be an allowable sequence. Then by Lemma \ref{lem:underlying_reduction}, the matrix $A(\pi_1, \ldots, \pi_{2n})$ has underlying rank two if and only if $\pi_1, \ldots, \pi_{2n}$ is realizable. Thus, there is a polynomial time reduction from the problem of determining whether an allowable sequence is representable to the problem of determining whether a matrix has underlying rank two. Thus, checking whether a matrix has underlying rank two is $\exists \R$-complete and thus NP-hard. 
\end{proof}

\section{Conclusion and Open Questions}
In this chapter, we introduced the concept of the underlying rank of a matrix, and introduced a set of techniques for estimating it. Many open questions, of both biological and mathematical interest, remain. We highlight some here.

In Section \ref{sec:radon}, we introduce the Radon rank of a matrix as a lower bound for its underlying rank. However, the Radon rank is not stable under localized noise:  shattering even one set of size $d+1$ is enough to raise the Radon rank to $d$, while only involving $d+1$ rows of the matrix.
% However, it might be possible to define a more mathematically robust version or Radon rank, for instance by counting the number of sets of size $d+1$ a matrix shatters and comparing to random matrices. 

In Section \ref{sec:comb_geo}, we use the theory of oriented matroids and allowable sequences to give lower bounds on underlying rank. Our work in this section leaves open interesting questions in oriented matroid theory. In particular, our Theorem \ref{thm:potential} states that if a matrix has underlying rank $d$ if its potential circuits of rank $d$ contain the circuits of an oriented matroid of rank $d+1$. This serves as a ``combinatorial relaxation" of the underlying rank problem. While computing the underlying rank problem is not computationally tractable, this relaxed version may be. 

\begin{question}
Is there a combinatorial characterization of the when a given set of potential circuits contain the circuits of an oriented matroid? 
Is there an efficient algorithm to determine whether this is the case?
\end{question}

These questions suggest building a more general theory of ``partial oriented matroids" which asks when a given set of sign vectors can be extended to be the set of topes or circuits of an oriented matroid. 
Progress in this area would also be helpful towards answering Question \ref{q:axioms}, since knowing an oriented matroid lives above a certain code also provides partial information about the matroid.

%% file: ThresholdLinear/ThresholdLinearIntro.tex
%auto-ignore 

% !TEX root = ../YourName-Dissertation.tex

\chapter{Introduction to Threshold-Linear Networks } \label{chapter:TLNs1}

In this chapter, we introduce threshold-linear networks and summarize past work on them. 
In Chapter \ref{chapter:nullclines}, we prove some basic results about nullclines of threshold-linear networks and how they shape trajectories, saving our main results for Chapter \ref{chapter:TLNs2}. 

\section{Past results}

Mathematical models of neural activity which are able to reproduce the rich variety of behavior observed in neural circuits must be nonlinear.
We study some of the simplest possible nonlinear models for neural activity, threshold-linear networks (TLNs). 
The firing rates $x_1, \ldots, x_n$ of neurons in a TLN are governed by the system of ordinary differential equations
\begin{align} \label{eqn:tln}
\frac{dx_i}{dt} &= -x_i + \left[\sum_{j = 1}^n W_{ij}x_j + b_i\right]_+,
\end{align}
which can be written in vector form as
\begin{align}
\frac{dx}{dt} &= -x + \left[Wx + b\right]_+. 
\end{align}
Here, $[y]_+ := \max\{y, 0\}$,  $W_{ij}$ is the strength of the input to neuron $i$ from neuron $j$, and $ b$ is a constant external drive. 
Variations of this model have been used for decades as models of neural systems, such at the horseshoe crab retina \cite{hartline1957inhibitory} and the mammalian visual cortex \cite{von1973self}. More generally, neurons whose response depends on whether their input exceeds a threshold  go back to the beginnings of theoretical neuroscience: the McCulloch-Pitts neuron is essentially a discrete threshold-linear neuron and the Hopfield network is essentially a discrete threshold-linear network \cite{mcculloch1943logical, hopfield1982neural}. 

While Equation \ref{eqn:tln} is a common model for neural activity, other threshold-linear models for neural activity appear in the literature, such as 
\begin{align} 
\label{eqn:tln_2}
\frac{d v}{dt} &= -v + W[ v]_+ +b
\end{align}
The equivalence of models \ref{eqn:tln} and \ref{eqn:tln_2} is proven in \cite{miller2012mathematical}, generalizing a result from \cite{beer2006parameter}. When $W$ is invertible and $b$ is constant, \cite{beer2006parameter} shows that if $x$ evolves according to \ref{eqn:tln} if and only if $ v = W x + b$ evolves according to \ref{eqn:tln_2}. The more general case, when $W$ may not be invertible and $b$ may vary over time, is tackled in \cite{miller2012mathematical}. 	

Because TLNs are built out of locally linear systems, they are surprisingly tractable mathematically as compared to other high-dimensional nonlinear systems. 
However, while solutions to systems of linear ODEs either converge to a stable fixed point or diverge to infinity, solutions to TLNs can exhibit the full range of nonlinear behavior including multistability, limit cycles, and chaos. 
Initial mathematical work on threshold-linear networks focused on their capacity for multistability \cite{hahnloser1998piecewise,feng1996qualitative}. When the weight matrix $W$ is symmetric, it is in fact the case that multistability is the only feature of nonlinear dynamics that threshold-linear networks exhibit, but this is not true in general \cite{hahnloser2000permitted}. 

The activity of a neural circuit is determined by a variety of factors, such as intrinsic neural dynamics, external input, and the structure of network connectivity itself.  
Our goal is to understand the role of connectivity in shaping neural activity. 
Thus, we study models of neural activity which isolate the role of connectivity. 
Combinatorial threshold-linear networks (CTLNs) are a special class of TLNs where $W$ is determined by a directed graph $G$ together with parameters $\varepsilon$ and $\delta$ satisfying $\delta > 0$, $0< \varepsilon< \frac{\delta}{\delta + 1}$, together with a constant external drive $\theta  >0   $ which is the same for all neurons.   
When $i\not \to j$, neuron $i$ strongly inhibits neuron $j$. When $i\to j$, neuron $i$ weakly inhibits neuron $j$. Values of $W_{ij}$ are given by the rule:

\vspace{-.15in}
\begin{align*}
W_{ij}(G, \varepsilon, \delta) = \begin{cases}
\;\;\; 0 &\mbox{ if }\,\, i = j\\
-1 + \varepsilon&\mbox{ if } j \to i \mbox{ in }\,\, G\\
-1 - \delta &\mbox{ if } j \not \to i \mbox{ in }\,\, G.
\end{cases} 
\end{align*}
We interpret this as a network of $n$ excitatory neurons, against a background of instantaneous non-specific inhibition.  
Notice that the network is truly all-to-all connected, even when the graph $G$ is missing many edges. 
In particular, even networks which look feedforward are recurrent. 
CTLNs can have a wide range of dynamics, ranging from multistability to limit cycles and chaos.

CTLNs fall into a more general category of competitive TLNs, networks whose weights are negative. 
%\begin{defn}
%A TLN is \emph{competitive} if $W_{ij} \leq 0$ and $W_{ii} = 0$ for all $i$ and $j$, and $\theta \geq 0$.  
%\end{defn}
Given a competitive TLN with weight matrix $W$ and constant external drive $\theta$, we can define a graph $G_W$ by 
\[i\to j \mbox{ in } G_W \iff 	W_{ij}  > -1.\] 
Note that if $W$ is the weight matrix of a CTLN with graph $G$, $G_W = G$. We view the class of TLNs with graph $G$ as relaxations of the CTLN on $G$.

Even though TLNs are nonlinear, they are locally or ``patchwork" linear, as illustrated in Figure \ref{fig:tln_eqn}.  That is, we can divide up the positive orthant into at most $2^n$ chambers based on whether $\sum_{j = 1}^n W_{ij}x_j + \theta \geq 0	$ or  $\sum_{j = 1}^n W_{ij}x_j + \theta <0 $ for each $i$. We label each chamber by the set $\sigma = \{i \mid \sum_{j = 1}^n W_{ij}x_j + \theta \geq 0\}$. Within chamber $\sigma$, the dynamics are governed by the purely linear dynamical system 
\begin{align*}
\frac{dx_i}{dt} =\begin{cases} -x_i + W_{ij}x_j + \theta \qquad &i\in \theta\\
 -x_i \qquad &i \notin\theta.
 \end{cases}
\end{align*}
Define $A^\sigma$ to be the matrix of this linear system. That is, if $i\in \sigma$, the $i^{th}$ row of $A^\sigma$ is equal to the $i^{th}$ row of $-I + W$, and if $i\notin \sigma$, the $i^{th}$ row of $A_\sigma$ is equal to the $i^{th}$ row of $-I$. 

\begin{figure}
\begin{center}
\includegraphics[width = 4 in]{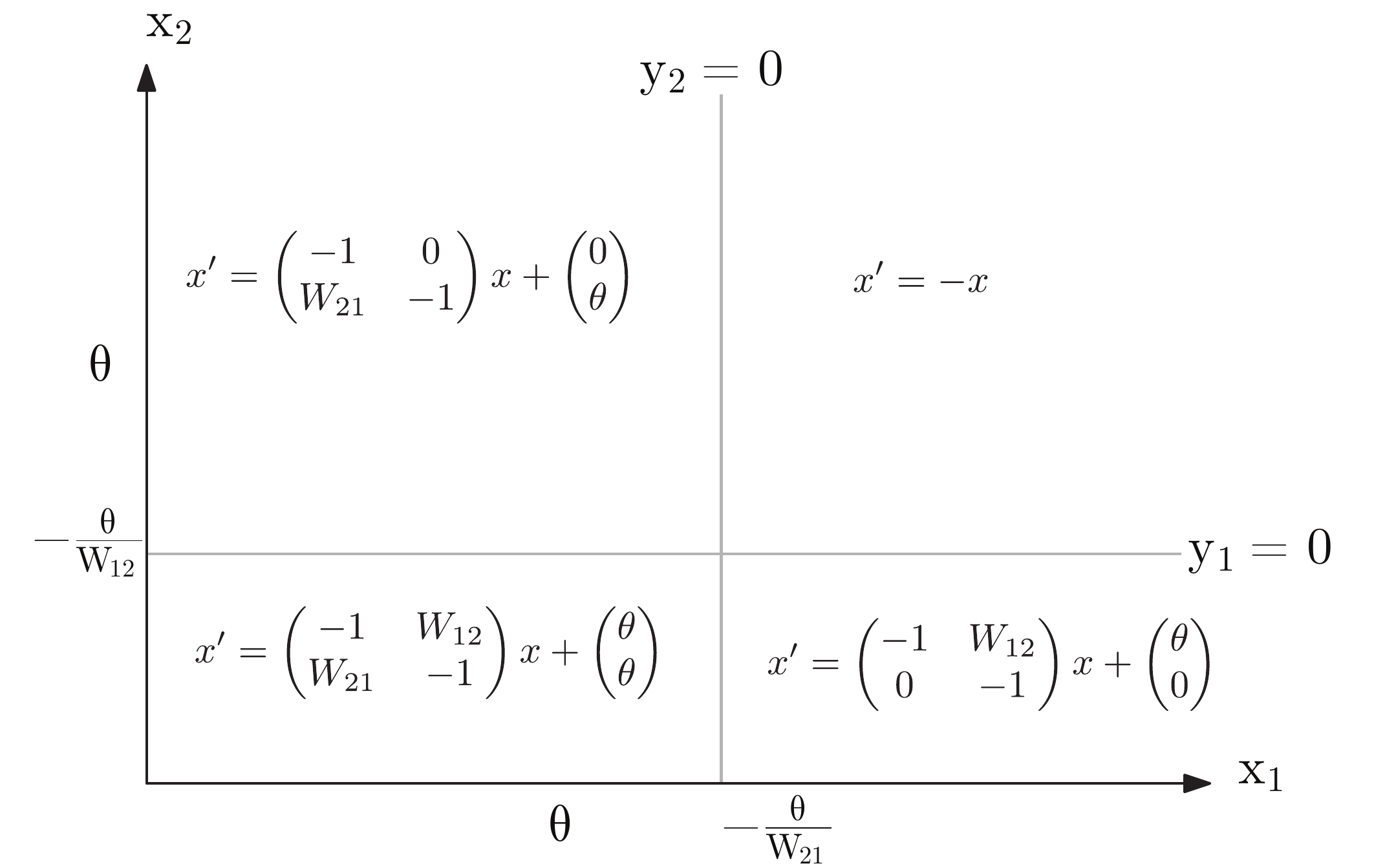}
\end{center}
\caption[TLNs are patchwork-linear networks]{TLNs are patchwork-linear networks, with a different linear system being active in each polyhedral chamber. \label{fig:tln_eqn}}
\end{figure}

The fixed points of a dynamical system are the points $x^*\in \R^d$ such that 
$\frac{dx_i}{dt}\Big|_{x^*} = \,0$ for all $i\in [n]$. 
Because of patchwork linearity, each fixed point of a TLN must be a fixed point of the relevant linear system. 
However, a fixed point of a linear system may fail to be a fixed point of the TLN if that fixed point falls outside of the appropriate linear chamber. 
The \emph{support} of a fixed point is its set of active neurons: $x^*$ is $\mathrm{supp}(x^*) := \{i \in [n]\mid x_i^* > 0\}$. 
There is at most one fixed point per support.  
The set of fixed point supports for the TLN defined by $W, b$ is denoted $\FP(W, b)$. 
It is of course possible to determine compute the fixed points supports of a TLN by first solving for the fixed point for each of the $2^n$ linear systems, and then checking whether it falls in the appropriate chamber. 
However, we can gain more insight by proving theorems which link the structure of the weight matrix $W$ to the properties of  $\FP(W, b)$. 
These results describe the relationship between a network's connectivity and its activity. 

A fixed point is \emph{stable} if all trajectories which start near it stay near it. 
Early mathematical work on TLNs focused on characerizing the \emph{permitted sets} of TLNs \cite{hahnloser2000permitted, curto2016pattern}. 
For a fixed weight matrix $W$, these are the sets $\sigma\subseteq[n]$ such that for \emph{some} input $b$, $\sigma$ is the support of a stable fixed point in $\FP(W, b)$. 
In particular, the permitted sets of a symmetric, competitive TLN form a simplicial complex: a subset of a permitted set must be permitted \cite{hahnloser2000permitted,  curto2016pattern}.  
Later work characterizes the full set of fixed points which are present for a particular input \cite{curto2016pattern}. 
In particular, for symmetric $W$, the \emph{stable} fixed points must form an antichain. 
This implies that a symmetric threshold-linear network has at most $ \binom n {\lfloor n/2 \rfloor}$ stable fixed points. 
A set $\sigma$ is the support of a stable fixed point of a symmetric CTLN with graph $G$ if and only if $\sigma$ is a maximal clique in $G$ \cite{curto2016pattern}.  
In Figure \ref{fig:sym}, we give an example of a symmetric graph with four stable fixed points, each of which corresponds to a maximal clique.

\begin{figure}[ht!]
\begin{center}
\includegraphics[width = 6 in]{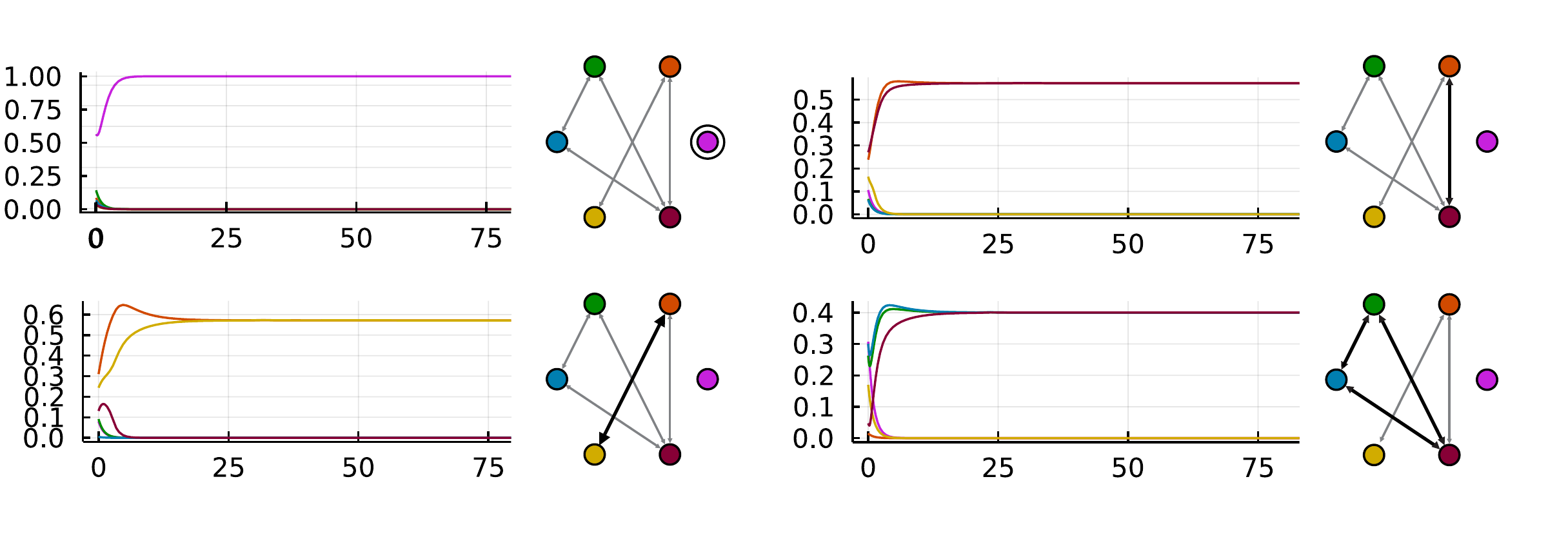}
\end{center}
\caption[Stable fixed points of a CTLN.]{The CTLN of a symmetric graph has a stable fixes point corresponding to each maximal clique. \label{fig:sym}} 
\end{figure}

Thus, by a bound in \cite{moon1965cliques}, a symmetric CTLN has at most $\mathrm O(3^{n/3})$ stable fixed points. 
Note that for large $n$, the gap between the upper and lower bounds becomes quite large. 
There are no known examples of symmetric, competitive TLNs which exceed the bound for CTLNs. 

Unstable fixed points of TLNs also play a role in shaping the dynamics of a TLN.   
Thus, a body of work has gone into characterizing $\FP(W, b)$ for competitive TLNs in general and CTLNs in particular \cite{curto2019fixed, curto2020combinatorial, curto2019robust}. 
For CTLNs, these results come in the form of \emph{graph rules}, which give constraints on $\FP(W(G, \varepsilon, \delta), \theta)$ based on properties of $G$ \cite{curto2019fixed}. 
In general, the parameters $\varepsilon$ and $\delta$ may affect  $\FP(W(G, \varepsilon, \delta), \theta)$, however, $G$ still provides strong constraints.

 %\emph{Graph rules in table? Don't want to write them all out, but do want to give some introduction.}

For general competitive TLNs, the relationship between $G_W$ and $\FP(W, \theta)$ is much looser. 
A graph $G$ is an \emph{robust motif} if for all $W, W'$ with $G_W = G_W' = G$, $\FP(W, \theta)$ is the same. With a finite family of exceptions, robust motifs fall into two very constrained infinite families, DAG1 and DAG2 \cite{curto2019robust}. A complete characterization of $\FP(W, \theta)$ for competitive TLNs on 3 neurons is given using oriented matroid theory in \cite{curto2020combinatorial}. 

\section{Dynamic attractors of TLNs}

Stable fixed points are a particular kind of attractor. Our goal in this chapter and the next two is to characterize the attractors of CTLNs and competitive TLNs more broadly. 
\begin{defn}
A set $A$ is an \emph{attracting set} of a dynamical system if 
\begin{enumerate}
\item 
it is invariant: if $x (0) \in A$, then $x(t) \in A$ for all $t \geq 0$. 
\item 
it has an open basin of attraction: there is an open set $U$ containing $A$ such that if $x (0)\in U$, then the distance from $x(t)$ to $A$ approaches zero as $t$ approaches infinity. 
\end{enumerate}
An \emph{attractor} is a \emph{minimal} attracting set. That is, $A$ is an attractor if it is an attracting set and there is no $B\subsetneq A$ that is also an attracting set. 
\end{defn}
There are many attractors beyond stable fixed points, termed \emph{dynamic attractors}. These include limit cycles and strange attractors.   
By the main result of \cite{hahnloser2000permitted}, if the weight matrix $W$ is symmetric and the matrix $(I-W)$ is copositive, then the TLN defined by $W$ has a set of stable fixed points, and all trajectories of the network must approach one of these states. In particular, competitive TLNs satisfy this copositivity condition. Thus all trajectories of symmetric competitive TLNs (such as CTLNs defined from undirected graphs) must approach stable fixed points. Thus, these networks do not have dynamic attractors. 

On the opposite side of the spectrum, \cite{morrison2016diversity} proves that if $G$ is an oriented graph with no sinks, the CTLN defined by $G$ has no stable fixed points. Because trajectories of these networks cannot approach stable fixed points and cannot diverge to infinity, these networks must have dynamic attractors.  In one case, the directed three-cycle, this dynamic attractor has been explicitly characterized as a stable limit cycle \cite{bel2021periodic}. Informal results connect the unstable fixed points of CTLNs to their dynamic attractors \cite{parmelee2022core, parmelee2021sequential}. Examples of dynamic attractors appear in Figures \ref{fig:two_cycles}, \ref{fig:two_cycles_same}, \ref{fig:baby_chaos}, and  \ref{fig:chaos_20}.  It is often difficult to predict how the activity of a CTLN will evolve over time: see Figure \ref{fig:transient} for examples of networks with early chaotic activity, which then fall into fixed points or limit cycles. 

\begin{figure}[ht!]
\includegraphics[width = 6 in]{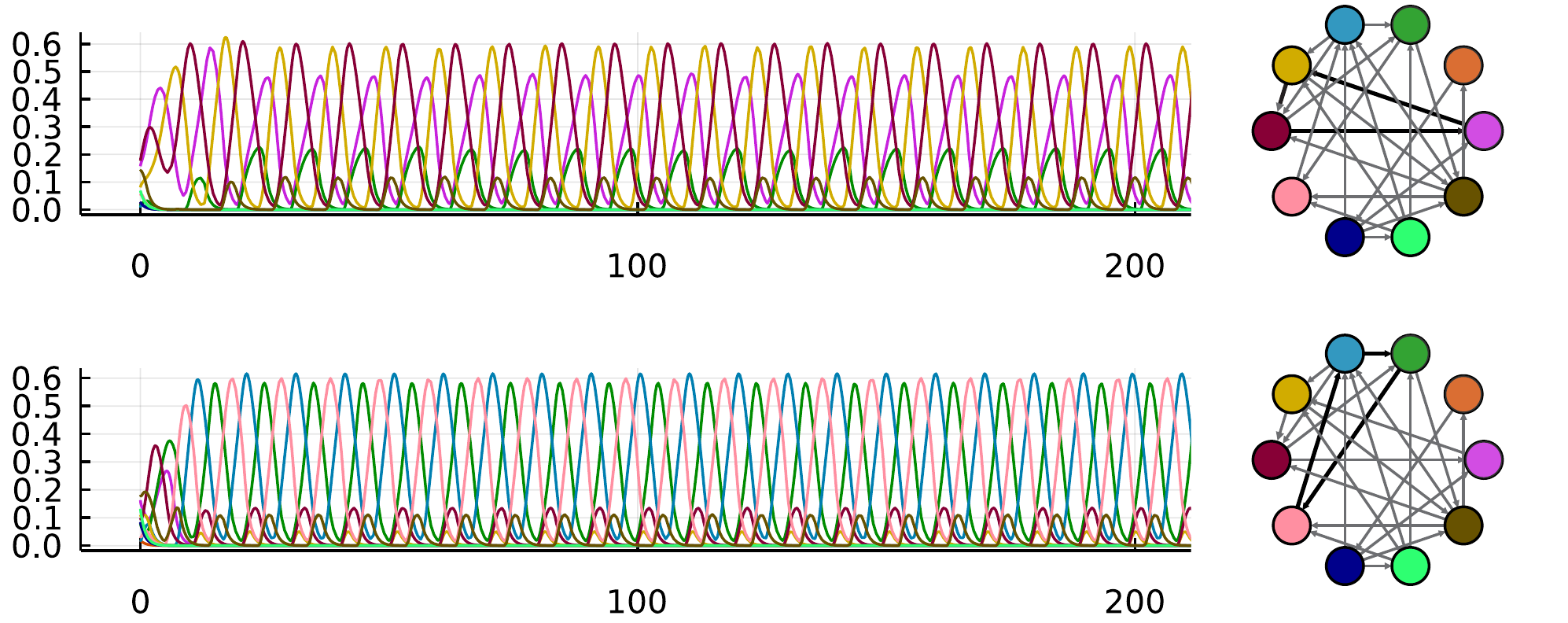}
\caption[Limit cycles of CTLNs]{A CTLN with at least two limit cycles, each corresponding to a different 3-cycle in the graph, which is shown with darker edges. The graph $G$ of this  $CTLN$  is oriented and has no sinks. Within this constraint, the edges were chosen randomly. For each $i, j$, the probability that there is either an edge from $i$ to $j$ or $j$ to $i$ is $0.5$. The  adjacency matrix of $G$ is
$ A = 
[0 0 0 0 0 1 0 1 0 0\mathrel{;} 0 0 0 0 0 0 0 0 0 1\mathrel{;} 0 0 0 1 0 1 0 0 1 0\mathrel{;} 0 0 0 0 0 0 1 1 1 1\mathrel{;} 1 0 0 1 0 0 0 0 1 0\mathrel{;} 0 0 0 1 1 0 0 0 0 1\mathrel{;} 0 0 1 0 0 0 0 0 1 1\mathrel{;} 0 1 0 0 0 0 0 0 0 0\mathrel{;} 0 0 0 0 0 0 0 1 0 0\mathrel{;} 0 0 1 0 1 0 0 1 0 0]
$. 
\label{fig:two_cycles}} 
\end{figure}

\begin{figure}[ht!]
\includegraphics[width = 6 in]{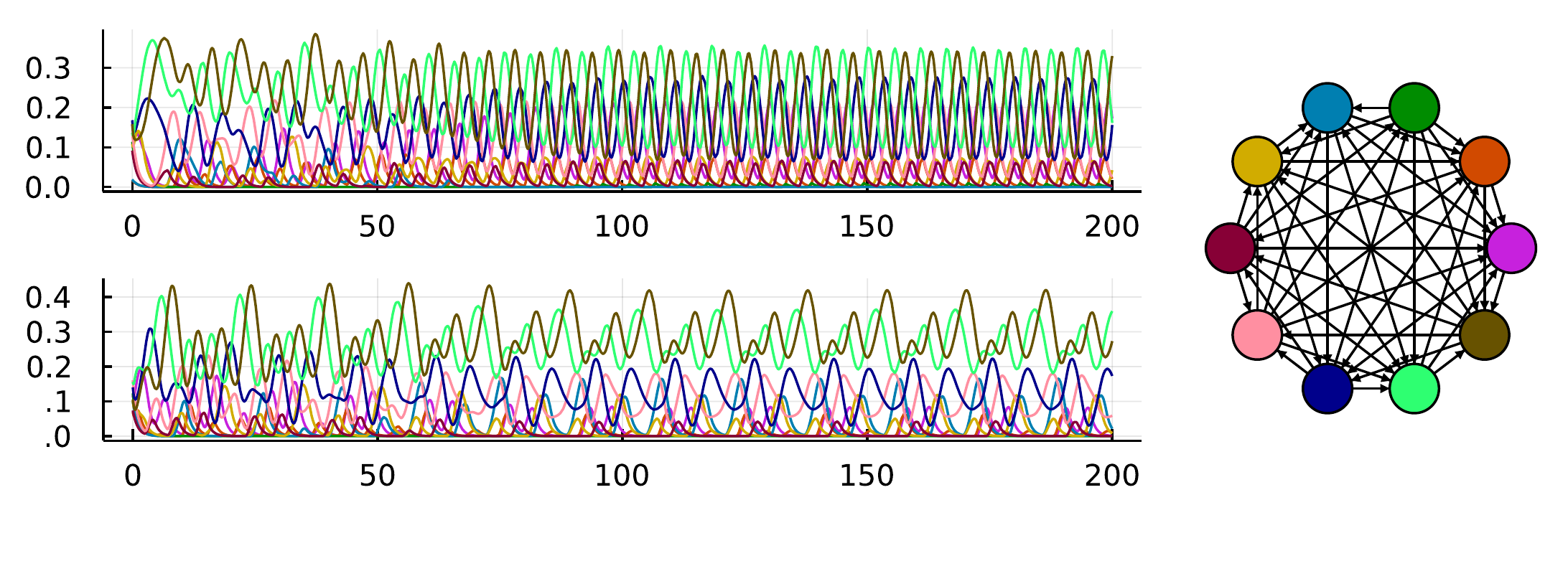}
\caption[Two limit cycles of CTLNs supported on the same neurons]{A CTLN with at least two limit cycles, each with the same set of active neurons, but with qualitatively different dynamics. 
The graph $G$ of this  $CTLN$  is oriented and has no sinks. Within this constraint, the edges were chosen randomly. 
For each $i, j$, there is either an edge from $i$ to $j$ or $j$. 
The  adjacency matrix of $G$ is 
$A = [0 1 1 1 0 1 1 0 1 0\mathrel{;} 0 0 1 1 1 0 1 0 1 0\mathrel{;} 0 0 0 0 0 1 1 1 0 0\mathrel{;} 0 0 1 0 1 1 1 0 0 1\mathrel{;} 1 0 1 0 0 1 1 0 0 1\mathrel{;} 0 1 0 0 0 0 0 1 0 1\mathrel{;} 0 0 0 0 0 1 0 1 0 1\mathrel{;} 1 1 0 1 1 0 0 0 0 1\mathrel{;} 0 0 1 1 1 1 1 1 0 0\mathrel{;} 1 1 1 0 0 0 0 0 1 0]$

\label{fig:two_cycles_same}}
\end{figure}

\begin{figure}[ht!]
\includegraphics[width = 6 in]{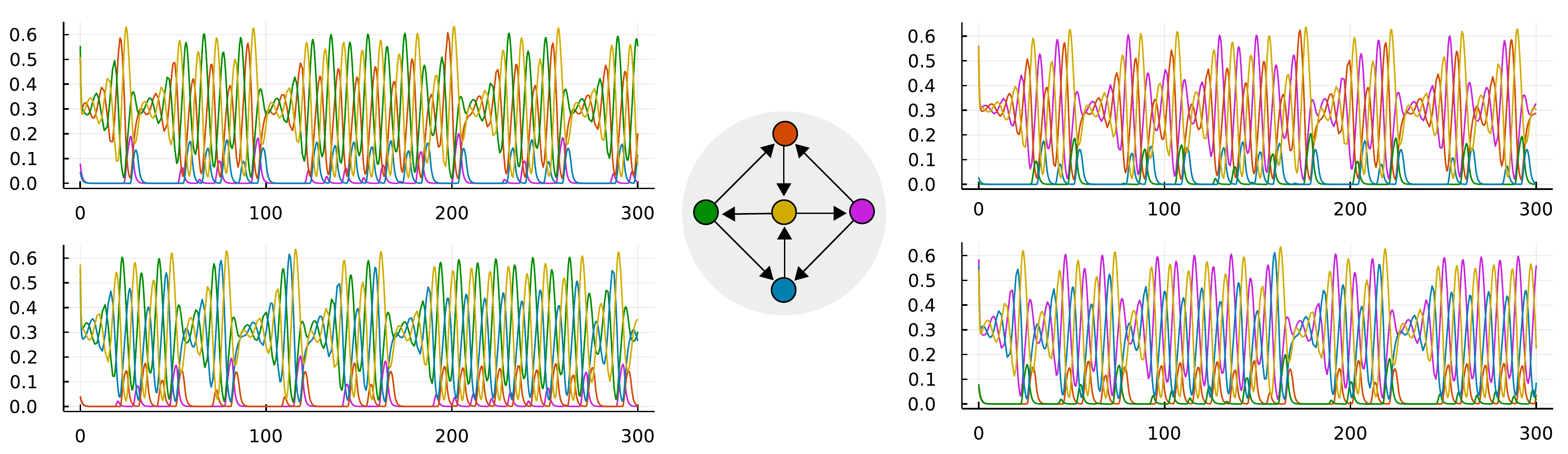}
\caption[Chaos in a small CTLN]{A five neuron CTLN with four chaotic attractors \label{fig:baby_chaos}}
\end{figure}

\begin{figure}[ht!]
\includegraphics[width = 6 in]{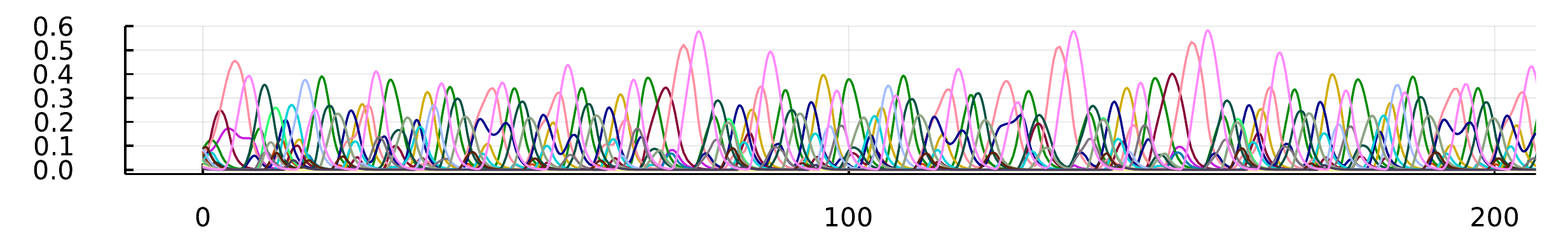}
\caption[Chaos in a larger CTLN]{Larger graphs can have more complicated chaotic dynamics, as pictured in this 20 neuron example.  The graph $G$ of this  $CTLN$  is oriented and has no sinks. Within this constraint, the edges were chosen randomly. For each $i \neq j$, the probability that there is either an edge from $i$ to $j$ or $j$ to $i$ is $0.5$.  \label{fig:chaos_20}}
\end{figure}

\begin{figure}[ht!]
\includegraphics[width = 6 in]{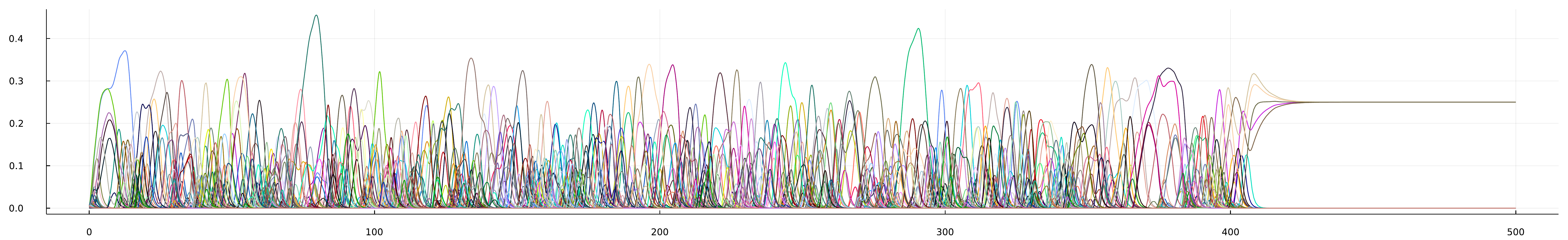}
\includegraphics[width = 6 in]{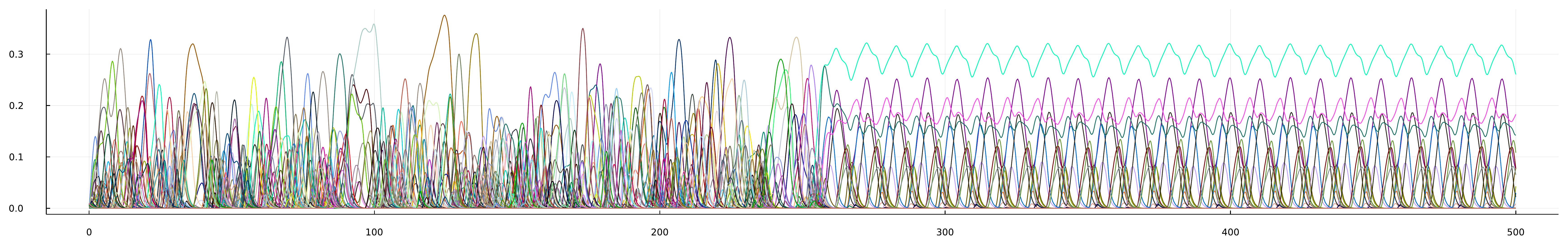}

\caption[Long transients in CTLNs]{It is often hard to predict how a network's activity will involve from the initial conditions. A network may have chaotic dynamics for a long time, before falling into a fixed point or a limit cycle, as illustrated by these two hundred-neuron examples. Illustrated here are two trajectories of the same CTLN. The graph is a random (Erd\H{o}s - Renyi) directed  graph with edge probability 0.5.   \label{fig:transient}}
\end{figure}

%The stability of a fixed point $x^*$ of a nonlinear system is determined by the eigenvalues of the Jacobian at $x^*$. In particular, $x^*$ is a stable fixed point if all of these eigenvalues have negative real part. Because TLN's are patchwork linear, the Jacobian at a point is just the matrix $A^\sigma$ of the linear system active in that chamber. Thus, the fixed point with support $\sigma$ of the TLN defined by $W, \theta$, if it exists, is stable if and only if all eigenvalues of the matrix $A^\sigma$ have negative real part.

% I don't know where this paragraph should go. Move it up? Move it down? 

\chapter{Nullcline Arrangements of Threshold-Linear Networks}
\label{chapter:nullclines}
\section{Introduction}\label{sec:nullclines}

In this chapter, we look at constraints on dynamic attractors of competitive TLNs provided by the arrangement of nullclines. In Section \ref{sec:nullclines}, we describe the nullclines of a TLN in terms of hyperplane arrangements. Next, Section \ref{sec:nullchamber_dynamics}, we describe how the chambers of the nullcline arrangement constrain the flow of trajectories of the system.  In particular, in Section \ref{sec:mixed_sign}, we prove our first constraint on trajectories of competitive TLNs in terms of the nullcline arrangement:
We prove that all trajectories of competitive TLNs approach and become trapped in a small region $\cA$ of phase space, corresponding to the set where at least one derivative is positive and at least one derivative is negative. This set is cut out by the hyperplanes defining the nullclines. 
\begin{ithm}
If $W,\theta$ defines a competitive TLN, $\cA$ is an attracting set, all trajectories approach it. 
\end{ithm}
This allows us to bound the total population activity of CTLNs: 
\begin{cor}
In CTLNs, \[\lim\inf_{t\to\infty} \sum_{i = 1}^n x_i \geq \frac{\theta}{1 + \delta}\] and  \[\lim\sup_{t\to\infty} \sum_{i = 1}^n x_i \leq \frac{\theta}{1 - \varepsilon}.\] 
\end{cor}
This supplements a previous bound on activity: if $x(0) \in \prod_{i = 1}^n[0, \theta]$, then $x(t) \in \prod_{i = 1}^n[0, \theta]$ for all $t >  0$ \cite{curto2019fixed}. Notice that neither bound implies the other. 

Next, in Section \ref{sec:nullcline_dynamics}, we explore in more generality how the nullcline chambers constrain trajectories. In Section \ref{sec:class2}, we use the nullcline arrangement to show that all trajectories of two-neuron competitive TLNs converge to a stable fixed point. On the other hand, we include an example from \cite{tang2005analysis} of a two-neuron non-competitive TLN which has a limit cycle. 

We begin by describing nullclines of threshold-linear networks in terms of hyperplane arrangements. 
A deeper account of the nullcline arrangement of a threshold-linear network, from the perspective of oriented matroid theory, can be found in \cite{curto2020combinatorial}.

The $i^{th}$ nullcline $\cN_i$ of a dynamical system is the set of all $x \in \R^n$ such that 
$\xdot_i = 0.$
Nullclines of linear systems are hyperplanes.
Because TLNs are threshold-linear, their nullclines are unions of pieces of hyperplanes.

\begin{defn}
Given a weight matrix $W\in \R^{n\times n}$, $\theta \in \R$, we define affine-linear functionals
\begin{align*}
e_i^*(x) &:= -x_i\\
h_i^*(x) &:= -x_i + \sum_{j=1}^n W_{ij}x_j + \theta.
\end{align*}
These define hyperplanes $E_1, \ldots, E_n, H_1, \ldots, H_n $. The hyperplane arrangement of the TLN defined by $W$ and $\theta$ is the set $$\cA(W, \theta) = \{E_1, \ldots, E_n,H_1, \ldots, H_n\},$$ with the orientation (i.e. choice of positive and negative half-spaces) of each hyperplane given by $e_i^*$, $h_i^*$.   
\end{defn}

We describe the nullclines in terms of the hyperplanes $E_1, \ldots, E_n, H_1, \ldots, H_n$. 

\begin{prop}
The $i^{th}$ nullcline of a TLN is 
$$	$$
\end{prop}

\begin{proof}
We show that $x_i' < 0$ if and only if  $\vec{x}\in H_i^-\cap E_i^-$, and $x_i' > 0$ if $x \in H_i^+\cup E_i^+$. If $x\in H_i^-\cap E_i^-$, then  $$\max\left\{-x_i, -x_i + \sum_{j=1}^n W_{ij}x_j + \theta\right\} = -x_i + \left[\sum_{j=1}^n W_{ij}x_j + \theta\right]_+ < 0,$$ so $x_i' < 0$.  Now, if $x \in H_i^+\cup E_i^+$, then 
$$\max\left\{-x_i, -x_i + \sum_{j=1}^n W_{ij}x_j + \theta\right\} = -x_i + \left[\sum_{j=1}^n W_{ij}x_j + \theta\right]_+ > 0,$$
so $x_i' > 0$. 
Now, by continuity, we must have $x_i' = 0$ on the boundary between $H_i^-\cap E_i^-$ and $H_i^+\cup E_i^+$, which is the set $\left(H_i\cap E_i^-\right)\cup \left(E_i \cap H_i^-\right)$. 
\end{proof}

\begin{ex}

Let $G$ be a the graph consisting of two vertices and a single directed edge from vertex 2 to vertex 1. Then the TLN arising from $G$ has the weight matrix 
$$W =  
\begin{pmatrix}
0& -1+ \varepsilon \\
-1-\delta & 0
\end{pmatrix}.$$

\label{ex:edge_null}
\begin{figure}[ht!]
\begin{center}
\includegraphics[width = 4 in]{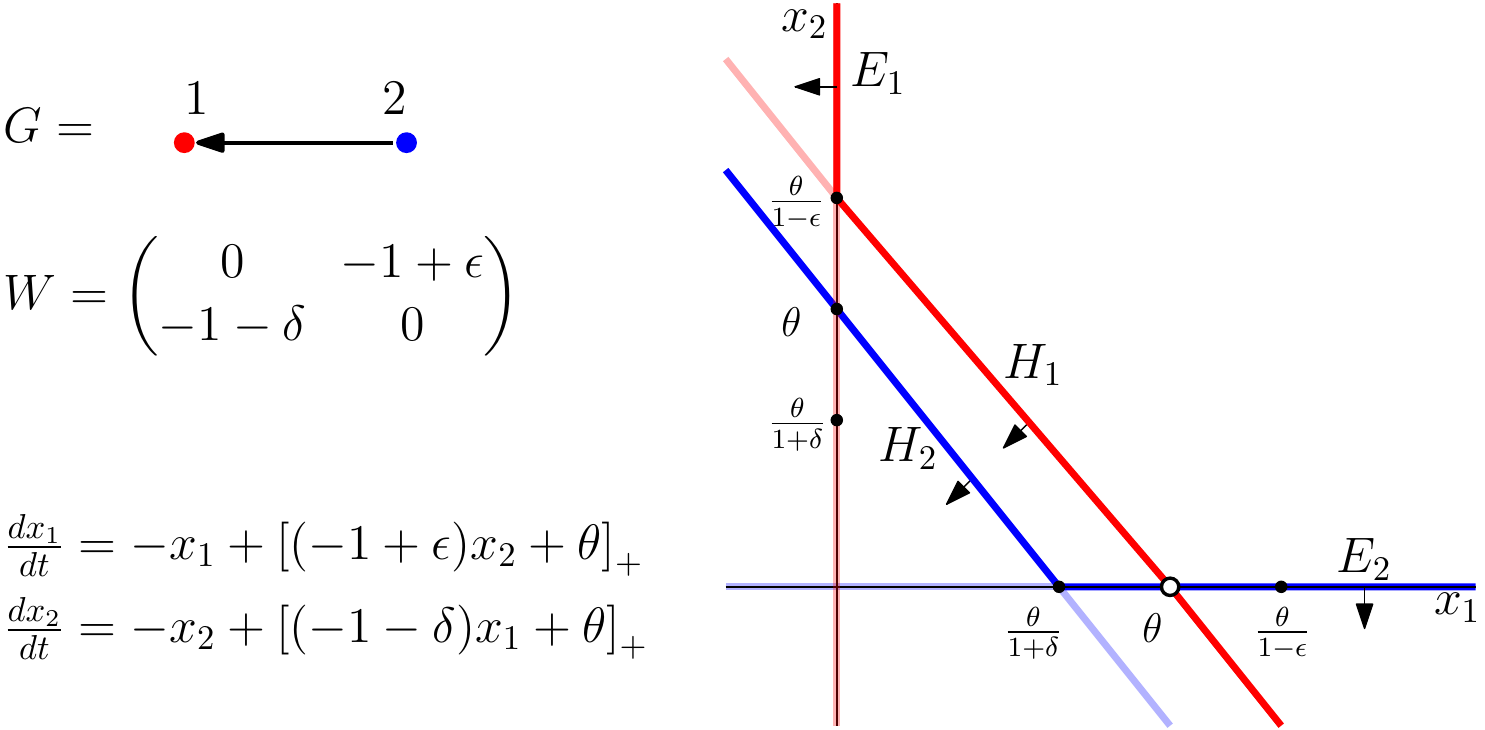}
\end{center}
\caption[Nullclines for the CTLN of a single directed edge]{Nullclines for the CTLN of a single directed edge. The $x_1$ nullclines, $E_1$ and $H_1$, are in red, $E_2$ and $H_2$ are in blue. Black arrows point towards the positive side of each hyperplane. \label{fig:nullcline_edge}}
\end{figure}

The hyperplanes $E_1, E_2, H_1, H_2$, shown in the Figure \ref{fig:nullcline_edge}, are given by the equations 
\begin{align*}
E_1 &= \{x | -x_1 = 0\}\\
E_2 &= \{x | -x_2 = 0\}\\
H_1 &= \{x | -x_1 + (-1+\varepsilon)x_2 + \theta = 0\}\\
H_2 &= \{x | -x_2 + (-1-\delta)x_2 + \theta = 0\}.
\end{align*}
The $x_1$ nullcline consists of the piece of $H_1$ on the negative side of $E_1$ together with the piece of $E_1$ on the negative side of $H_1$. Likewise, the $x_2$ nullcline consists of the piece of $H_2$ on the negative side of $E_2$ together with the piece of $E_2$ on the negative side of $H_2$. The pieces of $E_1, E_2, H_1$, and $H_2$ which act as nullclines are shown in darker colors in Figure \ref{fig:nullcline_edge}. The fixed point of this TLN is the point of intersection between these two nullclines, and appears as a white disk in the figure.

\end{ex}

Notice that in this example, the fixed point occurs at the intersection of both nullclines. This is true in general: the set of fixed points of a TLN is the intersection of all of the nullclines, since this is exactly the set where all derivatives are zero.

%%%%%%%%%%%%%%%%%%%%%%%%% Drawing Nullclines %%%%%%%%%%%%%%%%%%%%%%%%%%%%positive orthant. This allows us to describe the $H_i$ in terms of convex hulls, which is helpful when sketching the nullclines. 
The following proposition makes it easier to sketch nullclines, especially in three dimensions. 
\begin{prop}\label{prop:cvx_hull}
Let $W$ be a competitive TLN. Define $b_{ij} = \frac{-\theta}{W_{ij}}e_j$ for $i\neq j$, $ b_{ii} = \theta \,e_i$. Then
\begin{align*}
H_i \cap \R^n_{\geq 0} =  \conv\left(\{b_{ij}\mid j = 1, ..., n\}\right).
\end{align*}
\end{prop}

\begin{proof}
Recall that $H_i$ is the zero set of 
$$h_i^*(x) = -x_i + \sum_{j=1}^n W_{ij}x_j + \theta.$$
Along the $x_j$ axis generated by $e_j$, we can set the other variables to zero, so for $i\neq j$, we have
\begin{align*}
0 = h_i^*(x) =  W_{ij}x_j + \theta\\
x_j = \frac{-\theta}{W_{ij}}
\end{align*}
and for $i = j$, 
\begin{align*}
0 = h_i^*(x) =  -x_i  + \theta\\
x_i = \theta 
\end{align*}
Thus, for each $j$, $b_j$ is the intersection of $H_i$ with the $x_j$ axis.  Thus, these points and their affine span are contained in $H_i$. These points are affinely independent, and there are $n$ of them, so affine span is equal to $H_i$. That is, each point of $H_i$ can be written as 
$$x = \sum \lambda_j b_j$$
subject to the constraint that $\sum_{j = 1}^n \lambda_j = 1$. 

Now, if $x\in H_i\cap \R^n_+$, we can take each $\lambda_i \geq 0$. Thus, $x$ a convex combination of the points $\{b_j\mid j = 1, \ldots, n\}$. Thus, $H_i \cap \R^n_{\geq 0} =  \conv\left(\{ b_{ij}\mid j = 1, ..., n\}\right)$. 
\end{proof}

\begin{ex}
We use Proposition \ref{prop:cvx_hull} to sketch nullclines for the CTLN corresponding to the graph $G = (V = \{1, 2, 3\}, E = \{2\to 1, 2 \to 3\})$ in Figure \ref{fig:nullclines_3}. Notice that specializing to the case of CTLNs simplifies our descriptions of the points $ b_{ij}$. If $j \not\to i$, $b_{ij} = \frac{\theta}{1+ \delta}$ and if $j \to i$, $ b_{ij} = \frac{\theta}{1-\varepsilon}$. Also notice that $\frac{\theta}{1+ \delta} < \theta < \frac{\theta}{1-\varepsilon}$. Thus, the absence or presence of an edge $j \to i$ tells us how the point where $H_j$ intersects the $x_i$ relates to the point where $H_i$ intersects the $x_i$ axis. The fixed points lie at the points of intersection between all three nullclines. These are marked with white circles. 
\begin{figure}[ht!]
\includegraphics[width = 6 in]{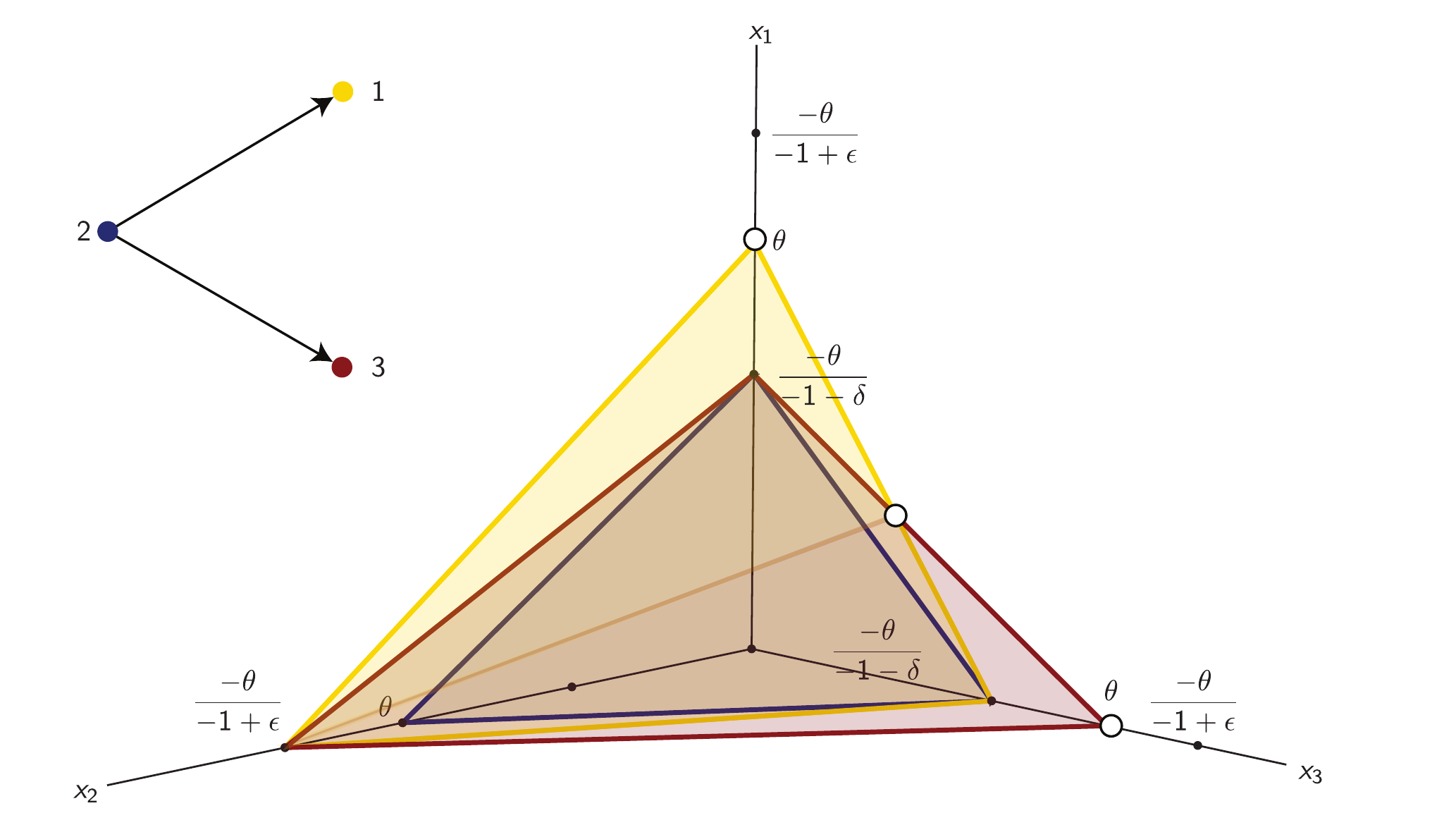}
\caption{Nullclines for a CTLN on three neurons.\label{fig:nullclines_3}}
\end{figure}
\end{ex}

\section{Nullcline chambers and dynamics}
\label{sec:nullchamber_dynamics}
In the previous section, we talked about how the hyperplane arrangement $\cA(W, \theta)$ controls the fixed points of a TLN. Now, we explore how the hyperplane arrangement controls other dynamic properties of the TLN. In particular, we discuss the \emph{chambers} of this arrangement. These chambers are defined by the sign of the derivative at each point of $\R^n$.  

\begin{defn}
A \emph{sign vector} $X$ of length $n$ is an element of $\{0, +, -\}^n$. 
\end{defn}

\begin{defn}
Given a TLN defined by $W, \theta$ and a sign vector $X$, we define the \emph{chamber} of $X$ to be the set 
\begin{align*}
A_{W, \theta}(X) = \{x \in \R^n \mid \sign(x_i') = X_i\}.
\end{align*}
\end{defn}
Because TLNs are threshold-linear, these chambers can be described in terms of the hyperplane arrangement. 
\begin{prop}
Let $X\in \{0,+, -\}^n$. 
\begin{align*}
A_{W, \theta}(X) = \bigcap_{X_i = +} (H_i ^+\cup E_i^+) \cap \bigcap_{X_j = -} (H_j^- \cap E_j^-) \cap \bigcap_{X_k = 0}\left( \left(H_i\cap E_i^-\right)\cup \left(E_i \cap H_i^-\right)\right).
\end{align*}
Within the strictly positive orthant, $\R^n_> = \{x\in \R^n\mid x_i > 0 \mbox{ for all } n \}$ this simplifies to 
\begin{align*}
A_{W, \theta}(X)\cap \R^n_> = \bigcap_{X_i = +} H_i ^+ \cap \bigcap_{X_j = -} H_j^-  \cap \bigcap_{X_k = 0} H_i
\end{align*}
\end{prop}
%
%\begin{ex}
%In Figure \ref{fig:edge_chambers}, we show the nullcline chambers for the graph with the single edge $1\to 2$. The four maximal nullcline chambers are the four possible sign vectors on $\{+, -\}^n$, $--, -+, +-, ++$. However, restricting to the positive orthant cuts off the $-+$, leaving us with three chambers $--, +-, ++$. 
%\begin{figure}[ht!]
%\includegraphics[width = 4 in]{ThresholdLinear/Figures/with_topes_labeled.pdf}
%\caption{Maximal nullcline chambers for the single directed edge. The positive side of each hyperplane is indicated with an arrow. Additionally, the the positive side of each nullcline is shaded. Each chamber is labeled with the appropriate sign vector. \label{fig:edge_chambers}}
%\end{figure}
%\end{ex}

%%%%%%%%%%%%%%%%%%%%%%%% MIXED SIGN CHAMBERS %%%%%%%%%%%%%%%%%%%%%%%%%%%%%%%

\subsection{The mixed-sign chambers are attracting }\label{sec:mixed_sign}
Now, we show that for competitive TLNs the chambers of $--\ldots-$ and $++\ldots+$ are repelling, and that all trajectories of the TLN eventually approach the union of the mixed-sign chambers. This constrains trajectories to a small subset of the positive orthant.

\begin{defn}
We define the \emph{mixed sign region} $\cA$ as the closure of the set 
$$\{x\in \R^{n}_{\geq 0} \mid \mbox{ there exist } \, i, j \mbox{ such that } \sign(\xdot_i) = - \mbox{ and } \sign(\xdot_j) = +\}.$$
That is, $\cA$ is the region where not all derivatives have the same sign. 
Similarly, we define the two pieces of the complement of $\cA$ as
\begin{align*}
\cA^+ &= \{x\in \R^{n}_{\geq 0} \mid  \sign(\xdot_i) = +  \mbox{ for all } i\}\\
\cA^- &= \{x\in \R^{n}_{\geq 0} \mid  \sign(\xdot_i) = -  \mbox{ for all } i\}.
\end{align*}
\end{defn}
Notice that $\cA$ is bounded by the nullclines. In particular, within the positive orthant, it is bounded by the hyperplanes $H_1, \ldots, H_n$. Since the fixed points are the points where all nullclines intersect, they are always on the boundary of the mixed-sign chamber. A two-dimensional example is illustrated in Figure \ref{fig:mixed_sign}. 

\begin{figure}
\begin{center}
\includegraphics[width = 5.5 in]{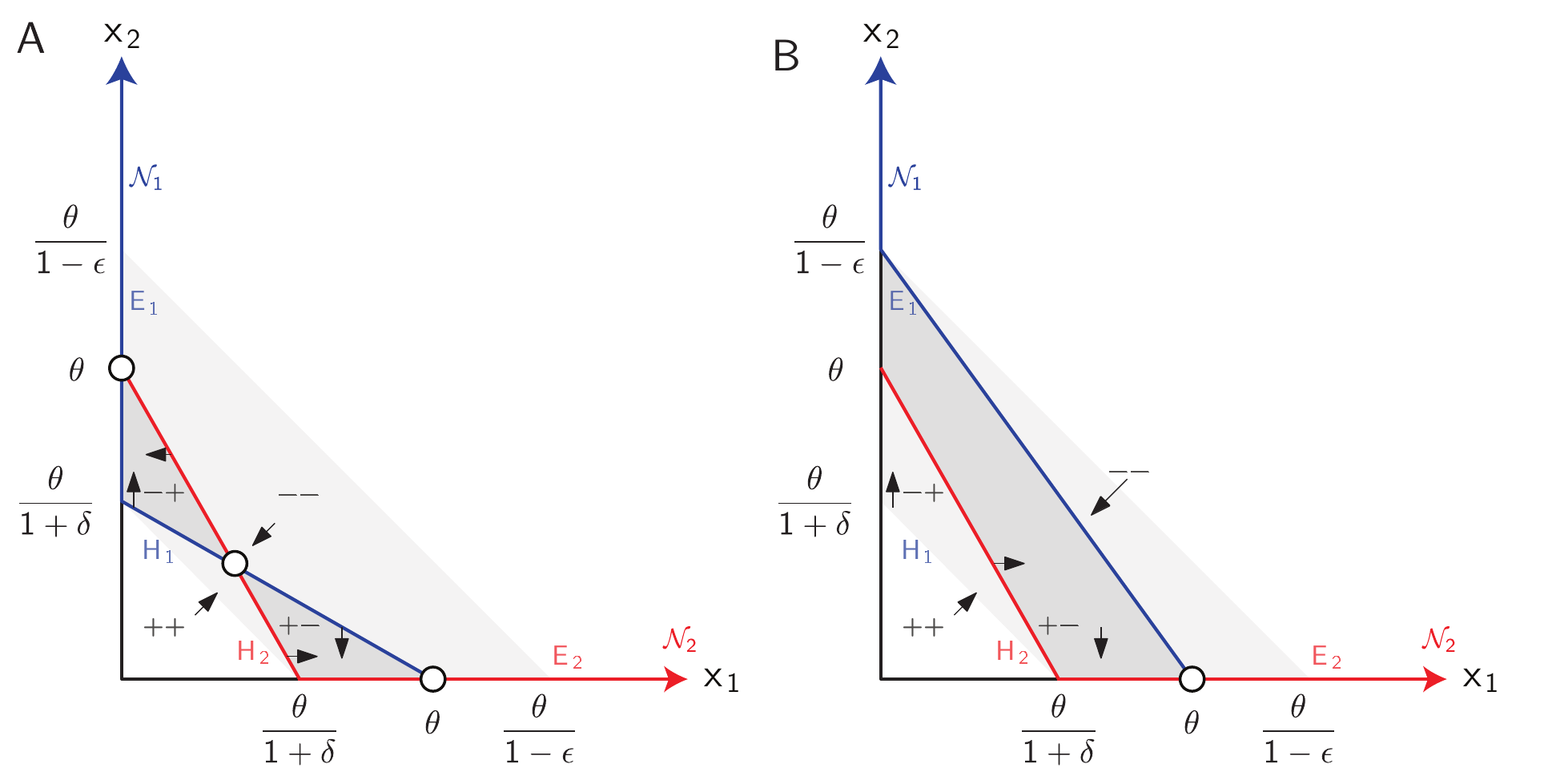}
\end{center}
\caption[The mixed sign region $\cA$ for a two CTLNs. ]{The mixed sign region $\cA$ for a two CTLNs,  an independent set (A) and the edge directed from 2 to 1 (B). The mixed sign region is shaded in dark gray. In light gray, we show the bounds from Corollary \ref{cor:total_pop}. Fixed points are shown as white circles--note that they all fall on the boundary of $\cA$, and the interior of the light gray region.  \label{fig:mixed_sign} }

\end{figure}

\begin{thm}\label{thm:mixed_sign}
For any competitive TLN, the mixed sign region $\cA$ is a globally attracting set. In particular, trajectories which enter $\cA$ cannot leave. 
\end{thm}

\begin{proof}
First, note that within $\cA^+$, all derivatives $\xdot_i$ are positive. Thus, the trajectory must approach the boundary between $\cA^+$ and $\cA$. Likewise, within $\cA^-$, all derivatives $\xdot_i$ are negative, thus  the trajectory must approach the boundary between $\cA^-$ and $\cA$. 

Next, we show that if a trajectory is on the boundary between either $\cA^+$ and $\cA$ or $\cA^-$ and $\cA$, it must enter $\cA$. We first consider the boundary between $\cA^+$ and $\cA$. Since $W$ is the weight matrix of a competitive TLN, all entries of the normal vector $ h_i$ to the hyperplane $H_i$ are negative for all $i$. 
On the boundary between $\cA^+$ and $\cA$,  $\xdot \geq 0$ in all coordinates. Thus, $ \xdot\cdot  h_i \leq 0$, with equality only achieved at fixed points. Thus, unless the trajectory is at a fixed point, it must cross from $\cA^+$ into $\cA$. Likewise, if the trajectory is on the boundary between $\cA^-$ and $\cA$, $\xdot \leq 0$ in all coordinates.   Thus, $ \xdot\cdot  h_i \geq 0$, with equality only achieved at fixed points. Thus, the trajectory must cross from $\cA^+$ into $\cA$. 
\end{proof}

We can use this result to derive bounds on the total population activity of competitive TLNs in general, and CTLNs in particular.  

\begin{cor}\label{cor:total_pop}
Let $m = \min_{i, j}\{(I - W_{ij})\}$ and $M = \max_{i, j}\{(I - W_{ij})\}$. 
In competitive TLNs, there is some time $T$ such that 
\[\frac{\theta}{M} < \sum_{i = 1}^n x_i < \frac{\theta}{m}\] 
for $t \geq T$. 
In CTLNs this result specializes to

\[\frac{\theta}{1 +\delta} < \sum_{i = 1}^n x_i < \frac{\theta}{1 -\varepsilon}.\]

\end{cor}

\begin{proof}
We show that, within the positive orthant, $\cA$ is contained between the hyperplanes $H_\ell = \{ x\mid \sum_{i = 1}^n x_i -\frac{\theta}{M} = 0\}$, oriented away from the origin, and $H_u = \{x\mid \sum_{i = 1}^n x_i - \frac{\theta}{m} =0\}$, oriented towards the origin. 
Further, we show that these hyperplanes contain no fixed points of the TLN. 
By Proposition \ref{prop:cvx_hull}, for each $i$, the hyperplane $H_i$ is the convex hull of the points $\mathbf b_{ij} = \frac{-\theta}{W_{ij}}\mathbf{e_j}$ for $i\neq j$, $\mathbf b_{ii} = \theta \,\mathbf e_i$. Each point $ b_{ij}$ is on the positive side of $H_\ell $ and on the negative side of the hyperplane defined by $H_u$. Thus, since all trajectories must approach $\cA$, they must also approach the region between $H_\ell$ and $H_u$. 

Now, notice that for each $i$, the intersection between $H_\ell$  and $H_i$ is the convex hull of the set of points $x_k = \frac{\theta}{M}$, $x_j = 0$ for $j\neq k$ which lie on the $H_i$. 
Notice that none of these vertices are fixed points, since singleton fixed points have $x_k = \theta$.
Further, note that there are no fixed points in the relative interior, since any nullcline which touches the relative interior of this convex hull must touch all vertices. 
The analogous statement holds for $H_u$ and points $x_k = \frac{\theta}{m}$, $x_j = 0$. 

Thus, since there are no fixed points on $H_\ell$ or $H_u$, on  the closure of $H_\ell^-$,  $\sum_{i=1}^n x_i$ is strictly increasing, and on the closure of $H_u^+$,   $\sum_{i=1}^n x_i$  is strictly decreasing. 
Therefore, trajectories cannot remain in  the closure of $H_u^+$ or $H_\ell^-$ forever, and must cross into the region  $H_u^+\cap H_\ell^-$. 
Thus eventually \[\frac{\theta}{M} < \sum_{i = 1}^n x_i < \frac{\theta}{m}.\] 

 %Thus, 
%\[\lim\inf_{t\to\infty} \sum_{i = 1}^n x_i \geq \frac{\theta}{M}\] and  \[\lim\sup_{t\to\infty} \sum_{i = 1}^n x_i \leq \frac{\theta}{m}.\] 
%If $W$ is the matrix of a CTLN, then $m = 1-\varepsilon$ and $M = 1+ \delta$, thus 
%\[\lim\inf_{t\to\infty} \sum_{i = 1}^n x_i \geq \frac{\theta}{1 +\delta}\] and  \[\lim\sup_{t\to\infty} \sum_{i = 1}^n x_i \leq \frac{\theta}{1 -\varepsilon}.\] 

\end{proof}

\begin{figure}[ht!]
\includegraphics[width = 6 in]{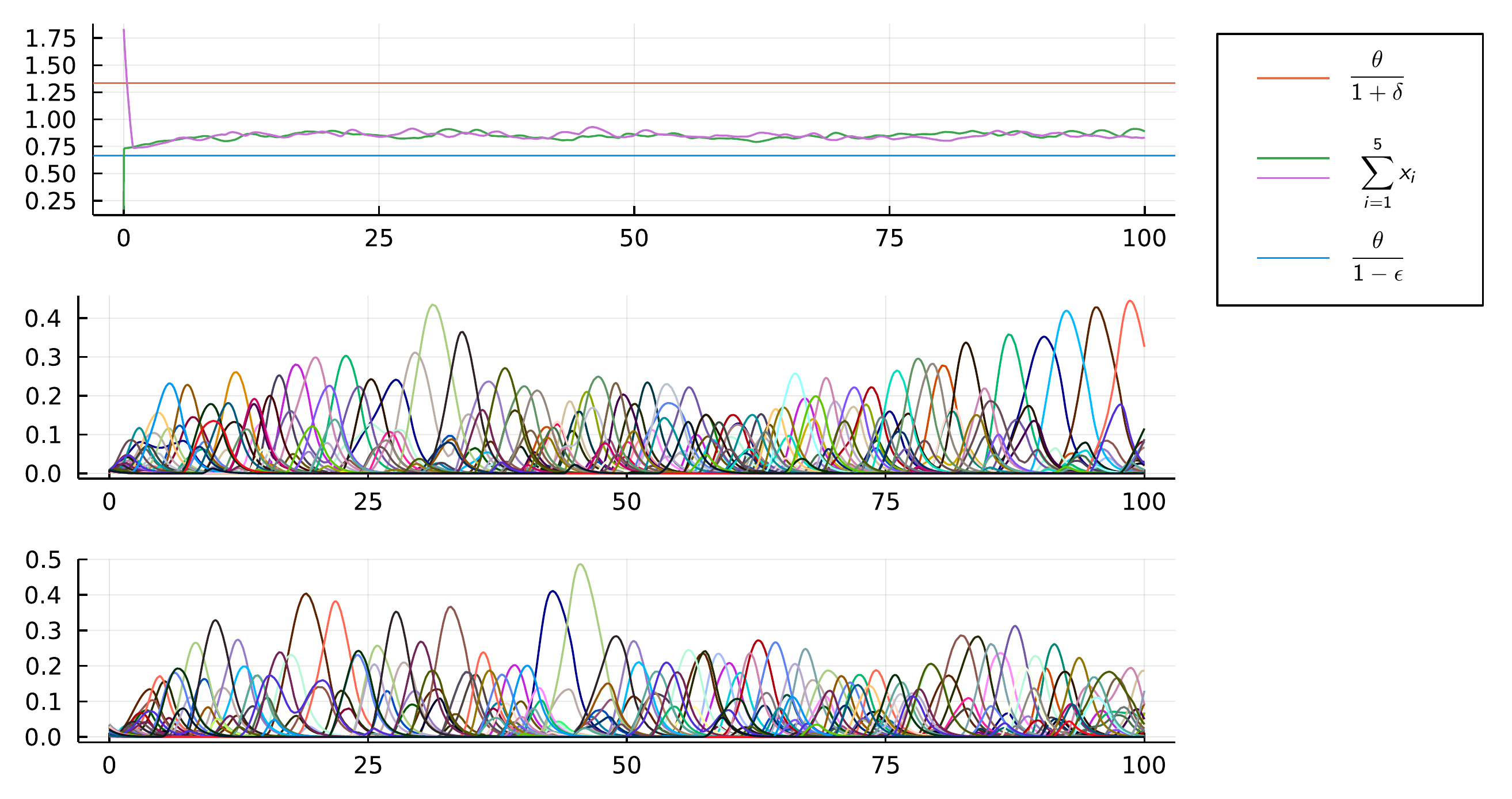}
\caption[Two trajectories of a CTLN with similar total population activity.]{Two trajectories of a CTLN. Note that total population activity  $\sum_{i = 1}^n x_i$ along both trajectories starts outside the region $\frac{\theta}{1 +\delta} < \sum_{i = 1}^n x_i < \frac{\theta}{1 -\varepsilon}$, but very rapidly enters the region.\label{fig:total_pop_dynamic}}
\end{figure}

\begin{figure}[ht!]
\includegraphics[width = 6 in]{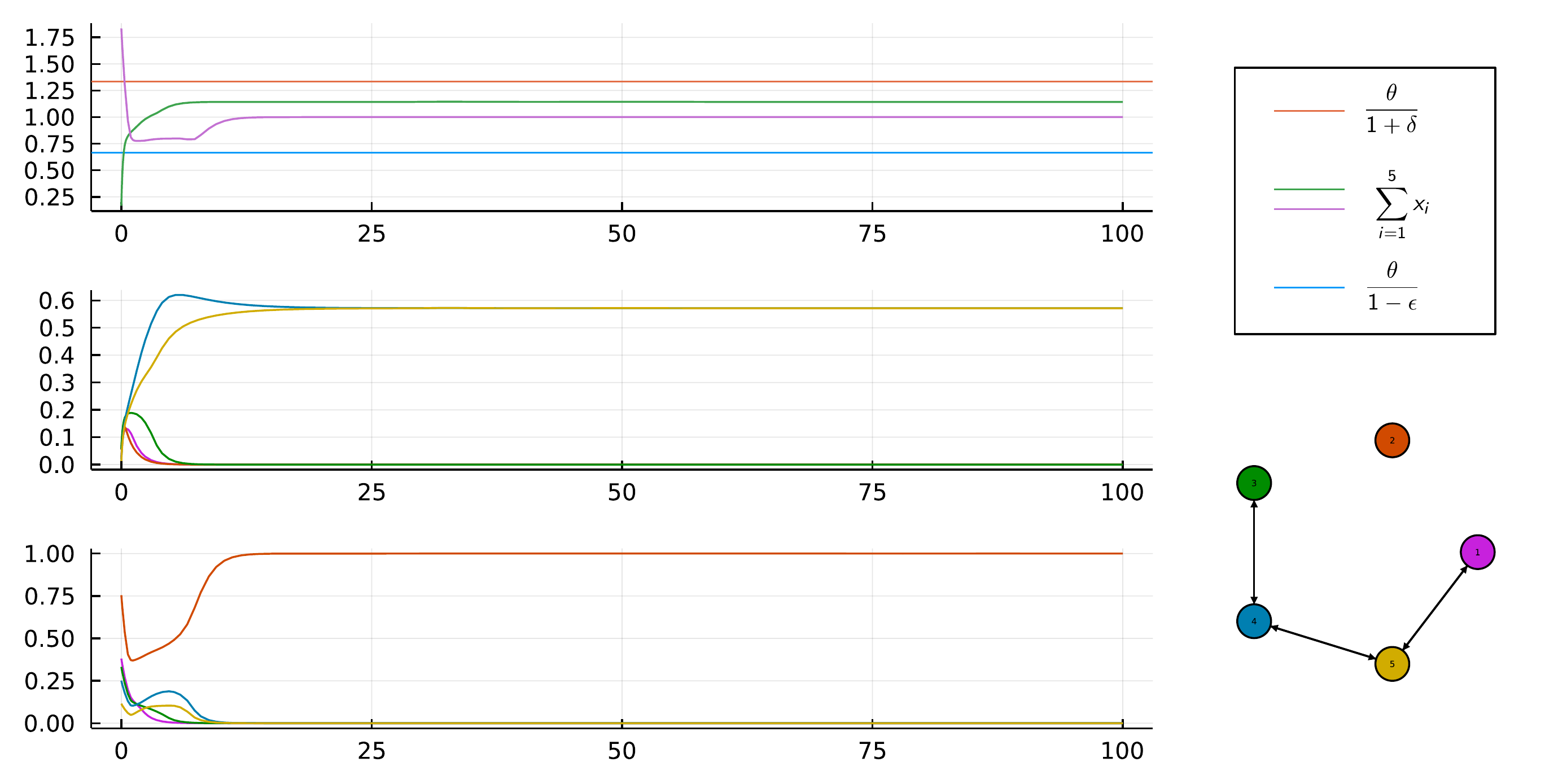}
\caption[Fixed points satisfy total population activity bounds.]{Two trajectories of a symmetric CTLN approach fixed points satisfying $\frac{\theta}{1 +\delta} < \sum_{i = 1}^n x_i < \frac{\theta}{1 -\varepsilon}$. \label{fig:total_pop_fp}}
\end{figure}

The total population activity $\sum_{i = 1}^n x_i$ often does not vary as much as these bounds allow. For instance, see Figures \ref{fig:total_pop_dynamic} and \ref{fig:total_pop_fp}.

%%%%%%%%%%%%%%%%%%%%%%%%%%% NULLCLINE CHAMBERS %%%%%%%%%%%%%%%%%%%%%%%%%%%%%%%%%%%

\subsection{How nullcline chambers shape dynamics}\label{sec:nullcline_dynamics}
\begin{defn}
Define a directed graph $T_{W, \theta} = (V, E)$ whose vertices correspond to nullcline chambers of the TLN defined by $W, \theta$ and whose edges correspond to possible transitions between chambers. That is, 
$V = \{ X\in \{+, -\}^n\mid A_{W, \theta, X} \neq \emptyset\}$, and the edge $(X\to Y) \in E$ if there is a point $x$ on the boundary between $X$ and $Y$ such that ${\xdot}$ points from $X$ to $Y$. 
\end{defn}

We now prove that trajectories follow the edges of $T_{W, \theta}$. 
\begin{lem}\label{lem:folloW_edges}
Any trajectory of the TLN defined by $W, \theta$ must follow a directed path in $T_{W, \theta}$. 
\end{lem}

\begin{proof}
Suppose a trajectory crosses directly from the chamber $A_{W, \theta}(X)$ to the chamber $A_{W, \theta}(Y)$. Then at some time $t$, $x(t)$ is on the boundary between $A_{W, \theta}(X)$ and $A_{W, \theta}(Y)$. The direction it is going at time $t$ is given by $\xdot(t)$. Thus, in order to cross from $A_{W, \theta}(X)$ to $A_{W, \theta}(Y)$, $\xdot(t)$ must point from $A_{W, \theta}(X)$ to $A_{W, \theta}(Y)$.
\end{proof}

\begin{ex}
We consider again the directed edge.The nullcline arrangement for this graph appears in Figure \ref{fig:trans_graph}. We have $T_{W, \theta} = (V, E)$ with $V = \{--, +-, ++\}$ and $E = \{(-- \to +-), (++ \to +-)$\}.  To see this, we first note that edges must be between adjacent chambers. The edge between $--$ and $+-$ corresponds to crossing the $x_1$ nullcline $H_1$. On this nullcline, we have $x_1' = 0$ and $x_2' < 0$. Thus, we are going straight down. Therefore, $\xdot$ points from the chamber $--$ to the chamber $+-$. The edge between $++$ and $+-$ corresponds to crossing the $x_2$ nullcline. On this nullcline, we have $x_2' = 0$, $x_1' > 0$. Thus, we are going right. Therefore, whenever we cross this nullcline, we must cross from the chamber $++$ to the chamber $+-$. 

\begin{figure}[ht!]
\begin{center}
\includegraphics[width = 4 in]{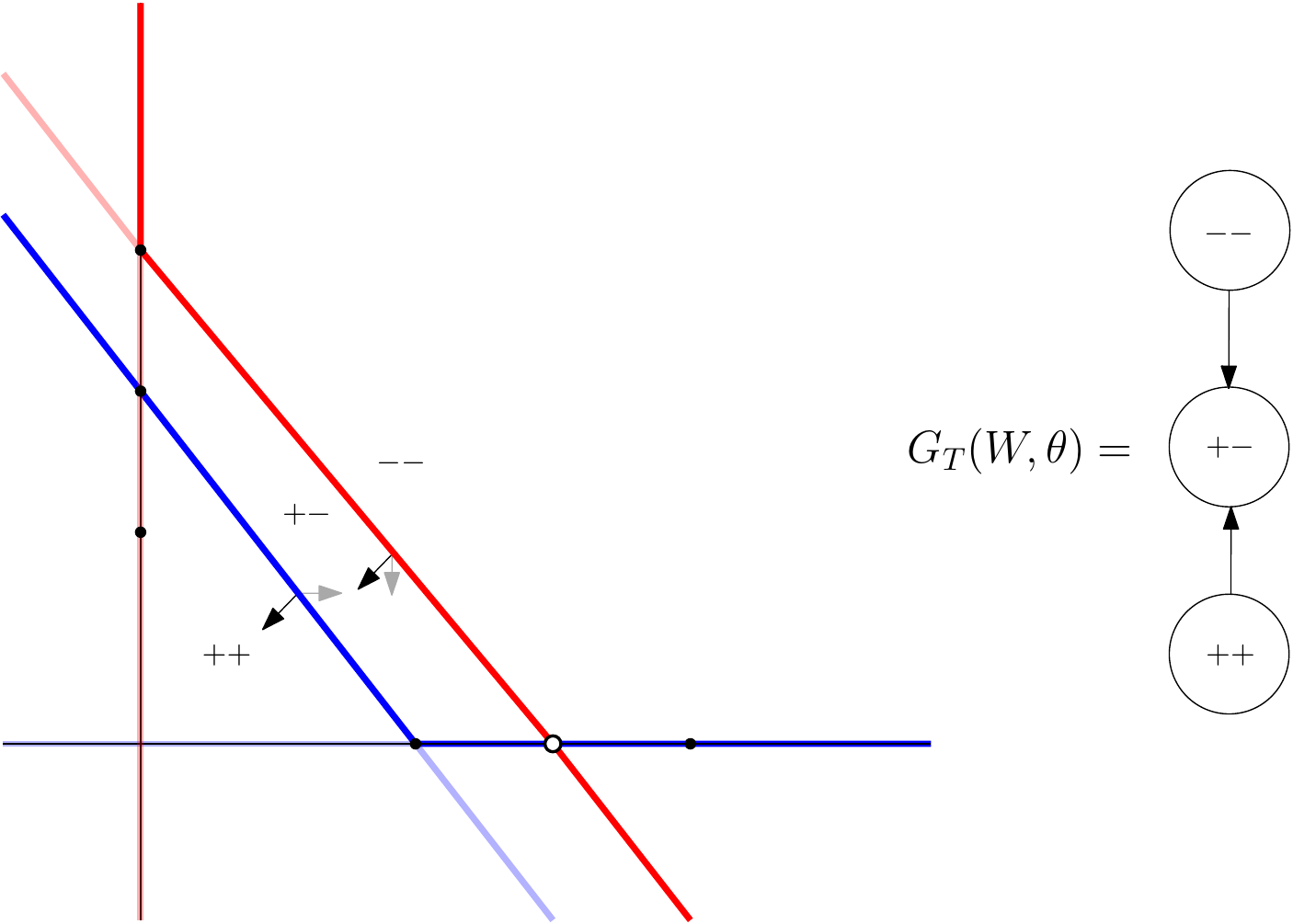}
\end{center}
\caption[The transition graph $T_{W, \theta}$ for the single directed edge.]{The transition graph $T_{W, \theta}$ for the single directed edge. The gray arrows are the vectors $\xdot$ on the nullclines.  \label{fig:trans_graph}}
\end{figure}
\end{ex}

Not only must transitions between chambers follow edges of $T_{W, \theta}$, dynamic attractors must follow directed cycles in  $T_{W, \theta}$. To prove this, it is sufficient to prove that single chambers of  $T_{W, \theta}$ to not contain dynamic attractors.

\begin{lem}\label{lem:single_chamber}
If a trajectory is contained on one side of the $i^{th}$ nullcline, then the value of $x_i(t)$ approaches a limit.  If a trajectory of a TLN is contained within a nullcline single chamber, then it approaches a fixed point. 
\end{lem}

\begin{proof}
If a trajectory is contained on one side of the $i^{th}$ nullcline, then the sign of $x_i'(t)$ along this trajectory is fixed. Thus, the function $x_i(t)$ is bounded and is either monotonically increasing or monotonically decreasing. Thus, $x_i(t)$ approaches a limit. 

Within a single chamber of $A_{W, \theta, X}$,  $x_i(t)$ approaches a limit for all $i = 1, ..., n$, thus the trajectory approaches a fixed point. 
\end{proof}

\begin{prop}\label{prop:directed_cycle} Any dynamic attractor of the TLN defined by $W, \theta$ must follow a directed cycle of $T_{W, \theta}$. 
\end{prop}

\begin{proof}
By Lemma \ref{lem:folloW_edges}, any trajectory must follow edges of $T_{W, \theta}$. By Lemma \ref{lem:single_chamber}, a dynamic attractor can never become trapped in a single chamber. Thus, it must transition between chambers infinitely many times. Since $T_{W, \theta}$ has only finitely many vertices of, the trajectory must thus return to a chamber is has visited before, thus following a directed cycle.  
\end{proof}

Proposition \ref{prop:directed_cycle} allows us to prove that a given TLN does not have a dynamic attractor. 

%%%%%%%%%%%%%%%%%%%%%%%%%%%%%%%%%% n = 2 %%%%%%%%%%%%%%%%%%%%%%%%%%%%%%%%%%%%%%%

\subsection{$N = 2$ competitive TLN classification}
\label{sec:class2}
\begin{figure}[ht!]
\begin{center}
\includegraphics[width = 4 in]{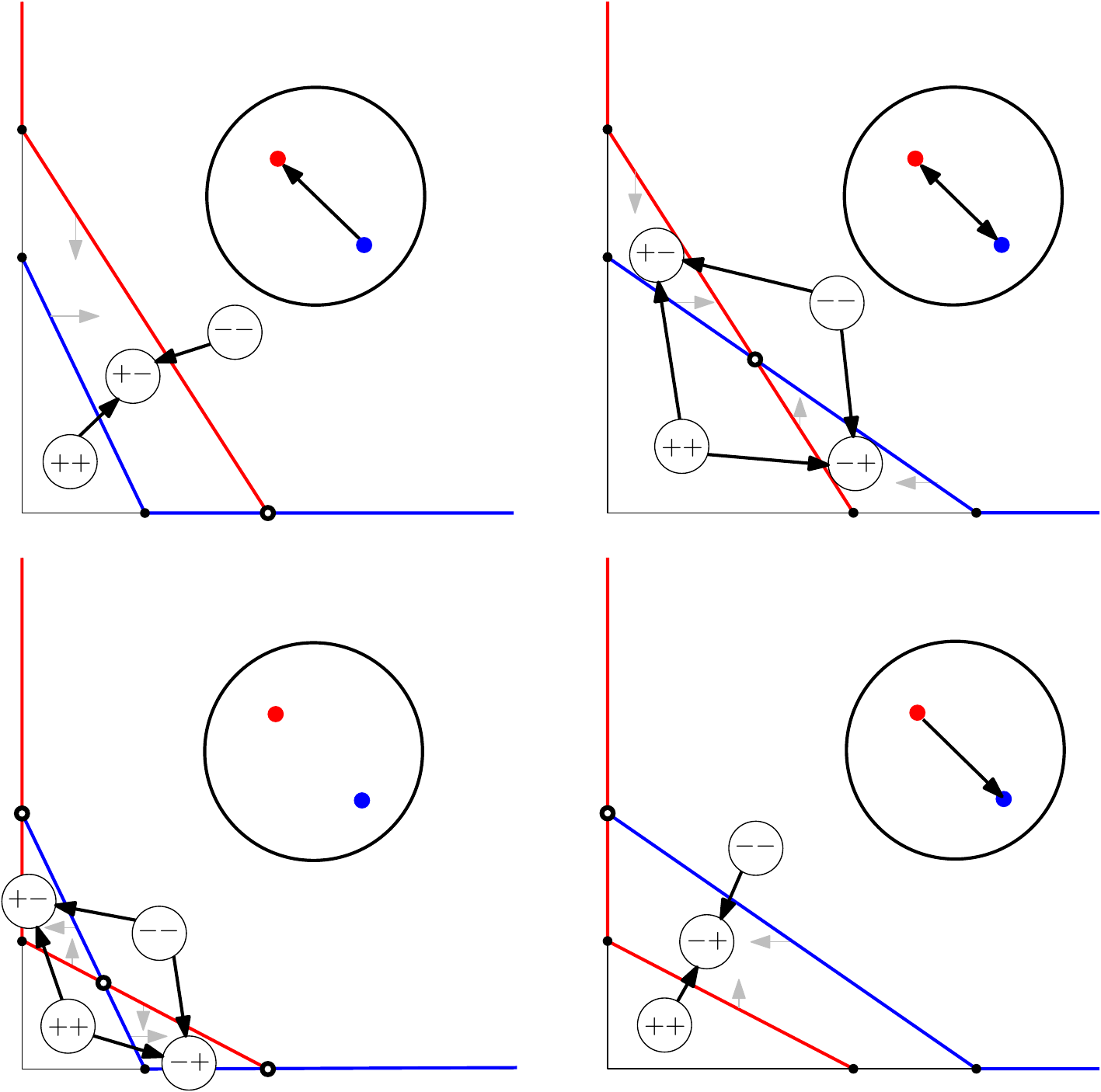}
\end{center}
\caption[Nullcline arrangements for all competitive TLNs on two neurons.]{In each subplot, we have the nullcline arrangement and transition graph $T_{W, \theta}$ corresponding to the connectivity graph in the upper right corner. Gray arrows on nullclines give the direction of the derivative on this piece of the nullclines. \label{fig:n_2_class}}
\end{figure}

The transition graph $T_{W, \theta}$ allows us to easily prove that competitive TLNs with two neurons do not have dynamic attractors.

\begin{prop}\label{prop:n=2}
All trajectories of competitive TLNs with two neurons converge to stable fixed points. 
\end{prop}

\begin{proof}
The transition graph $T_{W, \theta}$ is completely determined by the connectivity graph $G_{W, \theta}$. In particular, if $W_{12} \leq -1$, the intersection of $H_1$ with the $x_2$ axis is closer to the origin than the intersection of $H_2$ with the $x_2$ axis, and if $W_{12} \leq -1$ the intersection of $H_1$ with the $x_2$ axis is further from the origin than the intersection of $H_2$ with the $x_2$ axis. Similarly, the relationship between $W_{21}$ and -1 determines the relative positions of the intersections of $H_1$ and $H_2$ with the $x_1$ axis. These relationships are sufficient to completely determine all combinatorial properties of the hyperplane arrangement $E_1, E_2, H_1, H_2$, and thus the graph $T_{W, \theta}$.  All possible arrangements are shown in Figure \ref{fig:n_2_class}. We notice that none of these graphs have directed cycles. Thus, by Proposition \ref{prop:directed_cycle}, none of these graphs have dynamic attractors.  
\end{proof}
%%%%%%%%%%%%%%%%%%%%%%%%%% n = 2, noncompetitive %%%%%%%%%%%%%%%%%%%%%%%%%%%%%

\subsection{Non-competitive TLNs for $n = 2$ }

While we have shown that competitive TLN's with 2 neurons, this result does not hold for general TLNs with two neurons. In particular, we consider the TLN defined by 
\begin{align*}
W = \begin{pmatrix}
2 & -1 \\ 2 & 0
\end{pmatrix} && \theta = \begin{pmatrix}
2 \\ 1 
\end{pmatrix},
\end{align*}
which is introduced in \cite{tang2005analysis}. This TLN does, in fact, have a limit cycle. We can see that this limit cycle corresponds to a directed cycle in its transition graph, as illustrated in Figure \ref{fig:n=2cycle}. 

\begin{figure}[ht!]
\begin{center}
\includegraphics[width = 3 in]{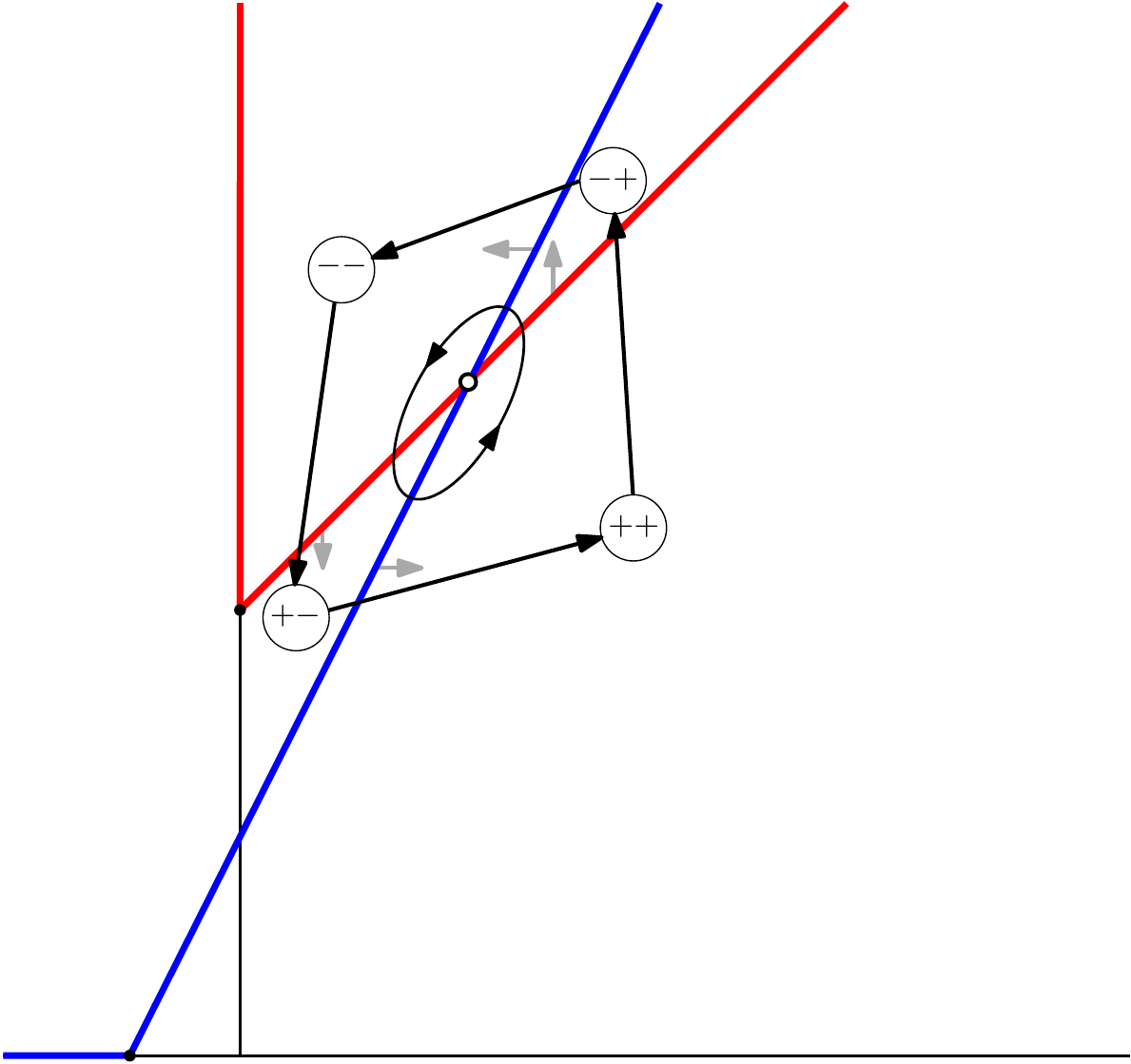}
\end{center}
\caption{A two neuron, non competitive TLN with a limit cycle.\label{fig:n=2cycle} }
\end{figure}

Additionally, notice that this limit cycle is not confined to the union of the two mixed-sign chambers, $-+$ and $+-$. Instead, the limit cycle visits all four nullcline chambers. Thus, we see that the assumption that networks are competitive is necessary for Theorem \ref{thm:mixed_sign}. 

%% file: ThresholdLinear/ThresholdLinearResults.tex
%auto-ignore 

% !TEX root = ../YourName-Dissertation.tex

\chapter{Constraining Dynamic Attractors of Threshold-Linear Networks } \label{chapter:TLNs2}
\section{Introduction}

In this section, we prove our main results about dynamic attractors of threshold-linear networks. 
In particular,  we aim to characterize which CTLNs have dynamic attractors and which do not, and to constrain the dynamic attractors of CTLNs. 
Our proof strategies resemble the one used to prove Theorem \ref{thm:mixed_sign} in the previous chapter, and many of our results involve showing that we can force trajectories of TLNs into smaller regions of space.

We organize our work around the following conjecture: any CTLN with a dynamic attractor must have a proper directed cycle, as defined below. 
\begin{defn}
A \emph{proper directed cycle} of a directed graph is a directed cycle which has at least one edge which is not bidirectional. 
\end{defn}
\begin{conj}Let $G$ be a graph with no proper directed cycle. Then no CTLN with graph $G$ has a dynamic attractor. \label{conj:graph_struct}
\end{conj}

Note that the converse of this conjecture does not hold: see Figure \ref{fig:dir_no_attr} for an example of a graph which has a strongly directed cycle, but no dynamic attractor. 

\begin{figure}
\begin{center}
\includegraphics{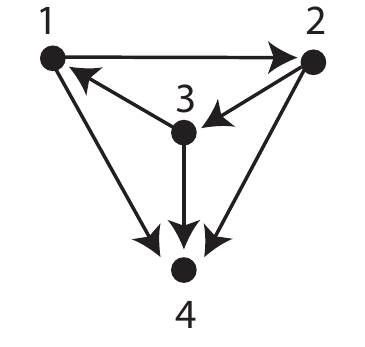}
\end{center}
\caption[A CTLN with a strongly directed cycle, but no dynamic attractor.]{While this graph has a strongly directed cycle $\{1,2,3\}$, its CTLN does not have any dynamic attractor. Instead, it has a single stable fixed point, supported on $\{4\}$.}
\label{fig:dir_no_attr}
\end{figure}

We prove this conjecture in special cases. First, we note that applying the main result of \cite{hahnloser2000permitted} to symmetric CTLNs establishes that these networks do not have dynamic attractors. Next, we are able to show in a special case that proper sources do not participate in dynamic attractors: 

\begin{ithm}
\label{thm:sources_die}
Let $G$ be a graph in which every vertex is reachable from a source. Then there is a source $j$ in $G$ such that $\lim_{t\to \infty} x_j(t) = 0$. 
\end{ithm}

We conjecture that proper sources die in general: 

\begin{conj}
\label{thm:sources_die}
Let $j$ be a source in a graph $G$. Then $\lim_{t\to \infty} x_j(t) = 0$. 
\end{conj}

We apply this result iteratively to show that directed acyclic graphs (DAGs), graphs with no directed cycles at all, do not have dynamic attractors.

\begin{ithm}
If $G$ is a directed acyclic graph, then no CTLN with graph $G$ has a dynamic attractor. \label{thm:dag}
\end{ithm}

\begin{figure}
\begin{center}
\includegraphics[width = 5 in]{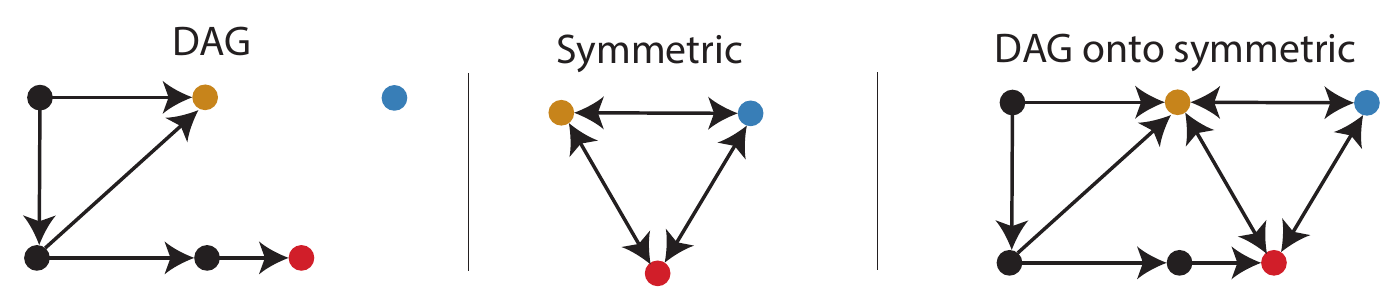}
\end{center}
\caption{Classes of graphs such that all trajectories of their CTLNs are guaranteed to converge to a  fixed point.}
\label{fig:graph_classes}
\end{figure}

Thus, we have shown that CTLNs with either symmetric or feedforward architectures do not support dynamic attractors.
It is natural to ask when it is possible to combine these two architectures to obtain a CTLN which we can guarantee has no dynamic attractors. 
One case in which we can do this is when a directed acyclic graph terminates on a symmetric graph.

\begin{ithm}
Let $G$ be a graph on vertex set $V$, and $\sigma, \tau$ be a partition of $V$. Suppose  the subgraph $G_\sigma$ is symmetric and connected, the subgraph $G_\tau$ is a directed acyclic graph, and that there are no edges from a vertex in $\tau$ back to a vertex in $\sigma$. 
Then no CTLN with graph $G$ has a dynamic attractor. \label{thm:dag_onto_symmetric}
\end{ithm}

Our Conjecture \ref{conj:graph_struct} describes the much broader set of ways we believe it is possible to combine feedforward and symmetric elements to obtain a CTLN for which all trajectories are guaranteed to approach  a fixed point. We summarize our convergence results in Figure \ref{fig:graph_classes}. 

We also apply our results to classify CTLNs on three neurons. We prove that only the directed three-cycle has a dynamic attractor. 

\begin{ithm}
The three-cycle is the only three neuron CTLN with persistent dynamic activity. 
\end{ithm}

Next, we explore to what extent our results extend to general competitive TLNs. In general, the graph $G$ does not determine the fixed point structure of the set of TLNs such that $G_W = G$ \cite{curto2019robust}. 
A graph $G$ is a \emph{robust motif} if all competitive TLNs with $G_W = G$ have the same set of fixed point supports. 
Almost all robust motifs fall into one of two infinite families, DAG1 and DAG2, with only a handful of exceptions. 
If $G$ is a robust motif, does the graph structure determine not just the fixed point structure of the graph, but also the presence of dynamic attractors?  Competitive TLNs with $G_W$ in either of these families each have a single stable fixed point. Further, in the CTLN case, these graphs have no dynamic attractors by our Theorems \ref{thm:dag} and \ref{thm:dag_onto_symmetric}, thus all trajectories approach the single stable fixed point.  It is natural to ask whether these graphs are ``dynamically robust": is it still true that all trajectories approach the single stable fixed point for all competitive TLNs such that $G_W$ is a member of DAG1 or DAG2? We conjecture that this is the case, and prove it in a special case, for graphs in a family we define and name \emph{shallow DAG1}. 

\begin{ithm}
Let $G$ be a member of shallow DAG1. Then in any competitive TLN with $G_W = G$, all trajectories approach the single stable fixed point. 
\end{ithm}

We also show that the result that trajectories of CTLNs whose graphs are in DAG1 have all trajectories approach the single stable fixed point is stable under small perturbations of the weight matrix. 

\begin{ithm}
Let $G$ be a member of DAG1. There is an open neighborhood around the CTLN weight matrix $W(G, \varepsilon, \delta)$ such that all trajectories of a TLN defined by a weight matrix $W$ in this neighborhood approach the unique fixed point. 
\end{ithm}

\setcounter{ithm}{0}

\section{Local Lyapunov Functions}
Our Proposition \ref{thm:mixed_sign} shows that all trajectories of TLNs approach the mixed sign chamber. In this section, we explore a wider variety of cases when we can prove that all trajectories of a dynamical system approach a smaller region of space. In particular, we introduce \emph{local Lyapunov functions}, a tool for showing all trajectories which are trapped in a region $B$ eventually approach a smaller region $A\subset B$. We begin with a few special cases, proving dynamic consequences of the combinatorial concept of graphical domination, before moving on to the general case. 

\subsection{Domination}

In this section, we review the concept of \emph{domination}, introduced in \cite{curto2019fixed}. While \cite{curto2019fixed} uses domination to establish constraints on fixed points of CTLNs, we show that it constrains dynamic attractors as well. 

\begin{defn}
Let $W$ be the weight matrix of a TLN on $n$ neurons. We say $k$ \emph{strongly dominates} $j$, written $k >_W j$,  if $\wt W_{ki} \geq \wt W_{ji}$ for all $i$.  

If $W$ is the weight matrix of a CTLN with graph $G$, then strong domination is equivalent to the following condition on $G$, termed \emph{graphical domination}: 
$k$ graphically dominates $j$, written $k >_G j$, if the following conditions hold: 
\begin{itemize}
\item For each vertex $i \neq j, k$, if $i\to j$, then $i \to k$. 
\item  $k\not\to j$
\item  $j \to k$
\end{itemize}
\end{defn}  

We also consider a weaker version of domination. 

\begin{defn}
Let $W$ be the weight matrix of a TLN on $n$ neurons. We say $k$ \emph{weakly dominates} $j$ if for all  $i\neq j, k$, $\wt W_{ki} \geq W_{ji}$, and $\wt W_{jk} = \wt W_{kj}$.   

If $W$ is the weight matrix of a CTLN with graph $G$, then weak domination is equivalent to the following condition on $G$, termed \emph{weak graphical domination}:  $k$ weakly graphically dominates $j$ if the following conditions hold

\begin{itemize}
\item For all $i$, if $i\to j$, then $i \to k$. 
\item $k\to j$ if and only if $j \to k$. 
\end{itemize}

\end{defn}

\begin{defn}
Neurons $j$ and $k$ are \emph{equal input} if $j$ weakly dominates $k$ and $k$ weakly dominates $j$. 
\end{defn}

\begin{ex}
Consider the graph with $V = \{1, 2, 3\}$, $E = \{(1 \to 2), (2\leftrightarrow 3), (3 \to 1)\}$, pictured in Figure \ref{fig:dom_ex}. Then 2 strongly dominates 1, 2 weakly dominates 3, and 1 and 3 have no domination relationship.
\begin{figure}[ht!]
\begin{center}
\includegraphics[width = 3 in]{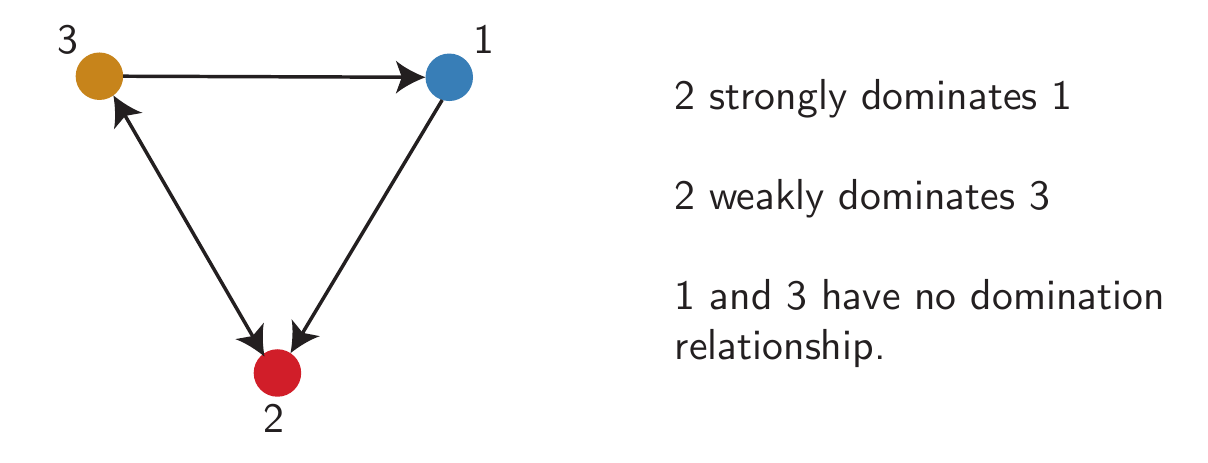}
\end{center}
\caption[Examples of strong and weak domination.]{In this graph, 2 strongly dominates 1, 2 weakly dominates 3, and 1 and 3 have no domination relationship. \label{fig:dom_ex}}
\end{figure}
\end{ex}

\subsubsection{Basic results}
We begin with the following observation about strong domination: 
\begin{obs}\label{obs:strong_domination}
Let $k$ strongly dominate $j$,  $\wt W_{ki} \geq \wt W_{ji}$. Then for all $x$ in the positive orthant, 
$$h_k^*(x) - h_j^*(x) = \sum_{i= 1}^n (\wt W_{ki}-\wt W_{ji})x_i \geq 0.$$
\end{obs}

We first see how this shapes the arrangement of nullclines. 

\begin{prop}\label{prop:dom_sep}
If $k$ strongly dominates $j$, then $H_j$ separates $H_k$ from the origin. 
\end{prop}

\begin{proof}
By Observation \ref{obs:strong_domination}, $h_k^*(x) - h_j^*(x) \geq 0$ throughout the positive orthant.
Thus, on the $H_j$ nullcline, where  $h_j^*(x) = 0$, we must have $h_k^*(x)\geq 0$. Therefore, within the positive orthant, $H_j$ is contained on the positive side of $H_k$, thus $H_j$ separates $H_k$ from the origin. 
\end{proof}

Domination has some simple consequences for fixed points. For instance, if any neuron $k$ strongly dominates $j$, $j$ is not contained within any fixed point support. For a more thorough account of how domination constrains fixed points, see  \cite{curto2019fixed}.
We now turn to the easiest constraints on dynamic attractors provided by domination.

\begin{prop}\label{prop:supersource}
If $i$ is a proper source in $G$, and $i\to j$ for all vertices $j$, then $\lim_{t\to\infty} x_i(t) = 0$. 
\end{prop}

\begin{proof}
Since $i$ is a proper source in $G$, and $i\to j$ for all $j\in [n]$, for any $j \in [n]$, $j$ strongly dominates $i$. Thus, $H_i$ separates all other $H_j$ from the origin. Thus, $H_i$ is the lower boundary of $\cA$. Thus, by Proposition \ref{thm:mixed_sign}, all trajectories approach the negative side of $H_i$. Any trajectory must either eventually cross $H_i$, or approach a fixed point on $H_i$. Since $i$ is proper source, it is not contained in any fixed points, to if a trajectory is approaching a fixed point, $\lim_{t\to\infty} x_i(t) = 0$. Otherwise, once our trajectory crosses $H_i$ at time $T$, for all $t >T$, $x_i' \leq 0$. Thus, since the quantity $x_i$ is monotonic, it approaches a fixed value. Thus, the trajectory must approach the $i^{th}$ nullcline. It cannot approach the interior of the nullcline $H_i$, thus it must approach $E_i$. Therefore, $\lim_{t\to\infty} x_i(t) = 0$.
\end{proof}

\begin{prop}\label{prop:supersink}
If $i$ is a proper sink in $G$, and $j\to i$ for all $j\in [n]$, then $\lim_{t\to\infty} x_i(t) = \theta$ and $\lim_{t\to\infty} x_j(t) = 0$ for all $j\neq i$. 
\end{prop}

\begin{proof}
Since $i$ is a proper sink in $G$ and $i$ receives an edge from every other neuron, $i$ dominates all other neurons $j\in [n]$. Thus, $H_j$ is separated from the origin from each $H_i$. Therefore, $H_i$ forms the outer boundary of $\cA$. Thus, by Proposition \ref{thm:mixed_sign}, all trajectories approach the positive side of $H_i$. Any trajectory must either eventually cross $H_i$, or approach a fixed point on $H_i$. Either way, the value of $x_i$ must approach a fixed value. Since $i$ dominates all other neurons, the only fixed point support of this network is the singleton set $\{i\}$. Thus, in this case $\lim_{t\to\infty} x_i(t) = \theta$ and  $\lim_{t\to\infty} x_j(t) = 0$ for all $j\neq i$. 
 
\end{proof}

%%%%%%%%%%%%%%%%%%%%%%%%% Dynamic domination  %%%%%%%%%%%%%%%%%%%%%%%%%%%%%%%%

\subsubsection{Dynamic domination relationships}

While Propositions \ref{prop:supersource} and \ref{prop:supersink} demonstrate that the purely combinatorial property of graphical domination has dynamic consequences, they only apply in extremely restrictive cases. In this next section, we show that domination has further-reaching consequences. We prove several results that show that domination relationships between the neurons $j$ and $k$ have consequences for the relationship between activity levels $x_j(t)$ and $x_k(t)$ over the long run. In order to do this, we must first prove the following lemma about threshold nonlinearities. 
\begin{lem}\label{lem:thresh_ineq}
Let $a \leq 0$. If $a \leq x - y$, then $a \leq [x]_+ - [y]_+$. Further, these inequalities can be made strict: if $b < 0$, $b < x - y$, then $b < [x]_+ - [y]_+$.
\end{lem}
\begin{proof}
We consider four cases:
\begin{itemize}
\item Case 1: $x \geq 0$ and $y \geq 0$. Then 
\begin{align*}
a \leq x- y = [x]_+ - [y]_+\\
b < x- y = [x]_+ - [y]_+
\end{align*}
\item Case 2: $x \geq 0$ and $y < 0$. Then 
\begin{align*}
a \leq 0 \leq [x]_+ = [x]_+ - [y]_+\\
b < 0 \leq [x]_+ = [x]_+ - [y]_+
\end{align*}
\item Case 3: $x < 0$ and $y \geq 0$. Then 
\begin{align*}
a \leq x - y < [x]_+ - [y]_+\\
b < x - y < [x]_+ - [y]_+
\end{align*}
\item 
Case 4: $x < 0$ and $y < 0$. Then 
\begin{align*}
a \leq 0 = [x]_+ - [y]_+\\
b <0 = [x]_+ - [y]_+
\end{align*}
\end{itemize}
\end{proof}
\begin{prop}[Weak domination] \label{prop:weak_dom}
If $k$ weakly dominates $j$, then if $x_k(T) \geq x_j(T)$, then $x_k(t) \geq x_j(t)$ for all $t \geq T$. 
\end{prop}

\begin{proof}
We show that if $k$ weakly dominates $j$, it not possible to cross from the $x_k > x_j$ side of the hyperplane $x_k = x_j$ to the $x_j > x_k$ side. To do this, we show that if 
$x_j = x_k$, then $x_k '- x_j' > 0$. We have 
\begin{align*}
x_k' - x_j' = -x_k + x_j + \left[\sum_{i = 1}^{n} W_{ki}x_i + \theta \right]_+ - \left[\sum_{i = 1}^{n} W_{ji}x_i + \theta \right]_+.
\end{align*}
Since $x_j = x_k$, 
\begin{align*}
x_k' - x_j' = \left[\sum_{i = 1}^{n} W_{ki}x_i + \theta \right]_+ - \left[\sum_{i = 1}^{n} W_{ji}x_i + \theta \right]_+.
\end{align*}
Now, since the threshold nonlinearity if a monotone transformation, 
\begin{align*}
0 \leq \left[\sum_{i = 1}^{n} W_{ki}x_i + \theta \right]_+ - \left[\sum_{i = 1}^{n} W_{ji}x_i + \theta \right]_+ 
\end{align*}
if and only if 
\begin{align*}
0 \leq 
\left(\sum_{i = 1}^{n} W_{ki}x_i + \theta \right) - \left(\sum_{i = 1}^{n} W_{ji}x_i + \theta \right).
\end{align*}
We have 

\begin{align*}
\left(\sum_{i = 1}^{n} W_{ki}x_i + \theta \right) - \left(\sum_{i = 1}^{n} W_{ji}x_i + \theta \right) = 
\sum_{i \neq j, k}^{n} (W_{ki}-W_{ji})x_i  + W_{kj} x_j -  W_{jk}x_k.
\end{align*}
Because $k$ weakly dominates $j$, $(W_{ki}-W_{ji})\geq 0$ for all $i$, so $\sum_{i \neq j, k}^{n} (W_{ki}-W_{ji})x_i \geq 0$. Because $x_j = x_k$ and because $k$ weakly dominates $j$, $W_{kj} x_j = W_{jk}x_k$. 
Thus
\begin{align*}
0\leq\sum_{i \neq j, k}^{n} (W_{ki}-W_{ji})x_i  + W_{kj} x_j -  W_{jk}x_k
\end{align*}
as desired.  
\end{proof}

\begin{cor}[Equal input]\label{cor:equal_input}
If $x_k$ and $x_j$ are equal input, then if $x_k(T) > x_j(T)$, then $x_k(t) > x_k(t)$ for all $t > T$. If $x_k(T) = x_j(T)$, then $x_k(t) = x_k(t)$ for all $t > T$.
\end{cor}

\begin{proof}
If $j$ and $k$ are equal-input, then $j$ weakly dominates $k$ and $k$ weakly dominates $j$. Thus, if  $x_k(T) \geq x_j(T)$, then $x_k(t) \geq x_j(t)$ for all $t > T$. If $x_k(T) = x_j(T)$, then $x_k(T) \geq x_j(T)$ and $x_k(T) \leq x_j(T)$, so by Proposition \ref{prop:weak_dom}, $x_k(t) \geq x_j(t)$ and $x_k(t) \leq x_j(t)$ for all $t \geq T$, so $x_k(t) = x_j(t)$ for all $t \geq T$. 
\end{proof}

\begin{prop}[Strong domination]
\label{prop:strong_dom}
If $k$ strongly dominates $j$, then for all $\eta > 0$, there exists some $T_\eta \geq 0$ such that for all $t \geq T_\eta$, $x_k(t) +\eta > x_j(t)$. \end{prop}

\begin{proof}
We prove that if $x_j > x_k$, then $x_k' - x_j' > 0$. By Observation \ref{obs:strong_domination}, we have $h_k^*(x) > h_j^*(x)$ within the positive orthant. 
Now, we apply Lemma \ref{lem:thresh_ineq} to show that this still holds once we apply the threshold nonlinearity: 
\begin{align*}
h_k^*(x)  -h_j^*(x)  &=  -x_k + y_k +  x_j - y_j \geq 0\\
x_k - x_j &\geq  y_k - y_j \\
x_k - x_j &\geq [y_k]_+ - [y_j]_+\\
x_k' - x_j'&\geq 0
\end{align*}
whenever $x_j  \geq x_k$. Further, this inequality can be made strict whenever $x_j > x_k$, since in this case, we must have $h_k^*(x)  -h_j^*(x) > 0$.

Now, we use this to prove the desired result. First, suppose that the trajectory starting at the point $x(0)$ has $x_k(T) \geq  x_j(T)$ at some time $T$. Then by Proposition \ref{prop:weak_dom} $x_k(t) \geq  x_j(t)$ for all time $t \geq T$. Thus, we can take $T_\eta= T$ for all $\eta > 0$. 

Now, suppose this does not happen, and that $x_k(t) <  x_j(t)$ for all $t \geq 0$. Then $x_k'(t) - x_j'(t) > 0$ for all $t\geq 0$. Thus, the quantity $x_k(t)-x_j(t)$ is monotonically increasing, and must thus approach a limit. This limit must be $x_k(t)-x_j(t) = 0$, since when $x_k(t) < x_j(t)$, $\xdot_k(t) -\xdot_j(t) >0$. Thus for all $\eta > 0$, there exists some $T_\eta \geq 0$ such that for all $t \geq T_\eta$, $x_k(t) +\eta > x_j(t)$.
\end{proof}

\subsection{General Local Lyapunov Functions}

The main idea in our proof of Proposition \ref{prop:strong_dom} is an example of something we will call a \emph{local Lyapunov function}. 
Lyapunov functions are typically used to show that all trajectories of a dynamical system approach a specified fixed point, conventionally taken to be the origin. Here, we generalize this style of argument by introducing local Lyapunov-like functions. We will use these to show that all trajectories which enter a region $B$ approach a smaller region $A$, as illustrated in Figure \ref{fig:llf}.

\begin{figure}[ht!]
\begin{center}
\includegraphics[width = 2.5 in]{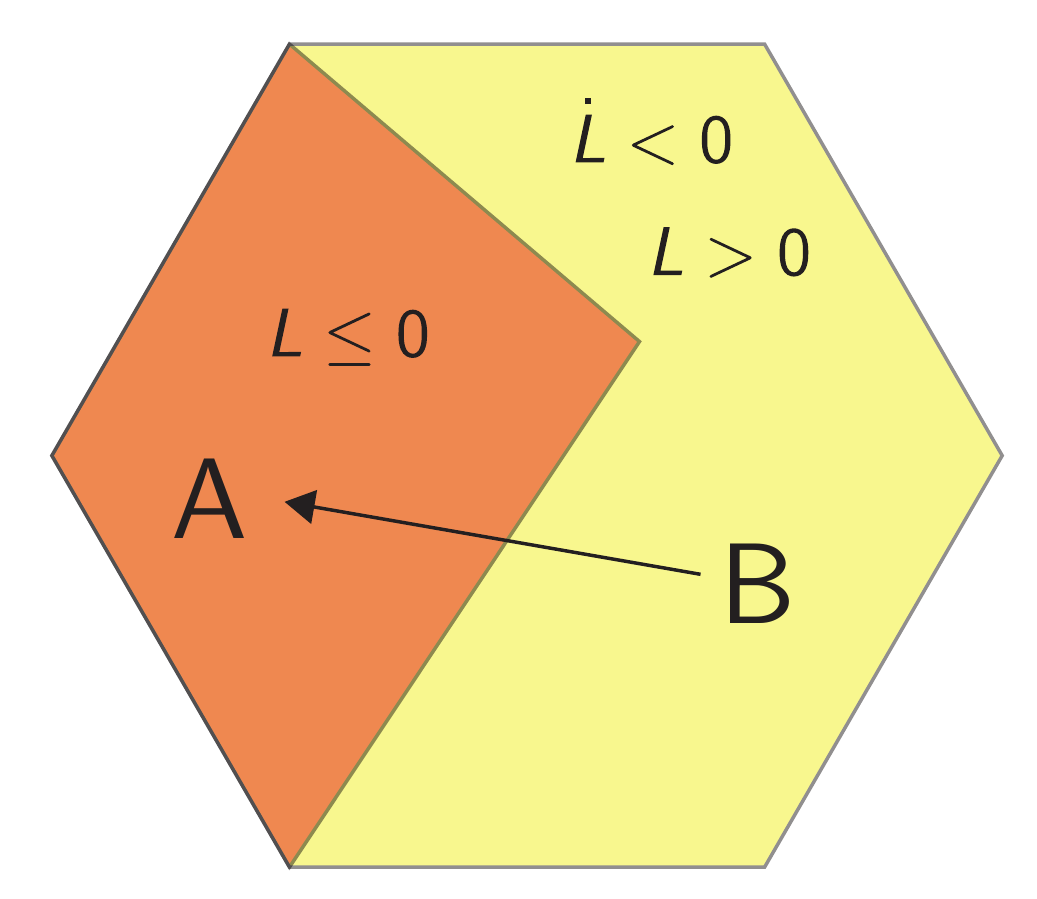}
\end{center}
\caption[Local Lyapunov functions]{If $L$ is a local Lyapunov function for $A$ and $B$, then all trajectories which enter $B$ must approach $A$. \label{fig:llf}}
\end{figure}
 
\begin{defn}Let $A\subseteq B\subseteq \R^n$ be compact sets, and $\xdot = g(x)$ be a dynamical system on $\R^n$.  A continuous,continuously differentiable function $L: B \to \R$ is a \emph{local Lyapunov-like function} for the pair $A, B$ if it satisfies the following: 
\begin{enumerate}
\item  $\dot{L} < 0$ on $B\setminus A$.
\item $L > 0$ on $B\setminus A$.  
\item $L\leq 0$ on $A$. 
 \end{enumerate} 
\end{defn}

For instance, our proof of Proposition \ref{prop:strong_dom} established that the $L(x) = x_j - x_k$ is a local Lyupanov function for $B = \cA$, $A = \cA \cap \{ x \mid x_k \geq x_j\}$. 

\begin{lem}
If $L$ is a local Lyapunov-like function for $A,B$, then any trajectory which remains in $B$ for all time approaches $A$. 
\end{lem}

\begin{proof}
Let $x_0\in B$. By assumption, the trajectory $x(t)$ with this  initial condition remains in $B$ for all $t\geq 0$. Now, suppose $x(t)$ does not ever enter $A$. Then $\dot L < 0$ for all time. This means that the value $L(x(t))$ is monotonically decreasing, so it must approach a limit as $t\to \infty$. This means that $\dot L(x(t)) \to 0$ as $t\to \infty$. Since the only place within $B$ where it is possible that $\dot L = 0$ is $A$, the trajectory must approach $A$. 

Now, suppose the trajectory enters $A$. We show that it must stay there. Since $\dot L  < 0$ on $B\setminus A$, by continuity of $\dot L $ we have $\dot L \leq 0$ on $\partial A$. Thus, if the trajectory were to cross out of $A$, the value of $L$ must decrease as we cross out. However, since $L \leq 0$ on $A$ and $L > 0$ on $B$, this is not possible. 
\end{proof}

\section{Main results: CTLNs}

In this section, we work towards a proof of the following result, which has Theorems \ref{thm:dag} and \ref{thm:dag_onto_symmetric} as consequences.

\begin{ithm}
\label{thm:sources_die}
Let $G$ be a graph in which every vertex is reachable from a source. Then there is a source $j$ in $G$ such that $\lim_{t\to \infty} x_j(t) = 0$. 
\end{ithm}

We begin by giving an informal summary of our argument. 
First, we prove that if $j$ is a source, and if $x_j \leq x_k$ for all other neurons $k$, then $\lim_{t\to \infty} x_j = 0$ (Lemma \ref{lem:smallest_sources_die}). We do this by showing that  within this set, the nullcline $H_j$ forms the lower boundary of the mixed-sign chamber (Figure \ref{fig:sources_die_arg} A), thus trajectories are trapped on the negative side of $H_j$. Thus once trajectories are trapped in this region, $\lim_{t\to  \infty}x_j$. 

\begin{figure}
\includegraphics{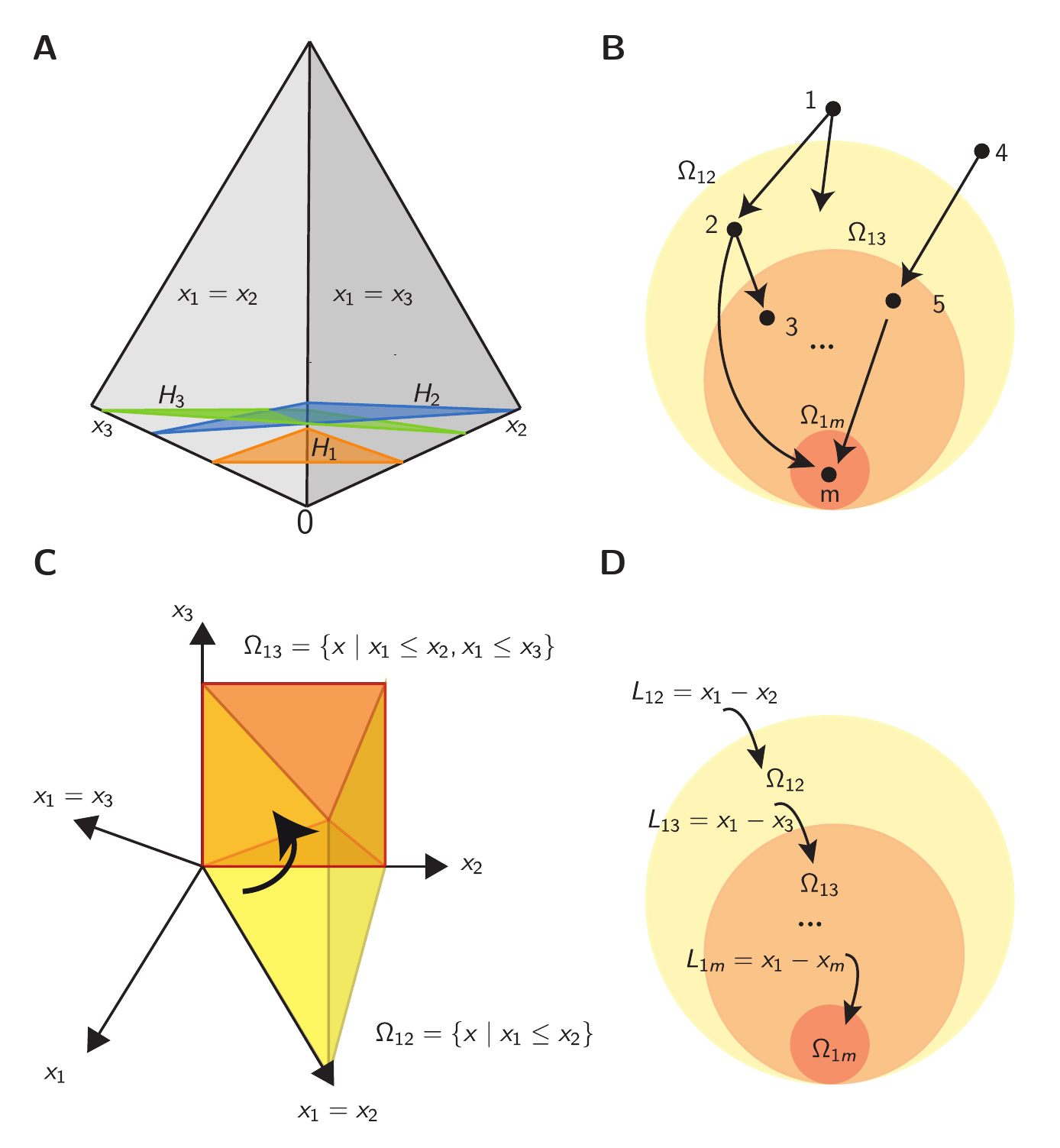}
\caption[Schematic for the proof of Theorem \ref{thm:sources_die}]{A) Within the region where the source numbered 1 is the smallest, the nullcline $H_1$ separates the other nullclines from the origin. B) We number our neurons along paths from the smallest source. C) The regions $\Omega_{1\ell}$ and $\Omega_{13}$. D) Our local Lyapunov functions.  \label{fig:sources_die_arg}}
\end{figure}

Next, we show that this condition is eventually satisfied for the smallest source: the source $j$ which has the $x_j \leq x_\ell$ for all other sources $\ell$ at our initial condition.
We show this using local Lyapunov functions for a nested sequence of sets $\Omega^{G}_1\supseteq \Omega^G_{2}\supseteq \cdots \supseteq\Omega^G_{n}$. These sets are defined by renumbering the neurons along paths from the source (Figure \ref{fig:sources_die_arg} B), and then requiring the activity of the source to be lower than that of other neurons along this path (Figure \ref{fig:sources_die_arg} C). Finally, we show that the function $L_\ell = x_j - x_\ell$ is a local Lyapunov function for the sets $A = \Omega^{G}_\ell$ and 
 $B =\Omega^{G}_{\ell -1}$. This means that trajectories get in $\Omega^{G}_{\ell -1}$ go towards $ \Omega^{G}_\ell$. Thus, trajectories must eventually enter the set where $x_j \leq x_{\ell}$ for all vertices $\ell$ reachable from $j$. 
This must also happen for every other source $i$. Thus, since $x_j$ is the smallest source, we must eventually have $x_j\leq x_i \leq x_\ell$ if $\ell$ is reachable from some other source $i$. Therefore, we must eventually enter the set where $x_j\leq x_\ell$ for all other neurons $\ell$. Thus, by Lemma \ref{lem:smallest_sources_die}, $\lim_{t\to \infty} x_j = 0$. 
Finally, the value of $x_j$ must become so low that its influence on the other neurons may be neglected. Thus, we let some other source take over as the smallest source. In this way, we show that the activity of all sources must decay to zero. 

In order to apply this result inductively in order to prove Theorems \ref{thm:dag} and \ref{thm:dag_onto_symmetric}, we wish to eventually neglect the activity of sources which have ``died". Thus, we will need to prove our results in a way which has some threshold for negligible activity which we allow ourselves to neglect. When we are proving Theorems \ref{thm:dag} and \ref{thm:dag_onto_symmetric} by induction, we will modify our graph by ignoring neurons whose activity is approaching zero. We will also need to deal with the fact that, if we have a local Lyapunov function for $A$ and $B$, this only guarantees that trajectories which remain in $B$ \emph{approach} $A$, not that they actually enter it. We handle both of these issues through our definition of $\Omega^{G}_\ell(\alpha, \beta)$. 

\begin{defn}\label{def:order_omega}
Let $G$ be a graph with vertices $V_G$. Let $L\subseteq V_G$.  Let $j$  be a particular source in $G|_L$.  Renumber the vertices in $L$ as follows: 
First, partially order the vertices by $k \leq_j \ell$ if the shortest path from $j$ to $k$ is shorter than the shortest path from $j$ to $\ell$, setting this length to be $\infty$ if $\ell$ is not reachable from $j$. Then, pick a linear extension of this order. Notice that under this order, we have labeled $j = 1$. Now, we can define
\begin{align*}
\Omega^{G_L}_\ell  &= \{x \mid x_j \leq x_k \mbox{ for all } k \leq \ell, k\in L, x_k = 0 \mbox{ for all } k \notin L\}\\
\Omega^{G_L}_\ell (\alpha, \beta) & = \{x \mid x_j \leq x_k + \alpha \mbox{ for all } k \leq \ell, k\in L, x_k \leq \beta, |\xdot_k|\leq \beta \mbox{ for all } k \notin L\}
\end{align*}
When $L = V_G$, we drop the parameter $\beta$, writing 
\begin{align*}
\Omega^{G_L}_\ell (\alpha) & = \{x \mid x_j \leq x_k + \alpha \mbox{ for all } k \leq \ell\}.
\end{align*}
\end{defn}

Notice that $\Omega^{G_L}_{|L|}$ is the set where the source numbered 1 has $x_1 \leq x_\ell$ for all $\ell \in L$, and $x_\ell = 0$ for all $\ell \notin L$. Our first lemma establishes that ``smallest sources die": if a trajectory stays in the set $\Omega^{G_L}_{|L|}$ where the source 1 is the neuron with lowest activity, then $x_1$ approaches zero. We prove a stronger result, that this holds on $ \Omega^{G_L}_{|L|}(\eta, \zeta)$ for some positive values of $\eta$ and $\zeta$.
 
\begin{lem}\label{lem:smallest_sources_die}
Let $L \subseteq V_G$ be our set of active neurons. Assume the neurons are numbered according to Definition \ref{def:order_omega} such that neuron 1 is a source in $G|_L$. Then there exist $\eta, \zeta > 0$ such that if $x \in \Omega^{G_L}_{|L|}(\eta, \zeta)$ for all $t \geq T$, then $\lim_{t\to \infty} x_1(t)= 0$. 
\end{lem}

%To motivate this claim, we first sketch the proof of an easier version with $\eta = 0$. Then we have $x_j \leq x_k$ for all $k$. By Proposition \ref{prop:weak_dom}, once $x_j \leq x_k$ at some time $T$,  $x_j \leq x_k$ for all $t > T$. Further, within the region where  $x_j \leq x_k$, the hyperplane $H_j$ forms the lower boundary of the mixed-sign chamber, separating all other hyperplanes $H_j$ from the origin. Thus, the entire mixed-sign chamber is contained on the negative side of $H_j$, so the value of $x_j$ must approach zero. Much of the work in the next proof goes into showing this argument can also be made to work for a positive $\eta > 0$. 

\begin{proof}

We first describe the value of $\eta$ and $\zeta$. If $1$ is a proper source, there are no fixed points on the set $ C := H_1\cap \R^n_\geq$, so $||\xdot||_1> 0$ on $H_1 \cap \R^n_\geq$. If $1$ is an isolated vertex, then the only fixed point $x^*$ on $H_1$ has support $\{1\}$, and therefore is separated from $\Omega^{G_L}_{|L|}$ by a positive distance $d$. Thus, $||\xdot||_1> 0$ on the compact set  $C := H_1\cap \R^n_\geq \setminus B_d(x^*)$ we get by deleting a $d$-neighborhood of this fixed point $x^*$.  Thus the continuous function $||\xdot||_1$ attains its minimum value on $C$. Therefore there exists some positive $\alpha > 0$ such that  $||\xdot||_1 \geq \alpha > 0$ on $C$. Let $\eta  = \frac{\alpha}{4n\delta }$. Let $\zeta = \delta \eta$. 

We first focus on the neurons in $L$ and show that if $x_1(T) < x_k(T) +\eta $ for some $T$, then $x_1'(t) \leq 0$ for all $t \geq T$. Thus $\lim_{t\to \infty}x_1(t) = 0$. For all $k$ such that $1\not \to k$,
\begin{align*}
h_k^*(x(t)) - h_1^*(x(t)) &= -x_k + x_1 +  (-1-\delta)x_1 - (-1-\delta)x_k + \sum_{i\neq 1, k} (W_{ki} - W_{1i})x_i  \\
h_k^*(x(t)) - h_1^*(x(t)) &= \delta(x_k - x_1) + \sum_{i\neq 1, k} (W_{ki} - W_{1i})x_i \\
h_k^*(x(t)) - h_1^*(x(t)) &\geq -\delta \eta. 
\end{align*}
Now, for $x(t)\in H_1 \cap \R^n_\geq$, $h_1^*(x(t)) = x_1'(t) = 0$. For all $x(t)$, $h_k^*(x(t)) \leq \xdot_k(t)$. Therefore, for $x(t)\in H_1 \cap \R^n_\geq$,
\begin{align*}
x_k'(t)  &\geq -\delta \eta. 
\end{align*}
If $1 \to k$, then the entirety of $H_1 \cap \R^n_\geq$ is contained on the positive side of $H_k$, so on $H_1 \cap \R^n_\geq$, 
\begin{align*}
\xdot_k(t)  \geq 0 \geq -\delta \eta\\ 
\xdot_k(t) \geq -\delta \eta. 
\end{align*}
Now, we note that this is also satisfied  by requirement for  the neurons outside $L$. 

Now, we show that it is only possible to cross the $H_1$ hyperplane from the positive to the negative side. We do this by showing that the dot product of $ \xdot(t)$ with the normal vector to $H_1$ is always negative. 
We have 
\begin{align*}
\sum_{i = 1}^n \xdot_i(t) = \sum_{x_i(t)\geq 0}\xdot_i(t) + \sum_{\xdot_i(t)< 0}\xdot_i(t) = \sum_{i = 1}^n | \xdot_i(t)| + 2 \sum_{x_i(t)< 0}\xdot_i(t) = ||\xdot(t)||_1+ 2 \sum_{x_i(t)< 0}\xdot_i(t).
\end{align*}
Since $\xdot_k(t)   \geq -\delta \varepsilon$ for all $k$,
\begin{align*}
2 \sum_{x_i(t)< 0}\xdot_i(t) \geq -2n\delta\varepsilon = -2n\delta \frac{\alpha}{4n\delta} = -\frac{\alpha}{2}.\\
\end{align*}

Thus 
\begin{align*}
\sum_{i = 1}^n \xdot_i = ||\xdot||_1+ 2 \sum_{x_i< 0}\xdot_i \geq \frac{\alpha}{2}  > 0.
\end{align*}
Now, recall that the normal vector to $H_i$ is $(-1, -1-\delta,-1-\delta, \ldots,  -1-\delta)$.  Thus, since $\xdot_1 = 0$ on $H_1\cap \R^n_\geq$, 
\begin{align*}
\xdot(t)\cdot h_1 = (-1-\delta)\sum_{i = 1}^n \xdot_i(t) < 0. 
\end{align*}
Thus, trajectories can only cross from the positive side to the negative side of $H_1$. Thus, the value of $x_1(t)$ eventually becomes monotone, and thus must approach a limit. This means the trajectory must approach the $x_1$ nullcline. The trajectory cannot approach $H_1$ because if $c\in  H_1$, $\xdot(t)\cdot h_1< 0$. Thus, the trajectory must approach $E_1$, so  $\lim_{t\to \infty}x_1(t) = 0$.

\end{proof}

Now, we show that the conditions of Lemma \ref{lem:smallest_sources_die} are ultimately achieved for some source, i.e.  for some source, $ \Omega^{G_L}_{|L|}$ is an attracting set. In particular, we show this for the smallest source, i.e. the source neuron $j$ such that $x_j \leq x_\ell$ for all other \emph{sources} $\ell$. 
We will prove by induction using this order that if $k$ is reachable along a path from a source $j$, then eventually $x_j  < x_k  + \eta$. We do this using local Lyapunov functions for a nested sequence of sets.

\begin{lem}\label{lem:lyapunov} For any $\beta > 0$, the function $L_\ell = x_j - x_{\ell} + \beta$ is a local Lyapunov function for 
$B = \Omega^{G_L}_{\ell}(\beta, \zeta)$ and  $A =\Omega^{G_L}_{\ell-1}(\alpha, \zeta)\cap B$ for $\beta = \frac{ \alpha\varepsilon}{2(\varepsilon + \delta)}$ and $\zeta = \frac{\varepsilon \alpha}{2n(\varepsilon + \delta)}$. 
\end{lem}

\begin{proof}
By construction, we have $ L_\ell \leq 0$ on $A$ and $L_\ell > 0$ on $B \setminus A$.  
Now, we show that $\dot L_\ell <0$ on $B\setminus A$. 
We consider 
\begin{align*}
\dot L_\ell = \xdot_j - \xdot_\ell = -x_j + x_\ell + \left[\sum_{i=1}^n W_{ji}x_i + \theta \right]_+ - \left[\sum_{i=1}^n W_{\ell i}x_i + \theta\right]_+. 
\end{align*}
By Lemma \ref{lem:thresh_ineq}, when $x_j \geq x_k$
\begin{align*}
\xdot_j - \xdot_\ell &\leq -x_j + x_\ell + \left( \sum_{i=1}^n W_{ji}x_i + \theta\right)  - \left(\sum_{i=1}^n W_{\ell i}x_i + \theta \right)_+\\
\xdot_j - \xdot_\ell &\leq -x_j + x_\ell + \sum_{i=1}^n (W_{ji}-W_{\ell i})x_i \\
\xdot_j - \xdot_\ell &\leq (-1 - W_{\ell j})x_j + (1 + W_{j\ell})x_\ell + \sum_{i \in L\setminus \{j, \ell\}  }(W_{ji}-W_{\ell i})x_i + \sum_{k \notin L }(W_{ji}-W_{\ell i})x_k
\end{align*}
Now, since $j$ is source in $G|_L$, $(W_{ji}-W_{\ell i}) \leq 0$ for each neuron $i \in L\setminus \{j, \ell\} $. 
 Further, $1 + W_{j\ell} = -\delta \leq 0$.  
 Since $x\in B$, the total activity of the neurons outside $L$ is bounded above, thus $ \sum_{k \notin L }(W_{ji}-W_{\ell i})x_k \leq  2n(\varepsilon + \delta) \frac{\varepsilon \alpha}{2n(\varepsilon +\delta)} = \frac{\varepsilon \alpha}{2}$.  
 
Now, we consider two cases, if $\ell$ recieves an edge directly from $j$ or not. In the fist case,  then  $(-1 - W_{\ell j})  = -\varepsilon <0$, so 
$$\xdot_j - \xdot_\ell \leq -\varepsilon x_j + \frac{\varepsilon\alpha}{2}.$$
Since $x_j \geq \alpha$ in $B\setminus A$, we have 
 $\xdot_j - \xdot_\ell \leq 0$ whenever $x_j \geq x_k$, thus $\dot L < 0$ on $B\setminus A$. 

Next, if $\ell$ does not receive an edge directly from $j$, by our ordering of neurons, we have that $\ell$ does receive an edge from some neuron $k < \ell$. For $x\in B$, we have $x_j \leq x_k + \beta$. We pull the $k^{th}$ term out of the summation in the right hand side, 
rewriting it as 
\begin{align*}
\xdot_j - \xdot_\ell &\leq (-1 - W_{\ell j})x_j + (1 + W_{j\ell})x_\ell + (W_{jk} - W_{\ell k})x_k  + \sum_{i\neq j, k,\ell \in L}(W_{ji}-W_{\ell i})x_i + \frac{\varepsilon\alpha}{2}\\
\xdot_j - \xdot_\ell &\leq \delta x_j - \delta x_\ell  - (\delta + \varepsilon)x_k  +\sum_{i\neq j, k,\ell \in L}(W_{ji}-W_{\ell i})x_i + \frac{\varepsilon\alpha}{2}\\
\end{align*}
We can drop the terms $ - \delta x_\ell$ and $\sum_{i\neq j, k,\ell \in L}(W_{ji}-W_{\ell i})x_i$, as they are guaranteed to be non-positive, obtaining the bound 
\begin{align*}
\xdot_j - \xdot_\ell &\leq \delta x_j  - (\delta + \varepsilon)x_k + \frac{\varepsilon\alpha}{2}.
\end{align*}
Now, for $x\in B$, we have $x_j \leq x_k + \beta$, so 
\begin{align*}
\xdot_j - \xdot_\ell &\leq \delta x_j  - (\delta + \varepsilon)(x_j - \beta)+ \frac{\varepsilon\alpha}{2}\\
\xdot_j - \xdot_\ell &\leq -\varepsilon x_j  + (\delta + \varepsilon)\beta + \frac{\varepsilon\alpha}{2}.
\end{align*}
Now, for all $x\notin A$, we have $x_j \geq x_k + \alpha \geq \alpha$, thus-
\begin{align*}
\xdot_j - \xdot_\ell &\leq -\varepsilon \alpha + (\delta + \varepsilon)\beta + \frac{\varepsilon\alpha}{2}= 0,
\end{align*}
as desired. 
\end{proof}

\begin{lem}\label{lem:source_becomes_smallest_1}
Let $L \subseteq V_G$ be our set of active neurons,  numbered according to Definition \ref{def:order_omega} such that neuron 1 is a source in $G|_L$. Let $\ell$ be a vertex which is reachable from $1$ along a directed path. Then for each $\eta \geq 0$, there exists some  $T \geq 0$ and some $\zeta> 0$ such that if $x_k(t), |\xdot_k(t)| \leq \zeta$ for all  time $t$ and all inactive neurons $k\notin L$ such that $x(t) \in \Omega_\ell^{G|_L}(\eta, \zeta)$ for all $t \geq T$. 
\end{lem}

\begin{proof}
We prove this by induction, with the vertices labeled according to Definition \ref{def:order_omega}. If $\ell = 1$ this result is trivially true. 

Now, we assume as an inductive hypothesis that this result holds for all vertices $k < \ell$, and show that it must hold for $\ell$.  Then by Lemma \ref{lem:lyapunov}, we have that there is a local Lyapunov function $L$ for $B = \Omega_{\ell-1}^{G|_L}(\eta/2, \zeta)$ and  $A =\Omega_\ell^{G|_L}(\beta, \zeta)\cap B$ for $\beta = \frac{ \eta\varepsilon}{2(\varepsilon + \delta)}.$ Now, by the inductive hypothesis, there exists a time $T'$ such that $x(t) \in B$ for all $t \geq T'$. Thus, all trajectories must enter and remain in $B$, so they approach $A$. Thus, after some time $T$ all trajectories must actually enter and become trapped in the set  $\Omega_\ell^{G|_L}(\eta, \zeta)$. 
\end{proof}

Now, we have shown that each source eventually becomes (approximately) smaller than all vertices which are reachable from it. However, to apply Lemma \ref{lem:smallest_sources_die}, we need to show the source eventually becomes  (approximately) smaller than \emph{all} active neurons. 

\begin{lem}\label{lem:smallest_source_becomes_smallest}
Let $L \subseteq V_G$ be our set of active neurons,  numbered according to Definition \ref{def:order_omega} such that neuron 1 is a source in $G|_L$. Suppose $x_1(0) < x_k(0)$ for all other sources $k$ in $G|_L$. Then for each $\eta \geq 0$, there exists $\zeta> 0$ such that if $x_j \leq \zeta$ for all $j\notin L$ and $T \geq 0$ such that $x(t) \in \Omega_\ell^{G|_L}(\eta, \zeta)$ for all $\ell \in L$, $t \geq T$. 
\end{lem}

\begin{proof}
By assumption, every vertex $\ell$ is reachable from some source $k$ along a directed path.  Thus by numbering the vertices according to Definition  \ref{def:order_omega} such that $k$ is sent to $1$, we can apply Lemma  \ref{lem:source_becomes_smallest_1}, to show that exists a time $T_{k\ell}$ such that $x_k \leq x_\ell + \eta$ for all $t\geq T_{k\ell}$. Now, since  $x_1(0) < x_k(0)$ for all other sources $k$, there is some $\zeta$ such that if  $x_j \leq \zeta$ for all $j\notin L$, the order of $x_1$ and $x_k$ cannot switch, thus $x_1$ will remain smaller than $x_k$ for all other sources $k$. Thus, by taking $T \geq T_{k\ell}$ for each vertex $\ell$ and source $k$, we have that 
$x_j \leq x_j \leq x_\ell + \eta$  for each vertex $\ell$. Thus, for  $t \geq T$, we have  $x(t) \in \Omega_\ell^{G|_L}(\eta, \zeta)$, as desired. 
\end{proof}

Finally, we are ready to prove Theorem \ref{thm:sources_die}. 

\begin{proof}
Let $G$ be a graph for which every vertex is reachable from a source, and consider a CTLN with graph $G$. Let $j$ be the source in $G$ which satisfies $x_j \leq x_k$ for all other sources  at the initial condition. By Lemma \ref{lem:smallest_source_becomes_smallest}, there exists a time $T$ such that $x\in \Omega_\ell^{G}(\eta)$ for all neurons $\ell$ and the values of $\eta$ an used in Lemma \ref{lem:smallest_sources_die}. Therefore, $\lim_{t\to\infty} x_{j} = 0$. 
\end{proof}

Now, we use this result to prove Theorems \ref{thm:dag} and \ref{thm:dag_onto_symmetric}, which we restate here. 

\begin{ithm}
If $G$ is a directed acyclic graph, then no CTLN with graph $G$ has a dynamic attractor. \label{thm:dag}
\end{ithm}

\begin{ithm}
Let $G$ be a graph on vertex set $V$, and $\sigma, \tau$ be a partition of $V$. Suppose  the subgraph $G_\sigma$ is symmetric and connected, the subgraph $G_\tau$ is a directed acyclic graph, and that there are no edges from a vertex in $\tau$ back to a vertex in $\sigma$. 
Then no CTLN with graph $G$ has a dynamic attractor. \label{thm:dag_onto_symmetric}
\end{ithm}

 The first part of our proof is to show that, in the case of a DAG, all neurons other than the sinks have their activity approach zero, and in the case of a symmetric graph, all neurons outside the symmetric part have their activity approach zero. We prove this for a more general family of graphs, which we define below. 

\begin{defn} 
Let $G$ be a graph and $\omega, \tau \subseteq V_G$. We say that $(\omega, \tau)$ is an \emph{inductively stable DAG decomposition} of $G$ if  $\omega \sqcup \tau$ is a partition of the vertices of $G$ satisfying the following inductive definition  
\begin{enumerate}
\item $G|_\omega$ is a DAG
\item $G|_\tau$ contains all sinks of $G$ 
\item there are no edges from $\tau$ back to $\omega$ 
\item every vertex of $\tau$ is reachable along a directed path from a vertex in $\omega$. 
\item for any source $i\in \omega$, $\omega \setminus \{i\}, \tau$ is an inductively stable DAG decomposition of $G|_{V_G\setminus \{i\}}$. 
\end{enumerate}
\end{defn}

\begin{lem}\label{lem:dag_decomp}
Let $G$ be a graph with inductively stable DAG decomposition $( \omega, \tau)$. Then for any $i\in \omega$, $\lim_{t\to \infty}x_i = 0$. 
\end{lem}

\begin{proof}

Let $G$ be a graph with inductively stable DAG decomposition $( \omega, \tau)$. We first prove that every vertex of $G$ is reachable from a source. First, if $i\in \omega$, then a path which follows edges backwards must remain in $\omega$, and must eventually terminate, since $G_\omega$ is a DAG. The vertex where this path terminates must be a source.  Now, by assumption, every vertex in $\tau$ is reachable along a path from a vetex in $\omega$. Thus, from a vertex in $\omega$, we can trace a path backwards into $\tau$, and then back to a source. 

Now, by Theorem \ref{thm:sources_die}, there is some source $i$ of $G$ whose activity must approach zero. Further, for this to be possible, all trajectories of the TLN must approach the $i^{th}$ nullcline. 
Thus we are guaranteed that for each value of $\zeta$, there exists a time $T$ such $x_i(t)\leq \zeta$, and $|\xdot_i(t)| \leq \zeta$ for all $t \geq T$. 
Now, by assumption, $\omega\setminus\{i\}, \tau$ is an inductively stable DAG decomposition of $G|_{V_G\setminus i}$.  
By choosing $\zeta$ small enough, we can satisfy the conditions of Lemmas \ref{lem:smallest_sources_die} and \ref{lem:lyapunov} when restricting to the graph  $G|_{V_G\setminus i}$. 
We can apply this argument inductively to show that for each $i\in \omega$, $\zeta > 0$, there is a time $T$ such that $x_i(t) \leq \zeta$ for all $t\geq T$. Thus, for any $i\in \omega$, $\lim_{t\to \infty}x_i = 0$. 
\end{proof}

Thus, we have shown that whenever $G$ has an inductively stable DAG decomposition $\omega$, $\tau$, all activity of the CTLN on $G$ becomes concentrated on $G|_\tau$. Now, we wish to prove that if $G_|\tau$ is symmetric, then all trajectories of the $CTLN$ on $G$ approach a fixed point. 
In order to prove this, we first prove the following lemma, which allows us to neglect the impact of neurons in $\omega$ upon those of $\tau$ once they have decayed sufficiently. 

\begin{lem}\label{lem:lyapunov_hell}
Let $\xdot = g(x)$ be a dynamical system on a compact set $A$ with a Lyapunov like function $L: A \to \R$ satisfying the following: 

\begin{itemize}
\item $\dot L \leq 0$ for all $x\in A$ 
\item $\dot L = 0$ only on a finite set of fixed points $X = \{x\in A\mid g(x) = 0\}$, 
\item $L$ is bounded below on $A$.  
\end{itemize}

Then for every $\varepsilon > 0$, there exists $\delta > 0$ such that if $\dot x = g'(x, t)$ is a dynamical system satisfying $|g(x) - g'(x, t)| < \delta$, then there exists $T$ such that all trajectories of of the system $\dot x = g'(x, t)$ satisfy $|x(t) - X| < \varepsilon$ for all $t > T.$

\end{lem}

\begin{figure}
\includegraphics[width = 6 in]{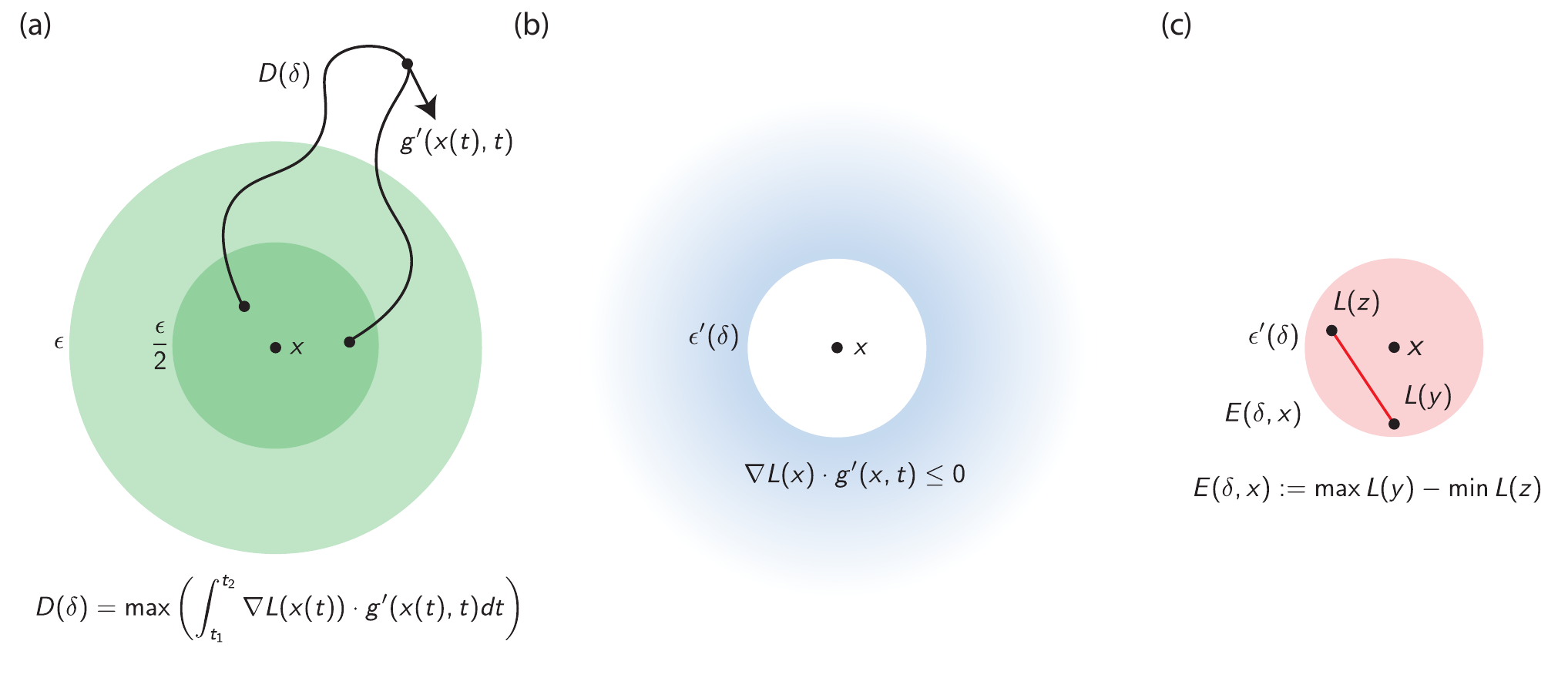}
\caption[Schematic for the proof of Lemma \ref{lem:lyapunov_hell}]{(a) The quantity $D(\delta)$. (b) The quantity $\varepsilon'(\delta)$. (c) The quantity $E(\delta, x)$.  \label{fig:lyapunov}}
\end{figure}

\begin{proof}
We consider three quantities, $D(\delta)$, $\varepsilon'(\delta)$,and $E(\delta)$, represented in Figure \ref{fig:lyapunov}.  First, $D(\delta)$ gives the maximum value of $\int_{t_1}^{t_2} \nabla L(x(t)) \cdot g'(x(t),t)dt$ can assume along trajectories of the system $\xdot = g'(x, t)$  which begin and end in an $\varepsilon/2$ neighborhood of $X$, but leave an $\varepsilon$ neighborhood of $X$ along the way. We show that $D(\delta)$ approaches a negative limit as $\delta \to 0$. Notice that  
\begin{align*}\int_{t_1}^{t_2} \nabla L(x(t)) \cdot g'(x(t),t)dt \leq (t_2 - t_1)  \max_{x \in A, d(x, X) \geq \varepsilon/2} (\nabla L(x(t)) \cdot g(x(t),t)).
\end{align*}
Now, since $\nabla L(x) \cdot g'(x,t)$ attains its negative maximum value $\alpha < 0$ on the compact set $\{x \in A\mid d(x, X) \geq \varepsilon\}$ , and  $g'(x, t) \to g(x)$ as $ \delta\to 0$, we must have that  $\max_{x \in A, d(x, X) \geq \varepsilon/2} (\nabla L(x) \cdot g'(x,t)) \to \alpha < 0$ as $\delta \to 0$.  Further, any trajectory which begins in an $\varepsilon/2$ neighborhood of $X$, leaves an $\varepsilon$ neighborhood of $X$, and returns to an  an $\varepsilon/2$ neighborhood of $X$ must have length at least $\varepsilon$, provided we have chosen $\varepsilon$ small enough that all fixed points in $X$ are separated by a distance of at most $\varepsilon$.  

Finally, since $g(x)$ is bounded above, so is $g'(x, t)$. Thus for a trajectory of length $\varepsilon$, we must have $t_2 - t_1 \geq \varepsilon / \max_{x\in A}(|g'(x,t)|)$. This must approach the positive limit $t_2 - t_1 \geq \varepsilon / \max_{x\in A}(|g(x)|)$ as $\delta\to 0$. Thus, the product $(t_2 - t_1)  \max_{x \in A, d(x, X) \geq \varepsilon/2} (\nabla L(x) \cdot g'(x,t))$ approaches a negative limit at $\delta \to 0$. Thus, so does $D(\delta)$. 

Now, we consider $\varepsilon'(\delta)$, which we define as the maximum value of $\varepsilon$ such that $ \nabla L(x) \cdot g'(x,t) < 0$ for all $x$ such that $|x - X| \geq \varepsilon'$. Again, we have that $\varepsilon' \to 0$ as $\delta \to 0$, since $ \nabla L(x) \cdot g'(x,t)\to \nabla L(x) \cdot g(x)$, and  $\nabla L(x) \cdot g(x) < 0$ for all $x\notin X$. Finally, we define 
\begin{align*}E(\delta, x):=\max_{y \in B_{\varepsilon'(\delta)}(x)} L(y) - \min_{z \in B_{\varepsilon'(\delta)}(x)}  L(z) 
\end{align*}
 Notice that  as $\delta \to 0$, $E(\delta, x) \to 0$ by continuity of $L$. Now, let $$E(\delta) := \max_{x \in X} E(\delta, x).$$ Since we are taking this maximum over a finite set, we also have that $E(\delta) \to 0$ as $\delta \to 0$. 
 
Thus, for $\delta$ small enough, we have $|E(\delta)| < |D(\delta)|$. When this is achieved, we claim that it is impossible for a trajectory to start within $B_{\varepsilon'}(X)$, leave  $B_{\varepsilon}(X)$, and return to $B_{\varepsilon'}(X)$. We choose $\delta$ such that $\varepsilon' < \varepsilon/2$. Then along any trajectory starts from and returns from  $B_{\varepsilon'}(X)$, the value of $L$ must decrease along the trajectory by at least $D(\delta)$. However, since the most the value of $L$ can change within $B_{\varepsilon'}(x)$ is $E(\delta) < D(\delta)$, this is not possible. 

Finally, we show that this means that there exists $T$ such that all trajectories of of the system $\dot x = g'(x, t)$ satisfy $|x(t) - X| < \varepsilon$ for all $t > T.$ By our definition of $\varepsilon'(\delta)$, we have that $ \nabla L(x) \cdot g'(x,t) < 0$ outside of $B_{\varepsilon'}(X)$. Thus, since $L$ is bounded below, all trajectories must visit the set $B_{\varepsilon'}(X)$. Now, for each particular fixed point $x^*\in X$, it is not possible for a trajectory to leave the set $B_\varepsilon(x^*)$ and return to $B_{\varepsilon'}(x^*)$. Thus, since there are only finitely many fixed points $x^*$, our trajectory must eventually get stuck in  $B_\varepsilon(x^*)$ for one of them.

\end{proof}

\begin{proof}[Proof of Theorem \ref{thm:dag}]
Let $G$ be a DAG. Let $\tau$ be the set of all sinks of $G$, and $\omega$ be $V_G\setminus\tau$. Notice that this is an inductively stable DAG decomposition of $G$. Then by Lemma \ref{lem:dag_decomp}, $\lim_{t\to \infty}x_i(t) = 0$ for all $i\in \omega$. 
Now, since $G|_\tau$ is symmetric, it has  a Lyapunov-like function $L$ by \cite{hahnloser2000permitted}. 
Therefore, by considering the effect of $x_i, i\in \omega$ as a time varying external input on $x_j, j\in \tau$, by our Lemma \ref{lem:lyapunov_hell}, for each $\eta > 0$, we can find a time $T'$ such that $\{x_j\}_{j\in \tau}$  is contained within an $\eta$ neighborhood of a fixed point of the CTLN defined by $G|_\tau$. 
Thus, as $T\to \infty$, the trajectory approaches a fixed point whose support is contained within $\tau$. 
\end{proof}

\begin{proof}[Proof of Theorem \ref{thm:dag_onto_symmetric}]
Let $G$ be a graph and $\omega, \tau$ a partition of the vertices of $G$ such that $G|_\tau$ is symmetric and connected, $G|_\omega$ is a DAG, and there are no edges from $\tau$ to $\omega$. Then by Lemma \ref{lem:dag_decomp}, $\lim_{t\to \infty}x_i(t) = 0$ for all $i\in \omega$. 
Now, since $G|_\tau$ is symmetric, it has  a Lyapunov-like function $L$ by \cite{hahnloser2000permitted}. 
Therefore, by considering the effect of $x_i, i\in \omega$ as a time varying external input on $x_j, j\in \tau$, by our Lemma \ref{lem:lyapunov_hell}, for each $\eta > 0$, we can find a time $T'$ such that $\{x_j\}_{j\in \tau}$  is contained within an $\eta$ neighborhood of a fixed point of the CTLN defined by $G|_\tau$. 
Thus, as $T\to \infty$, the trajectory approaches a fixed point whose support is contained within $\tau$. 
\end{proof}

\section{Classification of 3 neuron CTLNs}

We close by using the results of the previous section to show that, among CTLNs with 3 neurons, only the directed 3 cycle has a dynamic attractor, which is a limit cycle. 

\begin{ithm}
The three-cycle is the only three neuron CTLN with persistent dynamic activity. 
\end{ithm}

We split this work into several cases, shown in Figure \ref{fig:n_3_class}.

\begin{figure}[ht!]
\begin{center}
\includegraphics[width = 5 in]{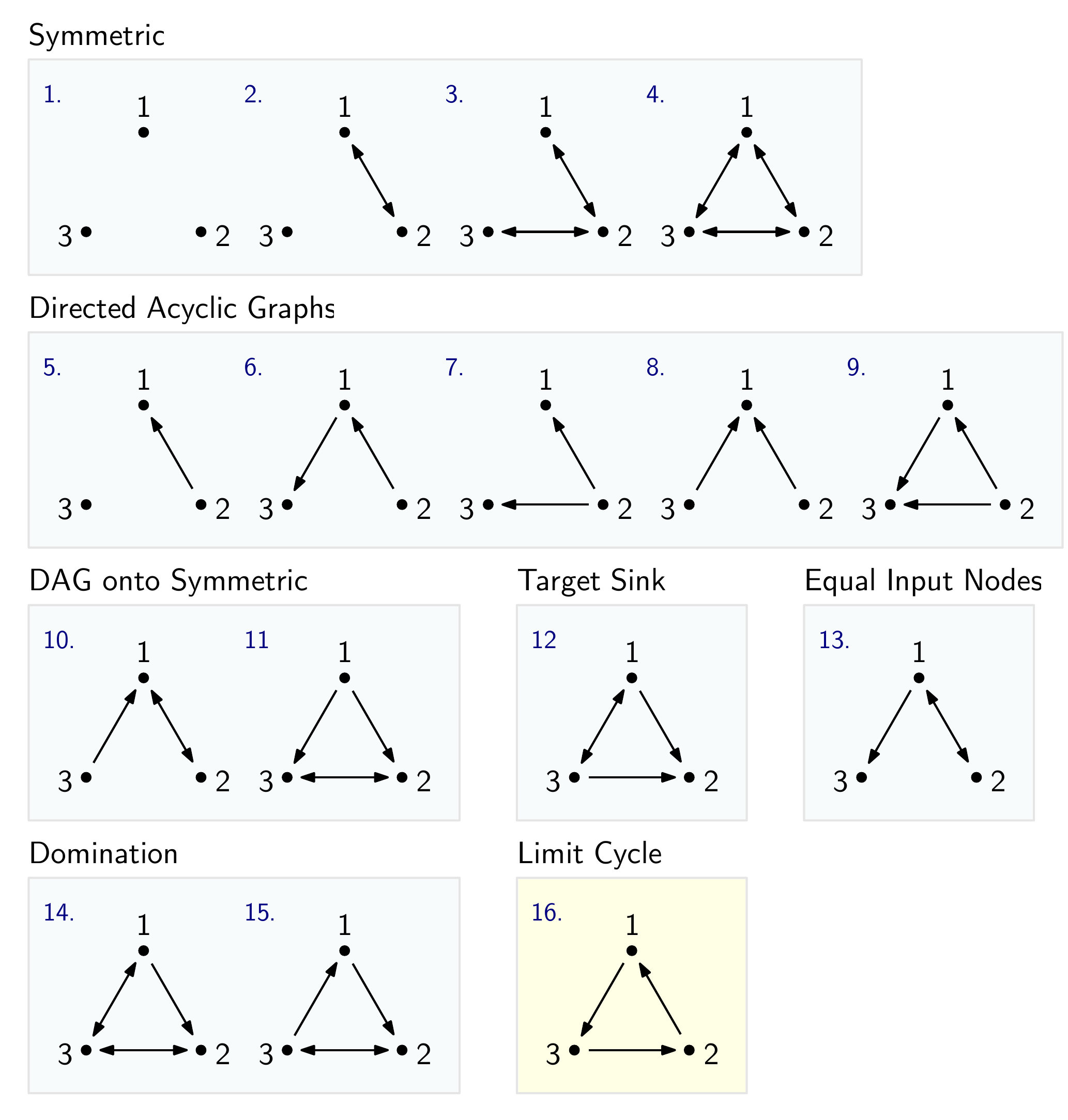}
\end{center}
\caption[Our classification of 3 neuron graphs.]{Our classification of 3 neuron graphs. Only Graph 16, the directed 3-cycle (in a yellow box), has a dynamic attractor, which is a limit cycle. \label{fig:n_3_class}}
\end{figure}

\begin{itemize}
\item \textbf{Case 1: Symmetric graphs.} Graphs 1-4 are symmetric. Thus, by \cite{hahnloser2000permitted}, all trajectories of their CTLNs approach stable fixed points. 
\item 
\textbf{Case 2: Directed acyclic graphs}
Graphs 5-9 are directed acyclic graphs. 
Thus, by Theorem \ref{thm:dag}, all trajectories of their CTLNs approach stable fixed points. 

\item \textbf{Case 3: DAG onto symmetric}
Graphs 10 and 11 can be decomposed as the union of a directed acyclic graph and a connected symmetric graph, such that every vertex of the symmetric part is reachable from the DAG part. 
Thus, by Theorem \ref{thm:dag_onto_symmetric}, all trajectories of their CTLNs approach stable fixed points.
 
\item \textbf{Case 4: Target Sink}
In graph 12, neuron 2 is a proper sink and receives an edge from every other neuron. 
Thus by Proposition \ref{prop:supersink}, all of its trajectories approach a stable fixed point. 

\begin{figure}[ht!]
\begin{center}
\includegraphics[width = 5 in]{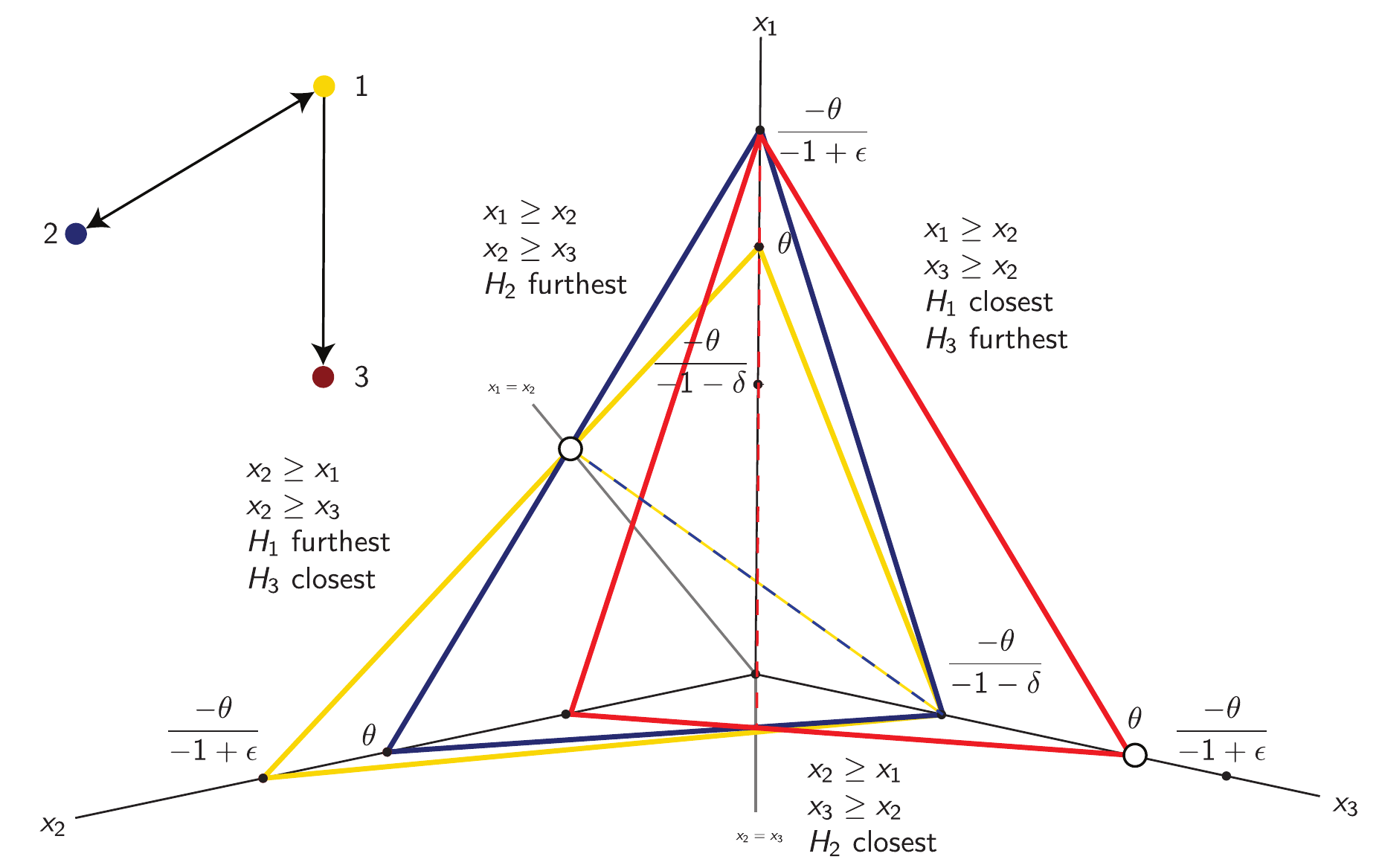}
\end{center}
\caption{Nullcline arrangement for graph 13. \label{fig:eq_input}}
\end{figure}

\item \textbf{Case 5: Equal input nodes}
In Graph 13, the pairs of vertices $\{1, 2\}$ and $\{2, 3\}$ receive equal input. Thus, by Corollary \ref{cor:equal_input}, the ordering of $x_1, x_2$ and $x_2, x_3$ is fixed by the initial conditions. 
Now, we consider the chambers defined by the hyperplanes $x_1 = x_2$, $x_2 = x_3$.  
We examine the nullcline arrangement within each chamber and show that no chamber can support a dynamic attractor. See Figure \ref{fig:eq_input} for an illustration. 

\begin{itemize}

\item 
Within the chamber where $x_2 \geq x_1$ and $x_3 \geq x_2$, the nullcline $H_2$ is closest to the origin. 

\item Within the chamber $x_2 \geq x_1$ and $x_2 \geq x_3$ the nullcline $H_1$ is furthest from the origin.

\item Within the chamber where $x_1 \geq x_2$ and $x_3 \geq x_2$,the hyperplane $H_3$ is furthest from the origin, and the hyperplane $H_1$ is closest to the origin. 

\item Within the chamber where $x_1 \geq x_2$ and $x_2 \geq x_3$, the hyperplane $H_2$ is furthest from the origin. 
\end{itemize}
Thus, within each chamber defined by the hyperplanes $x_1 = x_2$, $x_2 = x_3$, the sign of at least one derivative $x_i'$ is fixed. Thus, the value of $x_i$ approaches a limit. Thus, the derivative $x_i'$ approaches zero, so trajectories approach the $i^{th}$ nullcline.  If they approach $E_i$, then by Proposition \ref{prop:n=2}, activity of this CTLN converges to a stable fixed point. 
Now, we can't approach anything other than a fixed point on $H_i$ because $H_i$ is repelling in this chamber. 

\item \textbf{Case 6: Domination relationships}
We deal with graphs 14 and 15 separately. Their nullcline arrangements are illustrated in Figure \ref{fig:dom_cases}. 
\begin{itemize}
\item We first consider graph 14. Notice that $x_2$ strongly dominates $x_1$ and that $x_2$ and $x_3$ are equal input. Thus, the ordering of $x_2$ and $x_3$ is fixed by the initial conditions. Now, on the $x_2\leq x_3$ side of the hyperplane defined by $x_2 = x_3$, the nullcline $H_2$ is furthest from the origin. On the $x_2 \geq x_3$ side, the nullcline $H_1$ is closest to the origin and the nullcline $H_3$ is furthest from the origin. Thus, by Lemma \ref{lem:single_chamber}, dynamics within each chamber approaches a stable fixed point. 
\item 
We next consider Graph 15. Notice that $x_2$ strongly dominates $x_1$ and that $x_2$ weakly dominates $x_3$  Thus, the ordering of $x_2$ and $x_3$ is fixed by the initial conditions. And, eventually, $x_2$ becomes greater than $x_1$.  

Thus, we restrict our analysis to the side of the $x_1= x_2$ hyperplane where $x_1 \leq x_2$. 
Now, on the $x_2\leq x_3$ side of the hyperplane defined by $x_2 = x_3$, the nullcline $H_2$ is furthest from the origin. On the $x_2 \geq x_3$ side, the nullcline $H_1$ is closest to the origin. Thus, by Lemma \ref{lem:single_chamber}, dynamics within each chamber approaches a stable fixed point. 
\end{itemize}

\end{itemize}

\begin{figure}[ht!]
\begin{center}
\includegraphics[width = 5 in]{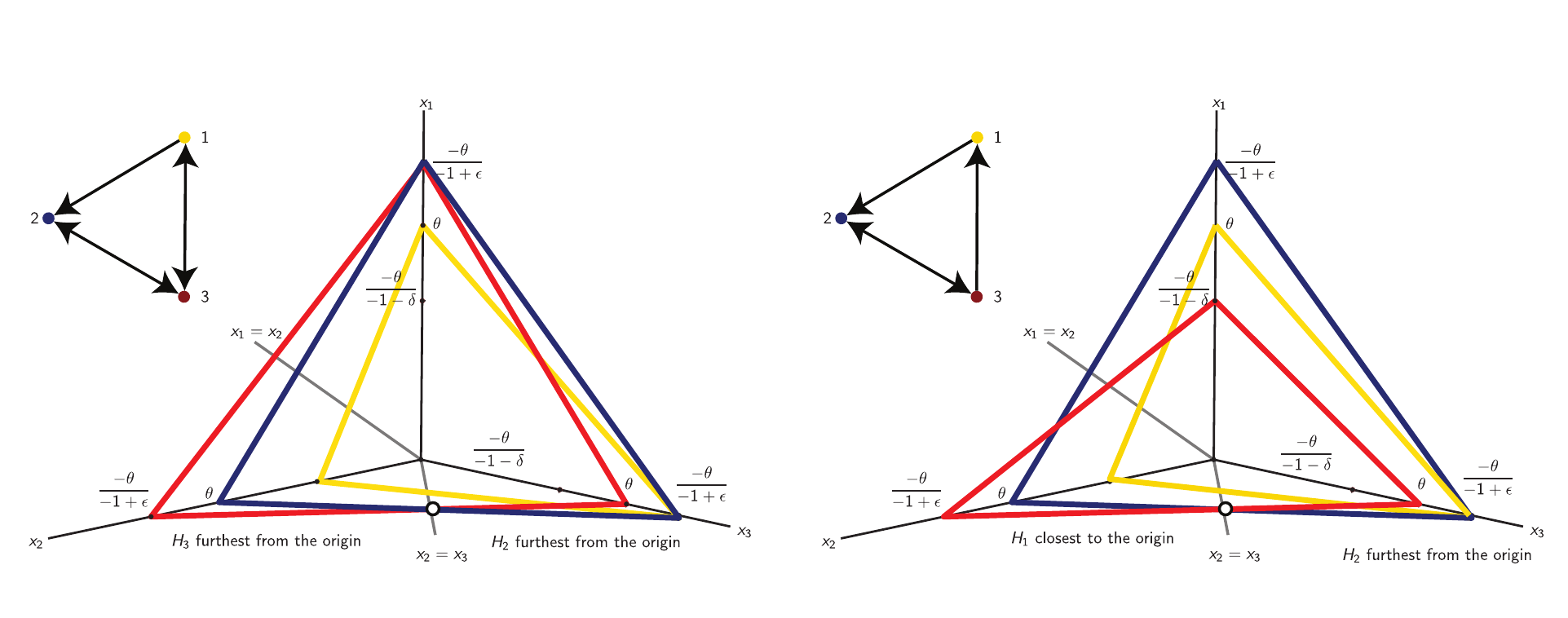}
\end{center}
\caption{Nullcline arrangement for graphs 14 and 15. \label{fig:dom_cases}}
\end{figure}

\section{Main results: competitive TLNs}
In \cite{curto2019robust}, Curto, Langdon, and Morrison relax the strict CTLN conditions, and investigate how and when the graph $G_W$ constrains the dynamics of a competitive TLN. In particular, they say that a graph $G$ is a \emph{robust motif} if $G$ fully determines $FP(W)$ for all weight matrices $W$ such that $G_W = G$. Their Theorem 1.6 gives a complete characterization of which graphs $G$ are robust motifs. 

\begin{thm*}[Theorem 1.6, \cite{curto2019robust}]
 $G$ is a robust motif if and only if one of the following holds:
 \begin{enumerate}
\item $G$ belongs to one of the infinite families DAG1 or DAG2, or
\item  $G$ is one of the four invariant permitted motifs (the singleton, the 2-clique, the independent set of size 2, or the 3-cycle).
\end{enumerate}
\end{thm*}

In particular, TLNs in the infinite families DAG1 and DAG2 each have a single fixed point, which is stable. However, the fixed point structure of a dynamical system does not fully determine its activity: in principle, a dynamical system with a unique fixed point which is stable may also have a limit cycle or a chaotic attractor. Thus, we would like to show that the families DAG1 and DAG2 are robust in the following sense:

\begin{conj}\label{conj:dyn_robust}
Let $(W, \theta)$ define a TLN such that $G_W$ is a member of DAG1 or DAG2. Then the unique fixed point is always a global attractor.
\end{conj}

We are not able to prove this conjecture for the full families DAG1 or DAG2. However, we define a smaller infinite family of graphs, \emph{shallow DAG1}, and show that graphs in this infinite family are dynamically robust in this sense.

\begin{defn}
Let $G$ be a member of DAG1 with source node $s$ and target node $u$. We say $G$ is a member of \emph{shallow DAG1} if  it satisfies the additional property that if $i\to j$, then $i = s$ or $j = u$. 
 That is, if a node $j$ is not the target, then it may only receive input from the source $s$. See Figure \ref{fig:dag1}. 
\end{defn}

\begin{figure}[ht!]
\begin{center}
\includegraphics{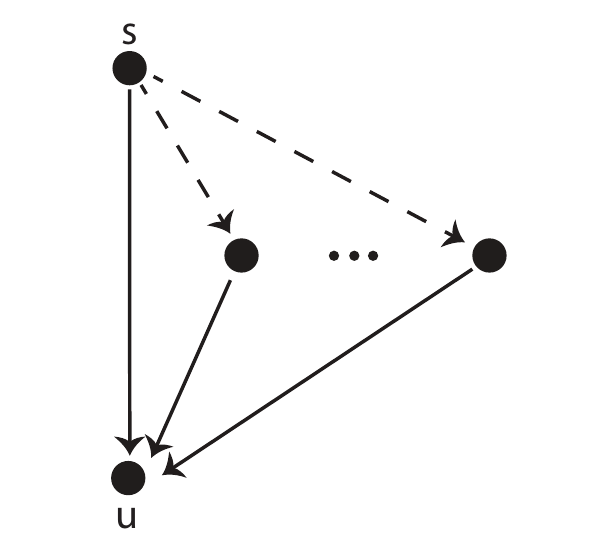}
\caption{A graph is a member of shallow DAG1 if only the target receives from a non-source vertex.
\label{fig:dag1}}
\end{center}
\end{figure}

We prove Conjecture \ref{conj:dyn_robust} in the special case where $G$ is a member of shallow DAG1. 

\begin{ithm}\label{thm:dag1global}
Let $G$ be a member of shallow DAG 1, and  $(W, \theta)$ be a TLN with  $G_W = G$.  This TLN has a unique fixed point which is a global attractor. 
\end{ithm}

We also prove a second result, that if $G$ is any member of DAG1 and the coefficients are close enough to the CTLN coefficients, then all trajectories of the CTLN approach a fixed point. 

\begin{ithm}\label{prop:near_ctln}
Let $G$ be a member of DAG1. There is an open set in $\R^{n\times n}$ containing the CTLN coefficients for $G$ such that if the weight matrix $W$ falls in this set, the TLN has a unique fixed point which is a global attractor. 
\end{ithm}

\subsection{Shallow DAG1}

In order to prove Theorem \ref{thm:dag1global},  we generalize shallow DAG1 to a larger family of TLNs using strong domination. In order to do this, we need a notion of strong domination with respect to a set $\sigma\subseteq[n]$.

\begin{defn}
Neuron $k$ strongly dominates neuron $j$ with respect to $\sigma \subseteq [n]$ if $\wt W_{ji} \leq \wt W_{ki}$ for all $i\in \sigma$. 
\end{defn}

\begin{defn} Let $W$ be the weight matrix of a TLN. We say that $W$ satisfies the \emph{shallow DAG1 domination property} if there exist neurons $s$ and $u$ such that the following conditions hold:
\begin{enumerate}
\item  $u$ strongly dominates $s$ with respect to $[n]$
\item  for all $i\notin\{s, u\}$,  $u$ strongly dominates $i$ with respect to $[n]\setminus \{s\}$. 
\end{enumerate}

\end{defn}

\begin{prop}
If $G$ is a member of shallow DAG1, and $G_W = G$, then $W$ satisfies the shallow DAG1 domination property. 
Further, suppose that $G$ is a graph such that $W$ satisfies the shallow DAG1 property whenever $G_W = G$. Then $G$ is a member of shallow DAG1. 
\end{prop}

\begin{proof}
First, suppose $G$ is a member of shallow DAG1. Because $s$ is a source and $u$ is a target, then $W_{ui} >-1$ for all $i\neq u$, and $W_{si} < -1$ for all $i\neq s$. Thus $\wt W_{ui} > \wt W_{si}$ for all $i$, thus $u$ strongly dominates $s$. 
Now, because if $i\to j$, either $i = s$ or $j = u$, we have $W_{ij} < -1$ for all $j\neq s, i$ and $W_{uj} > -1$ for all $j\neq u$. Thus, $\wt W_{uj} > \wt W_{ij}$. 

Now, we show that if $G$ is not a member of shallow DAG1, there is some weight matrix $W$ such that $G = G_W$, but $W$ does not satisfy the shallow DAG1 domination property. If $s$ is not a source, then there is some $i$ such that $i\to s$. Then we can choose $W_{ui}$ such that $W_{ui} <W_{si}$.  Likewise, if $u$ is not a target,  then there is some $i$ such that $i\not \to u$. Then we can choose $W_{ui} < W_{si}$. In both of these cases, we violate the condition that $u$ strongly dominates $s$. 

Now, suppose that there is some $i\neq s$, $j\neq s$ such that $i\to j$. Then $W_{ji} > -1$. Then we can choose $W_{ji} > W_{ui}$, which means that $u$ does not strongly dominate $j$ with respect to $[n]\setminus \{s\}$. 
\end{proof}

We will prove that if a TLN satisfies the shallow DAG1 domination property, it has a unique fixed point which is globally attracting. This means if the graph  $G$ is a member of the shallow DAG1 family, then any TLN with $G_W = G$ has a unique fixed point which is globally attracting. Thus, these graphs are dynamically robust in the sense of Conjecture \ref{conj:dyn_robust}. 

\begin{prop}\label{prop:dag1domglobal}
Let $W$ be the weight matrix of a TLN satisfying the shallow DAG1 domination property. 
Then all activity of the TLN defined by $W$ approaches a stable fixed point at which $u$ is the only active neuron. 
\end{prop}

Notice that this implies Theorem \ref{thm:dag1global}, since all TLNs whose graphs are in DAG1 satisfy the shallow DAG1 domination property. 

%We begin with a simple lemma about strong domination. 
%
%\begin{lem}  \label{lem:strong_domination} 
%If $k$ strongly dominates $j$ with respect to $\sigma$, then 
%$$h_k^*(x) - h_j^*(x) > 0$$
%for all 
%$
%x\in \{\R^n \mid 
% x_i = 0 \mbox{ if } i\notin \sigma, 
%x_i \geq 0 \mbox{ if } i\in \sigma
%\}\setminus\{0\}
%$
%In particular, when $\sigma = [n]$, $h_k^*(x) > h_j^*(x)$ on the entire positive orthant, excluding the origin. 
%\end{lem}
%
%\begin{proof}
%We have 
%
%\begin{align*}
%h_k^*(x) - h_j^*(x)  &= \left(-x_k+ \sum_{i=1}^n W_{ki}x_i  + \theta \right) - \left(-x_j + \sum_{i=1}^n W_{ji}x_{i} + \theta \right) \\
%h_k^*(x) - h_j^*(x)  &= \sum_{i = 1}^n \wt W_{ki}x_i -  \sum_{i = 1}^n \wt W_{ji} x_i  = \sum_{i = 1}^n  (\wt W_{ki} - \wt W_{ji} )x_i 
%\end{align*}
%By the definition of strong domination, all coefficients $\wt W_{ki} - \wt W_{ji} $ of this sum are positive for $i\in \sigma$. Thus  when $ x\in \{\R^n \mid 
% x_i = 0 \mbox{ if } i\notin \sigma, 
%x_i \geq 0 \mbox{ if } i\in \sigma
%\}\setminus\{0\}$, 
%we have $h_k^*(x) - h_j^*(x) < 0 $. 
%
%\end{proof}
%
%This lemma has the following geometric interpretation: if $k$ strongly dominates $j$ with respect to $[n]$, then $H_j \cap H_k \cap \R^n_{> 0} = \emptyset$. Further, $H_k$ is separated from the origin by $H_j$. If  $k$ strongly dominates $j$ with respect to $\sigma$, then $H_k$ is separated from the origin by $H_j$ within the surface where $x_i$ is set to zero for all $i\in [n]\setminus \sigma$.  

\begin{lem}\label{lem:strong_dom_deriv}
If $k$ strongly dominates $j$, then either:
\begin{align*}
\xdot_k  > \xdot_j 
\end{align*}
or

\begin{align*}
y_j < 0,\,\, \xdot_j = -x_j
\end{align*}
\end{lem}

\begin{proof}
We have 
\begin{align*}
\xdot_k - \xdot_j &= -x_k +x_j + [y_k]_+ - [y_j]+\\
\xdot_k - \xdot_j &= h_k^*(x) - h_j^*(x) + ([y_k]_+ -y_k) - ([y_j]_+- y_j) 
\end{align*}
By Observation \ref{obs:strong_domination}, the term $h_k^*(x) - h_j^*(x)$ is always positive. By the definition of the threshold nonlinearity, the term $[y_k]_+ -y_k$ is always positive. Thus, we can only have $\xdot_k - \xdot_j \leq 0$ when the term $[y_j]_+- y_j$  is positive, which occurs when $y_j < 0$, $\xdot_j = -x_j$. 
\end{proof}

The next lemma is main part of the proof of Proposition \ref{prop:dag1domglobal}. 
Essentially, we are showing that for TLNs which satisfy the shallow DAG1 domination property, the function $x_u - x_s$ is a Lyapunov-like function within $\cA$. 
We will use this to show that trajectories which enter $\cA$ must approach a fixed point. 

\begin{lem}\label{lem:lyapanov}
If a TLN satisfies the shallow DAG1 domination property, then $\xdot_u - \xdot_s >  0$ on the interior of  $\cA$, with equality achieved on the boundary of $\cA$ on the set where $\xdot_u = \xdot_s = 0$. 
\end{lem}

\begin{figure}[ht!]
\begin{center}
\includegraphics[width = 4 in]{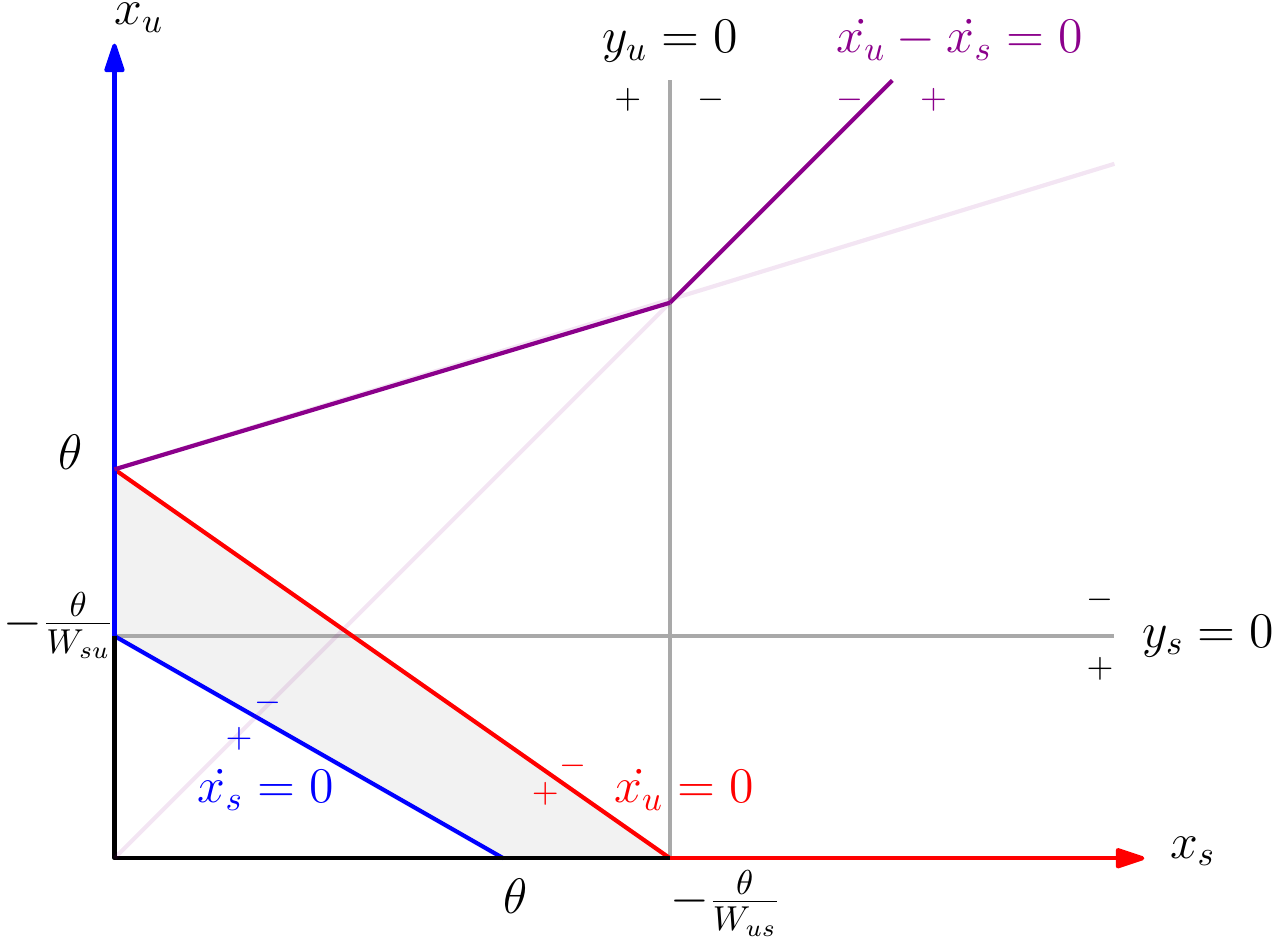}
\end{center}
\caption[Linear chambers for graphs satisfying the DAG1 domination property. ]{Linear chambers and nullclines restricted to the $x_s, x_u$ plane for graphs satisfying the DAG1 domination property. . Nullclines are shown in red and blue, boundaries of linear chambers are shown in gray. The surface $\dot{x_s} = \dot{x_u}$ is shown in green.  \label{fig:cham}   }
\end{figure}

\begin{proof}
We prove this statement by considering the surface where $\xdot_u = \xdot_s$, and show that it does not intersect with the interior of $\cA$. Because the equations for $\xdot_s$ and $\xdot_u$ are threshold-linear, the surface  $\xdot_u = \xdot_s$ consists of at most four hyperplanes, one in each linear chamber. The surface $\dot{x_s} = \dot{x_u}$ and the linear chambers are illustrated in two dimensions in Figure \ref{fig:cham}. 

We first consider the case where $y_s\geq 0$, which covers the linear chambers $y_s \geq 0, y_u \geq 0$ and $y_s \geq 0, y_u< 0$. By Lemma \ref{lem:strong_dom_deriv}, we have $\xdot_u > \xdot_s$. 
%We first show that the surface  $\xdot_u = \xdot_s$  does not actually intersect with the chambers $\sum_{i=1}^n W_{ui}x_i + \theta\geq 0,     \,\, \sum_{i=1}^n W_{si}x_i + \theta \geq 0$ or $\sum_{i=1}^n W_{ui}x_i + \theta < 0,     \,\, \sum_{i=1}^n W_{si}x_i + \theta \geq 0$. 
%
%\begin{itemize}
% 
% \item 
% %$++$: 
% 
% $\sum_{i=1}^n W_{ui}x_i + \theta\geq 0,     \,\, \sum_{i=1}^n W_{si}x_i + \theta \geq 0$. \\ 
% In this case, the threshold nonlinearities have no effect, and 
% \begin{align*}
%\xdot_u- \xdot_s &= h_u^*(x) - h_s^*(x)%\xdot_u - \xdot_s &= (-1 - W_{su})x_u +(1+W_{us})x_s +  \sum_{i\neq s,u}^n (W_{ui}-W_{si})x_i 
% \end{align*}
%Because $u$ strongly dominates $s$, by  Lemma \ref{lem:strong_domination} this expression is always positive for all $x\in \R^n_\geq \setminus \{0\}$.   Thus, $\xdot_u-\xdot_s > 0$  within this chamber. 
%
%\item 
%% $-+$: 
%$\sum_{i=1}^n W_{ui}x_i + \theta < 0,     \,\, \sum_{i=1}^n W_{si}x_i + \theta \geq 0$. \\
% In this case, only the threshold nonlinearity for $u$ has an effect, so 
% \begin{align*}
%\xdot_u- \xdot_s &= x_s - x_u  - \sum_{i=1}^n W_{si}x_i -\theta\\
%\xdot_u- \xdot_s&= h_u^*(x) - h_s^*(x)- \left(\sum_{i=1}^n W_{ui}x_i -\theta\right)
% \end{align*}
%Again by Lemma \ref{lem:strong_domination}, the sum of the first four terms  is positive throughout the positive orthant. Where this case applies, the quantity  $\sum_{i=1}^n W_{ui}x_i -\theta$ is negative. 
%Thus $\xdot_u -\xdot_s > 0$ within this chamber as well. 
% \end{itemize}
% 
% 
Thus, the surface  $\xdot_u = \xdot_s$  does not actually intersect either of these chambers or their closures. 

We next consider the region where $y_s < 0$, which covers the chambers $y_s < 0, y_u \geq 0$ and $y_s < 0, y_u < 0$. We let $b\neq s, u$ be some bystander neuron.  We show that the surface where $\xdot_u = \xdot_s$ is strictly confined to the negative side of the $x_b$ nullcline $H_b$. To do this, we consider the quantity $\xdot_u - \xdot_s - \xdot_b$, and show that  $\xdot_u - \xdot_s - \xdot_b > 0$ when  $\xdot_b> 0$. 

Notice that when $\xdot_b > 0$, we must have $y_b > 0$. Thus, 
\begin{align*}
 \xdot_u - \xdot_s - \xdot_b &= -x_u + x_s + x_b + [y_u]_+-[y_s]_+- [y_b]\\
 \xdot_u - \xdot_s - \xdot_b &= -x_u + x_s + x_b + [y_u]_+- y_b\\
  \xdot_u - \xdot_s - \xdot_b &\geq -x_u + x_s + x_b + y_u- y_b\
 \end{align*}
 Now, we expand this to 
    \begin{align*}
\xdot_u-\xdot_s -\xdot_b &= x_b + x_s - x_u + \left(\sum_{i=1}^n W_{ui}x_i +\theta\right)  -\left( \sum_{i=1}^n W_{bi}x_i + \theta\right)\\
\xdot_u-\xdot_s -\xdot_b &= (1 -W_{bs}+W_{us})x_s +(-1-W_{bu})x_u + (1 +W_{ub})x_b + \sum_{i\neq u, s,b}^n (W_{ui}-W_{bi})x_i 
\end{align*}

The coefficients $  (1 -W_{bs}+W_{us})$, $(-1-W_{bu})$, and  $(1 +W_{ub})$  are positive because $u$ strongly dominates $s$ with respect to $[n]$.  The coefficients $(W_{ui}-W_{bi})$ are positive because $b$ strongly dominates $u$ with respect to $[n]\setminus \{s\}$.  Thus, this quantity is positive within the positive orthant.  Thus, when $\xdot_b > 0$, we must also have $\xdot_u - \xdot_s > 0$. This means that the surface where $\xdot_u = \xdot_s$ is contained strictly within the negative side of  $H_b$.

Finally, by Lemma \ref{lem:strong_dom_deriv}, if $\xdot_u = \xdot_s$, then $\xdot_s = -x_s\leq 0$, so we must have $\xdot_u \leq 0$. This means that the surface $\xdot_u = \xdot_s$ is confined to the negative side of $H_u$, but this might not be strict. 
 
Since our bystander neuron was arbitrary, we have shown that within the mixed-sign chamber $\cA$, $\xdot_u \geq \xdot_s$. On the piece of the boundary of the mixed sign chamber made up by $H_u$, we have $\xdot_u = \xdot_s = 0$ along the surface defined by $x_s = 0, -x_u + \sum_{i = 1}W_{ui}x_i +\theta = 0$. 
\end{proof}

Finally, we are ready to prove Proposition \ref{prop:dag1domglobal}. 
\begin{proof}

By Theorem \ref{thm:mixed_sign}, all trajectories approach $\cA$. First, we consider trajectories which do not ever enter $\cA$. Such trajectories remain in either $\cA^+$ or $\cA^-$ for all time. Then along these trajectories, the signs of each derivative $\xdot_i$ are fixed. Thus, for each $i$, the value of $x_i$ approaches a constant. Thus, these trajectories approach a limit, which must be a fixed point. This must be the unique fixed point of the network. 

Next, we consider trajectories which do enter $\cA$. By Lemma \ref{lem:lyapanov}, once the trajectory has entered $\cA$, it is subject to the constraint $\xdot_u - \xdot_s > 0$. Then this trajectory must approach the limit on the boundary of $\cA$ where $\xdot_u = \xdot_s = 0$. On this set, by Lemma \ref{lem:strong_dom_deriv}, $x_s = 0$. Thus, this trajectory must approach the surface where $x_s = 0$. 

Now, by condition (2) of the shallow DAG1 decomposition property, for all $i\notin \{s, u\}$, we have that $u$ strongly dominates $i$ with respect to $[n]\setminus \{s\}$. Thus, within the plane where $x_s = 0$, for each $i\neq u$, $H_i$ strictly separates $H_u$ from the origin. Because this separation is strict, there must be some $\varepsilon > 0$ such that $H_i$ separates $H_u$ from the origin for all $x\in \R^n_{\geq 0}$ satisfying $0 \leq x_s < \varepsilon$. Because the trajectory approaches this surface where $x_s =0$, after some time, the trajectory must satisfy $x_s < \varepsilon$. At this point, to remain within the mixed-sign region $\cA$, the trajectory must have $\xdot_u > 0$. 

Thus, the trajectory must approach the nullcline for $x_u$. In particular, since we are increasing as we approach the nullcline, we must approach $H_u$. Thus, the trajectory also gets trapped on the negative side of each $H_i$. Therefore, since the signs of all derivatives are fixed, the trajectory must approach a limit, which must be a fixed point.  This must be the unique fixed point of the network. 
 \end{proof}

\subsection{Near CTLNs}

If a CTLN is a member of either DAG1 or DAG2, then by Theorem \ref{thm:dag_onto_symmetric} the unique fixed point is globally attracting. However, this result does not automatically extend to general TLNs. We are able to prove, however, that if a TLN has a graph $G_W$ that is a member of DAG1, and if the weights in $W$ are not too different from the CTLN weights, then the unique fixed point is a global attractor.

\begin{ithm}\label{prop:near_ctln}
Let $G$ be a member of DAG1 with a target vertex $n$. Let $W$ be a weight matrix satisfying $G_W = G$ and  for all $i,k\neq n$, $W_{nk} - W_{ik}> \frac{1}{d_{in}(i)}  (W_{in} + 1)$. Then all trajectories of a TLN defined by $W$ approach the unique fixed point. 
\end{ithm}

Notice that this is a relaxation of the CTLN condition. The quantity $ \frac{1}{d_{in}(i)}  (W_{in} + 1)$ is negative, since $n$ is a sink. Now, in the CTLN case, since $n$ is a target, the quantity $W_{nk} - W_{ik}$ is always non-negative. We relax the CTLN condition by allowing the terms $W_{nk} - W_{ik}$ to be negative, but not \emph{too negative}, relative to 
$\frac{1}{d_{in}(i)}  (W_{in} + 1)$. The condition  $W_{nk} - W_{ik}> \frac{1}{d_{in}(i)}  (W_{in} + 1)$ is met on an open set of the parameter space containing the CTLN parameters. 

Our proof resembles that of Theorem \ref{thm:sources_die}, using local Lyapunov functions for a sequence of nested sets. We illustrate this nested sequence in Figures \ref{fig:nested} and \ref{fig:omegas} . 
\begin{figure}
\includegraphics[width =6 in]{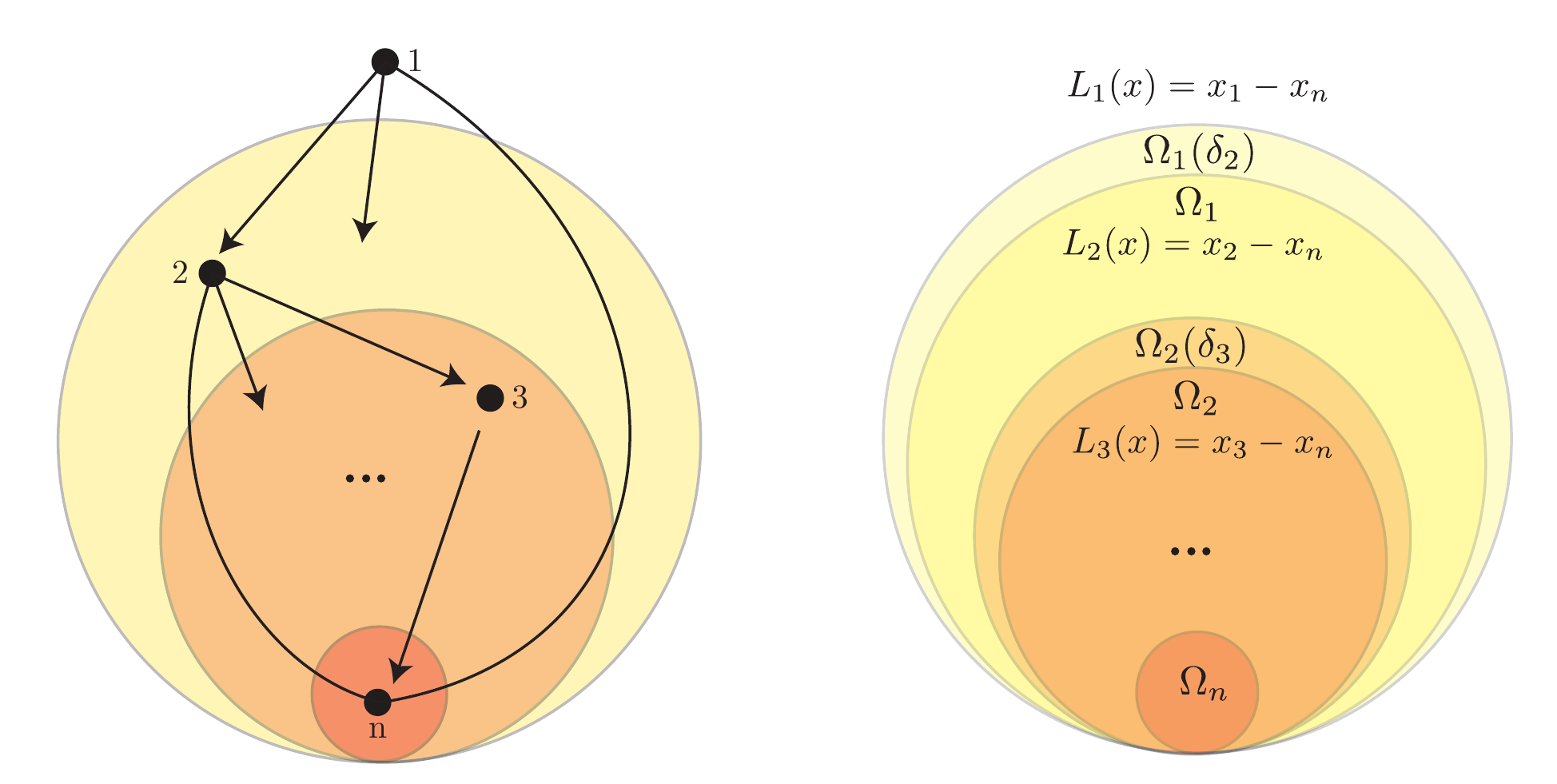}

\caption[Proof idea for Theorem \ref{prop:near_ctln}.]{Proof idea for Proposition \ref{prop:near_ctln}. We label the vertices of $G$ with $1, \ldots, n$ according to a topological order (right).  A sequence of local Lyapunov-like functions  $L_1, \ldots, L_{n-1}$ (left) force trajectories into a nested sequence of regions $\Omega_1 \supset \Omega_2 \supset \cdots \supset \Omega_{n-1}$. \label{fig:nested}}
\end{figure}

\begin{figure}
\begin{center}
\includegraphics[width =4  in]{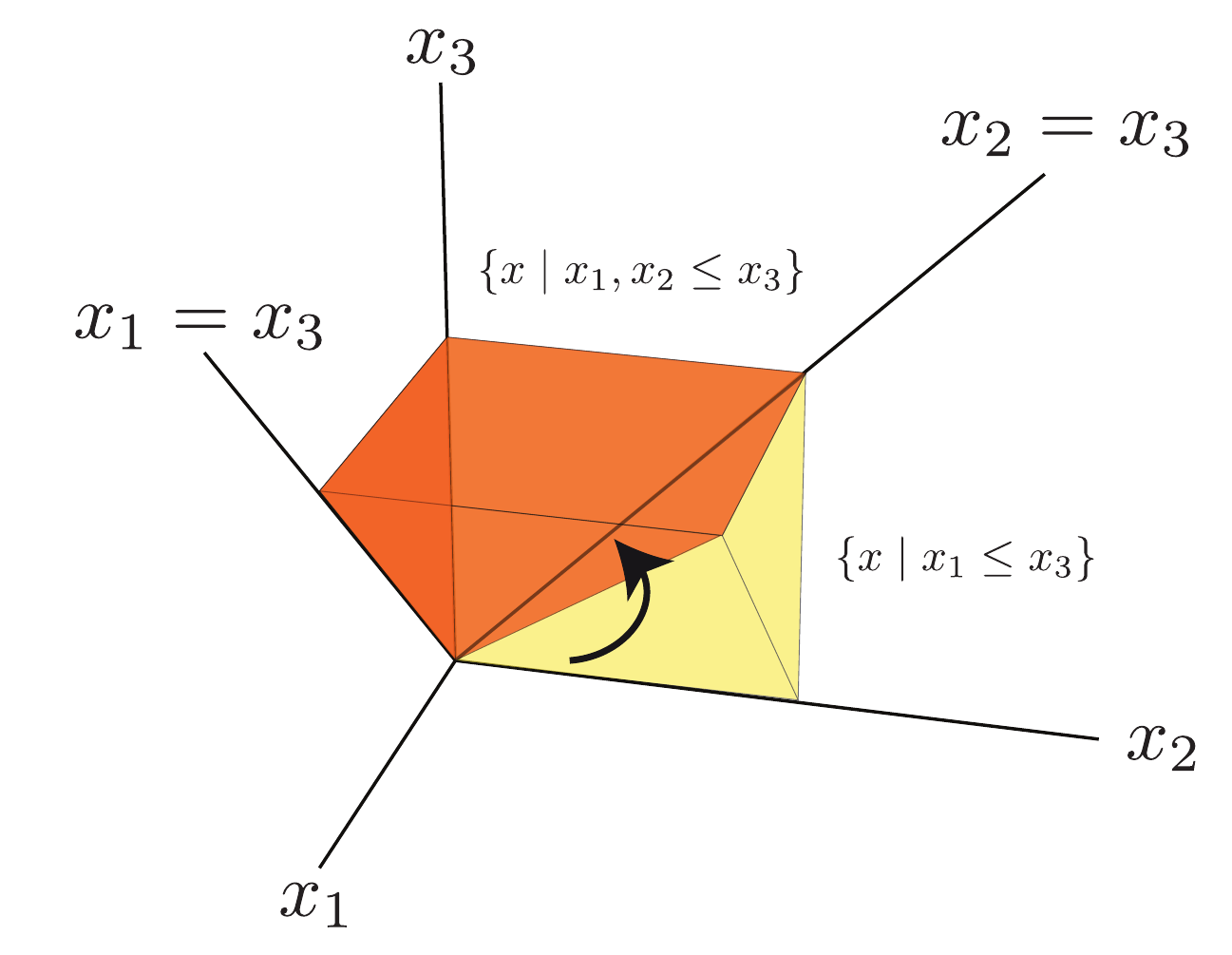}
\end{center}
\caption[Illustration of the regions $\Omega_1$ and $\Omega_2$.]{The regions $\Omega_1$ and $\Omega_2$ for a three-neuron network are obtained by intersecting the cones $\{x\mid x_1 \leq x_3\}$ and $\{x \mid x_1, x_2 \leq x_3\}$, illustrated here, with $\cA$. \label{fig:omegas} }
\end{figure}

\begin{defn}
Fix a topological ordering of the vertices of $G$ such that the source is numbered 1 and the target is numbered  $n$. 
Define 
\begin{align*}
\Omega_i(\varepsilon) :&= \{ x\in \R^n_{\geq 0}\mid x_j \leq x_n +\varepsilon\mbox{ for all } j \leq i\}\cap \cA\\
\Omega_i :&= \Omega_i(0)\\
\Omega_0 :&= \cA\\
\end{align*}
\end{defn}

Now, we will eventually prove that for some $\delta_i$, $L = x_i - x_n$ is a local Lyapunov function for $A = \Omega_i$, $B = \Omega_{i-1}(\delta_i)$. We will need to prove two lemmas first. 

\begin{lem}\label{lem:i_vs_n}
Under the assumptions of Theorem  \ref{prop:near_ctln},
there exists a constant $\delta_i := \delta_i(W) > 0$ such that if $x\in \Omega_{i-1}(\delta_i)$, then $h_i^*(x) - h_n^*(x) < 0$. 
\end{lem}

\begin{proof}
We have 
\begin{align*}
  h_i^*(x) - h_n^*(x) = -(1 + W_{ni})x_i + (1 + W_{in})x_n + \sum_{k\to i} (W_{ik} - W_{nk})x_k + \sum_{j\not\to i}(W_{ij}-W_{nj})x_j .
\end{align*}
Notice that except for those of $ \sum_{k\to i} (W_{ik} - W_{nk})x_k$, all coefficients of this sum are non-positive.  Thus 

\begin{align*}
  h_i^*(x) - h_n^*(x) \leq  (1 + W_{in})x_n + \sum_{k\to i} (W_{ik} - W_{nk})x_k  + C,
\end{align*}

where $C = -(1 + W_{ni})x_i + \sum_{j\not\to i}(W_{ij}-W_{nj})x_j $. 

Let $a = \max_{k\to i}(W_{ik}-W_{nk})$. If $a\leq0$, then all terms of the sum are non-positive, thus we have   $h_i^*(x) - h_n^*(x) \leq 0$ throughout the positive orthant. Otherwise, we use the assumption that $x\in \Omega_{i-1}(\delta_i )$ to show that 
\begin{align*}
  h_i^*(x) - h_n^*(x) \leq (1 + W_{in})x_n + \sum_{k\to i} a x_k + C \\
  h_i^*(x) - h_n^*(x) \leq (1 + W_{in})x_n + \sum_{k\to i} a (x_n + \delta_i) +C \\
  h_i^*(x) - h_n^*(x) \leq (1 + W_{in} +d_{in}(i)a)x_n + d_{in}(i)a\delta_i + C 
\end{align*}
The term $(1 + W_{in} +d_{in}(i)a)x_n$ is negative by the assumption that  $W_{nk} - W_{ik}> \frac{1}{d_{in}(i)}  (W_{in} + 1)$. On the other hand, the term $ d_{in}(i)a\delta_i $ is positive. We show that for the appropriate $\delta_i $,  $x\in \Omega_{i-1}(\delta_i) $ implies that 
$ (1 + W_{in} +d_{in}(i)a)x_n + d_{in}(i)a\delta_i  \leq 0$. 

By Corollary \ref{cor:total_pop}, $x\in \cA$ implies that $\sum_{\ell = 1}^n x_\ell \geq \frac{1}{-\min(W_{ij})}$. Let $c := \frac{1}{-\min(W_{ij})}$. Then $\sum_{\ell = 1}^n x_\ell \geq c$. Thus, there is some neuron $\ell$ such that $x_\ell \geq c/n$. Now, we consider two cases. First, if $\ell \leq i$, then we have $x_\ell \leq x_n + \delta_i$, so $x_n \geq c/n - \delta_i$. 
Then \begin{align*}
  h_i^*(x) - h_n^*(x) \leq (1 + W_{in} +d_{in}(i)a)\left(c/n - \delta_i \right)+ d_{in}(i)a\delta_i  + C
\end{align*}
  All terms of this are negative, 
thus, for any value of $\delta_i$, we have $h_i^*(x) - h_n^*(x) \leq 0$. 

Now, we consider the second case, when $\ell \geq i$. In this case, we expand the term $C$:

\begin{align*}
h_i^*(x) - h_n^*(x) \leq (1 + W_{in} +d_{in}(i)a)x_n + d_{in}(i)a\delta_i +  -(1 + W_{ni})x_i + \sum_{j\not\to i}(W_{ij}-W_{nj})x_j.
\end{align*}

Now, let $m = \max_{\ell \geq i}(W_{i\ell}-W_{n\ell})$. We are guaranteed $m \leq 0$, since we chose a topological ordering to label our neurons. Thus, we have

\begin{align*}
h_i^*(x) - h_n^*(x) \leq  d_{in}(i)a\delta_i + m x_\ell \\
h_i^*(x) - h_n^*(x) \leq  d_{in}(i)a\delta_i + m c/n .\\
\end{align*}

Thus, for $\delta_i \leq \frac{-mc/n}{d_{in}(i) a}$, we have $h_i^*(x) - h_n^*(x) $ as desired.

\end{proof}

\begin{lem}
\label{lem:lyapunov_2}
Under the assumptions of Theorem \ref{prop:near_ctln},
there exists a constant $\delta_i := \delta_i(W)>0$ such that the function $L(x) = x_{i} - x_n$ is a local Lyapunov function for $A = \Omega_{i}, B = \Omega_{i-1}(\delta_i)$. 
\end{lem}

\begin{proof}

First, we notice that  for any value of $\delta_i \geq 0$, $L > 0$ on $\Omega_{i-1}(\delta_i)\setminus  \Omega_{i}(0)$ and $L \leq 0$ on $\Omega_{i}(0)$. Now, we show  that for the value of $\delta_i = \delta_i(W)$ used in Lemma \ref{lem:i_vs_n},  $\dot L < 0$ on $\Omega_{i-1}(\delta_i )\setminus  \Omega_{i}(0)$.  
We have 
\begin{align*}
\dot L (x) = \xdot_{i} -\xdot_n = -x_{i} + x_n + [y_{i}]_+ - [y_n]_+
\end{align*}
Pick $x\in \Omega_{i-1}(\delta_i) \setminus  \Omega_{i}(0)$ 
Since $x\notin  \Omega_{i}(0)$, we have $x_{i} > x_n$. Thus, by Lemma \ref{lem:thresh_ineq}, we have 
\begin{align*}
\xdot_{i} -\xdot_n \leq h^*_i(x) - h_n^*(x) 
\end{align*}
Now, by Lemma \ref{lem:i_vs_n}, $h^*_i(x) - h^*(n) < 0$. Thus $\dot L < 0$ on $\Omega_{i-1}(\delta_i)\setminus  \Omega_{i}(0)$.

Finally, we check that $ \Omega_{i-1}(0)$ is a forward invariant set for our TLN. For $i = 1$, this follows from Theorem \ref{thm:mixed_sign}. For $i > 1$, this follows from the fact that $L = x_{i-1} - x_{n}$ is a local Lyapunov function for $A = \Omega_{i-1}$, $B  = \Omega_{i-2}(\delta_{i-1})$.

\end{proof}

Finally, we are ready to prove Theorem \ref{prop:near_ctln}. 

\begin{proof}[Proof of Theorem \ref{prop:near_ctln}]

We consider two cases, trajectories which do not enter $\cA$ and those which do. First, if a trajectory does not enter $\cA$, then signs of all derivatives are fixed, thus the trajectory approaches a fixed point. 

Next, we consider trajectories which enter $\cA$. We prove by induction along our topological order that for each $\varepsilon> 0$, $i = 1, \ldots, n-1$, there exists $T$ such that if $t \geq T$, $x\in  \Omega_i(\varepsilon)$. First, we prove the base case. Since the vertex $1$ is a source and $n$ is a target, $n$ strongly dominates $1$, so $h^*_1(x ) < h^*_n(x)$ for all $x$ in the positive orthant. Further, for $x\notin \Omega_1$, $x_1 > x_n$, so 
\begin{align*}
\xdot_1 - \xdot_n & = -x_1 + x_n + [y_1]_+ - [y_n]_+ \leq    -x_1 + x_n + y_1 - y_n\\
\xdot_1 - \xdot_n  &\leq  h^*_1(x ) - h^*_n(x) \leq 0.
\end{align*}
Thus $L= x_1 - x_n$ is a local Lyapunov like function for $\R^n_{\geq 0}, \Omega_1$. Thus, all trajectories approach $\Omega_1$. Therefore, there exists  $t \geq T$, $x(t) \in \Omega_i(\varepsilon)$. 

Now, assume as an inductive hypothesis that the statement has been proved for all $j \leq i$. Pick $\varepsilon > 0$. By Lemma \ref{lem:lyapunov_2}, there is a constant $\delta_i> 0$ so that  $L = x_{i+1} -x_n$ is a local Lyapunov function for $A = \Omega_i$, $B = \Omega_{i-1}(\delta) $. By the inductive hypothesis, we can choose $T_i$ so that for all $t\geq T_i$, $x\in \Omega_{i-1}(\delta_i)$. Therefore, by Lemma \ref{lem:lyapanov}, for $t> T$ the trajectory approaches $\Omega_i$. Thus, there is some $T$ such that for all $t\geq T$, $x(t) \in  \Omega_{i}(\varepsilon)$.

Thus, for all $\varepsilon > 0$, all trajectories enter the set $\Omega_{n-1}(\varepsilon)$. Now, we pick $\delta = \max_{i \in[n]} \delta_i$. Then by Lemma \ref{lem:i_vs_n}, $h^*i(x) - h_n^*(x) < 0$ for all $x\in \Omega_{n-1}(\delta)$. This shows that the $x_n$ nullcline is separated from the origin by each $x_i$ nullcline within $\Omega_{n-1}(\delta)$.  Thus $\Omega_{n-1}(\delta) \subseteq H_n^+$. By our inductive argument, all trajectories which enter $\cA$ eventually enter the region $\Omega_{n-1}(\delta)$. Once they enter this region, $\xdot_n > 0$ for all time. Thus, the value of $x_n$ becomes monotonically increasing. Thus, it must approach a limit on $H_n$. Since $H_n$ is contained on the negative side of all other $H_i$, the signs of all derivatives $\xdot_i$ are fixed. Thus, the trajectory must approach the unique fixed point of the network.

\end{proof}

\section{Conclusion and open questions}
In this chapter, we introduced local Lyapunov functions, and applied them to prove Theorems \ref{thm:dag} and \ref{thm:dag_onto_symmetric}:  CTLNs whose graphs are built out a DAG and a symmetric graph arranged in a particular way have no dynamic attractors. 
We hope that similar techniques may be able to prove stronger results, such as the following theorem about graphs with no strongly directed cycles. Recall that we have defined a strongly directed cycle as a cycle in a directed graph which can be followed in one direction, but not the other. 

\setcounter{conj}{0}
\begin{conj}Let $G$ be a graph with no proper directed cycle. Then no CTLN with graph $G$ has a dynamic attractor. \label{conj:graph_struct}
\end{conj}

In order to work towards a proof of Conjecture \ref{conj:graph_struct}, we make some observations about the structure of graphs with no proper directed cycles.

\begin{prop}
If $G$ has no proper directed cycle, then every strongly connected component of $G$ is symmetric. 
\end{prop} 

\begin{proof}
We prove the constrapositive, that if some strongly connected component of $G$ has a directed edge, then $G$ has a proper directed cycle. Let $G$ be a graph, and suppose $u$ and $v$ are in the same strongly connected component, and $u\to v$, $v\not \to u$. Because $u$ and $v$ are in the same strongly connected component, there is some directed path from $v$ to $u$. Adding the edge $u\to v$ to this path creates a proper directed cycle. 
\end{proof}

This means we can consider a graph $G$ with no proper directed cycles as being the union of a set of symmetric graphs $G_1, G_2, \ldots, G_m$, attached to one another in an overall directed acyclic graph as pictured in Figure \ref{fig:condensation}. This directed acyclic graph structure suggests we may be able to use inductive teqchniques, similar to those used in the proof of Theorem \ref{thm:dag_onto_symmetric} to prove that these CTLNs do not have dynamic attractors. 

\begin{figure}
\begin{center}
\includegraphics[width = 2 in]{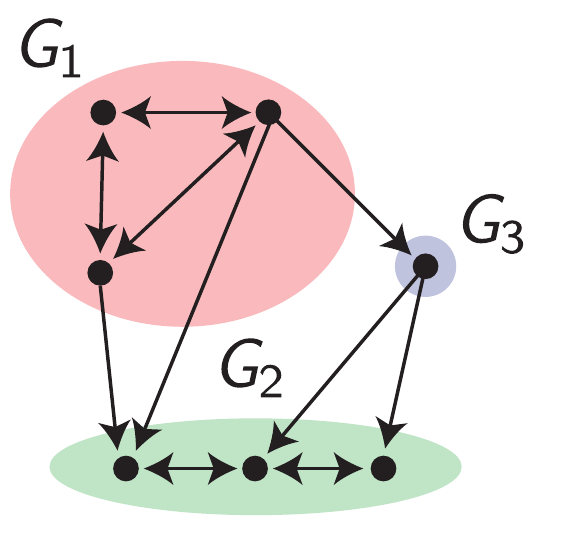}
\end{center}
\caption[Structure of graphs with no proper directed cycles.]{We can decompose any graph with no strongly directed cycle as a union of symmetric graphs $G_1, \ldots, G_m$, connected into a DAG.}
\label{fig:condensation}
\end{figure}

Our proofs of Theorems \ref{thm:dag} and \ref{thm:dag_onto_symmetric} depended on showing that ``sources die" in a particular case. We conjecture that this result holds in general. 
 
\begin{conj}
\label{thm:sources_die}
Let $j$ be a source in a graph $G$. Then $\lim_{t\to \infty} x_j(t) = 0$. 
\end{conj}

Further, we believe that more progress can be made on general competitive TLNs. In particular, while we were only able to prove dynamic robustness in the special case of shallow DAG1, we conjecture that similar result holds for all robust motifs.

\begin{conj}\label{conj:dyn_robust}
Let $(W, \theta)$ define a TLN such that $G_W$ is a member of DAG1 or DAG2. Then the unique fixed point is always a global attractor.
\end{conj}

Finally, while this chapter has focused on ruling out dynamic attractors based on graph structure, many questions remain about the relationship between structure and dynamics in CTLNs. 

\begin{question*}
Which features of a graph $G$ guarantee that the CTLN of $G$ has a limit cycle? Which features guarantee that $G$ has a chaotic attractor? 
\end{question*}

%% file: SupplementaryMaterial/Vita.tex
% Place your vita below.

Caitlin Lienkaemper graduated from Harvey Mudd College in 2017, where she earned a B.S. in Mathematics. She entered the Ph.D. program at the Pennsylvania State University to work with Carina Curto in 2017. 

\begin{center}Publications/Submitted Papers\end{center}

\begin{enumerate}

\item
``Obstructions to convexity in neural codes,'' with Anne Shiu and Zev Woodstock,
in Advances in Applied Mathematics Volume 85, Pages 31-59 (April 2017)
\item
``The geometry of partial fitness orders and an efficient method for detecting genetic interactions,''
with Lisa Lamberti, Dawn Drain, Niko Beerenwinkel, and Alex Gavryushkin, in 
Journal of  Mathematical Biology (May 2018) 
%\item ``Inductively pierced codes and neural toric ideals," preprint: arXiv:1811.04712 (November 2018).
\item ``Oriented matroids and combinatorial neural codes," with Alex Kunin and Zvi Rosen, preprint: arXiv:2002.03542 (February 2020). In revision for Journal of Combinatorial Theory. 
\item ``Order-forcing in Neural Codes," with R. Amzi Jeffs and Nora Youngs, 
preprint: arxiv:2011.03572 (November 2020). Submitted. 
\end{enumerate}